\documentclass[a4paper]{article}
\usepackage[svgnames]{xcolor}
\usepackage[a4paper,margin=0.7in]{geometry}
\usepackage[utf8]{inputenc}
\usepackage{url}
\usepackage{hyperref}

\usepackage{enumitem}

\usepackage{natbib}
\usepackage{rotating}

\usepackage{amsmath,amsfonts,amssymb,amsthm}
\usepackage{mathrsfs}
\usepackage{thmtools}
\usepackage{mathtools}
\usepackage{array}
\usepackage{cleveref}
\usepackage{stmaryrd}
\usepackage{mathpartir}
\usepackage[inference]{semantic}
\usepackage{enumitem}

\usepackage{ifthen}
\usepackage{ifpdf}

\usepackage{listings}
\usepackage{tikz}

\newcommand{\stktokens}{\textsc{StkTokens}}

\newcommand\lau[1]{{\color{purple} \sf \footnotesize {LS: #1}}\\}
\newcommand\dominique[1]{{\color{purple} \sf \footnotesize {DD: #1}}\\}
\newcommand\lars[1]{{\color{purple} \sf \footnotesize {LB: #1}}\\}

\renewcommand\lau[1]{}
\renewcommand\dominique[1]{}
\renewcommand\lars[1]{}

\declaretheorem[numbered=yes,name=Lemma,qed=$\blacksquare$]{lemma}
\declaretheorem[numbered=yes,name=Theorem,qed=$\blacksquare$]{theorem}
\declaretheorem[numbered=yes,name=Definition,qed=$\blacksquare$]{definition}

\newcommand{\defeq}{\stackrel{\textit{\tiny{def}}}{=}}
\newcommand{\defbnf}{::=}
\newcommand{\sem}[1]{\left\llbracket #1 \right\rrbracket}

\newcommand{\dom}{\mathrm{dom}}
\newcommand{\powerset}[1]{\mathcal{P}(#1)}

\newcommand{\npair}[2][n]{\left(#1,#2\right)}

\newcommand{\nsubeq}[1][n]{\overset{#1}{\subseteq}}
\newcommand{\nsupeq}[1][n]{\overset{#1}{\supseteq}}
\newcommand{\nequal}[1][n]{\overset{#1}{=}}

\newcommand\subsetsim{\mathrel{\ooalign{\raise.2ex\hbox{$\subset$}\cr
      \hidewidth\lower.8ex\hbox{\scalebox{0.9}{$\sim$}}\hidewidth\cr}}}
\newcommand\supsetsim{\mathrel{\ooalign{\raise.2ex\hbox{$\supset$}\cr
      \hidewidth\lower.8ex\hbox{\scalebox{0.9}{$\sim$}}\hidewidth\cr}}}

\newcommand{\fun}{\rightarrow}
\newcommand{\parfun}{\rightharpoonup}
\newcommand{\monnefun}{\xrightarrow{\textit{\tiny{mon, ne}}}}

\newcommand{\blater}{\mathop{\blacktriangleright}}

\newcommand{\tand}{\text{ and }}
\newcommand{\tor}{\text{ or }}
\newcommand{\totherwise}{\text{otherwise }}

\newcommand{\untrusted}{\mathrm{untrusted}}
\newcommand{\trusted}{\mathrm{trusted}}
\newcommand{\trust}{\var{tst}}

\newcommand{\sconeq}{\mathrel{\src{\approx_{\mathrm{ctx}}}}}
\newcommand{\tconeq}{\mathrel{\trg{\approx_{\mathrm{ctx}}}}}

\newcommand{\typesetlr}[1]{\mathcal{#1}}
\newcommand{\gc}{\var{gc}}
\newcommand{\lre}[1][\square,\gc]{\typesetlr{E}^{#1}}

\newcommand{\lrexj}[1][\square,\gc]{\typesetlr{E}^{#1}_{\mathrm{xjmp}}}

\newcommand{\lrrg}[2][\square,\gc]{\typesetlr{R}^{#1}_{#2}}
\newcommand{\lrr}[1][\square,\gc]{\lrrg[#1]{\untrusted}}
\newcommand{\lrrtrusted}[1][\square,\gc]{\lrrg[#1]{\trusted}}
\newcommand{\lro}[1][\square,\gc]{\typesetlr{O}^{#1}}
\newcommand{\lrol}{\lro[\preceq,(\ta,\stkb,\gsigrets,\gsigcloss)]}
\newcommand{\lror}{\lro[\succeq,(\ta,\stkb,\gsigrets,\gsigcloss)]}
\newcommand{\lrvg}[2][\square,\gc]{\typesetlr{V}^{#1}_{#2}}
\newcommand{\lrv}[1][\square,\gc]{\lrvg[#1]{\untrusted}}
\newcommand{\lrvtrusted}[1][\square,\gc]{\lrvg[#1]{\trusted}}

\newcommand{\lrrs}{\typesetlr{R}}
\newcommand{\lrm}{\typesetlr{M}}
\newcommand{\lrcomp}[1][\square,\gc]{\typesetlr{C}^{#1}}
\newcommand{\lrec}[1][\square,\gc]{\typesetlr{EC}^{#1}}
\newcommand{\lrheap}[1][]{\typesetlr{H}^{#1}}

\newcommand{\stpair}[3][]{
\ifthenelse{\equal{#1}{}}
{\left(\src{#2_S},#3_T\right)}
{\left(\src{#2},#3\right)}}

\newcommand{\lrfree}[1][\gc]{\typesetlr{F}^{#1}}
\newcommand{\lrstk}[1][\gc]{\typesetlr{S}^{#1}}

\newcommand{\memSatGeneric}[4]{#2 :_{#1}^{#4}#3}
\newcommand{\memSat}[3][n]{\memSatGeneric{#1}{#2}{#3}{\gc}}
\newcommand{\memSatStack}[3][n]{\npair[#1]{(#2)} \in \lrstk(#3)}
\newcommand{\memSatFStack}[3][n]{\npair[#1]{(#2)} \in \lrfree(#3)}

\newcommand{\Wor}{\mathrm{Wor}}
\newcommand{\World}{\mathrm{World}}
\newcommand{\Worlds}{\mathrm{World}_\text{private stack}}
\newcommand{\Worldh}{\mathrm{World}_\mathrm{heap}}
\newcommand{\Worldfs}{\mathrm{World}_\text{free stack}}

\newcommand{\RegionName}{\mathrm{RegionName}}

\newcommand{\Regions}{\mathrm{Region}_\mathrm{spatial}}
\newcommand{\Regionh}{\mathrm{Region}_\mathrm{shared}}

\newcommand{\spatial}{\mathrm{spatial}}
\newcommand{\spatialo}{\mathrm{spatial\_owned}}
\newcommand{\pure}{\mathrm{pure}}
\newcommand{\revoked}{\mathrm{revoked}}

\newcommand{\UPred}[1]{\mathrm{UPred}(#1)}
\newcommand{\URel}[1]{\mathrm{URel}(#1)}

\newcommand{\future}{\sqsupseteq}
\newcommand{\pub}{\mathrm{pub}}

\newcommand{\stdreg}[3][\square]{\iota^{\mathrm{std},#3,#1}_{#2}}
\newcommand{\stareg}[2][\stpair{\ms}{\ms}]{\iota^{\mathrm{sta},#2}_{#1}}
\newcommand{\codereg}[2][\mathrm{code},\square]{\iota^{#1}_{#2}}

\newcommand{\spa}{\mathrm{s}}
\newcommand{\spao}{\mathrm{so}}
\newcommand{\pur}{\mathrm{p}}

\newcommand{\sourcecolor}{\color{blue}}
\newcommand{\src}[1]{{\sourcecolor #1}}

\newcommand{\commoncolor}[1]{\color{black}}
\newcommand{\com}[1]{{\commoncolor{} #1}}

\newcommand{\targetcolor}[1]{\color{red}}
\newcommand{\trg}[1]{{\targetcolor{} #1}}

\newcommand{\zinstr}[1]{\texttt{#1}}
\newcommand{\oneinstr}[2]{
  \ifthenelse{\equal{#2}{}}
  {\zinstr{#1}}
  {\zinstr{#1} \; #2}
}
\newcommand{\twoinstr}[3]{
  \ifthenelse{\equal{#2#3}{}}
  {\zinstr{#1}}
  {\zinstr{#1} \; #2 \; #3}
}
\newcommand{\threeinstr}[4]{
  \ifthenelse{\equal{#2#3#4}{}}
  {\zinstr{#1}}
  {\zinstr{#1} \; #2 \; #3 \; #4}
}

\newcommand{\fourinstr}[5]{
  \ifthenelse{\equal{#2#3#4#5}{}}
  {\zinstr{#1}}
  {\zinstr{#1} \; #2 \; #3 \; #4 \; #5}
}

\newcommand{\scall}[4][]{  
\ifthenelse{\equal{#3#4}{}}
  {\ensuremath{\zinstr{\src{call}}_{#1}^{#2}}}
  {\ensuremath{\zinstr{\src{call}}_{#1}^{#2} \; #3 \; #4}}
}
\newcommand{\offpc}{\var{off}_{\mathrm{pc}}}
\newcommand{\offsigma}{\var{off}_\sigma}

\newcommand{\tfail}{\zinstr{\com{fail}}}
\newcommand{\thalt}{\zinstr{\com{halt}}}

\newcommand{\tjmp}[1]{\oneinstr{\com{jmp}}{#1}}
\newcommand{\tsetatob}[1]{\oneinstr{\com{seta2b}}{#1}}

\newcommand{\tjnz}[2]{\twoinstr{\com{jnz}}{#1}{#2}}
\newcommand{\tisptr}[2]{\twoinstr{\com{gettype}}{#1}{#2}}
\newcommand{\tgeta}[2]{\twoinstr{\com{geta}}{#1}{#2}}
\newcommand{\tgetb}[2]{\twoinstr{\com{getb}}{#1}{#2}}
\newcommand{\tgete}[2]{\twoinstr{\com{gete}}{#1}{#2}}
\newcommand{\tgetp}[2]{\twoinstr{\com{getp}}{#1}{#2}}
\newcommand{\tgetlin}[2]{\twoinstr{\com{getl}}{#1}{#2}}
\newcommand{\tmove}[2]{\twoinstr{\com{move}}{#1}{#2}}
\newcommand{\tstore}[2]{\twoinstr{\com{store}}{#1}{#2}}
\newcommand{\tload}[2]{\twoinstr{\com{load}}{#1}{#2}}
\newcommand{\tcca}[2]{\twoinstr{\com{cca}}{#1}{#2}}
\newcommand{\txjmp}[2]{\twoinstr{\com{xjmp}}{#1}{#2}}
\newcommand{\trestrict}[2]{\twoinstr{\com{restrict}}{#1}{#2}}
\newcommand{\tcseal}[2]{\twoinstr{\com{cseal}}{#1}{#2}}

\newcommand{\tsplice}[3]{\threeinstr{\com{splice}}{#1}{#2}{#3}}
\newcommand{\tlt}[3]{\threeinstr{\com{lt}}{#1}{#2}{#3}}
\newcommand{\tplus}[3]{\threeinstr{\com{plus}}{#1}{#2}{#3}}
\newcommand{\tminus}[3]{\threeinstr{\com{minus}}{#1}{#2}{#3}}

\newcommand{\tsplit}[4]{\fourinstr{\com{split}}{#1}{#2}{#3}{#4}}

\newcommand{\nats}{\mathbb{N}}
\newcommand{\ints}{\mathbb{Z}}

\newcommand{\ta}[1][]{T_{A#1}}

\newcommand{\update}[2]{[#1 \mapsto #2]}
\newcommand{\updReg}[2]{\update{\reg.#1}{#2}}

\newcommand{\shareddom}[1]{\mathrm{#1}}
\newcommand{\RegName}{\shareddom{RegisterName}}
\newcommand{\Addr}{\shareddom{Addr}}
\newcommand{\Seal}{\shareddom{Seal}}
\newcommand{\Symbol}{\shareddom{Symbol}}
\newcommand{\Perm}{\shareddom{Perm}}
\newcommand{\Caps}{\shareddom{Cap}}
\newcommand{\SealableCaps}{\shareddom{SealableCap}}
\newcommand{\Word}{\shareddom{Word}}
\newcommand{\Instr}{\shareddom{Instr}}
\newcommand{\Mem}{\shareddom{Memory}}
\newcommand{\Reg}{\shareddom{RegisterFile}}

\newcommand{\Conf}{\shareddom{Conf}}
\newcommand{\ExecConf}{\shareddom{ExecConf}}
\newcommand{\Linear}{\shareddom{Linear}}
\newcommand{\MemSeg}{\shareddom{MemorySegment}}
\newcommand{\StkFrame}{\shareddom{StackFrame}}
\newcommand{\Stack}{\shareddom{Stack}}

\newcommand{\scbnf}{\shareddom{sc}}
\newcommand{\cbnf}{\shareddom{c}}
\newcommand{\permbnf}{\shareddom{perm}}
\newcommand{\addrbnf}{\shareddom{a}}
\newcommand{\basebnf}{\shareddom{base}}
\newcommand{\aendbnf}{\shareddom{end}}
\newcommand{\rbnf}{\shareddom{r}}
\newcommand{\linbnf}{\shareddom{l}}
\newcommand{\sealbasebnf}{\sigma_\shareddom{base}}
\newcommand{\sealendbnf}{\sigma_\shareddom{end}}

\newcommand{\sstk}{\shareddom{stk}}

\newcommand{\sstkframe}{\shareddom{frame}}

\newcommand{\perm}{\var{perm}}
\newcommand{\lin}{\var{l}}

\newcommand{\addr}{\shareddom{a}}

\newcommand{\stkptr}[1]{\mathrm{stack\text{-}ptr}(#1)}

\newcommand{\retptrd}{\mathrm{ret\text{-}ptr\text{-}data}}
\newcommand{\retptrc}{\mathrm{ret\text{-}ptr\text{-}code}}

\newcommand{\seal}[1]{\shareddom{seal}(#1)}
\newcommand{\sealed}[1]{\shareddom{sealed}(#1)}

\newcommand{\failed}{\mathrm{failed}}

\newcommand{\halted}{\mathrm{halted}}

\newcommand{\targetdom}[1]{\mathrm{#1}}
\newcommand{\tRegName}{\targetdom{RegisterName}}

\newcommand{\context}{\mathscr{C}}

\newcommand{\plug}[2]{#1[#2]}

\newcommand{\step}[1][]{\rightarrow^{#1}}
\newcommand{\nstep}[2][n]{\step[#2]_{#1}}

\newcommand{\term}[1][-]{{\Downarrow_{#1}}}
\newcommand{\sterm}[2][-]{{\Downarrow_{#1}^{#2}}}

\newcommand{\var}[1]{\mathit{#1}}
\newcommand{\free}{\var{free}}
\newcommand{\rn}{\var{rn}}
\newcommand{\reg}[1][]{
\ifthenelse{\equal{#1}{}}{
\var{reg}
}{
\var{reg}^{(#1)}}}
\newcommand{\mem}{\var{mem}}
\newcommand{\ms}{\var{ms}}
\newcommand{\pc}{\var{pc}}
\newcommand{\stk}{\var{stk}}

\newcommand{\stkf}{\stk_{\var{frame}}}

\newcommand{\code}{\var{code}}
\newcommand{\priv}{\var{priv}}
\newcommand{\opc}{\var{opc}}

\newcommand{\vsc}{\var{sc}}
\newcommand{\cb}{\vsc}
\newcommand{\baddr}{\var{b}}
\newcommand{\eaddr}{\var{e}}
\newcommand{\aaddr}{\var{a}}

\newcommand{\sigret}[1][]{
  \ifthenelse{\equal{#1}{}}{
  \sigma_{\mathrm{ret}}
  }{
  \sigma_{\mathrm{ret},#1}
  }
}
\newcommand{\sigrets}[1][]{
  \ifthenelse{\equal{#1}{}}{
  \overline{\sigma_{\mathrm{ret}}}
  }{
  \overline{\sigma_{\mathrm{ret},#1}}
  }
}
\newcommand{\sigclos}[1][]{
  \ifthenelse{\equal{#1}{}}{
  \sigma_{\mathrm{clos}}
  }{
  \sigma_{\mathrm{clos},#1}
  }
}
\newcommand{\gsigcloss}{
  \overline{\sigma_{\mathrm{glob\_clos}}}}

\newcommand{\gsigrets}{
  \overline{\sigma_{\mathrm{glob\_ret}}}}

\newcommand{\sigcloss}[1][]{
  \ifthenelse{\equal{#1}{}}{
  \overline{\sigma_{\mathrm{clos}}}
  }{
  \overline{\sigma_{\mathrm{clos},#1}}
  }
}
\newcommand{\mscode}[1][]{
  \ifthenelse{\equal{#1}{}}{\ms_{\mathrm{code}}}
  {\ms_{\mathrm{code},#1}}
}
\newcommand{\mspad}[1][]{
  \ifthenelse{\equal{#1}{}}{\ms_{\mathrm{pad}}}
  {\ms_{\mathrm{pad},#1}}
}
\newcommand{\msdata}[1][]{
  \ifthenelse{\equal{#1}{}}{\ms_{\mathrm{data}}}
  {\ms_{\mathrm{data},#1}}
}

\newcommand{\constant}[1]{\mathrm{#1}}
\newcommand{\calllen}{\constant{call\_len}}
\newcommand{\stkb}{\constant{stk\_base}}
\newcommand{\retoffset}{\constant{ret\_pt\_offset}}
\newcommand{\pcreg}{\mathrm{pc}}
\newcommand{\rstk}{\mathrm{r}_\mathrm{stk}}

\newcommand{\rretc}{\mathrm{r}_\mathrm{ret code}}
\newcommand{\rretd}{\mathrm{r}_\mathrm{ret data}}
\newcommand{\rdata}{\mathrm{r}_\mathrm{data}}
\newcommand{\rtmp}[1]{\mathrm{r}_\mathrm{t#1}}

\newlist{enumproof}{enumerate}{10}
\setlist[enumproof]{label*=\arabic*.}

\newcommand{\plainlinearity}[1]{\mathrm{#1}}
\newcommand{\linear}{\plainlinearity{linear}}
\newcommand{\normal}{\plainlinearity{normal}}

\newcommand{\plainperm}[1]{\textsc{#1}}
\newcommand{\rwx}{\plainperm{rwx}}
\newcommand{\rx}{\plainperm{rx}}
\newcommand{\rw}{\plainperm{rw}}
\newcommand{\readonly}{\plainperm{r}}
\newcommand{\ro}{\readonly}
\newcommand{\noperm}{\plainperm{0}}

\newcommand{\comp}{\var{comp}}

\newcommand{\pwheap}[1][W]{#1.\mathrm{heap}}
\newcommand{\pwfree}[1][W]{#1.\mathrm{free}}
\newcommand{\pwpriv}[1][W]{#1.\mathrm{priv}}

\newcommand{\pregion}[1]{#1.\mathrm{region}}
\newcommand{\prv}[1]{#1.\mathrm{v}}

\newcommand{\erase}[2]{\lfloor #1 \rfloor_{\{#2\}}}

\newcommand{\plainfun}[2]{
  \ifthenelse{\equal{#2}{}}
  {\mathit{#1}}
  {\mathit{#1}(#2)}
}

\newcommand{\activeReg}[1]{\plainfun{active}{#1}}
\newcommand{\addressable}[1]{\plainfun{addressable}{#1}}
\newcommand{\callCond}[1]{\plainfun{callCondition}{#1}}
\newcommand{\decInstr}[1]{\plainfun{decodeInstruction}{#1}}
\newcommand{\decPerm}[1]{\plainfun{decodePerm}{#1}}
\newcommand{\encInstr}[1]{\plainfun{encodeInstruction}{#1}}
\newcommand{\encPerm}[1]{\plainfun{encocePerm}{#1}}
\newcommand{\encLin}[1]{\plainfun{encoceLin}{#1}}
\newcommand{\encType}[1]{\plainfun{encodeType}{#1}}
\newcommand{\exec}[1]{\plainfun{executable}{#1}}
\newcommand{\execCond}[2][\square,\gc]{\plainfun{executeCondition^{#1}}{#2}}
\newcommand{\isLinear}[1]{\plainfun{isLinear}{#1}}
\newcommand{\linCons}[1]{\plainfun{linearityConstraint}{#1}}
\newcommand{\linConsPerm}[2]{\plainfun{linearityConstraintPerm}{#1,#2}}

\newcommand{\nonExec}[1]{\plainfun{nonExecutable}{#1}}
\newcommand{\nonLinear}[1]{\plainfun{nonLinear}{#1}}
\newcommand{\nonZero}[1]{\plainfun{nonZero}{#1}}
\newcommand{\range}[1]{\plainfun{range}{#1}}
\newcommand{\readAllowed}[1]{\plainfun{readAllowed}{#1}}
\newcommand{\readCond}[2][\square,\gc]{\plainfun{readCondition^{#1}}{#2}}
\newcommand{\stackReadCond}[2][\square,\gc]{\plainfun{stackReadCondition^{#1}}{#2}}
\newcommand{\xReadCond}[2][\square,\gc]{\plainfun{readXCondition^{#1}}{#2}}
\newcommand{\writeCond}[2][\square,\gc]{\plainfun{writeCondition^{#1}}{#2}}
\newcommand{\stackWriteCond}[2][\square,\gc]{\plainfun{stackWriteCondition^{#1}}{#2}}

\newcommand{\updPcAddr}[1]{\plainfun{updatePc}{#1}}
\newcommand{\withinBounds}[1]{\plainfun{withinBounds}{#1}}
\newcommand{\writeAllowed}[1]{\plainfun{writeAllowed}{#1}}
\newcommand{\xjumpResult}[3]{\plainfun{xjumpResult}{#1,#2,#3}}
\newcommand{\purePart}[1]{\plainfun{purePart}{#1}}

\title{\stktokens{}: Enforcing Well-bracketed Control Flow and Stack Encapsulation using Linear Capabilities\\
  Technical Report with Proofs and Details}
\author{Lau Skorstengaard \and Dominique Devriese \and Lars Birkedal}
\begin{document}
\maketitle

This document is a technical report accompanying a paper by the same title and authors, published at POPL 2019.
It contains proofs and details that were omitted from the paper for space and presentation reasons.

\tableofcontents
Disclaimer: While the proofs in this technical report are done, the text can be lacklustre and from time to time out of date.
We will make the technical report up to date as soon as possible.
Further, the use of coloring in the later parts of this document has not always been done with as much care as it ought to. This means that not all source specific things are colored blue, but it should be evident from the context that they belong to the source language.

\section{The two capability machines}
\subsection{Domains}
\label{sec:domains}

\[
  \begin{array}{rrcl}
   \addrbnf,\basebnf \in & \Addr & \defeq & \nats \\
    \sealbasebnf, \sigma \in & \Seal & \defeq & \nats \\
    w \in &\Word & \defeq & \ints \uplus \Caps\\
    \permbnf \in& \Perm & \defbnf & \dots \\
    &\linbnf & \defbnf & \linear \mid \normal \\
    &\aendbnf & \in & \Addr \uplus \{\infty \} \\
    &\sealendbnf & \in & \Seal \uplus \{\infty \} \\
    \scbnf \in &\SealableCaps&\defbnf & ((\permbnf,\linbnf),\basebnf,\aendbnf,\addrbnf) \mid \seal{\sealbasebnf,\sealendbnf,\sigma}\\
    & & & {\sourcecolor{} \mid \stkptr{\permbnf,\basebnf,\aendbnf,\addrbnf}}\\ 
    & & & {\sourcecolor{} \; \mid \retptrd(\basebnf,\aendbnf) \mid \retptrc(\basebnf,\aendbnf,\addrbnf)}\\
    \cbnf \in&\Caps& \defbnf &  \scbnf \mid \sealed{\sigma,\scbnf}\\ 
    \rbnf \in & \RegName & \defbnf &\pcreg \mid \rretd \mid \rretc \mid \rstk \mid \rdata \mid \rtmp{1} \mid \rtmp{2} \mid \dots \\
    &\Reg & \defeq & \RegName \fun \Word\\
    &\Mem & \defeq & \Addr \fun \Word \\
    &\MemSeg & \defeq & \Addr \parfun \Word \\
    {\sourcecolor \sstkframe \in} & {\sourcecolor \StkFrame} & {\sourcecolor \defeq} & {\sourcecolor \Addr \times \MemSeg}\\
    \src{\sstk \in}& \src{ \Stack} & \src{ \defeq} & \src{ \StkFrame^*} \\
    \Phi \in & \ExecConf & \defeq & \Mem \times \Reg {\sourcecolor{} \; \times \; \Stack \times \MemSeg }\\
    &\Conf & \defeq & \ExecConf \uplus \{\failed\} \uplus \{\halted\}
  \end{array}
\]
The target language domains are all the non blue parts in the above. The source language domains are the black and blue parts in the above. Further
\begin{itemize}
\item $\linbnf$ defines domain $\Linear$
\item $\scbnf$ defines domain $\SealableCaps$
\item $\cbnf$ defines domain $\Caps$
\item $\rbnf$ defines the finite set $\RegName$. 
\item $\Perm$ is defined as the set of permissions in Figure~\ref{fig:perm-hier}.
\end{itemize}

In the source language, $\src{\Stack}$ is a call stack that contains the data for each call. The call stack consists of a number of $\src{\StkFrame}$'s that contains 1) the old pc and 2) caller's private stack.

 In both languages, the base address of the stack is known (the stack grows downwards in memory, so the base address marks the end of the stack). In the target language, this will be the base address of some linear capability. The base address can be any address on the machine, but it will have to remain the same during all of execution. We write this constant as
\[
  \stkb
\]

\subsubsection{Useful definitions}
\begin{definition}
  For a capability $c=((\_,\_),\baddr,\eaddr,\_)$ we say it has range $[\baddr,\eaddr]$ and we define
  \[
    \range{c} = [\baddr,\eaddr]
  \]
  Similarly for seals and stack pointers:
  \[
    \range{\seal{\sigma_\baddr,\sigma_\eaddr,\_}} = [\sigma_\baddr,\sigma_\eaddr]
  \]
  and
  \[
    \src{\range{\stkptr{\_,\baddr,\eaddr,\_}} = [\baddr,\eaddr]}
  \]
\end{definition}

\subsection{Syntax}
\label{sec:syntax}
The target machine is a simple capability machine with memory capabilities and sealed capabilities\footnote{In previous work, we used enter capabilities but for this the complexity introduced by mixing writable and executable memory is difficult to handle.} (inspired by CHERI). The syntax of the instructions of the target machine is defined as follows:
\[
\begin{array}{rcl}
n &\in & \ints \\
\com{r} &\in &  \tRegName \\
\com{\rn} &\defbnf &  \com{r} \mid n \\
\com{i} &\defbnf & \tfail \mid \thalt \mid \tjmp{\com{r}} \mid \tjnz{\com{r}}{\com{\rn}} \mid \tisptr{\com{r}}{\com{r}} \mid \tgeta{\com{r}}{\com{r}} \mid \tgetb{\com{r}}{\com{r}} \mid \\
      & &  \tgete{\com{r}}{\com{r}}\mid \tgetp{\com{r}}{\com{r}} \mid \tgetlin{\com{r}}{\com{r}} \mid \tmove{\com{r}}{\com{\rn}} \mid \tstore{\com{r}}{\com{r}} \mid\\
      & &  \tload{\com{r}}{\com{r}} \mid \tcca{\com{r}}{\com{\rn}} \mid \trestrict{\com{r}}{\com{\rn}} \mid \tlt{\com{r}}{\com{\rn}}{\com{\rn}} \mid \\
  & & \tplus{\com{r}}{\com{\rn}}{\com{\rn}} \mid \tminus{\com{r}}{\com{\rn}}{\com{\rn}} \mid \tsetatob{\com{r}} \mid \txjmp{\com{r}}{\com{r}} \mid \tcseal{\com{r}}{\com{r}} \mid \\ 
      & &   \tsplit{\com{r}}{\com{r}}{\com{r}}{\com{\rn}} \mid\tsplice{\com{r}}{\com{r}}{\com{r}} 
\end{array}
\]
$i$ defines the set $\Instr$


The source machine is also a capability machine with memory capabilities and sealed capabilities. Unlike the target machine, the source machine has a built in stack along with special stack and return tokens used in place of the actual capabilities. The syntax of the source machine language is as follows:
\[
  \begin{array}{rcl}
    \offpc,\offsigma & \in & \nats \\
    \src{i} & \defbnf &  \com{i} \mid \scall{\offpc,\offsigma}{r}{r}
  \end{array}
\]
There is one syntactic difference between the source language and the target language, namely the target language has an extra instruction in the $\scall{\offpc,\offsigma}{}{}$ instruction (the $o$ refers to the seal it uses, namely the offset in the seals made available by linking). 

The source machine allows for variable length instructions which we utilise for $\scall{\offpc,\offsigma}{}{}$. In other words, when $\scall{\offpc,\offsigma}{}{}$ is decoded, a series of addresses in memory must contain what corresponds to the call instruction (and be within the range of the pc capability). This is described in more detail when we present the decoding function.

\subsection{Permissions}
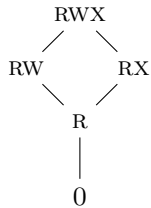
\begin{figure}[!h]
  \centering
  \begin{tikzpicture}[main node/.style={}]
    \node[main node] (rwx) {$\rwx$};
    \node[main node] (rx) [below right of=rwx] {$\rx$};

    \node[main node] (rw) [below left of=rwx,] {$\rw$};
    \node[main node] (r) [below right of=rw] {$\readonly$};
    \node[main node] (0) [below of=r] {$\noperm$};

    \path[every node/.style={font=\sffamily\small}]

    (rw) edge (r)
    (r) edge (0)

    (rwx) edge (rx)

    (rw) edge (rwx)
    (r) edge (rx);
  \end{tikzpicture}

  \caption{Permission hierarchy}
  \label{fig:perm-hier}
\end{figure}
\lau{13-09-2017: Note that most of the discussion below is now captured in Section~\ref{sec:linear-cap}}
\lau{01-09-2017: For now, I will add the linear capabilities such that a linear capability stays linear. It does, however, seem like it would be okay to have linear capabilities become non-linear (I at least don't see anything that would break down completely). One draw back would be that one could ``ruin'' the stack capability by making it non-linear.}
\dominique{7-9-2017: Err... you mean that we could allow non-linear capabilities to be made linear, right?  We definitely do not want to allow the stack or return capability to be made non-linear, as we want to prevent it from being aliased.}
\lau{07-09-2017: I actually did mean linear to non-linear (non-linear to linear does not make sense because the non-linear capability may have aliases). The operation that makes a linear capability non-linear would have to clear said capability. If a linear capability has been made non-linear, then someone who expects a linear capability can observe this difference and behave accordingly, i.e., fail the execution. It will, however, remain an invariant that there are no aliases for linear capabilities.}
\lau{07-09-2017: I had forgotten that return capabilities probably need to be linear. This would call for a ``load from offset'' operation to be able to load an in code stored seal (see the figure I drew about linking/seals).}

We assume functions $\decPerm{}$ and $\encPerm{}$.
\lau{15-09-2017: TODO, write more about these. Should be like in the local cap setting.}

\subsection{Operational Semantics}
The source machine is parameterized with a set of trusted addresses $\ta$. $\ta$ are the only addresses from which the $\scall{}{}{}$ will be interpreted. The source machine represents a virtual intermediate machine which we use to argue well-bracketedness and local state encapsulation. It is not meant to be the machine that the actual code is executed on. Further, it is the well-bracketedness and local state encapsulation of the compiled code, not the context we are concerned with. On the other hand, the target machine is not parameterized with $\ta$ as all the instructions on the target machine is available to the adversary.

\subsubsection{Notes}
Generally:
\begin{itemize}
\item Linear capabilities are cleared when they move around in memory.
\end{itemize}

Source language:
\begin{itemize}
\item Variable length instructions that match the length of the compiled instructions
  \begin{itemize}
  \item This is needed for correctness. 
  \item It is only used for the call instruction.
  \end{itemize}
\end{itemize}

\dominique{27-09-2017: I think we need to do something about tail calls in the source language?}
\lau{27-09-2017: In what sense? Do you want to allow them (I do not think they are allowed at the moment?)}
\dominique{8-11-2017: Yes, I would like to support them.  I think it would be a strength of our calling convention that we can support them (as I expect we can) and it would be good to demonstrate this.}

Target language:
\begin{itemize}
\item 
\end{itemize}

\subsubsection{Helpful functions, sets, and conventions}

\[
  \pi_\baddr(\vsc) =
  \begin{cases}
    \baddr & \text{if $\vsc = ((\_,\_),\baddr,\_,\_)$} \\ 
    \pi_\baddr(c) & \text{if $\vsc = \sealed{\_,c}$} \\ 
    \src{\baddr} & \src{\text{if $\vsc = \stkptr{\_,\baddr,\_,\_}$}} \\ 
    \src{\baddr} & \src{\text{if $\vsc = \retptrd(\baddr,\_)$}} \\ 
    \src{\baddr} & \src{\text{if $\vsc = \retptrc(\baddr,\_,\_)$}}
  \end{cases}
\]
\[
  \pi_\eaddr(\vsc) =
  \begin{cases}
    \eaddr & \text{if $\vsc = ((\_,\_),\_,\eaddr,\_)$} \\ 
    \pi_\eaddr(c) & \text{if $\vsc = \sealed{\_,c}$} \\ 
    \src{\eaddr} & \src{\text{if $\vsc = \stkptr{\_,\_,\eaddr,\_}$}} \\ 
    \src{\eaddr} & \src{\text{if $\vsc = \retptrd(\_,\eaddr)$}} \\ 
    \src{\eaddr} & \src{\text{if $\vsc = \retptrc(\_,\eaddr,\_)$}}
  \end{cases}
\]

\[
  \pi_\lin(\vsc) = 
  \begin{cases}
    \lin & \text{if $\vsc = ((\_,\_),\_,\eaddr,\_)$} \\ 
    \pi_\lin(c) & \text{if $\vsc = \sealed{\_,c}$} \\ 
    \src{\linear} & \src{\text{if $\vsc = \stkptr{\_,\_,\eaddr,\_}$}} \\ 
    \src{\linear} & \src{\text{if $\vsc = \retptrd(\_,\_)$}} \\ 
    \src{\normal} & \src{\text{if $\vsc = \retptrc(\_,\_,\_)$}}
  \end{cases}
\]

\dominique{4-12-2017: what about $\pi_a$?}
\lau{12-12-2017: I have added them as I needed them in the LR. I have not needed $\pi_a$ yet.}
\[
  \updPcAddr{\Phi} \defeq 
  \begin{cases}
    \Phi\update{\pcreg}{((\perm,\lin),\baddr,\eaddr,\aaddr+1)} & \text{if $\Phi(\pcreg) = ((\perm,\lin),\baddr,\eaddr,\aaddr)$}\\
    \failed & \totherwise
  \end{cases}
\]

\[
  \readAllowed{} \defeq \{ \rwx,\rw,\rx,\ro \}
\]

\[
  \writeAllowed{} \defeq \{ \rwx, \rw \}
\]

\[
  \isLinear{c} \defeq
  \begin{cases}
    \top & 
    \arraycolsep=0pt
    \begin{array}[t]{l}
      c = ((\_,\linear),\_,\_,\_) \tor\\
      c = \sealed{\_,\cb'} \tand \isLinear{\cb'} 
    \end{array}\\
    \src{\top} & 
    \sourcecolor\left.
    \arraycolsep=0pt
    \begin{array}[t]{l}
      c = \stkptr{\_,\_,\_,\_}\\
      c = \retptrd(\_,\_)\\
    \end{array}\right.\\
    \bot & \totherwise
  \end{cases}
\]

We now define 

\[
  \nonLinear{\cb} \defeq \neg \isLinear{\cb}
\]

\[
  \linCons{w} \defeq
  \begin{cases}
    0 & \isLinear{w} \\
    w & \totherwise
  \end{cases}
\]
\[
  \linConsPerm{\perm}{w} \defeq
  \begin{cases}
    \perm \in \writeAllowed{} & \isLinear{w} \\
    \mathrm{true} & \totherwise
  \end{cases}
\]

\[
  \exec{\cb} \defeq 
      \cb = ((\perm, \_), \_, \_, \_) \tand \perm \in \{\rwx,\rx\} 
\]

\[
  \nonExec{\cb} \defeq \neg \exec{\cb}
\]

\[
  \withinBounds{\vsc} \defeq 
  \begin{cases}
    \baddr \leq \aaddr \leq \eaddr & 
    \arraycolsep=0pt
    \begin{array}[t]{l}
      \vsc = ((\_,\_),\baddr,\eaddr,\aaddr) \src{\tor}\\
      \src{\vsc = \stkptr{\_,\baddr,\eaddr,\aaddr}}
    \end{array}\\
    \sigma_\baddr \leq \sigma_\aaddr \leq \sigma_\eaddr & \vsc = \seal{\sigma_\baddr,\sigma_\eaddr,\sigma} \\
    \bot & \totherwise
  \end{cases}
\]

\[
  \nonZero{w} \defeq
  \begin{cases}
    \bot & w \in \ints \tand w = 0 \\
    \top & \totherwise
  \end{cases}
\]

For convenience, we introduce the following notation:
\[
  \begin{array}{rcl}
    \Phi(r) & \defeq & \Phi.\reg(r)
  \end{array}
\]
where $r\in \RegName$.

For $\rn$ which can be a register or an integer, we take
\[
  \Phi(\rn) = n
\]
to mean
\[
  \text{either $n = \rn$ or $n = \Phi.\reg(\rn)$ and in either case $n \in \ints$}
\]

\subsubsection{Step relations}
\paragraph{Decode and encode functions}
We assume functions $\decInstr{} : \Word \fun \Instr$ and $\encInstr{} : \Instr \fun \ints$ where $\Instr$ is the set of target level instructions.

The $\decInstr{}$ function must be surjective and injective for all non-$\tfail$ instructions.
For any $c \in \Caps$, we have that $\decInstr{c} = \tfail$.
The $\encInstr{}$ function must be injective.
Further $\decInstr{}$ must be the left inverse of $\encInstr{}$ that is for all $i \in \Instr$
\[
  \decInstr{}(\encInstr{i}) = i
\]

These functions are used for both the target and source level machine. When we write instructions in places where words are required, we will assume that $\encInstr{}$ is implicit.

\paragraph{Step relation}
The instruction $\scall{\offpc,\offsigma}{}{}$ has length $\calllen$ which means that it does not fit in one memory address. In fact, $\scall{\offpc,\offsigma}{}{}$ should be seen as a different way of interpreting a series of instruction rather than an instruction on its own. We therefore introduce $\scall[0]{\offpc,\offsigma}{r_1}{r_2},\dots,\scall[\calllen-1]{\offpc,\offsigma}{r_1}{r_2}$ as aliases for the instructions that constitute $\scall{\offpc,\offsigma}{r_1}{r_2}$ (see Paragraph~\ref{par:call-impl} for details). We define the following condition that indicates that $\aaddr$ is the first address of a $\scall{\offpc,\offsigma}{r_1}{r_2}$ instruction in the configuration $\Phi$:
\[
  \sourcecolor\left.
    \callCond{\Phi,r_1,r_2,\offpc,\offsigma,\aaddr} = \left\{
      \begin{array}{l}
        \Phi.\mem(\aaddr) = \scall[0]{\offpc,\offsigma}{r_1}{r_2} \tand\\
        \vdots \\
        \Phi.\mem(\addr+\calllen-1) = \scall[\calllen-1]{\offpc,\offsigma}{r_1}{r_2}
      \end{array}
      \right.
  \right.
\]

We use the following step relation:
\lau{RESOLVED: Add $\ta$ to the source step.}
\lau{RESOLVED: Add $\stkb$ to source term and everything else where it is necessary. Consider how $\stkb$ and $\ta$ can be added without passing around these values in LR and source semantics (solution to parameterize both with it, so they need to be provided only once.)}
\begin{align*}
  \src{\Phi} & \; \src{\step[\src{\ta,\stkb}] \sem{\scall{\offpc,\offsigma}{r_1}{r_2}}(\Phi)} &  &\sourcecolor\left.\arraycolsep=0pt
                                                  \begin{array}[t]{l}
                                                    \text{if }\Phi(\pcreg) 
= ((\perm,\_),\baddr,\eaddr,\aaddr) \tand \\
                                                    \callCond{\Phi,r_1,r_2,\offpc,\offsigma,\aaddr} \tand\\
                                                    \lbrack \aaddr ,\aaddr + \calllen - 1\rbrack \subseteq \ta \tand \\
                                                    \lbrack\aaddr,\aaddr + \calllen-1\rbrack \subseteq [\baddr,\eaddr] \tand \\
                                                    \exec{\Phi(\pcreg)}
                                                  \end{array}\right.\\
  \Phi & \step[\src{\ta,\stkb}] \sem{\decInstr{\Phi.\mem(\aaddr)}}(\Phi) & &\left.\arraycolsep=0pt
                                                  \begin{array}[t]{l}
                                                    \text{if }\Phi(\pcreg) = ((\perm,\_),\baddr,\eaddr,\aaddr) \tand \\
                                                    \src{(\neg\callCond{\Phi,r_1,r_2,\offpc,\offsigma,\aaddr} \tor}\\
                                                    \src{\eaddr < \aaddr + \calllen-1) \tand} \\
                                                    \withinBounds{\Phi(\pcreg)} \tand \\
                                                    \exec{\Phi(\pcreg)}
                                                  \end{array}\right.\\
  \Phi& \step[\src{\ta,\stkb}] \failed & & \totherwise
\end{align*}
On the source machine, the instruction interpretation also takes $\gc$, the global constants.
It does, however, just pass it around and never makes changes to it, so we will leave it implicit.

\dominique{Note: The above does not allow executing code that is on the stack in the source machine, while this is in principle allowed in the target.
This is probably fine if we make sure a stack pointer can never be executable.}

\paragraph{Terminating computations}
We use the following notation to write that a configuration $\Phi$ successfully terminates:
\[
  \Phi\sterm[i]{\src{\ta,\stkb}} \defeq \Phi \nstep[i]{\src{\ta,\stkb}} \halted
\]
if we just know $\Phi$ terminates in some number of steps, then we write
\[
  \Phi\sterm{\src{\ta,\stkb}} \defeq \exists i \ldotp \sterm[i]{\src{\ta,\stkb}}
\]

\paragraph{Call implementation}
\label{par:call-impl}
This paragraph contains the implementation of $\scall{\offpc,\offsigma}{}{}$. That is each of the instructions in the implementation corresponds to $\scall[1]{\offpc,\offsigma}{r_1}{r_2} \dots \scall[\calllen-1]{\offpc,\offsigma}{r_1}{r_2}$, respectively.

\lau{16-10-2017: Implementation of $\scall{\offpc,\offsigma}{}{}$ may need registers for temporary registers - we should update the semantics accordingly.}
\[
  \begin{array}{l}
    \text{// push 42 on the stack (so it is non-empty).}\\
    \tmove{\rtmp{1}}{42}\\
    \tstore{\rstk}{\rtmp{1}}\\
    \tcca{\rstk}{(-1)}\\
    \text{// split the stack at its current address - rstk done.}\\
    \tgeta{\rtmp{1}}{\rstk}\\
    \tsplit{\rstk}{\rretd}{\rstk}{\rtmp{1}}\\
    \text{// load the seal for the return pointer through the pc capability.}\\
    \tmove{\rtmp{1}}{\pcreg}\\
    \tcca{\rtmp{1}}{(\offpc - 5)} \text{ // $\offpc$ is the offset from pc to the location of the seal for this code. \footnote{-5 is the offset from the previous instruction where the $\pcreg$ was moved to $\rtmp{1}$ to the first instruction of the call implementation.}}\\
    \tload{\rtmp{1}}{\rtmp{1}}\\
    \tcca{\rtmp{1}}{\offsigma}\text{ // $\offsigma$ is the offset for the seal used by this call.}\\
    \text{// seal the used stack frame as the data part of the return pointer pair.}\\
    \tcseal{\rretd}{\rtmp{1}}\\
    \text{// obtain the code part of the return pointer pair and seal it too.}\\
    \tmove{\rretc}{\pcreg}\\
    \tcca{\rretc}{5} \text{ //magic number is offset to return code}\\
    \tcseal{\rretc}{\rtmp{1}}\\
    \text{// now clear temporary register and jump to the adversary.}\\
    \tmove{\rtmp{1}}{0}\\
    \txjmp{r_1}{r_2}\\

    \text{// the following is the return code}\\
    \text{// check that the stack pointer is the same we handed out.}\\
    \tgetb{\rtmp{1}}{\rstk}\\
    \tminus{\rtmp{1}}{\rtmp{1}}{\stkb} \text{ //$\stkb$ is the stack base constant} \\
    \tmove{\rtmp{2}}{\pcreg}\\
    \tcca{\rtmp{2}}{5} \text{ //magic number is the offset to fail} \\
    \tjnz{\rtmp{2}}{\rtmp{1}} \\
    \tcca{\rtmp{2}}{1} \text{ //magic number is the offset to after fail} \\
    \tjmp{\rtmp{2}} \\
    \tfail \\
    \text{// join our stored private stack frame with the rest of the stack (this also finishes the stack pointer check).}\\
    \tsplice{\rstk}{\rstk}{\rdata} \\
    \text{// pop the magic number 42}\\
    \tcca{\rstk}{1}\\
    \text{// clear temporary registers used}\\
    \tmove{\rtmp{2}}{0}\\
    \text{// continue program after invocation.}
  \end{array}
\]
\dominique{Return code could be simplified if we had jz in addition to jnz.}
The call code does the following:
\begin{itemize}
\item Store 42 (it could be anything) to the stack (this ensures the stack is non-empty), and decrement the pointer according to convention.
\item Get the current address of the stack pointer and split according to it.
\item Retrieve the seal of the program.
\item cca the seal, so the seal to be used is active 
\item seal our private part of the stack capability.
\item Move the pc out of the pc register, adjust it to point to the first address of the return code, and seal it.
\item clear the temporary register.
\item cross jump to the two specified registers.
\item Upon return:
  \begin{itemize}
  \item get the base of the stack and check that it matches up with the global base of the stack.If not, fail.
  \item splice the returned stack pointer and the stack pointer for our private stack.
  \item adjust stack pointer to first empty address (the address with the end address is considered free)  
  \item clear the temporary registers. (Note that $\rtmp{1}$ is not cleared with 0 as it already contains 0 after the $\stkb$ check. If it didn't contain 0, then the execution would have failed.)

  \end{itemize}
\end{itemize}

For convenience, we will add the convention that the memory update $\update{\mem.a}{\scall{\offpc,\offsigma}{r_1}{r_2}}$ corresponds to
\[
  \begin{array}{l}
    \update{\mem.a}{\scall[0]{\offpc,\offsigma}{r_1}{r_2}}\\
    \update{\mem.a+1}{\scall[1]{\offpc,\offsigma}{r_1}{r_2}}\\
    \dots\\
    \update{\mem.a+\calllen-1}{\scall[\calllen-1]{\offpc,\offsigma}{r_1}{r_2}}\\
  \end{array}
\]

\subsubsection{Instruction Interpretation}
We have unified the two languages in the below definitions. Everything written in black is common for both source and target language. Everything written in \src{blue} is specific to the source language.

\noindent\textbf{fail and halt}\\
\begin{align*}
  \sem{\tfail}(\Phi) = & \; \failed \\
  \sem{\thalt}(\Phi) = & \; \halted
\end{align*}

\noindent\textbf{jmp and jnz}\\
\begin{align*}
  \sem{\tjmp{r}}(\Phi) = &  
                     \begin{cases}
                       \Phi\updReg{r}{w}\updReg{\pcreg}{\Phi(r)} & w = \linCons{\Phi(r)}
                     \end{cases}
\end{align*}

\begin{align*}
  \sem{\tjnz{r}{\rn}}(\Phi) = &       
                             \begin{cases}
                               \arraycolsep=0pt
                               \begin{array}[t]{rl}
                                 \Phi&\updReg{r}{w}\\
                                     &\updReg{\pcreg}{\Phi(r)}
                               \end{array} 
                                              & w = \linCons{\Phi(r)} \tand \nonZero{\Phi(\rn)}\\
                               \updPcAddr{\Phi} & \totherwise
                             \end{cases}
\end{align*}

\noindent\textbf{gettype}\\
In the definitions of the semantics below, we use a function $\encType{} : \Word \rightarrow \ints$. This is an encoding function for which the specific implementation does not matter. As the words of the two machines differ, we really need two functions which we call $\encType{}_{\var{src}}$ and $\encType{}_{\var{trg}}$. These two functions need to be related in the following way:
\begin{itemize}
\item $\encType{}(((\_,\_),\_,\_,\_))$, $\encType{}(\seal{\_,\_,\_})$, $\encType{}(\sealed{\_,\_})$, and $\encType{}(i)$ where $i\in\ints$ are all distinct.
\item For all $w \in \SealableCaps$ (only the words on the target machine), $\encType{}_{\var{trg}}(w) = \encType{}_{\var{src}}(w)$.
\item Finally, 
\[
\encType{}_{\var{src}}(\src{\stkptr{\_,\_,\_,\_}}) = \encType{}_{\var{trg}}(((\_,\_),\_,\_,\_))
\]
 and 
\[
\encType{}_{\var{src}}(\src{\retptrd(\_,\_)}) = \encType{}_{\var{src}}(\src{\retptrc(\_,\_,\_)}) = \encType{}_{\var{trg}}(((\_,\_),\_,\_,\_))
\]
\end{itemize}
In English this means that each type of word is represented by a distinct value and that the tokens on the source machine has the type of the capability they represent on the target machine.
\begin{align*}
  \sem{\tisptr{r_1}{r_2}}(\Phi) = & \; \updPcAddr{}(\Phi\updReg{r_1}{\encType{\Phi(r_2)}})
\end{align*}

\noindent\textbf{geta, getb, gete, getp, and getl}\\
We assume functions to encode and decode permissions as well as a function to encode linearity. The functions are used implicitly when a permission or linearity is used in a place where they need to be a word.

Specifically for the permission function, $\encPerm{} : \Perm \fun \ints$ and $\decPerm{} : \ints \fun \Perm $ encodes and decodes permissions, respectively. Where $\decPerm{}$ is the left inverse of $\encPerm{}$, $\encPerm{}$ is injective and for all $\perm \in \Perm$ $\encPerm{\perm} \neq -1$ (as this is used as an error value). $\decPerm{}$ is surjective.

For the linearity encoding, we make similar assumptions: $\encLin{} : \Linear \fun \ints$ encodes linearity. $\encLin{}$ is injective and for all $\lin \in \Linear$ $\encPerm{\lin} \neq -1$ (as this is used as an error value). As a capabilities linearity stays the same, we do not need a decoding function for linearity.
\begin{align*}
  \sem{\tgeta{r_1}{r_2}}(\Phi) = & 
                                   \begin{cases}
                                     \updPcAddr{\Phi\updReg{r_1}{\aaddr}} & 
                                     \arraycolsep=0pt
                                     \begin{array}[t]{l}
                                       \Phi(r_2) = ((\_,\_),\_,\_,\aaddr) \\
                                       \quad \tor \Phi(r_2) = \seal{\_,\_,\aaddr} \\ 
                                       \quad \src{\tor \Phi(r_2) = \stkptr{\_,\_,\_,\aaddr} } \\
                                     \end{array} \\
                                     \updPcAddr{\Phi\updReg{r_1}{-1}} & \totherwise
                                   \end{cases}
\end{align*}

\begin{align*}
  \sem{\tgetb{r_1}{r_2}}(\Phi) = & 
                                   \begin{cases}
                                     \updPcAddr{\Phi\updReg{r_1}{\baddr}} & 
                                     \arraycolsep=0pt
                                     \begin{array}[t]{l}
                                       \Phi(r_2) = ((\_,\_),\baddr,\_,\_) \\
                                       \quad \tor \Phi(r_2) = \seal{\baddr,\_,\_} \\
                                       \quad \src{\tor \Phi(r_2) = \stkptr{\_,\baddr,\_,\_} } \\
                                     \end{array} \\
                                     \updPcAddr{\Phi\updReg{r_1}{-1}} & \totherwise
                                   \end{cases}
\end{align*}

\begin{align*}
  \sem{\tgete{r_1}{r_2}}(\Phi) = & 
                                   \begin{cases}
                                     \updPcAddr{\Phi\updReg{r_1}{\eaddr}} & 
                                     \arraycolsep=0pt
                                     \begin{array}[t]{l}
                                       \Phi(r_2) = ((\_,\_),\_,\eaddr,\_) \\
                                       \quad \tor \Phi(r_2) = \seal{\_,\eaddr,\_} \\
                                       \quad \src{\tor \Phi(r_2) = \stkptr{\_,\_,\eaddr,\_} } \\
                                     \end{array} \\
                                     \updPcAddr{\Phi\updReg{r_1}{-1}} & \totherwise
                                   \end{cases}
\end{align*}
\lau{15-09-2017: Should gete and getb do something for seals?}
\dominique{4-12-2017: doesn't matter, I think.}

\begin{align*}
  \sem{\tgetp{r_1}{r_2}}(\Phi) = & 
                                   \begin{cases}
                                     \updPcAddr{\Phi\updReg{r_1}{\perm}} & 
                                     \arraycolsep=0pt
                                     \begin{array}[t]{l}
                                       \Phi(r_2) = ((\perm,\_),\_,\_,\_) \\
                                       \quad \src{\tor \Phi(r_2) = \stkptr{\perm,\_,\_,\_} } \\
                                     \end{array} \\
                                     \updPcAddr{\Phi\updReg{r_1}{-1}} & \totherwise
                                   \end{cases}
\end{align*}

\begin{align*}
  \sem{\tgetlin{r_1}{r_2}}(\Phi) = & 
                                   \begin{cases}
                                     \updPcAddr{\Phi\updReg{r_1}{\linear}} & 
                                     \arraycolsep=0pt
                                     \begin{array}[t]{l}
                                       \isLinear{\Phi(r_2)}
                                     \end{array} \\
                                     \updPcAddr{\Phi\updReg{r_1}{\normal}} & \nonLinear{\Phi(r_2)}
                                   \end{cases}
\end{align*}
\lau{19-12-2017: Resolved the comments below. Seals are all normal and we allow linearity of sealed capabilities to be retrieved as it is observable anyway.}
\lau{15-09-2017: Should getl do something for seals?}
\dominique{4-12-2017: doesn't matter, I think.}
\lau{15-09-2017: Do we want to allow getl and getp to work on sealed capabilities?}
\dominique{Linearity of sealed caps is already observable anyway, so no point in hiding it.  Permissions of the sealed capability should be hidden, I think, mostly out of conservativity (hide as much as possible!).}\\
\dominique{4-12-2017: Couldn't the definitions above use $\pi_b$, $\pi_e$, $\pi_l$ and $\pi_a$ and $\pi_p$ (the latter two to be added) ?}
\lau{12-12-2017: $\pi_\_$ also projects from sealed capabilities, so the case of sealed capabilities would have to be explicitly excluded (with getl as the exception).}

\noindent\textbf{move}\\
\begin{align*}
  \sem{\tmove{r}{\rn}}(\Phi) = & 
                              \begin{cases}
                                \updPcAddr{\Phi\updReg{r}{\rn}} & r \neq \pcreg \wedge \rn \in \ints \\
                                \updPcAddr{\Phi\updReg{\rn}{w}\updReg{r}{\Phi(\rn)}} & r \neq \pcreg \wedge w = \linCons{\Phi(\rn)}\\
                                \failed & \totherwise
                              \end{cases}
\end{align*}
(Notice that in the case where we are moving a linear capability and $r = \rn$ the order of the updates matter.)

\noindent\textbf{store}\\
\begin{align*}
  \sem{\tstore{r_1}{r_2}}(\Phi) = &
                                    \begin{cases}
                                      \updPcAddr{}\left(
                                        \arraycolsep=0pt
                                        \begin{array}{rl}
                                          \Phi&\updReg{r_2}{w_2}\\
                                              &\update{\mem.a}{\Phi(r_2)}
                                        \end{array}
\right) & 
                                      \arraycolsep=0pt
                                      \begin{array}[t]{l}
                                        \Phi(r_1) = ((\perm,\lin),\baddr,\eaddr,\aaddr) \tand \\
                                        \perm \in \writeAllowed{} \tand\\
                                        \withinBounds{\Phi(r_1)} \tand \\
                                        w_2 = \linCons{\Phi(r_2)} \tand r_2 \neq \pcreg\\
                                        \src{\aaddr \in \dom(\Phi.\mem)}
                                      \end{array}
                                      \\
                                      \sourcecolor\left.
                                      \updPcAddr{}\left(
                                      \arraycolsep=0pt
                                      \begin{array}{rl}
                                        \Phi&\updReg{r_2}{w_2}\\
                                            &\update{\ms_\stk.a}{\Phi(r_2)}
                                      \end{array}
                                      \right) \right.& 
                                      \sourcecolor\left.
                                      \arraycolsep=0pt
                                      \begin{array}[t]{l}
                                        \Phi(r_1) = \stkptr{\perm,\baddr,\eaddr,\aaddr} \tand \\
                                        \perm \in \writeAllowed{} \tand \\
                                        \withinBounds{\Phi(r_1)} \tand \\
                                        \aaddr \in \dom(\Phi.\ms_\stk) \tand \\
                                        w_2 = \linCons{\Phi(r_2)} \tand r_2 \neq \pcreg
                                      \end{array}\right.\\
                                      \failed & \totherwise
                                    \end{cases}
\end{align*}

\noindent\textbf{load}\\
\begin{align*}
  \sem{\tload{r_1}{r_2}}(\Phi) = & 
                                  \begin{cases}
                                    \updPcAddr{}\left(
                                      \arraycolsep=0pt
                                      \begin{array}{rl}
                                        \Phi&\update{\mem.a}{w_2}\\
                                            &\updReg{r_1}{w}
                                      \end{array}\right)& 
                                    \arraycolsep=0pt
                                    \begin{array}[t]{l}
                                      \Phi(r_2) = ((\perm,\lin),\baddr,\eaddr,\aaddr) \tand \\
                                      \perm \in \readAllowed{} \tand\\
                                      \withinBounds{\Phi(r_2)} \tand \\
                                      w = \Phi.\mem(a) \tand r_1 \neq \pcreg \tand \\
                                      w_2 = \linCons{w} \tand
                                      \linConsPerm{\perm}{w}
                                    \end{array}
                                    \\
                                    \sourcecolor\left.
                                    \updPcAddr{}\left(
                                      \arraycolsep=0pt
                                      \begin{array}{rl}
                                        \Phi& \update{\ms_\stk.a}{w_2}\\
                                            &\updReg{r_1}{w}
                                      \end{array}\right)\right.
                                    & 
                                    \sourcecolor\left.
                                    \arraycolsep=0pt
                                    \begin{array}[t]{l}
                                      \Phi(r_2) = \stkptr{\perm,\baddr,\eaddr,\aaddr} \tand \\
                                      \perm \in \readAllowed{} \tand \\
                                      \withinBounds{\Phi(r_2)} \tand \\
                                      \aaddr \in \dom(\Phi.\ms_\stk) \tand \\ 
                                      w = \Phi.\ms_\stk(a) \tand r_1 \neq \pcreg \tand \\
                                      w_2 = \linCons{w} \tand
                                      \linConsPerm{\perm}{w}
                                    \end{array}\right.
                                    \\
                                    \failed & \totherwise                                    
                                  \end{cases}
\end{align*} 

\noindent\textbf{cca}\\
\emph{Change Current Address}
\\\lau{15-09-2017: This is the old lea. I changed the name, so it actually reflects what the instruction does. We might also want to add a real lea at some point.}
\begin{align*}
  \sem{\tcca{r}{\rn}}(\Phi) = & 
                                  \begin{cases}
                                    \updPcAddr{\Phi\updReg{r}{c}} &  
                                    \arraycolsep=0pt
                                    \begin{array}[t]{l}
                                      \Phi(\rn) = n \tand \\
                                      \Phi(r) = ((\perm,\lin),\baddr,\eaddr,\aaddr) \tand \\
                                      c = ((\perm,\lin),\baddr,\eaddr,\aaddr + n) \tand\\
                                      r \neq \pcreg
                                    \end{array}
                                    \\
                                    \updPcAddr{\Phi\updReg{r}{s}} &  
                                    \arraycolsep=0pt
                                    \begin{array}[t]{l}
                                      \Phi(\rn) = n \tand \\
                                      \Phi(r) = \seal{\sigma_\baddr,\sigma_\eaddr,\sigma} \tand \\
                                      s = \seal{\sigma_\baddr,\sigma_\eaddr,\sigma + n}
                                    \end{array}
                                    \\
                                    \src{\updPcAddr{\Phi\updReg{r}{c}}} &  
                                    \sourcecolor\left.
                                    \arraycolsep=0pt
                                    \begin{array}[t]{l}
                                      \Phi(\rn) = n\tand \\
                                      \Phi(r) = \stkptr{\perm,\baddr,\eaddr,\aaddr} \tand \\
                                      c = \stkptr{\perm,\baddr,\eaddr,\aaddr + n}
                                    \end{array}\right.
                                    \\
                                    \failed & \totherwise
                                  \end{cases}
\end{align*}

\noindent\textbf{restrict}\\
This instruction uses the $\decPerm{}$ function.
\begin{align*}
  \sem{\trestrict{r_1}{\rn}}(\Phi) = &
                                      \begin{cases}
                                        \updPcAddr{}\left(
                                          \arraycolsep=0pt
                                          \begin{array}{rl}
                                          \Phi &\updReg{r_1}{c}
                                          \end{array} \right)
&
                                        \arraycolsep=0pt
                                        \begin{array}[t]{l}
                                          \Phi(r_1) = ((\perm,\lin),\baddr,\eaddr,\aaddr) \tand \\
                                          \Phi(\rn) = n \tand\\
                                          \decPerm{n} \sqsubseteq \perm \tand \\
                                          c = ((\decPerm{n},\lin),\baddr,\eaddr,\aaddr) \tand\\
                                          r_1 \neq \pcreg
                                        \end{array}
                                        \\
                                        \sourcecolor\left.
                                        \updPcAddr{}\left(
                                          \arraycolsep=0pt
                                          \begin{array}{rl}
                                          \Phi &\updReg{r_1}{c}
                                          \end{array}\right)\right.
                                        &
                                        \sourcecolor\left.
                                        \arraycolsep=0pt
                                        \begin{array}[t]{l}
                                          \Phi(r_1) = \stkptr{\perm,\baddr,\eaddr,\aaddr} \tand \\
                                          \Phi(\rn) = n \tand\\
                                          \decPerm{n} \sqsubseteq \perm \tand \\
                                          c = \stkptr{\decPerm{n},\baddr,\eaddr,\aaddr}
                                        \end{array}\right.
                                        \\
                                        \failed & \totherwise
                                      \end{cases}
\end{align*}

\noindent\textbf{lt}\\
\begin{align*}
  \sem{\tlt{r_0}{\rn_1}{\rn_2}}(\Phi) = &
                                                  \begin{cases}
                                                    \updPcAddr{\Phi\updReg{r_0}{1}} &
                                                    \arraycolsep=0pt
                                                    \begin{array}[t]{l}
                                                      \text{if for $i \in \{1,2\}$}\\
                                                      \quad\Phi(\rn_i) = n_i \tand\\
                                                      \quad n_1 < n_2
                                                    \end{array}\\
                                                    \updPcAddr{\Phi\updReg{r_0}{0}} &
                                                    \arraycolsep=0pt
                                                    \begin{array}[t]{l}
                                                      \text{if for $i \in \{1,2\}$}\\
                                                      \quad \Phi(\rn_i) = n_i \tand \\
                                                      \quad n_1 \not< n_2\\
                                                    \end{array}\\
                                                    \failed & \text{otherwise}
                                                  \end{cases}  
\end{align*}

\noindent\textbf{plus and minus}\\
\begin{align*}
  \sem{\tplus{r_0}{\rn_1}{\rn_2}}(\Phi) = &
                                                  \begin{cases}
                                                    \updPcAddr{\Phi\updReg{r_0}{n_1+n_2}} &
                                                    \arraycolsep=0pt
                                                    \begin{array}[t]{l}
                                                      \text{if for $i \in \{1,2\}$}\\
                                                      \quad \Phi(\rn_i) = n_i\\
                                                    \end{array}\\
                                                    \failed & \text{otherwise}
                                                  \end{cases}  
\end{align*}

\begin{align*}
  \sem{\tminus{r_0}{\rn_1}{\rn_2}}(\Phi) = &
                                                  \begin{cases}
                                                    \updPcAddr{\Phi\updReg{r_0}{n_1-n_2}} &
                                                    \arraycolsep=0pt
                                                    \begin{array}[t]{l}
                                                      \text{if for $i \in \{1,2\}$}\\
                                                      \quad\Phi(\rn_i) = n_i\\
                                                    \end{array}\\
                                                    \failed & \text{otherwise}
                                                  \end{cases}  
\end{align*}

\noindent\textbf{seta2b}\\
\begin{align*}
  \sem{\tsetatob{r_1}}(\Phi) = & 
                                \begin{cases}
                                  \updPcAddr{}\left(
                                    \Phi \updReg{r_1}{c}
                                    \right)
&
                                    \arraycolsep=0pt
                                    \begin{array}[t]{l}
                                      \Phi(r_1) = ((\perm,\lin),\baddr,\eaddr,\_) \tand \\
                                      c = ((\perm,\lin),\baddr,\eaddr,\baddr) \tand\\
                                      r_1 \neq \pcreg
                                    \end{array} \\
                                  \updPcAddr{}\left(
                                    \Phi \updReg{r_1}{c}
                                    \right)
&
                                    \arraycolsep=0pt
                                    \begin{array}[t]{l}
                                      \Phi(r_1) = \seal{\sigma_\baddr,\sigma_\eaddr,\_} \tand \\
                                      c = \seal{\sigma_\baddr,\sigma_\eaddr,\sigma_\baddr}
                                    \end{array} \\
\sourcecolor\left.
                                  \updPcAddr{}\left(
                                    \Phi\updReg{r_1}{c}
                                    \right)
\right.
&
\sourcecolor\left.
                                    \arraycolsep=0pt
                                    \begin{array}[t]{l}
                                      \Phi(r_1) = \stkptr{\perm,\baddr,\eaddr,\_} \tand \\
                                      c = \stkptr{\perm,\baddr,\eaddr,\baddr}
                                    \end{array}\right. \\
                                    \failed & \totherwise
                                \end{cases}
\end{align*}

\noindent\textbf{xjmp}\\
\lau{22-09-2017: Can we allow the callee to return any stack pointer? At the moment, it is not possible to create an empty stack pointer which may be a saving grace. If we at some point change this, so you can split capabilities into empty ones, then we need to consider whether it is safe for callees to return empty stacks (I think it is okay as long as we always have some stack usage).}
\begin{align*}
  \sem{\txjmp{r_1}{r_2}}(\Phi) = &
                                   \begin{cases}
                                       \Phi''
                                     &
                                     \begin{array}[t]{l}
                                       \Phi(r_1) = \sealed{\sigma_1,c_1} \tand
                                       \Phi(r_2) = \sealed{\sigma_2,c_2} \tand\\
                                       \sigma_1 = \sigma_2 \tand \\
                                       w_1 = \linCons{c_1} \tand \\
                                       w_2 = \linCons{c_2} \tand \\
                                       \Phi' = \Phi \updReg{r_1}{w_1} \updReg{r_2}{w_2} \tand\\
                                       \Phi'' = \xjumpResult{c_1}{c_2}{\Phi'}
                                     \end{array}\\
                                     \failed & \totherwise
                                   \end{cases}
\end{align*}
\begin{align*}
  \xjumpResult{c_1}{c_2}{\Phi} = 
  \begin{cases}
    \arraycolsep=0pt
    \begin{array}[t]{rl}
      \Phi & \updReg{\pcreg}{c_1}\\
           & \updReg{\rdata}{c_2}
    \end{array} &
    \begin{array}[t]{l}
      \src{c_1 \neq \retptrc(\_) \tand} \\
      \src{c_2 \neq \retptrd(\_) \tand} \\
      \nonExec{c_2}
    \end{array}
    \\&\\
    \sourcecolor\left.
      \arraycolsep=0pt
      \begin{array}[t]{rl}
        \Phi' & \updReg{\pcreg}{c_\opc}\\
              & \updReg{\rdata}{0}\\
              & \updReg{\rstk}{c_\stk}\\
              & \updReg{\rtmp{1}}{0}\\
              & \updReg{\rtmp{2}}{0}
      \end{array}\right.&
    \sourcecolor\left.
      \begin{array}[t]{l}
        c_1 = \retptrc(\baddr,\eaddr,\aaddr) \tand \\
        c_2 = \retptrd(\aaddr_\stk,\eaddr_{\stk,\priv}) \tand \\
        \Phi(\rstk) = \stkptr{\rw,\stkb,\eaddr_\stk,\_} \tand \\
        \stkb \leq \eaddr_\stk \tand\\
        \Phi = (\mem,\reg,\stkf::\stk,\ms_\stk) \tand \\
        \stkf = (\opc,\ms_{\stk,\priv}) \tand \\
        \opc = \aaddr \tand \\
        c_\opc = ((\rx,\normal),\baddr,\eaddr,\opc) \tand \\
        \dom(\ms_{\stk,\priv}) = [\eaddr_\stk+1,\eaddr_{\stk,\priv}] \tand\\
        e_\stk + 1 = a_\stk \tand\\
        c_\stk = \stkptr{\rw,\stkb,\eaddr_{\stk,\priv},\aaddr_\stk} \tand\\
        \Phi' = (\mem,\reg,\stk,\ms_{\stk,\priv} \uplus \ms_\stk) 
      \end{array}\right.
    \\
    \failed & \totherwise
  \end{cases}
\end{align*}

\noindent\textbf{cseal}\\
\begin{align*}
  \sem{\tcseal{r_1}{r_2}}(\Phi) = &
                                  \begin{cases}
                                    \updPcAddr{}\left(
                                    \arraycolsep=0pt
                                    \begin{array}{rl}
                                      \Phi&\updReg{r_1}{\vsc}
                                    \end{array}\right)
&
                                    \arraycolsep=0pt
                                    \begin{array}[t]{l}
                                      \Phi(r_1) \in \SealableCaps \tand \\
                                      \Phi(r_2) = \seal{\sigma_\baddr, \sigma_\eaddr,\sigma} \tand \\
                                      \sigma_\baddr \le \sigma \le \sigma_\eaddr \tand \\
                                      \vsc = \sealed{\sigma,\Phi(r_1)}
                                    \end{array}
                                    \\
                                    \failed & \totherwise
                                  \end{cases}
\end{align*}

\noindent\textbf{split and splice}\\
We would like splice and split to have following properties
\begin{enumerate}
\item No authority amplification - splitting or splicing capabilities should give you no more authority than you already had.
\item Split should be dual to splice in the sense that a split on a capability followed by a splice of the two resulting capabilities should yield the same capability.
\item Take the addresses governed by a linear capability to be a multiset. If this capability is split, then the union of the two multisets of addresses governed by the resulting capabilities should be the same as the first multiset. In other words, splice and split should not break linearity.
\end{enumerate}
Split cannot create ``empty capabilities'' (a capability that governs no segment of the memory, i.e.\ a capability where the base address is greater than the end address). We partly do not allow this out of convenience as it makes the implementation of call simpler. We do not need empty capabilities as they have no semantic value in the sense that they allow you to do essentially the same as a piece of data.
\begin{align*}
  \sem{\tsplit{r_1}{r_2}{r_3}{\rn_4}}(\Phi) = &
                               \begin{cases}
                                 \updPcAddr{}\left(
                                   \arraycolsep=0pt
                                   \begin{array}{rl}
                                     \Phi&\updReg{r_3}{w}\\
                                         &\updReg{r_1}{c_1}\\
                                         &\updReg{r_2}{c_2}
                                   \end{array}\right)
&
                                 \arraycolsep=0pt
                                 \begin{array}[t]{l}
                                   \Phi(r_3) = ((\perm,\lin),\baddr,\eaddr,\aaddr) \tand \\
                                   \Phi(\rn_4) = n \tand\\
                                   \baddr \leq n \tand n < \eaddr \tand \\
                                   c_1 = ((\perm,\lin),\baddr,n,\aaddr) \tand \\
                                   c_2 = ((\perm,\lin),n+1,\eaddr,\aaddr) \tand \\
                                   w = \linCons{\Phi(r_3)} \tand \\
                                   r_1,r_2,r_3 \neq \pcreg
                                 \end{array}\\
                                 \updPcAddr{} \left(
                                 \arraycolsep=0pt
                                 \begin{array}{rl}
                                   \Phi&\updReg{r_1}{c_1}\\
                                       &\updReg{r_2}{c_2}
                                 \end{array} \right)
&
                                 \arraycolsep=0pt
                                 \begin{array}[t]{l}
                                   \Phi(r_3) = \seal{\sigma_\baddr,\sigma_\eaddr,\sigma} \tand \\
                                   \Phi(\rn_4) = n \tand\\
                                   \sigma_\baddr \leq n \tand n < \sigma_\eaddr \tand \\
                                   c_1 = \seal{\sigma_\baddr,n,\sigma} \tand \\
                                   c_2 = \seal{n+1,\sigma_\eaddr,\sigma}
                                 \end{array}\\
                                   \sourcecolor\left.
                                   \updPcAddr{}\left(
                                   \arraycolsep=0pt
                                   \begin{array}{rl}
                                     \Phi&\updReg{r_3}{0}\\
                                               &\updReg{r_1}{c_1}\\
                                               &\updReg{r_2}{c_2}
                                   \end{array} \right)\right.
&
                                 \sourcecolor\left.
                                 \arraycolsep=0pt
                                 \begin{array}[t]{l}
                                   \Phi(r_3) = \stkptr{\perm,\baddr,\eaddr,\aaddr} \tand \\
                                   \Phi(\rn_4) = n \tand \\
                                   \baddr \leq n \tand n < \eaddr \tand \\
                                   c_1 = \stkptr{\perm,\baddr,n,\aaddr} \tand \\
                                   c_2 = \stkptr{\perm,n+1,\eaddr,\aaddr}
                                 \end{array} \right.\\
                                 \failed & \totherwise
                               \end{cases}
\end{align*}
Two important points about $\tsplice{}{}{}$ related to the calling convention: (1) Splice fails if two capabilities are not adjacent.
This means that if a caller tries to use a return pointer with a stack that is not immediately adjacent to the private stack, then it fails.
(2) Splice prohibit splicing with an empty capability!
This means that a callee cannot return an empty stack (this also means that it is impossible to make a call when all of the stack is used - this may indeed be undesirable, but without this restriction we need to handle other things).
Note: because $\tsplice{}{}{}$ does not allow empty stacks, it is not ``left inverse'' to $\tsplit{}{}{}{}$ (because of the empty case).
Intuitively, it is weird that a $\tsplit{}{}{}{}$ followed by a $\tsplice{}{}{}$ does not yield the same capability.
\begin{align*}
  \sem{\tsplice{r_1}{r_2}{r_3}}(\Phi) = &
                              \begin{cases}
                                \updPcAddr{}\left(
                                \arraycolsep=0pt
                                \begin{array}{r l}
                                  \Phi&\updReg{r_2}{w_2}\\
                                      &\updReg{r_3}{w_3}\\
                                      &\updReg{r_1}{c}
                                \end{array}\right)
&
                                \arraycolsep=0pt
                                \begin{array}[t]{l}
                                  \Phi(r_2) = ((\perm,\lin),\baddr_2,\eaddr_2,\_) \tand \\
                                  \Phi(r_3) = ((\perm,\lin),\baddr_3,\eaddr_3,\aaddr_3) \tand \\
                                  \eaddr_2 + 1 = \baddr_3 \tand \baddr_2 \leq \eaddr_2 \tand \baddr_3 \leq \eaddr_3 \tand \\
                                  c = ((\perm,\lin),\baddr_2,\eaddr_3,\aaddr_3) \tand\\
                                  w_2 = \linCons{\Phi(r_2)} \tand \\
                                  w_3 = \linCons{\Phi(r_3)}  \tand \\
                                  r_1,r_2,r_3 \neq \pcreg\\
                                \end{array}\\
                                \updPcAddr{}\left(
                                \arraycolsep=0pt
                                \begin{array}{r l}
                                  \Phi&\updReg{r_1}{c}
                                \end{array}\right)
&
                                \arraycolsep=0pt
                                \begin{array}[t]{l}
                                  \Phi(r_2) = \seal{\sigma_{\baddr,2},\sigma_{\eaddr,2},\_} \tand \\
                                  \Phi(r_3) = \seal{\sigma_{\baddr,3},\sigma_{\eaddr,3},\sigma} \tand \\
                                  \sigma_{\eaddr,2}+1 = \sigma_{\baddr,3} \tand \sigma_{\baddr,2} \leq \sigma_{\eaddr,2} \tand\\
                                  \sigma_{\baddr,3} \leq \sigma_{\eaddr,3} \tand \\
                                  c = \seal{\sigma_{\baddr,2},\sigma_{\eaddr,3},\sigma}
                                \end{array}\\
                                \sourcecolor\left.
                                \updPcAddr{}\left(
                                \arraycolsep=0pt
                                \begin{array}{r l}
                                  \Phi&\updReg{r_2}{0}\\
                                      &\updReg{r_3}{0}\\
                                      &\updReg{r_1}{c}
                                \end{array}\right)\right.
&
                                \sourcecolor
                                \arraycolsep=0pt
                                \begin{array}[t]{l}
                                  \Phi(r_2) = \stkptr{\perm,\baddr_2,\eaddr_2,\_} \tand \\
                                  \Phi(r_3) = \stkptr{\perm,\baddr_3,\eaddr_3,\aaddr_3} \tand \\
                                  \eaddr_2 + 1 = \baddr_3 \tand \baddr_2 \leq \eaddr_2 \tand \baddr_3 \leq \eaddr_3 \tand \\
                                  c = \stkptr{\perm,\baddr_2,\eaddr_3,\aaddr_3} \\
                                \end{array}\\
                                \failed & \totherwise
                              \end{cases}
\end{align*}

\noindent\textbf{call}\\

\begin{multline*}
  \sem{\scall{\offpc,\offsigma}{\src{r_1}}{\src{r_2}}}(\Phi) = \\
  \begin{cases}
    \mathit{xjumpResult}\left(c_1,c_2, 
     \arraycolsep=0pt
    \array[]{rl}
      \Phi'&\updReg{r_1}{w_1} \\
      &\updReg{r_2}{w_2} \\
      &\updReg{\rstk}{c_\stk}\\
      &\updReg{\rretc}{\sealed{\sigma_?,c_\opc}}\\
      &\updReg{\rretd}{\sealed{\sigma_?,c_\var{priv\_data}}}\\
      &\updReg{\rtmp{1}}{0}
      \endarray
     \right)
    & \arraycolsep=0pt
    \begin{array}[t]{l}
      r_1 \neq \rtmp{1} \tand r_2 \neq \rtmp{1} \tand\\
      \Phi(r_1) = \sealed{\sigma_1,c_1} \tand \\
      \Phi(r_2) = \sealed{\sigma_2,c_2} \tand \\
      \sigma_1 = \sigma_2 \tand \\
      \nonExec{c_2} \tand\\
      \Phi = (\mem,\reg,\stk,\ms_\stk) \tand\\
      \Phi(\rstk) = \stkptr{\rw,\baddr_\stk,\eaddr_\stk,\aaddr_\stk} \tand \\
      \baddr_\stk < \aaddr_\stk \leq \eaddr_\stk \tand \\
      \ms_{\stk,\priv} = \ms_\stk |_{[\aaddr_\stk,\eaddr_\stk]}\update{\aaddr_\stk}{42} \tand\\
      \ms_{\stk,\var{rest}} = \ms_\stk - \ms_\stk |_{[\aaddr_\stk,\eaddr_\stk]} \tand \\
      c_\stk = \stkptr{\rw,\baddr_\stk,\aaddr_\stk-1,\aaddr_\stk-1} \tand \\
      c_\var{priv\_data} = \retptrd(\aaddr_\stk,\eaddr_\stk) \tand \\
      \Phi(\pcreg) = ((\_,\_),\baddr,\eaddr,\aaddr) \tand \\
      \opc = \aaddr + \calllen \tand \\
      c_\opc = \retptrc(\baddr,\eaddr,\aaddr+\calllen) \tand \\
      \stk' = (\opc,\ms_{\stk,\priv}) :: \stk \tand\\
      \Phi' = (\mem,\reg,\stk',\ms_{\stk,\var{rest}})\tand\\
      \mem(\aaddr+\offpc) = \seal{\sigma_\baddr,\sigma_\eaddr,\sigma_\aaddr}\tand\\
      \baddr \leq  \aaddr+\offpc \leq \eaddr \tand \\
      \sigma_? = \sigma_\aaddr + \offsigma \tand\\
      \sigma_\baddr \leq \sigma_? \leq \sigma_\eaddr \tand\\
      w_1 = \linCons{\Phi(r_1)} \tand \\
      w_2 = \linCons{\Phi(r_2)}
    \end{array}\\
    \failed & \totherwise
  \end{cases}
\end{multline*}

Note: the caller may have split part of the stack pointer off and even pass the fragments split off to the callee in registers.
This behavior is in principle fine.
Source semantics will define that only the non-split-off part will be encapsulated.
The parts that were split off and passed to the adversary are not protected, as expected.

\subsection{Components}
\label{sec:components}

A component can be either a component with a main (i.e. a program that still needs to be linked with library implementations) or one without (i.e. a library implementation that will be used by other components). 
It contains code memory, data memory a list of imported symbols, a list of exported symbols, a list of seals used for producing return capability pairs and a list of seals used for producing closures.

We define a component as follows:
\begin{align*}
  s &\in \Symbol\\
  \var{import} &\mathrel{::=} a \mapsfrom s\\
  \var{export} &\mathrel{::=} s \mapsto w\\
  \var{comp}_0 &\mathrel{::=} (\mscode,\msdata,\overline{\var{import}},\overline{\var{export}},\sigrets,\sigcloss,A_\linear)\\
  \var{comp} &\mathrel{::=} \var{comp}_0\\
               &\mid  (\var{comp}_0,c_{\mathrm{main},c}, c_{\mathrm{main},d})
\end{align*}
\clearpage
We define inductively when a component is valid ($\ta \vdash \var{comp}$) by the below inference rules:
\begin{mathpar}
  \inference{
    \mscode(a) = \seal{\sigma_\baddr,\sigma_\eaddr,\sigma_\baddr} & [\sigma_\baddr,\sigma_\eaddr] = (\sigrets \cup \sigcloss)
  }{
    \sigrets,\sigrets[\mathrm{owned}],\sigcloss,\ta \vdash_{\mathrm{comp-code}} \mscode,a
  }
  \and
  \inference{
    \mscode(a) \in \ints\\
{    \begin{multlined}
      ([a \cdots a + \calllen-1] \subseteq \ta \wedge \mscode([a \cdots a + \calllen-1]) = \scall[0..\calllen-1]{\offpc,\offsigma}{r_1}{r_2})
      \Rightarrow\\ (\mscode(a+\offpc) = \seal{\sigma_\baddr,\sigma_\eaddr,\sigma_\baddr} \wedge \sigma_\baddr+\offsigma \in \sigrets[\mathrm{owned}])
    \end{multlined}
}  }{
    \sigrets,\sigrets[\mathrm{owned}],\sigcloss,\ta \vdash_{\mathrm{comp-code}} \mscode,a
  }
  \and
  \inference{
    \text{$\mscode$ has no hidden calls}\\
    \sigrets \mathrel{\#} \sigcloss \\
    \exists d_\sigma : \dom(\mscode) \rightarrow \powerset{\Seal} \ldotp
    \sigrets = \biguplus_{a \in \dom(\mscode)} d_\sigma(a) \tand\\
    \forall a \in \dom(\mscode) \ldotp 
    \sigrets,d_\sigma(a),\sigcloss,\ta \vdash_{\mathrm{comp-code}} \mscode,a\\
    \exists a \ldotp \mscode(a) = \seal{\sigma_\baddr,\sigma_\eaddr,\_} \wedge [\sigma_\baddr,\sigma_\eaddr] \neq \emptyset
  }{
    \sigrets,\sigcloss,\ta\vdash_{\mathrm{comp-code}} \mscode
  }
\end{mathpar}

\begin{mathpar}
  \inference{
  }{
    A_{\mathrm{code}},A_{\mathrm{own}},A_{\mathrm{non-linear}},\sigrets,\sigcloss \vdash_{\mathrm{comp-value}} z
  }
  \and
  \inference{
    \permbnf \sqsubseteq \rw&
    \lin = \linear \Rightarrow \emptyset \subset [\baddr,\eaddr] \subseteq A_{\mathrm{own}} &
    \lin = \normal \Rightarrow [\baddr,\eaddr] \subseteq A_{\mathrm{non-linear}}
  }{
    A_{\mathrm{code}},A_{\mathrm{own}},A_{\mathrm{non-linear}},\sigrets,\sigcloss \vdash_{\mathrm{comp-value}} ((\perm,\lin),\baddr,\eaddr,\aaddr)
  }
  \and
  \inference{
    A_{\mathrm{code}},A_{\mathrm{own}},A_{\mathrm{non-linear}},\sigrets,\sigcloss \vdash_{\mathrm{comp-value}} \vsc \\
    \sigma \in \sigcloss
  }{
    A_{\mathrm{code}},A_{\mathrm{own}},A_{\mathrm{non-linear}},\sigrets,\sigcloss \vdash_{\mathrm{comp-value}} \sealed{\sigma,\vsc}
  }
\end{mathpar}

\begin{mathpar}
  \inference{
    [\baddr,\eaddr] \subseteq A_{\mathrm{code}} & 
    \sigma \in \sigcloss
  }{
    A_{\mathrm{code}},A_{\mathrm{non-linear}},\sigrets,\sigcloss \vdash_{\mathrm{comp-export}} s \mapsto \sealed{\sigma,((\rx,\normal),\baddr,\eaddr,\aaddr)}
  }
  \and
  \inference{
    A_{\mathrm{code}},\emptyset,A_{\mathrm{non-linear}},\sigrets,\sigcloss \vdash_{\mathrm{comp-value}} w
  }{
    A_{\mathrm{code}},A_{\mathrm{non-linear}},\sigrets,\sigcloss \vdash_{\mathrm{comp-export}} s \mapsto w
  }
\end{mathpar}

\begin{mathpar}
  \inference{
    \dom(\mscode) = [\baddr,\eaddr]&
    [\baddr-1,\eaddr+1] \mathrel{\#} \dom(\msdata)\\
    \mspad = [\baddr-1\mapsto 0] \uplus [\eaddr+1 \mapsto 0]\\
    \sigrets,\sigcloss,\ta \vdash_{\mathrm{comp-code}} \mscode \\
    \exists A_\mathrm{own} : \dom(\msdata) \rightarrow \powerset{\dom(\msdata)} & \dom(\msdata) = A_{\mathrm{non-linear}} \uplus A_\linear \\
    A_\linear = \biguplus_{a \in \dom(\msdata)} A_{\mathrm{own}}(a) \\
    \forall a \in \dom(\msdata)\ldotp \dom(\mscode),A_{\mathrm{own}}(a),A_{\mathrm{non-linear}},\sigrets,\sigcloss \vdash_{\mathrm{comp-value}} \msdata(a)\\
    \overline{\var{export}} = \overline{s_{\mathrm{export}} \mapsto w_{\mathrm{export}}} & \overline{\var{import}} = \overline{a_{\mathrm{import}} \mapsfrom s_{\mathrm{import}}}& \{\overline{a_{\mathrm{import}}}\} \subseteq \dom(\msdata)\\
    \overline{\dom(\mscode), A_{\mathrm{non-linear}}, \sigrets,\sigcloss \vdash_{\mathrm{comp-export}} w_{\mathrm{export}}}\\
    \overline{s_{\mathrm{import}}} \mathrel{\#} \overline{s_{\mathrm{export}}} & (\emptyset \neq \dom(\mscode) \subseteq \ta) \vee (\dom(\mscode) \mathrel{\#} \ta \wedge \sigrets = \emptyset)&
    \dom(\msdata) \mathrel{\#} \ta
  }{
    \ta \vdash (\mscode\uplus \mspad,\msdata,\overline{\var{import}},\overline{\var{export}},\sigrets,\sigcloss,A_\linear)
  }
  \and
  \inference{
    \var{comp}_0 = (\mscode,\msdata,\overline{\var{import}},\overline{\var{export}},\sigrets,\sigcloss,A_\linear)\\
    \ta \vdash \var{comp}_0 & (\_ \mapsto c_{\mathrm{main},c}), (\_ \mapsto c_{\mathrm{main},d}) \in \overline{\var{export}}
  }{
    \ta \vdash (\var{comp}_0,c_{\mathrm{main},c}, c_{\mathrm{main},d})
  }
\end{mathpar}
where the following definition is used
\begin{definition}[No hidden calls]
  \label{def:no-hidden-calls}
  We say that a memory segment $\mscode$ has no hidden calls iff
  \[
    \begin{array}{l}
      \forall \aaddr \in \dom(\mscode) \ldotp \\
      \quad \forall i \in [0,\calllen - 1] \\
      \qquad \mscode(a + i) = \scall[i]{\offpc,\offsigma}{r_1}{r_2} \Rightarrow \\
      \qquad \quad (\dom(\mscode) \supseteq [\aaddr-i,\aaddr + \calllen - i - 1]
      \wedge \mscode([\aaddr-i,\aaddr + \calllen - i -1]) =
      \scall[0..\calllen-1]{\offpc,\offsigma}{r_1}{r_2}) \vee\\
      \qquad \quad \exists j \in [\aaddr-i,\aaddr + \calllen - i -1] \cap \dom(\mscode) \ldotp \mscode(j) \neq  \scall[j-\aaddr-i]{\offpc,\offsigma}{r_1}{r_2}
    \end{array}
  \]
\end{definition}

\subsection{Linking}
\label{sec:linking}

\begin{mathpar}
  \inference{
    \var{comp}_1 = (\mscode[1], \msdata[1], \overline{\var{import}_1}, \overline{\var{export}_1}, \sigrets[1], \sigcloss[1],A_{\linear,1})\\
    \var{comp}_2 = (\mscode[2], \msdata[2], \overline{\var{import}_2}, \overline{\var{export}_2}, \sigrets[2], \sigcloss[2],A_{\linear,2})\\
    \var{comp}_3 = (\mscode[3], \msdata[3], \overline{\var{import}_3}, \overline{\var{export}_3}, \sigrets[3], \sigcloss[3],A_{\linear,3})\\
    \mscode[3] = \mscode[1] \uplus \mscode[2] \\
    \msdata[3] = (\msdata[1] \uplus \msdata[2])[a \mapsto w \mid (a \mapsfrom s) \in (\overline{\var{import}_1} \cup \overline{\var{import}_2}), (s \mapsto w) \in \overline{\var{export}}_3] \\
    \overline{\var{export}_3} = \overline{\var{export}_1} \cup \overline{\var{export}_2}&
    \overline{\var{import}_3} = \{ a \mapsfrom s \in (\overline{\var{import}_1} \cup \overline{\var{import}_2}) \mid s \mapsto \_ \not\in \overline{\var{export}_3} \}\\
    \sigrets[3] = \sigrets[1] \uplus \sigrets[2] &
    \sigcloss[3] = \sigcloss[1] \uplus \sigcloss[2] &
    A_{\linear,3} = A_{\linear,1} \uplus A_{\linear,2}\\
    \dom(\mscode[3]) \mathrel{\#} \dom(\msdata[3]) & \sigrets[3] \mathrel{\#} \sigcloss[3]
  } {
    \var{comp}_3 = \var{comp}_1 \bowtie \var{comp}_2
  }
  \and
  \inference{
    \var{comp}_0'' = \var{comp}_0 \bowtie \var{comp}_0'
  }{
    (\var{comp}_0'',c_{\mathrm{main},c}, c_{\mathrm{main},d}) = \var{comp}_0 \bowtie (\var{comp}_0',c_{\mathrm{main},c}, c_{\mathrm{main},d})
  }
  \and
  \inference{
    \var{comp}_0'' = \var{comp}_0 \bowtie \var{comp}_0'
  }{
    (\var{comp}_0'',c_{\mathrm{main},c}, c_{\mathrm{main},d}) = (\var{comp}_0,c_{\mathrm{main},c}, c_{\mathrm{main},d}) \bowtie \var{comp}_0'
  }
\end{mathpar}


\subsection{Programs, contexts, initial execution configuration}
\label{sec:programs}

A program is intuitively a component that is ready to be executed, i.e. it must have an empty import list and a pair of capabilities to be used as main.
A context for a given component is any other component that can be linked with it to produce a program.

\begin{definition}[Programs and Contexts]
  We define a program to be a component $(\var{comp}_0,c_{\mathrm{main},c}, c_{\mathrm{main},d})$ with an empty import list.

  A context for a component $\var{comp}$ is another component $\var{comp}'$ such that $\var{comp} \bowtie \var{comp}'$ is a program.
\end{definition}

\begin{definition}[Initial execution configuration]
  \begin{mathpar}
    \inference{
      c_{\mathrm{main},c} = \sealed{\sigma_1, c_{\mathrm{main},c}'} &
      c_{\mathrm{main},d} = \sealed{\sigma_2, c_{\mathrm{main},d}'} &
      \sigma_1 = \sigma_2\\
      \nonExec{c_{\mathrm{main},d}'}&
      \reg(\pcreg) = c_{\mathrm{main},c}' & 
      \reg(\rdata) = c_{\mathrm{main},d}' \\
      \src{\reg(\rstk) = \stkptr{\rw,\baddr_\stk,\eaddr_\stk,\eaddr_\stk}} & 
      \trg{\reg(\rstk) = ((\rw,\linear),\baddr_\stk,\eaddr_\stk,\eaddr_\stk)} \\
      \reg(\RegName \setminus \{\pcreg,\rdata,\rstk\}) = 0\\
      \range{\ms_\stk} = \{0\}&
      \mem = \mscode \uplus \msdata \trg{\;\uplus\; \ms_\stk} \uplus \ms_\var{frame} \\
      [\baddr_\stk,\eaddr_\stk] = \dom(\ms_\stk ) &
      [\baddr_\stk-1,\eaddr_\stk+1]\mathrel{\#} (\dom(\mscode) \cup \dom(\msdata)) &
      \overline{\var{import}} = \emptyset
    }{
      ((\mscode, \msdata, \overline{\var{import}}, \overline{\var{export}}, \sigrets, \sigcloss,A_\linear),  c_{\mathrm{main},c}, c_{\mathrm{main},d}) \rightsquigarrow (\mem, \reg\src{, \emptyset, \ms_\stk})
    }
  \end{mathpar}
\end{definition}
\lau{07-07-2018: TODO why don't the above definition care about the stack base?}

\begin{definition}[Plugging a component into a context]
  When $\var{comp'}$ is a context for component $\var{comp}$ and $\var{comp}' \bowtie \var{comp} \rightsquigarrow \Phi$, 
  then we write $\plug{\var{comp'}}{\var{comp}}$ for the execution configuration $\Phi$.
\end{definition}

\begin{lemma}
  \label{lem:context-trusted-addr-ampl}
  For components $\context$ and $\comp$,
  if
  \begin{itemize}
  \item $\emptyset \vdash \context$
  \item $\dom(\comp.\mscode) \vdash \comp$
  \item $\plug{\context}{\comp}$ is defined
  \end{itemize}
  Then
  \[
    \dom(\comp.\mscode) \vdash \context
  \]
\end{lemma}
\begin{proof}
  Follows by definition.
\end{proof}

\section{Compiler}
\label{sec:compiler}

The compiler is the identity.

\section{Logical Relation}
In the following definitions, \src{blue} is used to indicate values related to the source machine. This is unlike previous definitions, where \src{blue} was used to indicate source language specific parts of definitions.
\subsection{Worlds}
\label{subsec:worlds}
\begin{theorem}
  \label{thm:recursive-domain-eq}
  There exists a c.o.f.e.\ $\Wor$ and preorder $\future$ such that $(\Wor,\future)$ and there exists an isomorphism $\xi$ such that
  \[
    \xi : \Wor \cong \blater (\Worldh \times \Worlds \times \Worldfs)
  \]
  and for $\hat{W},\hat{W}' \in \Wor$
  \[
    \hat{W'} \future \hat{W}\text{ iff }\xi(\hat{W'}) \future \xi(\hat{W})
  \]
\end{theorem}

Where $\Worlds$, $\Worldh$, and $\Worldfs$ are defined as follows
\[
  \Worldh = \RegionName \parfun (\Regions + \Regionh)
\]
and
\[
  \Worlds = \RegionName \parfun (\Regions \times \Addr)
\]
and
\[
  \Worldfs = \RegionName \parfun \Regions
\]
where $\RegionName = \nats$.
\begin{multline*}
  \Regionh = 
  \{\pure \} \times (\Wor \monnefun \URel{\MemSeg^2}) \times \\
  (\Seal \parfun \Wor \monnefun \URel{\SealableCaps \times \SealableCaps})
\end{multline*}
and
\[
  \Regions = \left\{
  \begin{array}{l}
    \{\spatial \} \times (\Wor \monnefun \URel{\MemSeg^2}) \uplus \\
    \{\spatialo \} \times (\Wor \monnefun \URel{\MemSeg^2})\uplus \\ 
    \{\revoked\}
  \end{array} \right.
\]
where $\spatial$ and $\spatialo$ are regions governing segments of memory addressed by linear capabilities.
$\spatialo$ signifies that this region is addressable.
$\spatial$ signifies that the region is not owned and can thus not be addressed.
At the same time it signifies that if something else addresses it, it is a $\linear$ capability.
Finally, $\pure$ signifies that the region is only addressed by non-linear capabilities.
Notice that no region allows for both linear and non-linear capabilities to address it.
Notice also that $\pure$ regions have an additional component that allows them to claim ownership of part of the address space of seals and impose a relational invariant on everything signed with those seals.

\lau{I suppose you are thinking of $\nequal$ for UPreds?
  The way I have seen it before, it is defined as ``up to $n-1$ (including)''.
  More precisely: Define the $k$-cut of $p \in \UPred{A}$ to be \newcommand{\floor}[1]{\lfloor #1 \rfloor} \newcommand{\kcut}[2][k]{\floor{#2}_{#1}}
  \[
  \kcut{p} = \{ (j,a) \in p | j < k \}.
  \]
  We use this to define a family of equivalences $(\nequal)_{n=0}^{\infty}$ to be, for $p$ and $q$ in $\UPred{A}$
  \[
    p \nequal q \textit{ iff } \kcut[n]{p} = \kcut[n]{q}
  \]
  If we pick a different definition, we should check that it is indeed a well-defined $n$-equality (for instance, it should be the total relation for $\nequal[0]$).
} 


We introduce a bit of notation for projecting out each part of the world:
\[
  \begin{array}{l}
    \pwheap = \pi_1(W)\\
    \pwpriv = \pi_2(W)\\
    \pwfree = \pi_3(W)
  \end{array}
\]
as well as projections for the regions:
\[
  \begin{array}{l}
  \prv{(v,s,\phi_\pub,\phi,H)} = v \\
  \prv{\revoked} = \revoked     
  \end{array}
\]

We define erasure for worlds as follows:
\[
\erase{(W_{\mathrm{heap}},W_{\mathrm{priv}},W_{\mathrm{free}})}{S} =
  \left( \erase{W_{\mathrm{heap}}}{S},\erase{W_{\mathrm{priv}}}{S},\erase{W_{\mathrm{free}}}{S} \right)
\]
where erasure for each part of a world is defined as follows:
\[
\erase{W_{\mathrm{heap}}}{S} = \lambda r \ldotp \left\{
    \begin{array}{l}
      \prv{W_{\mathrm{heap}}(r)} \in S\\
      \bot
    \end{array}
  \right.
\]
\[
\erase{W_{\mathrm{priv}}}{S} = \lambda r \ldotp \left\{
    \begin{array}{l}
      \prv{\pregion{W_{\mathrm{priv}}(r)}} \in S \\
      \bot
    \end{array}
  \right.
\]
\[
\erase{W_{\mathrm{free}}}{S} = \lambda r \ldotp \left\{
    \begin{array}{l}
      \prv{W_{\mathrm{free}}(r)} \in S\\
      \bot
    \end{array}
  \right.
\]

The $\activeReg{}$ function takes a world and filters away all the revoked regions, so
\[
  \activeReg{W} = \erase{W}{\spatial,\spatialo,\pure}
\]

Disjoint union of worlds. Joins together two alike worlds with strictly disjoint ownership over $\spatialo$-regions. Two worlds can be joined together if all of their three parts agree on the region names and each of their regions can be joined together.
\[
  W_1 \oplus W_2 = W
  \text{ iff }
  \begin{array}[t]{l}
    \dom(\pwheap) = \dom(\pwheap[W_1]) = \dom(\pwheap[W_2]) \tand \\
    \dom(\pwfree) = \dom(\pwfree[W_1]) = \dom(\pwfree[W_2]) \tand \\
    \dom(\pwpriv) = \dom(\pwpriv[W_1]) = \dom(\pwpriv[W_2]) \tand \\
    \forall r \in \dom(\pwheap) \ldotp \pwheap(r) = \pwheap[W_1](r) \oplus \pwheap[W_2](r) \tand \\
    \forall r \in \dom(\pwfree) \ldotp \pwfree(r) = \pwfree[W_1](r) \oplus \pwfree[W_2](r) \tand \\
    \forall r \in \dom(\pwpriv) \ldotp \pi_1(\pwpriv(r)) = \pi_1(\pwpriv[W_1](r)) \oplus \pi_1(\pwpriv[W_2](r))
  \end{array}
\]
$\oplus$ on regions is defined as follows
\begin{align*}
  (\pure,H,H_\sigma) \oplus (\pure,H,H_\sigma) =  & \; (\pure,H,H_\sigma) \\
  (\spatial,H) \oplus (\spatial,H) =  & \; (\spatial,H) \\
  (\spatialo,H) \oplus (\spatial,H) = & \; (\spatial,H) \oplus (\spatialo,H)\\
                                           =  & \; (\spatialo,H)
\end{align*}
and for all other cases $\oplus$ is undefined. Specifically, $\oplus$ is not defined when both sides are a $\spatialo$-region. It is further not defined if the two sides do not agree on region type or heap or sealed value relations.

\begin{lemma}
  \label{lem:oplus-assoc-comm}
  $\oplus$ is associative and commutative.
  Also, left-hand sides in the commutativity and associativity laws are defined whenever the right-hand sides are defined and vice versa.
\end{lemma}
\begin{proof}
  Follows easily from the definitions.
\end{proof}

\begin{lemma}[$\oplus$ and future worlds]
  \label{lem:oplus-future-distr}
  If $W' \future W_1 \oplus W_2$, then there exist $W_1',W_2'$ such that $W' = W_1' \oplus W_2'$ and $W_1' \future W_1$ and $W_2' \future W_2$.
\end{lemma}
\begin{proof}
  We define $W_1'$ and $W_2'$ to have the same regions as $W'$ with a possibly different visibility.
  For regions that are present in $W_1$ and $W_2$, we give them the same visibility in $W_1'$ and $W_2'$ respectively.
  For regions that are new in $W'$, we make them $\pure$ or $\spatial$ in both $W_1'$ and $W_2'$ if they are in $W'$ and we make them $\spatialo$ in $W_1'$ but $\spatial$ in $W_2'$ if they are spatial in $W'$.
  It is then easy to check that the required equations hold.
\end{proof}

We also define a second disjoint union operator of worlds:
\[
  W_1 \uplus W_2 = W
  \text{ iff }
  \begin{array}[t]{l}
    \dom(\pwheap) = \dom(\pwheap[W_1]) \uplus \dom(\pwheap[W_2]) \tand \\
    \dom(\pwfree) = \dom(\pwfree[W_1]) \uplus \dom(\pwfree[W_2]) \tand \\
    \dom(\pwpriv) = \dom(\pwpriv[W_1]) \uplus \dom(\pwpriv[W_2]) \\
  \end{array}
\]
The two operators $\uplus$ and $\oplus$ are quite different.
The difference is most clear in the treatment of pure regions: $\uplus$ allows both worlds to have the same pure region, while $\oplus$ forbids this.
To understand this different treatment ($W_1 \uplus W_2$ and $W_1 \oplus W_2$), you should understand that the two are intended for different usages of worlds.
The $W_1 \oplus W_2$ operator treats the worlds as specifications of authority: taking the disjoint union of worlds specifying non-exclusive ownership of a block of memory is allowed and produces a new world that also specifies non-exclusive ownership of world. 
The $W_1 \uplus W_2$  operator treats worlds as specifications of memory contents: taking the disjoint union of worlds specifying the presence of the same memory range is not allowed.
The latter operator is used in the logical relation for components which specifies (among other things) that the world should specify the presence of the component's data memory.
Linking two components then produces a new component with both components' data memory.
The linked component is then valid in a world that has the combined memory presence specifications, not the combined authority.
In other words, $\oplus$ specifies disjoint authority distribution, while $\uplus$ specifies disjoint memory allocation.

Note also that this picture is further complicated by our usage of non-authority-carrying $\spatial$ regions.
They are really only there in a world $W$ as a shadow copy of a $\spatialo$ region in another world $W'$ that $W$ will be combined with.
The shadow copy is used for specifying when a memory satisfies a world: the memory should contain all memory ranges that anyone has authority over, not just the ones whose authority belongs to the memory itself.
For example, if a register contains a linear pointer to a range of memory, then the register file will be valid in a world where the corresponding region is $\spatialo$, while the memory will be valid in a world with the corresponding region only $\spatial$.
However, for the memory to satisfy the world, the block of memory needs to be there, i.e. the memory should contain blocks of memory satisfying every region that is $\spatialo$, $\pure$, but also just $\spatial$ (because it may be $\spatialo$ in, for example, the register file's world).

\begin{lemma}
  \label{lem:uplus-assoc-comm}
  $\uplus$ is associative and commutative.
\end{lemma}
\begin{proof}
  Follows easily from the definitions.
\end{proof}

\begin{lemma}[Odd distributivity of $\oplus$ and $\uplus$]
  \label{lem:oplus-distr-uplus}
  \begin{equation*}
    (W_1\oplus W_2) \uplus (W_3 \oplus W_4) = (W_1 \uplus W_3) \oplus (W_2 \uplus W_4)
  \end{equation*}
  Also, the left expression is defined iff the right expression is.
\end{lemma}
\begin{proof}
  Follows by definition-chasing.
\end{proof}

\subsection{Future world}
The future world relation becomes:
\begin{mathpar}
  \inferrule{ \text{for $i \in \{\mathrm{heap},\mathrm{free},\mathrm{priv} \}$} \\ \exists m_i : \RegionName \fun \RegionName, \text{ injective}\ldotp \dom(W'.i) \supseteq \dom(m_i(W.i)) \wedge \forall r \in \dom(W.i)\ldotp W'.i(m_i(r)) \future W.i(r) }
            { W' \future W }
\end{mathpar}



Future regions allow $\spatial$ regions to become $\revoked$.
Also: the future region relation allows $\spatial$ regions to become $\spatialo$, which expresses that our system is affine, rather than linear.
\lau{22-3-2018: TODO need extra rules for regions with seal interpretations}
\begin{mathpar}
  \inferrule{ }{ \revoked \future (\spatial,\_)}
\and
  \inferrule{ }{ \revoked \future \revoked }
  \and
  \inferrule{ }{(\spatialo,H) \future (\spatial,H)}
  \and
  \inferrule{ }{(\spatialo,H) \future (\spatialo,H)}
  \and
  \inferrule{ }{(\spatial,H) \future (\spatial,H)}
\end{mathpar}

\begin{definition}[The pure part of a world]
  \label{def:purePart}
  For any world $W$, we define
  \begin{align*}
    \purePart{W} &\defeq (\purePart{\pwheap},\purePart{\pwpriv},\purePart{\pwfree})\\
    \purePart{W_\var{heap}} &\defeq
                       \begin{cases}
                         W_{\var{heap}}(r) & \text{if } W_\var{heap}(r) = (pure,\var{sm})\\
                         (\spatial,\var{sm}) & \text{if } W_\var{heap}(r) = (\spatial,\var{sm})\\
                         (\spatial,\var{sm}) & \text{if } W_\var{heap}(r) = (\spatialo,\var{sm})\\
                         \revoked & \text{if } W_\var{heap}(r) = \revoked\\
                       \end{cases}\\
    \purePart{W_\var{priv}} &\defeq
                       \begin{cases}
                         ((\spatial,\var{sm}),\opc) & \text{if } W_\var{priv}(r) = ((\spatial,\var{sm}),\opc)\\
                         ((\spatial,\var{sm}),\opc) & \text{if } W_\var{priv}(r) = ((\spatialo,\var{sm}),\opc)\\
                         (\revoked,\opc) & \text{if } W_\var{priv}(r) = (\revoked,\opc)\\
                       \end{cases}\\
    \purePart{W_\var{free}} &\defeq
                       \begin{cases}
                         W_{\var{free}}(r) & \text{if } W_{\var{free}}(r) = (pure,\var{sm})\\
                         (\spatial,\var{sm}) & \text{if } W_{\var{free}}(r) = (\spatial,\var{sm})\\
                         (\spatial,\var{sm}) & \text{if } W_{\var{free}}(r) = (\spatialo,\var{sm})\\
                         \revoked & \text{if } W_{\var{free}}(r) = \revoked\\
                       \end{cases}
  \end{align*}
\end{definition}

\begin{lemma}[purePart is duplicable]
  \label{lem:purePart-duplicable}
 For all $W$, we have that $W = W \oplus \purePart{W}$.
\end{lemma}
\begin{proof}
  Follows from the definition of $\purePart{}$ and $\oplus$.
\end{proof}

\begin{lemma}[purePart is idempotent]
  \label{lem:purePart-idempotent}
  For all $W$, we have that $\purePart{W} = \purePart{}(\purePart{W})$ and $\purePart{W} = \purePart{W} \oplus \purePart{W}$.
\end{lemma}
\begin{proof}
  The first part follows easily from the definition.
  The second statement follows from the first, together with Lemma~\ref{lem:purePart-duplicable}.
\end{proof}

\begin{lemma}[purePart respects $\oplus$]
  \label{lem:purePart-oplus}
  For all $W_1,W_2$, we have that $\purePart{W_1\oplus W_2} = \purePart{W_1} \oplus \purePart{W_2} = \purePart{W_1} = \purePart{W_2}$.
  Also, all worlds in the equations above are defined when $W_1\oplus W_2$ is defined (but not necessarily vice versa).
\end{lemma}
\begin{proof}
  Follows from the definition of $\purePart{}$ and $\oplus$.
\end{proof}

\begin{lemma}[purePart is monotone]
  \label{lem:purePart-mono}
 For all $W'\future W$, we have that $\purePart{W'} \future \purePart{W}$.
\end{lemma}
\begin{proof}
  Follows easily from the definition of $\purePart{}$ and $\future$.
\end{proof}

\begin{lemma}[Increasing authority is the future]
  \label{lem:oplus-future}
  For all $W_1,W_2$, we have that:
  $W_1 \oplus W_2 \future W_1$
\end{lemma}
\begin{proof}
  Follows easily from the definitions.
\end{proof}

\begin{lemma}[Adding memory is the future]
  \label{lem:uplus-future}
  For all $W_1,W_2$, we have that:
  $W_1 \uplus W_2 \future W_1$
\end{lemma}
\begin{proof}
  Follows easily from the definitions.
\end{proof}

\begin{lemma}[Purity is a thing of the past]
  \label{lem:world-fut-purePart}
  For all $W$, we have that $W \future \purePart{W}$.
\end{lemma}
\begin{proof}
  Consequence of Lemmas~\ref{lem:purePart-duplicable}
  and~\ref{lem:oplus-future}.
\end{proof}

\begin{lemma}[Partial authority is better than nothing]
  \label{lem:auth-partial-betterthannothing}
  If $W = W_1 \oplus W_2$, then $W_1 \future \purePart{W}$.
\end{lemma}
\begin{proof}
  Follows easily from the definitions.
\end{proof}

\subsection{Memory satisfaction}
\lau{14-02-2018: No matter what the memory satisfaction ends up looking like, remember to handle the revoked case.}
\lau{RESOLVED add frame to memory satisfaction and remove it from observation relation.}
Memory satisfaction for new worlds:
\[
  \memSat{\ms_S,\ms_\stk,\stk,\ms_T}{W} \text{ iff } 
  \left\{
    \begin{array}{l}
      \stk = (\opc_0,\ms_0):: \dots :: (\opc_m,\ms_m) \wedge \\
      \ms_S \uplus \ms_\stk \uplus \ms_0 \uplus \dots \uplus \ms_m  \text{ is defined} \wedge\\
      W = W_{\var{stack}} \oplus W_{\var{free\_stack}} \oplus W_{\var{heap}} \wedge\\
      \exists \ms_\var{T,stack}, \ms_\var{T,free\_stack}, \ms_\var{T,heap}, \ms_{T,f}, \ms_{S,f}, \ms_S',\overline{\sigma}\ldotp \\
      \quad \ms_S = \ms_{f,S} \uplus \ms_S' \wedge \\
      \quad \ms_T = \ms_\var{T,stack} \uplus \ms_\var{T,free\_stack} \uplus \ms_\var{T,heap} \uplus \ms_{T,f} \wedge \\
      \quad \dom(\ms_{\var{T,stack}} \uplus \var{T,free\_stack}) = [\baddr_\stk,\eaddr_\stk] \wedge \\
      \quad \{\baddr_\stk -1,\eaddr_\stk + 1\} \in \dom(\ms_{T,f}) \wedge \\
      \quad \memSatStack{\stk,\ms_\var{T,stack}}{W_{\var{stack}}} \wedge \\
      \quad \memSatFStack{\ms_\stk,\ms_\var{T,free\_stack}}{W_{\var{free\_stack}}} \wedge \\
      \quad \npair{(\overline{\sigma},\ms_S',\ms_\var{T,heap})} \in \lrheap(\pwheap)(W_{\var{heap}})
    \end{array}
  \right.
\]

\[
  \memSatStack{\stk,\ms_T}{W} \text{ iff } 
  \left\{
    \begin{array}{l}
      W_\var{stack} = \pwpriv \wedge \\
      \stk = (\opc_0,\ms_0), \dots (\opc_m,\ms_m) \wedge \\
      \forall i \in \{0,\dots,m\} \ldotp (\dom(\ms_i) \neq \emptyset \wedge\\
      \quad \forall i < j \ldotp \forall a \in \dom(\ms_i) \ldotp \forall a' \in \dom(\ms_j) \ldotp \stkb < a < a') \wedge\\
      \exists R_\ms : \dom(\activeReg{W_\var{stack}}) \fun \MemSeg \times \Addr \times \MemSeg \ldotp \\
      \quad \ms_T = \biguplus_{r \in \dom(\activeReg{W_\var{stack}})} \pi_3(R_\ms(r)) \wedge \\
      \quad \ms_0 \uplus \dots \uplus \ms_m = \biguplus_{r \in \dom(\activeReg{W_\var{stack}})} \pi_1(R_\ms(r)) \wedge \\
      \quad \exists R_W : \dom(\activeReg{W_\var{stack}}) \fun \World \ldotp \\
      \qquad W = \bigoplus_{r \in \dom(\activeReg{W_\var{stack}})} R_W(r) \wedge \\
      \qquad \forall r \in \dom(\activeReg{W_\var{stack}}), n' < n\ldotp \\
      \qquad \quad \npair[n']{(\pi_1(R_\ms(r)),\pi_3(R_\ms(r))} \in W_\var{stack}(r).H \; \xi^{-1}(R_W(r)) \wedge\\
      \qquad \quad \pi_2(R_\ms(r)) = W_\var{stack}(r).\opc \wedge\\
      \qquad \quad \exists i \ldotp \opc_i = W_\var{stack}(r).\opc \wedge \ms_i = \pi_1 (R_\ms(r))
    \end{array}
  \right.
\]

\[
  \memSatFStack{ms_\stk,\ms_T}{W} \text{ iff } 
  \left\{
    \begin{array}{l}
      W_\var{stack} = \pwfree \wedge \\
      \exists R_\ms : \dom(\activeReg{W_\var{stack}}) \fun \MemSeg \times \MemSeg \wedge \\
      \quad \ms_T = \biguplus_{r \in \dom(\activeReg{W_\var{stack}})} \pi_2(R_\ms(r)) \wedge \\
      \quad \ms_\stk = \biguplus_{r \in \dom(\activeReg{W_\var{stack}})} \pi_1(R_\ms(r)) \wedge \\
      \quad \stkb \in \dom(\ms_T) \wedge \stkb \in \dom(\ms_\stk) \wedge \\
      \quad \exists R_W : \dom(\activeReg{W_\var{stack}}) \fun \World\ldotp\\
      \qquad W = \oplus_{r \in \dom(\activeReg{W_\var{stack}})} R_W(r) \wedge\\
      \qquad \forall r \in \dom(\activeReg{W_\var{stack}}),n' < n \ldotp \\
      \qquad \quad \npair[n']{R_\ms(r)} \in  W_\var{stack}(r).H \; \xi^{-1}(R_W(r))
    \end{array}
  \right.
\]

\[
  \lrheap(\pwheap)(W') = 
  \left\{
    \npair{(\overline{\sigma},ms,\ms_T)} \middle|
    \begin{array}{l}
      \exists R_\ms : \dom(\activeReg{\pwheap}) \fun \MemSeg \times \MemSeg \wedge \\
      \quad \ms_T = \biguplus_{r \in \dom(\activeReg{\pwheap})} \pi_2(R_\ms(r)) \wedge \\
      \quad \ms = \biguplus_{r \in \dom(\activeReg{\pwheap})} \pi_1(R_\ms(r)) \wedge \\
      \quad \exists R_W : \dom(\activeReg{\pwheap}) \fun \World\ldotp\\
      \qquad W' = \oplus_{r \in \dom(\activeReg{\pwheap})} R_W(r) \wedge\\
      \qquad \forall r \in \dom(\activeReg{\pwheap}), n' < n \ldotp \\
      \qquad \quad \npair[n']{R_\ms(r)} \in  \pwheap(r).H \; \xi^{-1}(R_W(r)) \wedge\\
      \exists R_\var{seal} : \dom(\activeReg{\pwheap}) \fun \powerset{\Seal} \wedge\\
      \quad \biguplus_{r \in \dom(\activeReg{\pwheap})} R_\var{seal}(r)) \subseteq \overline{\sigma} \wedge\\
      \quad \dom(\pwheap(r).H_\sigma) = R_\var{seal}(r)
    \end{array}
  \right.
\]


\begin{lemma}[Combined independent heap memory satisfies disjoint world]
  \label{lem:combined-memory-disjoint-world}
  If $\npair{(\overline{\sigma_1},ms_{S,1},\ms_{T,1})}\in\lrheap(\pwheap[W_1])(W)$ and $\npair{(\overline{\sigma_2},ms_{S,2},\ms_{T,2})} \in \lrheap(\pwheap[W_2])(W)$, then $\npair{(\overline{\sigma_1}\uplus\overline{\sigma_2},ms_{S,1}\uplus \ms_{S,2},\ms_{T,1}\uplus \ms_{T,2})} \in \lrheap(\pwheap[W_1 \uplus W_2])(W)$.
\end{lemma}
\begin{proof}
  Unfolding the definitions, it's easy to construct the memory and seal partitions $R_{\ms,3}$ and $R_{\var{seal},3}$ and $R_{W,3}$ from the corresponding partitions of the separate memories, seals and worlds.
\end{proof}

\subsection{Relation}
Two expression relations: one for sealed code-data pairs being jumped to and one for capabilities being jumped to in the regular way.
The argument for having one relation relate pairs of pairs of capabilities and the other relate pairs of capabilities is that that is how xjump and regular jumps work: xjump takes pairs while regular jumps take single capabilities.
\begin{align*}
  \lrexj(W) &= \left\{ \npair{\stpair[.]{v_{c,S},v_{d,S}}{v_{c,T},v_{d,T}}} \middle| 
    \begin{array}{l}
      \forall n' \leq n, \src{\reg_S}, \reg_T, \src{\ms_S}, \ms_T, \src{\ms_\stk}, \src{\stk} \ldotp\\
      \quad \forall W_\lrrs , W_\lrm \ldotp \\
      \qquad\npair[n']{\stpair{\reg}{\reg}} \in \lrr(\{\rdata\}) (W_\lrrs ) \wedge\\
      \qquad\memSat[n']{\stpair[.]{\ms_S,\stk,\ms_\stk}{\ms_T}}{W_\lrm} \wedge \\
      \qquad\Phi_S = \src{(\ms_S,\reg_S,\stk, \ms_\stk)}\wedge\\
      \qquad\Phi_T = (\ms_T,\reg_T) \wedge\\
      \qquad W \oplus W_\lrrs \oplus W_\lrm \text{ is defined }\\
      \qquad\Rightarrow \exists \Phi_S',\Phi_T'\ldotp\\
      \quad\qquad \Phi_S' = \xjumpResult{v_{c,S}}{v_{d,S}}{\Phi_S} \tand\\
      \quad\qquad\Phi_T' = \xjumpResult{v_{c,T}}{v_{d,T}}{\Phi_T}\tand\\
      \quad\qquad\npair[n']{\left(\Phi_S', \Phi_T' \right)}\in \lro
    \end{array}
    \right\}\\
  \lre(W) &= \left\{ \npair{\stpair[.]{v_{c,S}}{v_{c,T}}} \middle| 
    \begin{array}{l}
      \forall n' \leq n, \src{\reg_S}, \reg_T, \src{\ms_S}, \ms_T, \src{\ms_\stk}, \src{\stk} \ldotp\\
      \quad \forall W_\lrrs , W_\lrm \ldotp \\
      \qquad\npair[n']{\stpair{\reg}{\reg}} \in \lrr(W_\lrrs ) \wedge\\
      \qquad\memSat[n']{\stpair[.]{\ms_S,\stk,\ms_\stk}{\ms_T}}{W_\lrm} \\
      \qquad\Phi_S = \src{(\ms_S,\reg_S,\stk, \ms_\stk)}\\
      \qquad \Phi_S' = \Phi_S \updReg{\pcreg}{v_{c,S}}\\
      \qquad\Phi_T = (\ms_T,\reg_T)\\
      \qquad\Phi_T' = \Phi_T\updReg{\pcreg}{v_{c,T}}\\
      \qquad W \oplus W_\lrrs \oplus W_\lrm\\
      \qquad\Rightarrow\npair[n']{\left(\Phi_S', \Phi_T' \right)}\in \lro
    \end{array}
    \right\}
\end{align*}

\lau{09-07-2018: Fix these definitions' formatting like in paper} 
\[
  \lro[\preceq,(\ta,\stkb,\_,\_)] = \left\{ \npair{\left(\array{l}\src{(\ms_S,\reg_S,\stk_S,\ms_{\stk,S})},\\(\ms_T,\reg_T)\endarray\right)} \middle|
    \begin{array}{l}
      \forall i \leq n \ldotp \\
      \quad \src{(\ms_S,\reg_S,\stk_S,\ms_{\stk,S})} \sterm[i]{\ta,\stkb} \\\qquad\Rightarrow (\ms_T,\reg_T) \term\\
    \end{array}
\right\}
\]
\[
  \lro[\succeq,(\ta,\stkb,\_,\_)] = \left\{ \npair{\left(\array{l}\src{(\ms_S,\reg_S,\stk_S,\ms_{\stk,S})},\\(\ms_T,\reg_T)\endarray\right)} \middle|
    \begin{array}{l}
      \forall i \leq n \ldotp \\ 
      \quad (\ms_T ,\reg_T) \term[i] \\\qquad\Rightarrow \src{(\ms_S,\reg_S,\stk_S,\ms_{\stk,S})} \sterm{\ta,\stkb}
    \end{array}
\right\}
\]


\[
  \lrrg{\trust}(R)(W) = \left\{ \npair{\stpair{\reg}{\reg}} \middle|
    \begin{array}{l}
      \exists S : (\RegName \setminus (\{\pcreg \} \cup R))\fun \World \ldotp \\
      \quad W = \bigoplus_{r \in (\RegName\setminus (\{\pcreg,\rdata \} \cup R))} S(r) \wedge \\
      \quad \forall r \in \RegName \setminus (\{\pcreg \} \cup R)\ldotp\\
      \qquad\npair{\stpair[.]{\src{\reg_S(r)}}{\reg_T(r)}} \in \lrvg{\trust}(S(r))
    \end{array}
            \right\}
\]
We write $\lrr(W)$ to mean $\lrr(\emptyset)(W)$. That is, if we do not need to exclude extra registers, then we simply omit that argument.
\lau{RESOLVED check what conditions need to have a ``square'' i.e. need approximation}
\[
  \lrv(W) =
  \begin{array}[t]{l}
    \left\{ \npair{\stpair[.]{i}{i}} \;\middle|\; i \in \ints \right\}\cup \\
    \hspace{-2cm}\left\{ \npair{\left(\arraycolsep=0pt\array{l}\src{\sealed{\sigma,\vsc_S}},\\ \sealed{\sigma,\vsc_T} \endarray\right)} \;\middle| \;
    \begin{array}{l}
      (\isLinear{\src{\vsc_S}} \text{ iff } \isLinear{\vsc_T}) \wedge\\
      \exists r \in \dom(\pwheap), \sigrets,\sigcloss,\mscode \ldotp \pwheap(r) = (\pure,\_,H_\sigma) \tand \\
      \quad H_\sigma \; \sigma \nequal H^\mathrm{code,\square}_\sigma \; \sigrets \; \sigcloss \; \mscode \; \gc \; \sigma \tand \\
      \quad \npair[n']{\stpair[.]{\vsc_S}{\vsc_T}} \in H_\sigma \; \sigma \; \xi^{-1}(W) \text{ for all $n' < n$}\wedge\\
      \quad (\isLinear{\src{\vsc_S}} \Rightarrow \\
      \qquad\forall W' \future W, W_o, n' < n, \npair[n']{\stpair[.]{\vsc_S'}{\vsc'_T}} \in H_\sigma \; \sigma \; \xi^{-1}(W_o) \ldotp \\
      \qquad \quad \npair[n']{\src{\vsc_S},\src{\vsc_S'},\vsc_T,\vsc_T'} \in \lrexj(W'\oplus W_o)) \wedge \\
      \quad (\nonLinear{\src{\vsc_S}} \Rightarrow \\
      \qquad \forall W' \future \purePart{W}, W_o, n' < n, \npair[n']{\stpair[.]{\vsc_S'}{\vsc'_T}} \in H_\sigma \; \sigma \; \xi^{-1}(W_o) \ldotp \\
      \qquad \quad \npair[n']{\src{\vsc_S},\src{\vsc_S'},\vsc_T,\vsc_T'} \in \lrexj(W'\oplus W_o))

    \end{array}
    \right\}\cup\\
    \hspace{-2cm}\left\{ \npair{\left(\arraycolsep=0pt\array{l} \src{\seal{\sigma_\baddr,\sigma_\eaddr,\sigma}},\\ \seal{\sigma_\baddr,\sigma_\eaddr,\sigma} \endarray \right)} 
    \; \middle| \;
    \begin{array}{l}
      [\sigma_\baddr,\sigma_\eaddr] \mathrel{\#} (\gsigrets \cup \gsigcloss) \tand\\
      \forall \sigma' \in [\sigma_\baddr,\sigma_\eaddr] \ldotp \exists r \in \dom(\pwheap) \ldotp \\
      \quad \pwheap(r) = (\pure,\_,H_\sigma) \tand H_\sigma \; \sigma' \nequal (\lrv \circ \xi)
    \end{array}
    \right\} \cup \\
    \hspace{-2cm}\left\{ \npair{\left(\arraycolsep=0pt\array{l} \src{\stkptr{\perm,\baddr,\eaddr,\aaddr}},\\ ((\perm,\linear),\baddr,\eaddr,\aaddr) \endarray \right)} \;\middle|\;
    \begin{array}{l}
      \begin{array}{r l l}
        \perm \not\in \{\rx,\rwx\} \wedge\\
        \perm \in \readAllowed{} &\Rightarrow& \npair{[\baddr,\eaddr]} \in \stackReadCond{W} \wedge \\
        \perm \in \writeAllowed{} &\Rightarrow& \npair{[\baddr,\eaddr]} \in \stackWriteCond{W}
      \end{array}
    \end{array}
    \right\} \cup \\
    \hspace{-2cm}\left\{ \npair{\left(\arraycolsep=0pt\array{l} \src{((\perm,\lin),\baddr,\eaddr,\aaddr)},\\ ((\perm,\lin),\baddr,\eaddr,\aaddr) \endarray \right)} \;\middle|\; 
    \begin{array}{l}
      [b,e] \mathrel{\#} \ta \tand\\
      \begin{array}{r l l }
        \perm \in \readAllowed{} &\Rightarrow& \npair{[\baddr,\eaddr]} \in \readCond{\lin,W} \wedge\\
        \perm \in \writeAllowed{} &\Rightarrow& \npair{[\baddr,\eaddr]} \in \writeCond{\lin,W} \wedge\\
        \perm \neq \rwx \wedge \\
        \perm = \rx &\Rightarrow& \array[t]{l}\npair{[\baddr,\eaddr]} \in \execCond{W} \wedge\\
        \npair{[\baddr,\eaddr]} \in \xReadCond{W} \wedge \\
                                  \lin = \normal \\ \endarray
      \end{array}
    \end{array}
    \right\}
  \end{array}
\]

\[
  \lrvtrusted[\square,\gc](W) =
  \begin{array}[t]{l}
    \lrv(W)\cup \\
    \left\{ \npair{\left(\arraycolsep=0pt\array{l} \src{\seal{\sigma_\baddr,\sigma_\eaddr,\sigma}},\\ \seal{\sigma_\baddr,\sigma_\eaddr,\sigma} \endarray \right)} 
    \; \middle| \;
    \begin{array}{l}
      \gc = (\ta,\stkb,\gsigrets,\gsigcloss)  \wedge \\
      \exists r \in \dom(\pwheap) \ldotp \\
      \quad \pwheap(r) \nequal \codereg{\sigrets,\sigcloss,\code,(\ta,\stkb,\gsigrets,\gsigcloss)} \wedge\\
      \quad \dom(\code) \subseteq \ta \wedge [\sigma_\baddr,\sigma_\eaddr] \subseteq (\sigrets\cup\sigcloss) \wedge\\
      \quad \sigrets \subseteq \gsigrets \wedge \sigcloss \subseteq \gsigcloss
    \end{array}
    \right\} \cup \\
    \left\{ \npair{\left(\arraycolsep=0pt\array{l} \src{((\perm,\normal),\baddr,\eaddr,\aaddr)},\\ ((\perm,\normal),\baddr,\eaddr,\aaddr) \endarray \right)} \;\middle|\; 
    \begin{array}{l}
      \perm \sqsubseteq \rx \wedge \\
      \gc = (\ta,\stkb,\gsigrets,\gsigcloss)  \wedge \\
      {} [\baddr,\eaddr] \subseteq \ta \wedge\\
      \npair{[\baddr,\eaddr]} \in \xReadCond[\square,\gc]{W} 
    \end{array}
    \right\}
  \end{array}
\]

\begin{lemma}[Untrusted is trusted]
  \label{lem:untrusted-supset-trust}
  \begin{itemize}
  \item $\lrvtrusted(W) \supseteq \lrv(W)$
  \item $\lrrtrusted(W) \supseteq \lrr(W)$
  \end{itemize}
\end{lemma}
\begin{proof}
  Follows easily by definition.
\end{proof}

Note: the case for sub-$\rx$ capabilities in the trusted value relation allows for pointers to trusted code blocks.
Such pointers will not satisfy the read condition, which requires the standard region, which is defined in terms of the untrusted value relation.
Trusted code blocks contain trusted seal capabilities which do not satisfy the untrusted code relation.
An alternative might be to introduce a trust parameter to the readcondition and standard region and merge the case with the regular case for rx capabilities.

\subsection{Permission based conditions}
\lau{28-02-2018: These definitions need to be updated to fit with the new worlds}.
\[
  \addressable{\lin,W} =
  \begin{cases}
    \{ r \mid W(r) = (\pure,\_) \} & \text{if $\lin = \normal$} \\
    \{ r \mid W(r) = (\spatialo,\_) \}  & \text{otherwise (i.e. $\lin = \linear$)} \\
  \end{cases}
\]

\[
  \readCond{\lin,W} = \left\{ \npair{A} \middle| 
    \begin{array}{l}
      \exists S \subseteq \addressable{\lin, \pwheap} \ldotp \\
      \quad \exists R : S \fun \powerset{\nats} \ldotp\\
      \qquad \biguplus_{r\in S} R(r) \supseteq A \wedge\\
      \qquad (\lin = \linear \Rightarrow \forall r\ldotp |R(r)|  = 1) \wedge\\
      \qquad \forall r \in S \ldotp \pwheap(r).H \nsubeq \stdreg{R(r),\gc}{\pur}.H
    \end{array}
  \right\}
\]
\lau{RESOLVED: Make this condition require a set of adjacent regions. Check whether needed for the stack condition.}
where $\iota_A$ is a standard region defined in Section~\ref{sec:standard-regions}.
\[
  \stackReadCond{W} = \left\{ \npair{A} \middle| 
    \begin{array}{l}
      \exists S \subseteq \addressable{\linear, \pwfree} \ldotp \\
      \quad \exists R : S \fun \powerset{\nats} \ldotp\\
      \qquad \forall r \in S \ldotp |R(r)| = 1\\
      \qquad \biguplus_{r\in S} R(r) \supseteq A \wedge \\
      \qquad \forall r \in S \ldotp \pwfree(r).H \nsubeq \stdreg{R(r),\gc}{\pur}.H
    \end{array}
  \right\}
\]
where $\iota_A$ is a standard region defined in Section~\ref{sec:standard-regions}.
\[
  \xReadCond{W} = \left\{ \npair{A} \middle| 
    \begin{array}{l}
      \exists r \in \addressable{\normal, \pwheap} \ldotp \\
      \qquad \pwheap(r) \nequal \codereg{\_,\_,\code,\gc}\wedge\\
      \qquad \dom(\code) \supseteq A 
    \end{array}
  \right\}
\]

\begin{definition}
  \label{def:address-stratified}
  We say that a region $\iota = (\_,H,\_)$ is address stratified iff
  \[
    \begin{array}{l}
      \forall n, \src{\ms_S},\ms_T,\src{\ms_S'},\ms_T',s,\hat{W}\ldotp \\
      \quad \npair{\stpair{\ms}{\ms}}, \npair{\stpair[.]{\ms_S'}{\ms_T'}} \in H \; \hat{W} \wedge \\
      \quad \dom(\src{\ms_S}) = \dom(\ms_T) = \dom(\src{\ms_S'}) = \dom(\ms_T') \\
      \quad \Rightarrow \\
      \qquad \forall \aaddr \in \dom(\ms_S) \ldotp \npair{(\src{\ms_S}\update{\aaddr}{\src{\ms_S'}(\aaddr)},\ms_T\update{\aaddr}{\ms_T'(\aaddr)})} \in H \; \hat{W}
    \end{array}
  \]
\end{definition}

\[
  \writeCond{\lin,W} = \left\{ \npair{A} \middle| 
    \begin{array}{l}
      \exists S \subseteq \addressable{\lin, \pwheap} \ldotp \\
      \quad \exists R : S \fun \powerset{\nats} \\
      \qquad \biguplus_{r\in S} R(r) \supseteq A \wedge\\
      \qquad (\lin = \linear \Rightarrow \forall r\ldotp |R(r)|  = 1) \wedge\\
      \qquad \forall r \in S \ldotp \pwheap(r).H \nsupeq \stdreg{R(r),\gc}{\pur}.H \wedge\\
      \qquad \quad \pwheap(r) \text{ is address-stratified}
    \end{array}
  \right\}
\]
where $\iota_A$ is a standard region defined in Section~\ref{sec:standard-regions}.
\[
  \stackWriteCond{W} = \left\{ \npair{A}) \middle| 
    \begin{array}{l}
      \exists S \subseteq \addressable{\linear, \pwfree} \ldotp \\
      \quad \exists R : S \fun \powerset{\nats} \\
      \qquad \biguplus_{r\in S} R(r) \supseteq A \wedge \\
      \qquad \forall r \in S \ldotp |R(r)| = 1 \wedge \\
      \qquad \forall r \in S \ldotp \pwfree(r).H \nsupeq \stdreg{R(r),\gc}{\pur}.H \wedge\\
      \qquad \quad \pwfree(r) \text{ is address-stratified}
    \end{array}
  \right\}
\]
where $\iota_A$ is a standard region defined in Section~\ref{sec:standard-regions}.
\dominique{22-6-2018: address-stratified is useless for length-one regions.}
Note: this new version of execCond uses the expression relation for regular jumps since it expresses the validity of jumping to an address pointed to by an executable capability. 
\begin{align*}
  \execCond{W} &=
  \left\{ \npair{A} \middle|
    \begin{array}{l}
      \forall n' < n, W' \future \purePart{W}, \aaddr \in [\baddr',\eaddr'] \subseteq A \ldotp\\
      \quad \npair[n']{\stpair[.]{((\rx,\normal),\baddr',\eaddr',\aaddr)}{((\rx,\normal),\baddr',\eaddr',\aaddr)}} \in \lre(W')
    \end{array}
    \right\}
\end{align*}

\subsection{Standard regions}
\label{sec:standard-regions}
Standard region:
\[
  \stdreg{A,\gc}{v} \defeq (v,H_A^{\mathrm{std},\square} \; \gc) , v \in \{\spa,\spao\}
\]
for readability, we use $\spao$ short for $\spatialo$, $\spa$ short for $\spatial$, and $\pur$ as short for $\pure$.
\[
  \stdreg{A,\gc}{\pur} \defeq (\pur,H_A^{\mathrm{std},\square} \; \gc, \lambda \_ \; \_ \ldotp \emptyset)
\]
where $H^{\mathrm{std},\square}_A$ is defined as follows:
\[
  H_A^{\mathrm{std},\square} \; \gc \; \hat{W} \defeq \left\{ \npair{\ms_S,\ms_T} \middle|
    \begin{array}{l}
      \dom(\ms_S) = \dom(\ms_T) = A \wedge \\
      \exists S : A \fun \World \ldotp \xi(\hat{W}) = \oplus_{\aaddr \in A} S(\aaddr) \wedge\\
      \quad \forall \aaddr \in A \ldotp \npair{(\ms_S(\aaddr),\ms_T(\aaddr))} \in \lrv(S(\aaddr))
    \end{array}
  \right\}
\]

\[
  \stareg[\stpair{\ms}{\ms},\gc]{v,\square} = (v,H^\mathrm{sta,\square}_{\stpair{\ms}{\ms}} \; \gc)
\]
\[
  H^\mathrm{sta,\square}_{\stpair{\ms}{\ms}} \; \gc = \left\{ \npair{\stpair{\ms}{\ms}} \middle| 
    \begin{array}{l}
      \exists S : \dom(\ms) \fun \World \ldotp \xi(W) = \oplus_{\aaddr \in \dom(\ms)} S(\aaddr) \wedge\\
      \quad \forall \aaddr \in \dom(\ms) \ldotp \npair{(\ms_S(\aaddr),\ms_T(\aaddr))} \in \lrv(S(\aaddr))
    \end{array}
\right\}
\]

\[
  \codereg{\sigrets,\sigcloss,\code,\gc} \defeq (\pure,
H^\mathrm{code,\square} \; \sigrets \; \sigcloss \; \code \; \gc,
H^\mathrm{code}_\sigma \; \sigrets \; \sigcloss \; \code \; \gc)
\]

\begin{multline*}
  H^\mathrm{code} \; \sigrets \; \sigcloss \; \code \; (\ta,\_,\gsigrets,\gsigcloss) \; \hat{W} =\\
  \left\{\npair{\arraycolsep=0pt\left(\array{l}\code \uplus \mspad,\\ \code \uplus \mspad\endarray\right)} \middle|
    \begin{array}{l}
    \dom(\code) = [\baddr,\eaddr] \wedge \\
      ([\baddr - 1, \eaddr + 1] \subseteq \ta \wedge \sigrets \subseteq \gsigrets \wedge \sigcloss \subseteq \gsigcloss \wedge \trust = \trusted) \vee \\
      \quad ([\baddr-1,\eaddr+1]\mathrel{\#} \ta \wedge \sigrets = \emptyset \wedge \trust =\untrusted) \wedge \\
      \mspad = [\baddr-1 \mapsto 0] \uplus [\eaddr + 1 \mapsto 0]\wedge\\
      \sigrets,\sigcloss,\ta \vdash_{\mathrm{comp-code}} \code \wedge\\
      \forall a \in \dom(\code)\ldotp\\
      \quad\npair{(\code(a),\code(a))} \in \lrvg{\trust}(\purePart{\xi(\hat{W})})
    \end{array}
  \right\}
\end{multline*}

\begin{multline*}
  H^\mathrm{code,\square}_\sigma \; \sigrets \; \sigcloss \; \code \;
  (\ta,\stkb) \; \sigma \; \hat{W}= \\
  \begin{array}[t]{l}
\left\{
    \begin{array}{l}
\left. \npair{\arraycolsep=0pt\left(\array{l}\retptrc(\baddr,\eaddr,\aaddr'+\calllen),\\((\rx,\normal),\baddr,\eaddr,\aaddr)\endarray\right)} \middle| \right. \\
      \begin{array}{l}
        \sigrets \subseteq \gsigrets \tand \\
        \dom(\code) \subseteq \ta \tand\\
        \decInstr{\code([\aaddr',\aaddr' + \calllen-1])} = \overline{\scall{\offpc,\offsigma}{r_1}{r_2}} \tand \\
        \aaddr = \aaddr' + \retoffset \tand \\
        \code(\aaddr'+\offpc) = \seal{\sigma_b,\sigma_e,\sigma_b} \tand \sigma = \sigma_b + \offsigma \in \sigrets \tand\\
        \lbrack \aaddr',\aaddr' + \calllen -1 \rbrack \subseteq \lbrack \baddr, \eaddr \rbrack
      \end{array}
    \end{array}
      \right\} \uplus \\
\left\{
    \begin{array}{l}
\left. \npair{\arraycolsep=0pt\left(\array{l}\retptrd(\baddr,\eaddr),\\((\rw,\linear),\baddr,\eaddr,\baddr-1)\endarray\right)} \middle| \right. \\
      \begin{array}{l}
        \sigrets \subseteq \gsigrets \tand \\
        \dom(\code) \subseteq \ta \tand\\
        \exists r \in \addressable{\linear,\pwpriv[\xi(\hat{W})]} \ldotp \pwpriv[\xi(\hat{W})](r).H \nequal (\stareg[(\ms_S,\ms_T),(\ta,\stkb)]{\spao,\square}, \aaddr'+\calllen) \tand \\
        \quad \dom(\ms_S) = \dom(\ms_T) = [\baddr,\eaddr] \tand\\
        \quad \decInstr{\code([\aaddr',\aaddr' + \calllen-1])} = \overline{\scall{\offpc,\offsigma}{r_1}{r_2}} \tand \\
        \quad \code(\aaddr'+\offpc) = \seal{\sigma_b,\sigma_e,\sigma_b} \tand \sigma = \sigma_b + \offsigma \in \sigrets
      \end{array}
    \end{array}
    \right\} \\
  \end{array}\\
    \text{ if } \sigma \in \sigrets\\
\end{multline*}
and
\begin{multline*}
  H^\mathrm{code,\square}_\sigma \; \sigrets \; \sigcloss \; \code \;
  (\ta,\stkb) \; \sigma \; \hat{W}= \\
  \left\{
    \begin{array}{l}
\left. \npair{(\vsc, \vsc' )} \middle| \right. \\
      \begin{array}{l}
        (\dom(\code) \mathrel{\#} \ta \tand \npair{(\vsc,\vsc')} \in \lrv \; \xi(\hat{W})) \tor\\
        (\dom(\code) \subseteq \ta \tand \sigcloss \subseteq \gsigcloss \tand \sigrets \subseteq \gsigrets \tand \\
         \quad((\exec{\vsc} \wedge \npair{(\vsc,\vsc')} \in \lrvtrusted \; \xi(\hat{W})) \vee\\
         \quad\ (\nonExec{\vsc} \wedge\npair{(\vsc,\vsc')} \in \lrv \; \xi(\hat{W}))))
      \end{array}
    \end{array}
  \right\}\\
  \text{ if } \sigma \in \sigcloss
\end{multline*}

\subsection{Reasonable Components}
\label{sec:reasonability}

Take a set of trusted addresses $\ta$ and sets of return pointer and closure seals $\gsigrets$ and $\sigcloss$.
We define that a word $w$ is reasonable up to $n$ steps in memory $\ms$ and
free stack $\ms_\stk$ if $n=0$
or the following implications hold.
\begin{definition}[Reasonable word]
  \label{def:reasonable-word}
  \begin{itemize}
  \item If $w = \seal{\sigma_\baddr,\sigma_\eaddr,\_}$, then
    $[\sigma_\baddr,\sigma_\eaddr] \mathrel{\#} (\gsigrets \cup \gsigcloss)$
  \item If $w = ((\perm,\_),\baddr,\eaddr,\_)$, then $[\baddr,\eaddr] \# \dom(\ta)$
  \item If $w = \sealed{\sigma,\vsc}$ and $\sigma \not\in (\gsigrets\cup
    \gsigcloss)$ then $\vsc$ is reasonable up to $n - 1$ steps.
  \item If $w = ((\perm,\_),\baddr,\eaddr,\_)$ and $\perm \in \readAllowed{}$
    and $n > 0$, then $\ms(\aaddr)$ is reasonable up to $n - 1$ steps for all
    $\aaddr \in ([\baddr,\eaddr] \setminus \ta)$
  \item If $w = \stkptr{\perm,\baddr,\eaddr,\_}$ and $\perm \in \readAllowed{}$
    and $n > 0$, then $\ms_\stk(\aaddr)$ is reasonable up to $n - 1$ steps for
    all $\aaddr \in [\baddr,\eaddr]$
  \end{itemize}
\end{definition}

\begin{definition}[Reasonable configuration]
  \label{def:reasonable-conf}
  We say that an execution configuration $\Phi$ is reasonable up to $n$ steps with $(\ta,\stkb,\gsigrets,\gsigcloss)$ 
  iff for $n' \leq n$:
  \begin{itemize}
  \item \emph{Guarantee stack base address before call} If
    \begin{itemize}
    \item $\Phi$ points to $\src{\scall{\offpc,\offsigma}{r_1}{r_2}}$ in $\ta$
      for some $\src{r_1}$ and $\src{r_2}$
    \end{itemize}
    Then all of the following hold:
    \begin{itemize}
    \item $\src{\Phi}(\src{r_\stk}) =
      \src{\stkptr{\_,\stkb,\_,\_}}$
    \item $r_1 \neq \rtmp{1}$
    \item $n' = 0$ or $\Phi(\pcreg) + \calllen$ behaves reasonably up to $n'-1$ steps
    \end{itemize}
  \item \emph{Use return seals only for calls, use closure seals appropriately} If
    \begin{itemize}
    \item $\src{\Phi}$ points to $\src{\tcseal{r_1}{r_2}}$ in $\ta$ and $\Phi(r_2) = \seal{\sigma_\baddr,\sigma_\eaddr,\sigma}$
    \end{itemize}
    Then one of the following holds:
    \begin{itemize}
    \item $\src{\Phi}$ is inside $\src{\scall{\offpc,\offsigma}{r_1'}{r_2'}}$ and $\sigma \in \gsigrets$
    \item $\sigma \in \gsigcloss$ and one of the following holds:
      \begin{itemize}
      \item $\exec{\Phi(r_1)}$ and $n' = 0$ or $\Phi(r_1)$ behaves reasonably up to $n' - 1$ steps.
      \item $\nonExec{\Phi(r_1)}$ and $n' = 0$ or $\Phi(r_1)$ is reasonable up to $n' - 1$ steps in memory $\Phi.\ms$ and free stack $\Phi.\ms_\stk$.
      \end{itemize}
    \end{itemize}
  \item \emph{Don't store private stuff...} If
    \begin{itemize}
    \item $\src{\Phi}$ points to $\src{\tstore{r_1}{r_2}}$ in $\ta$, then
    \end{itemize}
    Then $n' = 0$ or $\Phi.\reg(r_2)$ is reasonable in memory $\Phi.\mem$ up to $n' -1$ steps.
  \item \emph{Don't leak private stuff...} If
    \begin{itemize}
    \item $\Phi \step[\ta,\stkb] \Phi'$
    \end{itemize}
    Then one of the following holds:
    \begin{itemize}
    \item All of the following hold:
      \begin{itemize}
      \item $\Phi'.\reg(\pcreg) =
        ((\perm,\lin),\baddr,\eaddr,\aaddr')$ and $\Phi.\reg(\pcreg) =
        ((\perm,\lin),\baddr,\eaddr,\aaddr)$
      \item $\Phi$ does not point to $\src{\txjmp{r_1}{r_2}}$ for some $\src{r_1}$ and $\src{r_2}$
      \item $\Phi$ does not point to $\src{\scall{\offpc,\offsigma}{r_1}{r_2}}$ for some $\src{r_1}$ and $\src{r_2}$, $\offpc$, $\offsigma$
      \item $n' = 0$ or $\Phi'$ is reasonable up to $n'-1$ steps
      \end{itemize}
    \item
      \begin{itemize}
      \item $\Phi$ points to $\src{\scall{\offpc,\offsigma}{r_1}{r_2}}$ for some $\src{r_1}$ and $\src{r_2}$
      \item $n' = 0$ or $\Phi.\reg(r)$ is reasonable in memory $\Phi.\mem$ and free stack $\Phi.\ms_\stk$ up to $n'-1$ steps for all $r \neq \pcreg$
      \end{itemize}
    \item
      \begin{itemize}
      \item $\Phi$ points to $\src{\txjmp{r_1}{r_2}}$ for some $\src{r_1}$ and $\src{r_2}$
      \item $n' = 0$ or $\Phi.\reg(r)$ is reasonable in memory $\Phi.\mem$ and free stack $\Phi.\ms_\stk$ up to $n'-1$ steps for all $r \neq \pcreg$
      \end{itemize}
    \end{itemize}
  \end{itemize}
\end{definition}

  \begin{lemma}
    \label{lem:ec-reasonable-downwards-closed}
    For all $n' \le n$ if
    \begin{itemize}
    \item $\Phi$ is reasonable up to $n$ steps
    \end{itemize}
    Then
    \begin{itemize}
    \item $\Phi$ is reasonable up to $n'$ steps
    \end{itemize}
  \end{lemma}
  \begin{proof}
    Follows from the definition.
  \end{proof}

  \begin{definition}[Reasonable pc]
    \label{def:reasonable-pc}
    We say that an executable capability $c =
    ((\perm,\normal),\baddr,\eaddr,\aaddr)$ behaves reasonably up to $n$ steps
    if for any $\Phi$ such that
    \begin{itemize}
    \item $\Phi.\reg(\pcreg) = c$
    \item $\Phi.\reg(r)$ is reasonable up to $n$ steps in memory $\Phi.\mem$ and
      free stack $\Phi.\ms_\stk$ for all $r \neq \pcreg$
    \item $\Phi.\mem$, $\Phi.\ms_\stk$ and $\Phi.\stk$ are all disjoint
    \end{itemize}
    We have that $\Phi$ is reasonable up to $n$ steps.
  \end{definition}

  \begin{definition}[Reasonable component]
    \label{def:reasonable-component}
    We say that a component
    $(\mscode,\msdata,\overline{\var{import}},\overline{\var{export}},\sigrets,\sigcloss,A_\linear)$
    is reasonable if the following hold: For all $(s \mapsto \sealed{\sigma,\vsc}) \in
    \overline{c_{\mathrm{export}}}$, with $\exec{\vsc}$, we have that $\vsc$
    behaves reasonably up to any number of steps $n$.

    We say that a component
    $(\var{comp}_0,c_{\mathrm{main},c}, c_{\mathrm{main},d})$ is reasonable if $\var{comp}_0$ is reasonable.
  \end{definition}

\begin{lemma}
  \label{lem:code-reg-and-mem-sat}
  If $\pwheap(r_\code) \nequal \codereg{\sigrets,\sigcloss,\code,\gc}$ and
  $\memSat{\ms_S,\stk,\ms_\stk,\ms_T}{W}$, then there exists $W''$
  and $W'$ such that $W = W' \oplus W''$ and
  $\npair[n']{(\code,\code)} \in H^\mathrm{code} \; \sigrets \; \sigcloss \; \code \; \gc \; \xi^{-1}(W')$ for all $n' < n$
\end{lemma}
\begin{proof}
  By definition of $\memSat{\ms_S,\stk,\ms_\stk,\ms_T}{W}$, we get that $W = W_{\var{stack}} \oplus W_{\var{free\_stack}} \oplus W_{\var{heap}}$ and
  \begin{equation*}
    \npair{(\overline{\sigma},\ms_S',\ms_\var{T,heap})} \in
    \lrheap(\pwheap)(W_{\var{heap}})
  \end{equation*}
  for some $\ms_S' \subseteq \ms_S$ and $\ms_{T,\var{heap}}\subseteq \ms_T$.

  By definition of $\npair{(\overline{\sigma},\ms_S',\ms_\var{T,heap})} \in
  \lrheap(\pwheap)(W_{\var{heap}})$, we get an
  $R_\ms : \dom(\activeReg{\pwheap}) \fun \MemSeg \times \MemSeg$ and $\exists R_W : \dom(\activeReg{\pwheap}) \fun \World$ such that
  \begin{itemize}
  \item $W' = \oplus_{r \in \dom(\activeReg{\pwheap})} R_W(r)$ and
  \item $\npair[n']{R_\ms(r_{\var{code}})} \in \pwheap(r_{\var{code}}).H \; \xi^{-1}(R_W(r_{\var{code}}))$ for all $n' < n$
  \end{itemize}

  Since $\pwheap(r_\code) \nequal \codereg{\sigrets,\sigcloss,\code,\gc}$, this implies that also $\npair[n']{R_\ms(r_{\var{code}})} \in  \codereg{\sigrets,\sigcloss,\code,\gc}.H~ \xi^{-1}(R_W(r_{\var{code}}))$, i.e. $\npair[n']{R_\ms(r_{\var{code}})} \in  H^\mathrm{code} \; \sigrets \; \sigcloss \; \code \; \gc~ \xi^{-1}(R_W(r_{\var{code}}))$ as required.
\end{proof}

\begin{lemma}[Untrusted source values are reasonable]
  \label{lem:untrusted-source-values-are-reasonable}
  If
  \begin{itemize}
  \item $\gc = (\ta,\stkb,\gsigrets,\gsigcloss)$
  \item $\npair{(w,\_)} \in \lrv(W_w)$
  \item $\memSat{\ms_S,\stk,\ms_\stk,\_}{W_M}$
  \item $\purePart{W_w} \oplus \purePart{W_M}$ is defined
  \end{itemize}
  then, with respect to $\ta,\gsigrets,\gsigcloss$, $w$ is reasonable up to $n$ steps in memory $\ms_S$ and free stack $\ms_\stk$.
\end{lemma}
\begin{proof}
  Induction over $n$.
  For $n=0$ all words are reasonable.
  For $n>0$, we need to prove four implications:
  \begin{itemize}
  \item If $w = \seal{\sigma_\baddr,\sigma_\eaddr,\_}$, then $[\sigma_\baddr,\sigma_\eaddr] \mathrel{\#} (\gsigrets \cup \gsigcloss)$:

    This follows directly from $\npair{(w,\_)} \in \lrv(W_w)$ by definition of $\lrv$.

  \item If $w = ((\perm,\_),\baddr,\eaddr,\_)$, then $[\baddr,\eaddr] \# \ta$:

    This follows directly from $\npair{(w,\_)} \in \lrv(W_w)$ by definition of $\lrv$.

  \item If $w = \sealed{\sigma,\vsc}$ and $\sigma \not\in (\gsigrets\cup \gsigcloss)$ then $\vsc$ is reasonable for $n-1$ steps:

    By $\npair{(w,\_)} \in \lrv(W_w)$, we get some region $r \in \dom(\pwheap)$ and a set of return seals, a set of closure seals and a code memory: $\sigrets,\sigcloss,\mscode$ such that $\pwheap(r) = (\pure,\_,H_\sigma)$ and $H_\sigma \nequal H^\mathrm{code,\square}_\sigma \; \sigrets \; \sigcloss \; \mscode \; \gc)$ and $\npair[n']{(\sealed{\sigma,\vsc},\_)} \in H_\sigma \; \sigma \; \xi^{-1}(W)$ for all $n' < n$, so in particular for $n' = n-1$.
    It follows easily from the above that also $\npair[n-1]{(\vsc,\_)} \in H^\mathrm{code,\square}_\sigma \; \sigrets \; \sigcloss \; \mscode \; \gc \; \sigma \; \xi^{-1}(W)$.

    This gives us three cases for $\vsc$.
    The two first cases $\vsc = \retptrd(\_,\_)$ and $\vsc = \retptrc(\_,\_,\_)$ are easily discharged as the reasonability definition puts no requirements on return pointers.
    For the final case, $\sigma \in \sigcloss$ and either
    \begin{itemize}
    \item $\dom(\mscode) \mathrel{\#} \ta \tand \npair[n-1]{(\vsc,\_)} \in \lrv (W)$; or
    \item $\dom(\mscode) \subseteq \ta$, $\sigcloss \subseteq \gsigcloss$, and $\npair{(\vsc,\_)} \in \lrvg{\trust}(W)$ for $\trust = \untrusted$ iff $\nonExec{\vsc}$.
    \end{itemize}
    In the first case, the result follows from the induction hypothesis.
    In the second case, we have a contradiction with $\sigma \not\in (\gsigrets\cup \gsigcloss)$.

  \item If $w = ((\perm,\_),\baddr,\eaddr,\_)$ and $\perm \in \readAllowed{}$ and $n > 0$, then $\ms_S(\aaddr)$ is reasonable up to $n - 1$ steps for all $\aaddr \in ([\baddr,\eaddr] \setminus \ta)$:

    From $\npair{(w,\_)} \in \lrv(W_w)$, we get $\npair[n]{[b,e]} \in \readCond{\perm,\lin}$.

    From $\npair[n]{[b,e]} \in \readCond{\perm,\lin}$ we get $S\subseteq \addressable{\lin,\pwheap}$ and $R : S \fun \powerset{\nats}$ such that
    \begin{itemize}
    \item $\biguplus_{r\in S} R(r) \supseteq [b,e]$
    \item $(\lin = \linear \Rightarrow \forall r\ldotp |R(r)|  = 1)$
    \item $\forall r \in S \ldotp \pwheap(r).H \nsubeq \stdreg{R(r),\gc}{\pur}.H$
    \end{itemize}
    given $a \in [b,e]$ we know by the above that there exists $r$ such that $a
    \in R(r)$.

    By $\memSat{\ms_S,\stk,\ms_\stk,\_}{W_M}$, we get $R_\ms : \dom(\activeReg{\pwheap}) \fun \MemSeg \times \MemSeg$ such that $\ms_S = \biguplus_{r \in \dom(\activeReg{\pwheap})} \pi_1(R_\ms(r))$.
    
    Further, we get $R_W : \dom(\activeReg{\pwheap}) \fun \World$ such that $W' = \oplus_{r \in \dom(\activeReg{\pwheap})} R_W(r)$ and $\forall r \in \dom(\activeReg{\pwheap}) \ldotp \npair{R_\ms(r)} \in \pwheap(r).H \; \xi^{-1}(R_W(r))$.
    For $W_M = W' \oplus W''$ for some $W''$.

    In particular, we have $\npair{R_\ms(r)} \in \pwheap(r).H \; \xi^{-1}(R_W(r))$.
    
    We know $\pwheap(r).H \nsubeq \stdreg{R(r),\gc}{\pur}.H$, so $\npair[n']{R_\ms(r)} \in \stdreg{R(r),\gc}{\pur}.H(\xi^{-1}(R_W(r)))$ for $n' < n$.
    This gives us $\npair[n']{(\pi_1(R_\ms(r))(a),\_)} \in \lrv(R_W'(a))$ where $R_W' : [b,e] \fun \World$ such that $\biguplus_{a \in \dom(\pi_1(R_\ms(r)))} R_W'(a) = R_W(r)$.

    At this point we apply the induction hypothesis which is possible as we have
    the following:
    \begin{itemize}
    \item $\npair[n-1]{(\pi_1(R_\ms(r))(a),\_)} \in \lrv(R_W'(a))$

      We get this by the above.
    \item $\memSat[n-1]{\ms_S,\stk,\ms_\stk,\_}{W_M}$

      This follows by assumption and Lemma~\ref{lem:downwards-closed}
    \item $\purePart{R_W'(a)} \oplus \purePart{W_M}$

      Which follows by definition of $\purePart{}$ and the fact that $R_W'(a)$ is $W_M$ with part of its ownership, but $\purePart{}$ strips away the ownership making the two compatible.
    \end{itemize}
    which gives us that, with respect to $\ta,\gsigrets,\gsigcloss$, $\pi_1(R_\ms(r))(a)$ is reasonable up to $n-1$ steps in memory $\ms_S$ and free stack $\ms_\stk$.
    Which is what we wanted as $\pi_1(R_\ms(r))(a) = \ms_S(a)$.

  \item If $w = \stkptr{\perm,\baddr,\eaddr,\_}$ and $\perm \in \readAllowed{}$ and $n > 0$, then $\ms_\stk(\aaddr)$ is reasonable up to $n - 1$ steps for all $\aaddr \in [\baddr,\eaddr]$:

    This case is proven in the same way the previous case was.
    The main difference is that the free part of the world is used instead of the heap part:

    From $\npair{(w,\_)} \in \lrv(W_w)$ and $\perm \in \readAllowed{}$, we get $\npair{[b,e]} \in \stackReadCond{\perm}$.

    From $\npair{[b,e]} \in \stackReadCond{\perm}$ we get $S\subseteq
    \addressable{\linear,\pwfree}$ and $R : S \fun \powerset{\nats}$ such that
    \begin{itemize}
    \item $\biguplus_{r\in S} R(r) \supseteq [b,e]$
    \item $\forall r\ldotp |R(r)|  = 1$
    \item $\forall r \in S \ldotp \pwheap(r).H \nsubeq \stdreg{R(r),\gc}{\pur}.H$
    \end{itemize}
    given $a \in [b,e]$ we know by the above that there exists $r$ such that $a \in R(r)$.

    By $\memSat{\ms_S,\stk,\ms_\stk,\_}{W_M}$, we get $R_\ms : \dom(\activeReg{\pwfree}) \fun \MemSeg \times \MemSeg$ such that $\ms_\stk = \biguplus_{r \in \dom(\activeReg{\pwfree})} \pi_1(R_\ms(r))$.
    
    Further, we get $R_W : \dom(\activeReg{\pwfree}) \fun \World$ such that $W' = \oplus_{r \in \dom(\activeReg{\pwfree})} R_W(r)$ and $\forall r \in \dom(\activeReg{\pwfree}) \ldotp \npair{R_\ms(r)} \in \pwfree(r).H \; \xi^{-1}(R_W(r))$.
    For $W_M = W' \oplus W''$ for some $W''$.

    In particular, we have $\npair{R_\ms(r)} \in \pwfree(r).H \; \xi^{-1}(R_W(r))$.
    
    We know $\pwfree(r).H \nsubeq \stdreg{R(r),\gc}{\pur}.H$, so $\npair[n']{R_\ms(r)} \in \stdreg{R(r),\gc}{\pur}.H(\xi^{-1}(R_W(r)))$ for $n' < n$.
    This gives us $\npair[n']{(\pi_1(R_\ms(r))(a),\_)} \in \lrv(R_W'(a))$ where $R_W' : [b,e] \fun \World$ such that $\biguplus_{a \in \dom(\pi_1(R_\ms(r)))} R_W'(a) = R_W(r)$.

    At this point we apply the induction hypothesis which is possible as we have
    the following:
    \begin{itemize}
    \item $\npair[n-1]{(\pi_1(R_\ms(r))(a),\_)} \in \lrv(R_W'(a))$

      We get this by the above.
    \item $\memSat[n-1]{\ms_S,\stk,\ms_\stk,\_}{W_M}$

      This follows by assumption and Lemma~\ref{lem:downwards-closed}
    \item $\purePart{R_W'(a)} \oplus \purePart{W_M}$

      Which follows by definition of $\purePart{}$ and the fact that $R_W'(a)$ is $W_M$ with part of its ownership, but $\purePart{}$ strips away the ownership making the two compatible.
    \end{itemize}
    which gives us that, with respect to $\ta,\gsigrets,\gsigcloss$, $\pi_1(R_\ms(r))(a)$ is reasonable up to $n-1$ steps in memory $\ms_S$ and free stack $\ms_\stk$.
    Which is what we wanted as $\pi_1(R_\ms(r))(a) = \ms_\stk(a)$.
  \end{itemize}
\end{proof}

\begin{lemma}[Untrusted register files are reasonable]
  If
  \begin{itemize}
  \item $\gc = (\ta,\stkb,\gsigrets,\gsigcloss)$
  \item $\npair{(\reg_S,\_)} \in \lrr(W_w)$
  \item $\memSat{\ms_S,\stk,\ms_\stk,\_}{W_M}$
  \item $\purePart{W_w} \oplus \purePart{W_M}$ is defined
  \end{itemize}
  then, with respect to $\ta,\gsigrets,\gsigcloss$,
  \begin{itemize}
  \item $\reg_S(r)$ is reasonable up to $n$ steps in memory $\ms_S$ and free stack $\ms_\stk$ for all $r \not\in \{\pcreg\}$
  \item $\ms_S$, $\ms_\stk$ and $\stk$ are all disjoint
  \end{itemize}
\end{lemma}
\begin{proof}
  We prove the two results separately:
  \begin{itemize}
  \item $\reg_S(r)$ is reasonable up to $n$ steps in memory $\ms_S$ and free stack $\ms_\stk$ for all $r \not\in \{\pcreg\}$:

    We have that $W_w = \bigoplus_{r \in \RegName \setminus \{\pcreg\}} S(r)$ and $\npair{(\reg_S(r),\reg_T(r))} \in \lrv(S(r))$ for all $r \not\in \{\pcreg\}$.

    For all $r \neq \pcreg$, we have that $\purePart{S(r)} \oplus \purePart{W_M}$ is defined by Lemmas~\ref{lem:purePart-oplus} and~\ref{lem:oplus-assoc-comm}.

    The result then follows directly from Lemma~\ref{lem:untrusted-source-values-are-reasonable}.

  \item $\ms_S$, $\ms_\stk$ and $\stk$ are all disjoint:
    
    This follows from $\memSat{\ms_S,\stk,\ms_\stk,\_}{W_M}$.
  \end{itemize}
\end{proof}

\begin{lemma}[Reasonable things don't need to be trusted]
  \label{lem:trusted-and-reasonable-is-untrusted}
  If
  \begin{itemize}
  \item $\npair{(w,w')} \in \lrvg{\trusted}(W_w)$
  \item $w$ is reasonable up to $n$ steps in memory $\ms_S$ and free stack $\ms_\stk$, with respect to $\code,\ta,\gsigrets,\gsigcloss$.
  \item $n > 0$
  \item Theorem~\ref{thm:ftlr} holds up to $n$ steps.
    \dominique{Note: I mention this assumption explicitly to avoid cyclic reasoning in the FTLR proof.}
  \end{itemize}
  Then
  \begin{itemize}
  \item $\npair{(w,w')} \in \lrv(W_w)$
  \end{itemize}
\end{lemma}
\begin{proof}
  Assume that $\npair{(w,w_T)} \in \lrvg{\trusted}(W_w)$

  By definition of $\lrvg{\trusted}$, it is clear that also $\npair{(w,w_T)} \in \lrv(W_w)$, except in the following two cases:
  \begin{itemize}
  \item
    \begin{itemize}
    \item $w = \src{\seal{\sigma,\sigma_\baddr,\sigma_\eaddr}}$
    \item $w_T = \seal{\sigma,\sigma_\baddr,\sigma_\eaddr}$ 
    \item $r \in \dom(\pwheap)$
    \item $\pwheap(r) \nequal \codereg{\sigrets,\sigcloss,\code,(\ta,\stkb,\gsigrets,\gsigcloss)}$
    \item $\dom(\code) \subseteq \ta$ 
    \item $[\sigma_\baddr,\sigma_\eaddr] \subseteq (\sigrets\cup\sigcloss)$ 
    \item $\sigrets \subseteq \gsigrets$
    \item $\sigcloss \subseteq \gsigcloss$
    \end{itemize}

    By definition of reasonability, we get immediately that $[\sigma_b,\sigma_e]\mathrel{\#} (\gsigrets\cup \gsigcloss)$, contradicting the last three facts above.
  \item 
    \begin{itemize}
    \item $w = \src{((\perm,\normal),\baddr,\eaddr,\aaddr)}$, 
    \item $w_T =
      ((\perm,\normal),\baddr,\eaddr,\aaddr)$
    \item $\perm \sqsubseteq \rx$
    \item $\gc = (\ta,\stkb,\gsigrets,\gsigcloss)$
    \item $\npair{[\baddr,\eaddr]} \in \xReadCond[\square,\gc]{W}$
    \end{itemize}

    If $\perm = \noperm$, then $\npair{(w,w_T)} \in \lrv(W_w)$ follows directly.
    
    If $\perm \neq \noperm$, then by definition of reasonability, we get immediately that $[b,e] \mathrel{\#} \ta$.

    By definition of $\lrv$, it suffices to show that $\npair{[b,e]} \in \readCond{\normal,W}$ and $\npair{[b,e]} \in \execCond{\normal,W}$.
    The former follows directly from Lemma~\ref{lem:xReadCond-outside-ta-implies-readCond}.
    The latter follows by Theorem~\ref{thm:ftlr}, Lemmas~\ref{lem:conds-shrinkable} and~\ref{lem:monotonicity} and definition of $\execCond{}$.
  \end{itemize}
\end{proof}

\subsection{Fundamental Theorem of Logical Relations}
\begin{theorem}[FTLR]
  \label{thm:ftlr}
  For all $n,W,\lin,\baddr,\eaddr,\aaddr$,
  If
  \begin{itemize}
  \item $\npair{[\baddr,\eaddr]} \in \xReadCond{W}$
  \end{itemize}
  and one of the following sets of requirements holds:
  \begin{itemize}
  \item \begin{itemize}
    \item $[\baddr,\eaddr] \subseteq \ta$
    \item $({((\rx,\normal),\baddr,\eaddr,\aaddr)}$ behaves reasonably up to $n$ steps.
    \end{itemize}
  \item
    \begin{itemize}
    \item $[\baddr,\eaddr] \mathrel{\#} \ta$
    \end{itemize}
  \end{itemize}
  Then
  \[
    \npair{\left(\arraycolsep=0pt\array{l}((\rx,\normal),\baddr,\eaddr,\aaddr),\\
      ((\rx,\normal),\baddr,\eaddr,\aaddr)\endarray\right)} \in \lre(W)
  \]
\end{theorem}
Note: we don't require the readcondition in the trusted case because trusted code pointers point to trusted code blocks which may require seals for trusted seals which are not in the untrusted value relation.

\begin{lemma}[Untrusted environments produce safe closures]
  \label{lem:lre-implies-lrexj}
  If
  \begin{itemize}
  \item $\npair{(w_1,w_1')} \in \lre(W_1)$ or ($\nonExec{w_1}$ and $\nonExec{w_1'}$)
  \item $w_1 \neq \retptrc(\_)$
  \item $\npair{(w_2,w_2')} \in \lrv(W_2)$ or ($\exec{w_2}$ and $\exec{w_2'}$)
  \end{itemize}
  Then
  \begin{itemize}
  \item $\npair{((w_1,w_2),(w_1',w_2'))} \in \lrexj(W_1 \oplus W_2)$
  \end{itemize}
\end{lemma}
\begin{proof}
  Take $n' \leq n$, $\src{\reg_S}, \reg_T, \src{\ms_S}, \ms_T, \src{\ms_\stk}, \src{\stk}$, $W_\lrrs , W_\lrm$ and assume
  \begin{itemize}
  \item $\npair[n']{\stpair{\reg}{\reg}} \in \lrr(\{\rdata\}) (W_\lrrs )$
  \item $\memSat[n']{\stpair[.]{\ms_S,\stk,\ms_\stk}{\ms_T}}{W_\lrm}$
  \item $\Phi_S = \src{(\ms_S,\reg_S,\stk, \ms_\stk)}$
  \item $\Phi_T = (\ms_T,\reg_T)$
  \item $W_1 \oplus W_2 \oplus W_\lrrs \oplus W_\lrm$ is defined 
  \end{itemize}
  Then we need to prove that for $\Phi_S'$,$\Phi_T'$
  \begin{itemize}
  \item $\Phi_S' = \xjumpResult{w_1}{w_2}{\Phi_S}$
  \item $\Phi_T' = \xjumpResult{w_1'}{w_2'}{\Phi_T}$
  \item $\npair[n']{\left(\Phi_S', \Phi_T' \right)}\in \lro$
  \end{itemize}

  Using the fact that $w_1 \neq \retptrc(\_)$, we know by definition of $\xjumpResult{w_1}{w_2}{\Phi_S}$ that we must be in one of the following two cases:
  \begin{itemize}
  \item $w_1 \neq \retptrc(\_)$, $w_2 \neq \retptrd(\_)$, $\nonExec{w_2}$ and
    $\Phi_S' = \Phi_S \updReg{\pcreg}{w_1}\updReg{\rdata}{w_2}$.

    From $\nonExec{w_2}$ and $\npair{(w_2,w_2')} \in \lrv(W_2)$ or $\exec{w_2}$, it follows that also $\nonExec{w_2'}$ and $\Phi_T' = \Phi_T \updReg{\pcreg}{w_1'}\updReg{\rdata}{w_2'}$

    We can now combine $\npair[n']{\stpair{\reg}{\reg}} \in \lrr(\{\rdata\}) (W_\lrrs )$ and $\npair{(w_2,w_2')} \in \lrv(W_2)$ into $\npair[n']{(\reg_S\update{\rdata}{w_2},\reg_T\update{\rdata}{w_2'})} \in \lrr (W_\lrrs \oplus W_2)$, using Lemmas~\ref{lem:monotonicity},~\ref{lem:oplus-future},~\ref{lem:oplus-assoc-comm}, and~\ref{lem:downwards-closed}.

    With our other assumptions above, $\npair{(w_1,w_1')} \in \lre(W_1)$ then gives us that $\npair[n']{(\Phi_S',\Phi_T')}\in \lro$, as required.
    On the other hand, if $\nonExec{w_1}$ and $\nonExec{w_1'}$, then $\Phi_S' \step[\gc] \failed$ and $\Phi_T' \step \failed$ and the result follows by definition of $\lrol$ and $\lror$.

  \item Otherwise: $\Phi_S' = \failed$.
    
    From $w_1 \neq \retptrc(\_)$ and ($\npair{(w_2,w_2')} \in \lrv(W_2)$ or ($\exec{w_2}$ and $\exec{w_2'}$)), the only way we can get to this case is that $\exec{w_2}$ and $\exec{w_2'}$.

    It follows that $\Phi_T' = \failed$.
    
    The result follows by definition of $\lrol$ and $\lror$.
  \end{itemize}
\end{proof}

\begin{lemma}[Safe values are safe to execute]
  \label{lem:safe-vals-safe-exec}
  If
  \begin{itemize}
  \item $\npair{(w_1,w_1')} \in \lrv(W_1)$
  \item $\npair[n']{(w_2,w_2')} \in \lrv(W_2)$ or ($\exec{w_2}$ and $\exec{w_2'}$)
  \item $n' < n$
  \end{itemize}
  Then
  \begin{itemize}
  \item $\npair[n']{((w_1,w_2),(w_1',w_2'))} \in \lrexj(W_1 \oplus W_2)$
  \end{itemize}
\end{lemma}
\begin{proof}
  By Lemma~\ref{lem:lre-implies-lrexj}, it suffices to prove the following: 
  \begin{itemize}
  \item $\npair[n']{(w_1,w_1')} \in \lre(W_1)$ or ($\nonExec{w_1}$ and $\nonExec{w_1'}$)
  \item $w_1 \neq \retptrc(\_)$.
  \end{itemize}

  The latter follows immediately from $\npair{(w_1,w_1')} \in \lrv(W_1)$ by definition.
  It also follows that either ($\nonExec{w_1}$ and $\nonExec{w_1'}$) or $w_1 = w_1' = ((\perm,\normal),\baddr,\eaddr,\aaddr)$ and $\npair{[\baddr,\eaddr]} \in \execCond{W_1}$.
  In the latter case, it follows by definition of $\execCond{}$ that $\npair[n']{(w_1,w_1')} \in \lre(W_1)$ as required.
\end{proof}

\subsection{Related components}
\label{sec:related-components}

\begin{align*}
  \lrcomp(W) &=
  \left\{\begin{aligned}
     &\npair{\var{comp},\var{comp}} \;\mid \;\\
      &\qquad\var{comp} = (\mscode,\msdata,\overline{a_{\mathrm{import}} \mapsfrom s_{\mathrm{import}}},\overline{s_{\mathrm{export}} \mapsto w_{\mathrm{export}}},\sigrets,\sigcloss) \tand \\
      &\qquad\text{For all } W' \future W \ldotp \\
      &\qquad\quad\text{If } \overline{\npair[n']{(w_{\mathrm{import}},w_{\mathrm{import}})}} \in \lrv(\purePart{W'}) \text{ for all $n' < n$}\\
      &\qquad\quad\text{and } \msdata' = \msdata{}[\overline{a_{\mathrm{import}} \mapsto w_{\mathrm{import}}}] \\
      &\qquad\quad\text{then } \npair{(\sigrets\uplus\sigcloss,\mscode\uplus \msdata', \mscode\uplus\msdata')} \in \lrheap(\pwheap)(W') \tand\\
      &\qquad\quad\overline{\npair{(w_{\mathrm{export}},w_{\mathrm{export}})}} \in \lrv(\purePart{W'})
  \end{aligned}
    \right\}\\
  &\cup \left\{
    \begin{multlined}
     \npair{(\var{comp}_0,c_{\mathrm{main},c}, c_{\mathrm{main},d}),(\var{comp}_0,c_{\mathrm{main},c}, c_{\mathrm{main},d})} \;\mid \;\\
     \npair{(\var{comp}_0,\var{comp}_0)} \in \lrcomp(W) \tand\\
     \{(\_ \mapsto c_{\mathrm{main},c}),(\_ \mapsto c_{\mathrm{main},d})\} \subseteq \overline{w_{\mathrm{export}}}
       \end{multlined}
    \right\} 
\end{align*}

\begin{lemma}[Compatibility lemma for linking]
  \label{lem:compat-linking}
  If
  \begin{itemize}
  \item $\npair{(\var{comp}_1,\var{comp}_1)} \in \lrcomp(W_1)$
  \item $\npair{(\var{comp}_2,\var{comp}_2)} \in \lrcomp(W_2)$
  \end{itemize}
  then
  \begin{itemize}
  \item $\npair{(\var{comp}_1 \bowtie \var{comp}_2,\var{comp}_1\bowtie\var{comp}_2)} \in \lrcomp(W_1 \uplus W_2)$.
  \end{itemize}
\end{lemma}
\begin{proof}
  First consider the case where one of the two pairs of components have a pair of ``main'' capabilities.
  By the definition of $\lrcomp$, these capabilities have to be part of that components' export list.
  It is then easy to see that the other component cannot have ``main'' capabilities (otherwise linking is not defined) and it is sufficient to prove the result for the underlying components (without the ``main'' capabilities).
  Therefore, we can restrict ourselves to the case where both pairs of components are of the form $\var{comp}_0$, i.e. no ``main'' capabilities.

  By definition of $\bowtie$, we have that the following hold:
  \begin{itemize}
  \item $\var{comp}_1 = (\mscode[1], \msdata[1], \overline{\var{import}_1}, \overline{\var{export}_1}, \sigrets[1], \sigcloss[1],A_{\linear,1})$
  \item $\var{comp}_2 = (\mscode[2], \msdata[2], \overline{\var{import}_2}, \overline{\var{export}_2}, \sigrets[2], \sigcloss[2],A_{\linear,2})$
  \item $\var{comp}_3 = (\mscode[3], \msdata[3], \overline{\var{import}_3}, \overline{\var{export}_3}, \sigrets[3], \sigcloss[3],A_{\linear,3})$
  \item $\mscode[3] = \mscode[1] \uplus \mscode[2]$
  \item $\msdata[3] = (\msdata[1] \uplus \msdata[2])[a \mapsto w \mid (a \mapsfrom s) \in (\overline{\var{import}_1} \cup \overline{\var{import}_2}), (s \mapsto w) \in \overline{\var{export}}_3]$
  \item $\overline{\var{export}_3} = \overline{\var{export}_1} \cup \overline{\var{export}_2}$
  \item $\overline{\var{import}_3} = \{ a \mapsfrom s \in (\overline{\var{import}_1} \cup \overline{\var{import}_2}) \mid s \mapsto \_ \not\in \overline{\var{export}_3} \}$
  \item $\sigrets[3] = \sigrets[1] \uplus \sigrets[2]$
  \item $\sigcloss[3] = \sigcloss[1] \uplus \sigcloss[2]$
  \item $A_{\linear,3} = A_{\linear,1} \uplus A_{\linear,2}$
  \item $\dom(\mscode[3]) \mathrel{\#} \dom(\msdata[3])$
  \item $\sigrets[3] \mathrel{\#} \sigcloss[3]$
  \end{itemize}

  Now take $W' \future (W_1\uplus W_2)$ and assume that $\overline{\npair[n']{(w_{\mathrm{import},3},w_{\mathrm{import},3})}} \in \lrv(\purePart{W'})$ for all $n' < n$.
  Take $\overline{\var{import}_3} = \overline{a_{\mathrm{import},3} \mapsfrom s_{\mathrm{import},3}}$.
  Take $\msdata[3]' = \msdata[3]{}[\overline{a_{\mathrm{import},3} \mapsto w_{\mathrm{import}}}]$
  Then it remains to show that
  \begin{equation*}
    \npair{(\sigrets[3]\uplus\sigcloss[3],\mscode[3]\uplus \msdata[3]', \mscode[3]\uplus\msdata[3]')} \in \lrheap(\pwheap[(W_1\uplus W_2)])(W')
  \end{equation*}
  and
  \begin{equation*}
    \overline{\npair{(c_{\mathrm{export},3},c_{\mathrm{export},3})}} \in \lrv(\purePart{W'})
  \end{equation*}
  We do this by complete induction on $n$, so that we can assume that $\overline{\npair[n']{(c_{\mathrm{export},3},c_{\mathrm{export},3})}} \in \lrv(\purePart{W'})$ for all $n' < n$.

  First, note that
  \begin{align*}
    \msdata[3]' &= \msdata[3]{}[\overline{a_{\mathrm{import},3} \mapsto w_{\mathrm{import}}}]\\
                &=(\msdata[1] \uplus \msdata[2])[a \mapsto w \mid (a \mapsfrom s) \in (\overline{\var{import}_1} \cup \overline{\var{import}_2}), (s \mapsto w) \in \overline{\var{export}}_3]\\
                  &\hspace{4cm}[\overline{a_{\mathrm{import},3} \mapsto w_{\mathrm{import}}}]\\
                &= \msdata[1][a \mapsto w \mid (a \mapsfrom s) \in \overline{\var{import}_1}, (s \mapsto w) \in \overline{\var{export}}_2][\overline{a_{\mathrm{import},3} \mapsto w_{\mathrm{import},3}} \mid \overline{a_{\mathrm{import},3}\mapsfrom \_} \in \overline{\var{import}_1}]\\
                &\uplus \msdata[2][a \mapsto w \mid (a \mapsfrom s) \in \overline{\var{import}_2}, (s \mapsto w) \in \overline{\var{export}}_1][\overline{a_{\mathrm{import},3} \mapsto w_{\mathrm{import},3}} \mid \overline{a_{\mathrm{import},3}\mapsfrom \_} \in \overline{\var{import}_2}]\\
                &= \msdata[1]' \uplus \msdata[2]'
  \end{align*}
  
  First, we prove that for all substituted values $w$ in the equations above, we have that
  \begin{equation*}
    \overline{\npair[n']{(w,w)}} \in \lrv(\purePart{W'})
  \end{equation*}
  for all $n' < n$.
  We know by assumption that this is true for the $\overline{w_{\mathrm{import}}}$ such that $\overline{a_{\mathrm{import},3}\mapsfrom \_} \in \overline{\var{import}_1}$ or $\overline{a_{\mathrm{import},3}\mapsfrom \_} \in \overline{\var{import}_2}$.
  On the other hand, we know by induction that this is true for the $w$ such that $(s \mapsto w) \in \overline{\var{export}}_1$ or $(s \mapsto w) \in \overline{\var{export}}_2$.
  
  Then, it follows from our assumptions ($\npair{(\var{comp}_1,\var{comp}_1)} \in \lrcomp(W_1)$ and $\npair{(\var{comp}_2,\var{comp}_2)} \in \lrcomp(W_2)$), and using Lemma~\ref{lem:uplus-future} that
  \begin{align*}
    \npair{(\sigrets[1]\uplus\sigcloss[1],\mscode[1]\uplus \msdata[1]', \mscode[1]\uplus\msdata[1]')} &\in \lrheap(\pwheap[W_1])(W')\\
    \npair{(\sigrets[2]\uplus\sigcloss[2],\mscode[2]\uplus \msdata[2]', \mscode[2]\uplus\msdata[2]')} &\in \lrheap(\pwheap[W_2])(W')
  \end{align*}
  It follows by Lemma~\ref{lem:combined-memory-disjoint-world} that 
  \begin{equation*}
    \npair{(\sigrets[3]\uplus\sigcloss[3],\mscode[3]\uplus \msdata[3]', \mscode[3]\uplus\msdata[3]')} \in \lrheap(\pwheap[(W_1\uplus W_2)])(W')
  \end{equation*}
  
  Finally, for each
  \begin{equation*}
    \npair{(c_{\mathrm{export}},c_{\mathrm{export}})} \in \overline{\npair{(c_{\mathrm{export},3},c_{\mathrm{export},3})}}
  \end{equation*}
  we have that $\npair{(c_{\mathrm{export}},c_{\mathrm{export}})}$ is either in $\overline{\npair{(c_{\mathrm{export},1},c_{\mathrm{export},1})}}$ or $\overline{\npair{(c_{\mathrm{export},2},c_{\mathrm{export},2})}}$.
  By unfolding $\npair{(\var{comp}_1,\var{comp}_1)} \in \lrcomp(W_1)$ and $\npair{(\var{comp}_2,\var{comp}_2)} \in \lrcomp(W_2)$ and using Lemma~\ref{lem:uplus-future}, this follows from the results above.
\end{proof}

\subsection{FTLR for components}
\label{sec:ftlr-for-comps}

\begin{lemma}[Untrusted code regions' seals specify value safety]
  \label{lem:untrusted-codereg-sealed-vals-safe}
  If
  \begin{itemize}
  \item $\dom(\code) \mathrel{\#} \ta$
  \item $\sigma \in \sigcloss$
  \end{itemize}
  Then
  \begin{equation*}
    H^\mathrm{code,\square}_\sigma \; \sigrets \; \sigcloss \; \code \; (\ta,\stkb) \; \sigma \nequal \lrv
  \end{equation*}
\end{lemma}
\begin{proof}
  Follows easily by definition.
\end{proof}

\begin{lemma}[Code region for untrusted components stronger than standard safe memory region]
  \label{lem:codereg-untrusted-stdreg}
  If
  \begin{itemize}
  \item $\dom(\mscode) \mathrel{\#} \ta$
  \end{itemize}
  then $\codereg{\sigrets,\sigcloss,\mscode,\gc}.H \nsubeq \stdreg{\dom(\mscode),\gc}{\pur}.H$ 
\end{lemma}
\begin{proof}
  Take $\hat{W}$, then we need to prove that
  \begin{equation*}
    H^\mathrm{code}\;\sigrets\;\sigcloss\;\mscode\;\gc \; \hat{W} \nsubeq H_{\dom(\mscode)}^{\mathrm{std},\square}\;\hat{W}
  \end{equation*}
  So take $\npair{\ms_S,\ms_T} \in H^\mathrm{code}\;\sigrets\;\sigcloss\;\mscode\;\gc \; \hat{W}$, then we know that
  \begin{itemize}
  \item $\ms_S = \ms_T = \mscode\uplus\mspad$
  \item $\dom(\mscode) =[\baddr,\eaddr]$
  \item $[\baddr-1,\eaddr+1] \subseteq \ta \vee ([\baddr-1,\eaddr+1]\mathrel{\#} \ta \wedge \sigrets = \emptyset)$, but we know that $\dom(\mscode) \mathrel{\#} \ta$, so we get that $\sigrets = \emptyset$.
  \item For all $a \in \dom(\code)$, $\npair{(\code(a),\code(a))} \in \lrv(\purePart{\xi(\hat{W})})$
  \end{itemize}

  By definition of $H_{\dom(\mscode)}^{\mathrm{std},\square}\;\hat{W}$, we need to prove that
  \begin{itemize}
  \item $\dom(\mscode) = \dom(\mscode) = \dom(\mscode)$
  \item there exists a $S : \dom(\mscode) \fun \World$ with $\xi(\hat{W}) = \oplus_{\aaddr\in\dom(\mscode)} S(\aaddr)$ and for all $\aaddr \in \dom(\mscode)$, we have that $\npair{(\ms_S(\aaddr),\ms_T(\aaddr))} \in \lrv(S(\aaddr))$
  \end{itemize}

  We take $S$ to map every address to $\purePart{\hat{W}}$, except for a single one that we map to $\hat{W}$.
  The result then follows by Lemmas~\ref{lem:purePart-duplicable}, \ref{lem:oplus-future} and \ref{lem:monotonicity} and the above fact that for all $a \in \dom(\code)$, $\npair{(\code(a),\code(a))} \in \lrv(\purePart{\xi(\hat{W})})$.
\end{proof}

\begin{lemma}[Code memory for untrusted components safely readable]
  \label{lem:code-memory-untrusted-readable}
  If
  \begin{itemize}
  \item $\dom(\mscode) \mathrel{\#} \ta$
  \item $\pwheap(r_{\mathrm{code}}) = \codereg{\sigrets,\sigcloss,\mscode,\gc}$
  \item $A \subseteq \dom(\mscode)$,
  \end{itemize}
  we have that $\npair{A} \in \readCond{\normal,W}$.
\end{lemma}
\begin{proof}
  Follows by defintion of $\readCond{}$ from Lemma~\ref{lem:codereg-untrusted-stdreg}.
\end{proof}

\begin{lemma}[Code memory for untrusted components safely readable]
  \label{lem:xReadCond-outside-ta-implies-readCond}
  If
  \begin{itemize}
  \item $A \mathrel{\#} \ta$
  \item $\npair{A} \in \xReadCond{W}$
  \end{itemize}
  we have that $\npair{A} \in \readCond{\normal,W}$.
\end{lemma}
\begin{proof}
  Follows by defintion of $\readCond{}$ and $\xReadCond{}$ from Lemma~\ref{lem:codereg-untrusted-stdreg}.
\end{proof}

\begin{lemma}[Code region sealing invariant for untrusted components implies untrusted safety]
  \label{lem:codereg-sealing-safety}
  If
  \begin{itemize}
  \item $\emptyset\neq\dom(\mscode) \mathrel{\#} \ta$
  \item $\sigma \in \sigcloss$
  \end{itemize}
  then $\codereg{\sigrets,\sigcloss,\mscode,\gc}.H_\sigma \sigma \nequal \lrv$ for all $n$
\end{lemma}
\begin{proof}
  The result follows easily by definition of $\codereg{}$ and $H^\mathrm{code,\square}_\sigma$.
\end{proof}

\begin{lemma}[non-linear words are pure]
  \label{lem:non-linear-pure}
  \hfill
  \begin{itemize}
  \item If $\npair{A} \in \readCond{\normal,W}$, then $\npair{A} \in \readCond{}(\normal, \purePart{W} )$.
  \item If $\npair{A} \in \writeCond{\normal,W}$, then $\npair{A} \in \writeCond{}(\normal,\purePart{W})$
  \item If $\npair{A} \in \execCond{W}$, then $\npair{A} \in \execCond{}(\purePart{W})$
  \item If $\npair{A} \in \xReadCond{W}$, then $\npair{A} \in \xReadCond{}(\purePart{W} )$.
  \item If $\npair{(w_1,w_2)} \in \lrv(W)$ and ($\nonLinear{w_1}$ or
    $\nonLinear{w_2}$), then
    $\npair{(w_1,w_2)} \in \lrv(\purePart{W})$.
  \end{itemize}
\end{lemma}
\begin{proof}
  Follows easily by inspecting the definitions of $\lrv$, $\readCond[]{}$, $\addressable{}$, $\writeCond[]{}$, $\execCond[]{}$ and $\xReadCond[]{}$ and using Lemma~\ref{lem:purePart-idempotent}.
\end{proof}

\begin{lemma}[permission-based conditions shrinkable]
  \label{lem:conds-shrinkable}
  \hfill
  \begin{itemize}
  \item If $\npair{A} \in \readCond{\lin,W}$  and $\emptyset \neq A' \subseteq A$, then
    $\npair{A'} \in \readCond{\lin,W}$.
  \item If $\npair{A} \in \writeCond{\lin,W}$  and $\emptyset \neq A' \subseteq A$, then
    $\npair{A'} \in \writeCond{\lin,W}$.
  \item If $\npair{A} \in \xReadCond{W}$  and $A' \subseteq A$, then
    $\npair{A'} \in \xReadCond{W}$.
  \item If $\npair{A} \in \execCond{W}$  and $\emptyset \neq A' \subseteq A$, then
    $\npair{A'} \in \execCond{W}$.
  \end{itemize}
\end{lemma}
\begin{proof}
  Follows easily from the definitions.
  In the case for $\execCond{}$ where $\lin = \linear$, we have a partition of the world into $W = \oplus_{a \in A} W_a(a)$ that we need to convert into a partition $W = \oplus_{a \in A'} W_a'(a)$.
  It suffices construct new $W_a'$ by taking $W_a$ but adding $\oplus_{a\in(A\setminus A')}$ to an $W_a(a)$ for an arbitrary $a \in A'$ to make this work.
\end{proof}

\begin{lemma}[permission-based conditions splittable]
  \label{lem:conds-splittable}
  If $A = [\baddr_1,\eaddr_1] \uplus [\baddr_2,\eaddr_2]$, then
  \begin{itemize}
  \item If $\npair{A} \in \readCond{\lin,W}$, then there exists $W_1, W_2$ such that $W = W_1\oplus W_2$ such that $\npair{[\baddr_1,\eaddr_1]} \in \readCond{\lin,W_1}$ and $\npair{[\baddr_2,\eaddr_2]} \in \readCond{\lin,W_2}$.
  \item If $\npair{A} \in \writeCond{\lin,W}$, then there exists $W_1, W_2$ such that $W = W_1\oplus W_2$ such that $\npair{[\baddr_1,\eaddr_1]} \in \writeCond{\lin,W_1}$ and $\npair{[\baddr_2,\eaddr_2]} \in \writeCond{\lin,W_2}$.
  \item If $\npair{A} \in \execCond{W}$, then there exists $W_1, W_2$ such that $W = W_1\oplus W_2$ such that $\npair{[\baddr_1,\eaddr_1]} \in \execCond{W_1}$ and $\npair{[\baddr_2,\eaddr_2]} \in \execCond{W_2}$.
  \end{itemize}
\end{lemma}
\begin{proof}
  If $\lin = \normal$, then Lemma~\ref{lem:non-linear-pure} tells us that it suffices to consider pure worlds $W = \purePart{W}$ and Lemma~\ref{lem:purePart-duplicable} tells us that then $W = W \oplus W$.
  The results then follow from the previous Lemma~\ref{lem:conds-shrinkable}.
  
  As such, we can limit ourselves to the case where $\lin = \linear$.

  The read- and write-conditions then require separate islands for each individual address in $A$ and we can distribute ownership of those islands according to whether those addresses are in $[\baddr_1,\eaddr_1]$ or $[\baddr_2,\eaddr_2]$, i.e. make the island $\pwheap(r)$ with $R(r) = \{a\}$ $\spatialo$ in $W_1$ and $\spatial$ in $W_2$ if $a \in [\baddr_1,\eaddr_1]$ and vice versa.
  It is then easy to check that our results hold.

  For the execute condition, we get a partition of the world into $W = \oplus_{a \in A} W_a(a)$ and we can define $W_1 = \oplus_{a\in[\baddr_1,\eaddr_1]} W_a(a)$ and likewise for $W_2$.
  The results then follow easily from the definitions.
\end{proof}

\begin{lemma}[permission-based conditions splicable]
  \label{lem:conds-splicable}
  If
  \begin{itemize}
  \item $[\baddr,\eaddr] = [\baddr_1,\eaddr_1] \uplus [\baddr_2,\eaddr_2]$
  \item $[\baddr,\eaddr]\mathrel{\#} \ta$
  \item $W_1\oplus W_2$ is defined.
  \end{itemize}
  then
  \begin{itemize}
  \item If $\npair{[\baddr_1,\eaddr_1]} \in \readCond{\lin,W_1}$ and $\npair{[\baddr_2,\eaddr_2]} \in \readCond{\lin,W_2}$, then $\npair{[\baddr,\eaddr]} \in \readCond{\lin,W_1 \oplus W_2}$
  \item If $\npair{[\baddr_1,\eaddr_1]} \in \writeCond{\lin,W_1}$ and $\npair{[\baddr_2,\eaddr_2]} \in \writeCond{\lin,W_2}$, then $\npair{[\baddr,\eaddr]} \in \writeCond{\lin,W_1 \oplus W_2}$.
  \item If $\npair{[\baddr_1,\eaddr_1]} \in \execCond{W_1}$ and $\npair{[\baddr_2,\eaddr_2]} \in \execCond{W_2}$, then $\npair{[\baddr,\eaddr]} \in \execCond{W_1 \oplus W_2}$.
  \end{itemize}

  If
  \begin{itemize}
  \item $[\baddr,\eaddr] = [\baddr_1,\eaddr_1] \uplus [\baddr_2,\eaddr_2]$
  \item $W_1\oplus W_2 \oplus W_M$ is defined.
  \item $\memSat{\ms_S,\stk,\ms_\stk,\ms_T}{W_M}$
  \end{itemize}
  then
  \begin{itemize}
  \item If $\npair{[\baddr_1,\eaddr_1]} \in \xReadCond{\lin,W_1}$ and
    $\npair{[\baddr_2,\eaddr_2]} \in \xReadCond{\lin,W_2}$, then
    $\npair{[\baddr,\eaddr]} \in \xReadCond{\lin,W_1 \oplus W_2}$
  \end{itemize}

  If $[\baddr,\eaddr] = [\baddr_1,\eaddr_1] \uplus [\baddr_2,\eaddr_2]$, then
  \begin{itemize}
  \item If $\npair{[\baddr_1,\eaddr_1]} \in \stackReadCond{W_1}$ and $\npair{[\baddr_2,\eaddr_2]} \in \stackReadCond{W_2}$, then $\npair{[\baddr,\eaddr]} \in \stackReadCond{W_1 \oplus W_2}$
  \item If $\npair{[\baddr_1,\eaddr_1]} \in \stackWriteCond{W_1}$ and $\npair{[\baddr_2,\eaddr_2]} \in \stackWriteCond{W_2}$, then $\npair{[\baddr,\eaddr]} \in \stackWriteCond{W_1 \oplus W_2}$.
  \end{itemize}
\end{lemma}
\begin{proof}
  The results for $\readCond{}$, $\writeCond{}$, $\stackReadCond{}$ and $\stackWriteCond{}$ follow by taking the union of the two sets $S$, and the union of $R$ for every $r$.

  The result for $\execCond{}$ follows easily by definition and using Lemmas~\ref{lem:purePart-oplus} and~\ref{lem:oplus-future}.

  For $\xReadCond{}$, we get $r_i \in \addressable{\normal,W_i}$ such that $\pwheap[W_i](r_i) \nequal \codereg{\_,\_,\code_i,\gc}$ such that $\dom(\code_i)\supseteq [\baddr_i,\eaddr_i]$.
  From $\memSat{\ms_S,\stk,\ms_\stk,\ms_T}{W_M}$, it follows that $\code_i = [\baddr_i',\eaddr_i']$ and if $r_1 \neq r_2$, then $[\baddr_1'-1,\eaddr_1+1] \mathrel{\#} [\baddr_2'-1,\eaddr_2'+1]$.
  Because $[\baddr,\eaddr] = [\baddr_1,\eaddr_1] \uplus [\baddr_2,\eaddr_2]$ and $[\baddr_i',\eaddr_i'] = \dom(\code_i)\supseteq [\baddr_i,\eaddr_i]$, we can derive that $r_1= r_2$.
  The result then follows by definition of $\xReadCond{}$.
\end{proof}

\begin{lemma}[FTLR for component code capabilities]
  \label{lem:ftlr-component-code-ptrs}
  If
  \begin{itemize}
  \item $\npair{\dom(\mscode)}\in\xReadCond{W}$
  \item ($\trust = \untrusted$ and $\dom(\mscode \mathrel{\#} \ta)$) or ($\trust = \trusted$ and $\dom(\mscode) \subseteq \ta$)
  \item $w = ((\rx,\normal),\baddr,\eaddr,\aaddr)$
  \item $[\baddr,\eaddr]\subseteq \dom(\mscode)$
  \end{itemize}

  then $\npair{(w,w)} \in \lrvg{\trust}(W)$.
\end{lemma}
\begin{proof}
  We distinguish two cases:
  \begin{itemize}
  \item $\trust = \untrusted$ and $\dom(\mscode) \mathrel{\#} \ta$:
    By definition of $\lrv(W)$, it suffices to show that:
    \begin{itemize}
    \item $[b,e] \mathrel{\#} \ta$: follows directly from $[\baddr,\eaddr]\subseteq \dom(\mscode)$ and $\dom(\mscode) \mathrel{\#} \ta$.
    \item $\npair{[b,e]}\in\readCond{\normal,W}$:
      this follows from Lemma~\ref{lem:code-memory-untrusted-readable}.
    \item $\npair{[b,e]}\in\xReadCond{W}$:
      by assumption and Lemma~\ref{lem:conds-shrinkable}.
    \item $\npair{[b,e]}\in\execCond{W}$:
      take $n' < n$, $W' \future \purePart{W}$, $\aaddr' \in [\baddr',\eaddr'] \subseteq [\baddr,\eaddr]$ and $w' = ((\rx,\normal),\baddr',\eaddr',\aaddr')$.
      Then we need to show that $\npair[n']{(w',w')} \in \lre(W')$.
      This now follows immediately from the (regular) FTLR (Theorem~\ref{thm:ftlr}), using the two above points, Lemma~\ref{lem:downwards-closed} and Lemma~\ref{lem:conds-shrinkable}.
    \end{itemize}

  \item $\trust = \trusted$ and $\dom(\mscode) \subseteq \ta$:
    By definition of $\lrvg{\trusted}(W)$ and Lemma~\ref{lem:conds-shrinkable} with the assumption that $\npair{\dom(\mscode)}\in\xReadCond{W}$.
  \end{itemize}
\end{proof}

\begin{lemma}[FTLR for component code-values]
  \label{lem:ftlr-comp-code}
  If
  \begin{equation*}
    \sigrets,\sigrets[\mathrm{owned}],\sigcloss,\ta\vdash_{\mathrm{comp-code}} w
  \end{equation*}
  and
  \begin{itemize}
  \item $\sigrets \subseteq \gsigrets$ and $\sigcloss \subseteq \gsigcloss$
  \item $\pwheap(r_{\mathrm{code}}) = \codereg{\sigrets,\sigcloss,\mscode,\gc}$
  \item ($\dom(\mscode) \mathrel{\#} \ta$ and $\trust = \untrusted$) or ($\dom(\mscode) \subseteq \ta$ and $\trust = \trusted$)
  \end{itemize}

  then $\npair{(w,w)} \in \lrvg{\trust}(W)$
\end{lemma}
\begin{proof}
  By induction on
  \begin{equation*}
    \sigrets,\sigrets[\mathrm{owned}],\sigcloss,\ta\vdash_{\mathrm{comp-code}} w
  \end{equation*}
  There are two cases to consider:
  \begin{itemize}
  \item
    $\mscode(a) = \seal{\sigma_\baddr,\sigma_\eaddr,\sigma_\baddr}$ and $[\sigma_\baddr,\sigma_\eaddr] \subseteq (\sigrets \cup \sigcloss)$:

    We distinguish two cases:
    \begin{itemize}
    \item $\trust = \trusted$:

      By definition of $\lrvg{\trusted}$, it suffices to prove that
      $\pwheap(r_{\mathrm{code}}) = \codereg{\sigrets,\sigcloss,\mscode,\gc}$, $\dom(\code)\subseteq \ta$ and $[\sigma_\baddr,\sigma_\eaddr] \subseteq (\sigrets \cup \sigcloss)$, all of which follow by assumption.

    \item $\trust = \untrusted$:

      In this case, we know that $\sigrets=\emptyset$.

      By definition of $\lrv$, it suffices to prove that
      \begin{equation*}
        \forall \sigma' \in [\sigma_\baddr,\sigma_\eaddr] \ldotp \exists r \in \dom(\pwheap) \ldotp \pwheap(r) = (\pure,\_,H_\sigma) \tand H_\sigma \; \sigma' \nequal \lrv
      \end{equation*}
      We take $r = r_{\var{code}}$ and the result then follows from Lemma~\ref{lem:codereg-sealing-safety}.
    \end{itemize}

  \item $\mscode(a) \in \ints$ and
    \begin{multline*}
      ([a \cdots a + \calllen-1] \subseteq \ta \wedge\mscode([a \cdots a + \calllen-1]) = \scall[0..\calllen-1]{\offpc,\offsigma}{r_1}{r_2}) \Rightarrow\\ (\mscode(a+\offpc) = \seal{\sigma_\baddr,\sigma_\eaddr,\sigma_\baddr} \wedge \sigma_\baddr+\offsigma \in \sigrets[\mathrm{owned}])
    \end{multline*}:

    Using $\mscode(a) \in \ints$, the result follows easily by definition.
  \end{itemize}
\end{proof}

\begin{lemma}[FTLR for component data-values]
  \label{lem:ftlr-comp-value}
  If
  \begin{equation*}
    \dom(\mscode), A_{\mathrm{own}}, A_{\mathrm{non-linear}}, \sigrets,\sigcloss \vdash_{\mathrm{comp-value}} w
  \end{equation*}
  and
  \begin{itemize}
  \item $\npair{A_{\mathrm{non-linear}}} \in \readCond{\normal,W}$
  \item $\npair{A_{\mathrm{non-linear}}} \in \writeCond{\normal,W}$
  \item $\npair{A_{\mathrm{own}}} \in \readCond{\linear,W}$
  \item $\npair{A_{\mathrm{own}}} \in \writeCond{\linear,W}$
  \item $\pwheap(r_{\mathrm{code}}) = \codereg{\sigrets,\sigcloss,\mscode,\gc}$
  \item $\dom(\mscode) \mathrel{\#} \ta$ or $\dom(\mscode) \subseteq \ta$
  \item $(A_{\mathrm{own}} \cup A_{\mathrm{non-linear}}) \mathrel{\#} \ta$
  \end{itemize}

  then $\npair{(w,w)} \in \lrv(W)$
\end{lemma}
\begin{proof}
  By induction on the judgement 
  \begin{equation*}
    \dom(\mscode), A_{\mathrm{own}}, A_{\mathrm{non-linear}}, \sigrets,\sigcloss \vdash_{\mathrm{comp-value}} w
  \end{equation*}
  We have the following cases:
  \begin{itemize}
  \item $w = z$: result follows trivially
  \item $w = ((\perm,\lin),\baddr,\eaddr,\aaddr)$,
    $\permbnf \sqsubseteq \rw$,
    $\lin = \linear \Rightarrow \emptyset \subset [\baddr,\eaddr] \subseteq A_{\mathrm{own}}$ and
    $\lin = \normal \Rightarrow [\baddr,\eaddr] \subseteq A_{\mathrm{non-linear}}$:
    
    By definition of $\lrv$, it suffices to prove that $[b,e] \mathrel{\#} \ta$, $\npair{[\baddr,\eaddr]} \in \readCond{\lin,W}$ and $\npair{[\baddr,\eaddr]}\in\writeCond{\lin,W}$.
    The result then follows easily from the assumptions, with Lemma~\ref{lem:conds-shrinkable}.
    
  \item $w = \sealed{\sigma,\vsc}$ and
    \begin{equation*}
      \dom(\mscode),A_{\mathrm{own}},A_{\mathrm{non-linear}},\sigrets,\sigcloss \vdash_{\mathrm{comp-value}} \vsc
    \end{equation*}
    and $\sigma \in \sigcloss$:

    First, we have by induction that $\npair{(\vsc,\vsc)} \in \lrv (W)$.

    By definition of $\lrv$ and by choosing island $r = r_{\mathrm{code}}$, it suffices to prove that
    \begin{itemize}
    \item $\npair[n']{(\vsc,\vsc)} \in H_\sigma \; \sigma \; \xi^{-1}(W)$ for all $n' < n$:

     By definition of $\codereg{\sigrets,\sigcloss,\mscode,\gc}$ and $H^\mathrm{code,\square}_\sigma$ and because we know that $\sigma\in\sigcloss$, it suffices to prove that
     one of the following holds: 
     \begin{itemize}
     \item $(\dom(\code) \mathrel{\#} \ta$ and $\npair[n']{(\vsc,\vsc)} \in \lrv \; \xi(\xi^{-1}(W))$
     \item $(\dom(\code) \subseteq \ta$ and $\npair[n']{(\vsc,\vsc)} \in \lrvtrusted \; \xi(\xi^{-1}(W))$
     \end{itemize}
     Both follow since $\xi(\xi^{-1}(W)) = W$ and $\lrv \; W \subseteq \lrvtrusted \; W$, $\dom(\mscode \mathrel{\#} \ta)$ or $\dom(\mscode) \subseteq \ta$, and the above fact that $\npair{(\vsc,\vsc)} \in \lrv (W)$ and Lemma~\ref{lem:downwards-closed}.

   \item $(\isLinear{\vsc} \text{ iff } \isLinear{\vsc})$: trivially fine
      
   \item If $\isLinear{\vsc}$, then 
     for all $W' \future W$, $W_o$, $n' < n$ and $\npair[n']{\stpair[.]{\vsc_S'}{\vsc'_T}} \in H_\sigma \; \sigma \; \xi^{-1}(W_o)$, we have that 
     \begin{equation*}
       \npair[n']{(\vsc,\vsc_S',\vsc,\vsc_T')} \in \lrexj(W'\oplus W_o))
     \end{equation*}

     By definition of $\codereg{\sigrets,\sigcloss,\mscode,\gc}$ and $H^\mathrm{code,\square}_\sigma$ and because we know that $\sigma\in\sigcloss$, we know that
     one of the following holds:
     \begin{itemize}
     \item $(\dom(\code) \mathrel{\#} \ta$ and $\npair[n']{(\vsc_S',\vsc_T')} \in \lrv \; \xi(\xi^{-1}(W_o))$
     \item $(\dom(\code) \subseteq \ta$ and $\sigcloss\subseteq \gsigcloss$ and $\sigrets \subseteq \gsigrets$ and $\npair[n']{(\vsc_S',\vsc_T')} \in \lrvg{\trust} \; \xi(\xi^{-1}(W_o))$ with $\trust = \trusted$ iff $\exec{\vsc_S'}$.
     \end{itemize}

     Lemma~\ref{lem:safe-vals-safe-exec} now allows us to conclude (using Lemma~\ref{lem:monotonicity}) that:
     \begin{equation*}
       \npair[n']{(\vsc,\vsc_S',\vsc,\vsc_T')} \in \lrexj(W'\oplus W_o))
     \end{equation*}

   \item 
       If $\nonLinear{\src{\vsc_S}}$ then for all $W' \future \purePart{W}$, $W_o$, $n' < n$, $\npair[n']{(\vsc_S',\vsc'_T)} \in H_\sigma \; \sigma \; \xi^{-1}(W_o)$, we have that
       \begin{equation*}
         \npair[n']{(\vsc,\src{\vsc_S'},\vsc,\vsc_T')} \in \lrexj(W'\oplus W_o))
       \end{equation*}

     By definition of $\codereg{\sigrets,\sigcloss,\mscode,\gc}$ and $H^\mathrm{code,\square}_\sigma$ and because we know that $\sigma\in\sigcloss$, we know that
     one of the following holds:
     \begin{itemize}
     \item $(\dom(\code) \mathrel{\#} \ta$ and $\npair[n']{(\vsc_S',\vsc_T')} \in \lrv \; \xi(\xi^{-1}(W_o))$
     \item $(\dom(\code) \subseteq \ta$ and $\sigcloss\subseteq \gsigcloss$ and $\sigrets \subseteq \gsigrets$ and $\npair[n']{(\vsc_S',\vsc_T')} \in \lrvg{\trust} \; \xi(\xi^{-1}(W_o))$ with $\trust = \trusted$ iff $\exec{\vsc_S'}$.
     \end{itemize}

     First, Lemma~\ref{lem:non-linear-pure} allows us to conclude that also $\npair{(\vsc,\vsc)} \in \lrv (\purePart{W})$.

     Next, Lemma~\ref{lem:safe-vals-safe-exec} now allows us to conclude (using Lemma~\ref{lem:monotonicity}) that:
     \begin{equation*}
       \npair[n']{(\vsc,\vsc_S',\vsc,\vsc_T')} \in \lrexj(W'\oplus W_o))
     \end{equation*}
    \end{itemize}
  \end{itemize}
\end{proof}

\begin{lemma}[FTLR for component exports]
  \label{lem:ftlr-comp-export}
  If
  \begin{equation*}
    \dom(\mscode), A_{\mathrm{non-linear}}, \sigrets,\sigcloss \vdash_{\mathrm{comp-export}} w
  \end{equation*}
  and
  \begin{itemize}
  \item $\npair{A_{\mathrm{non-linear}}} \in \readCond{\normal,W}$
  \item $\npair{A_{\mathrm{non-linear}}} \in \writeCond{\normal,W}$
  \item $\pwheap(r_{\mathrm{code}}) = \codereg{\sigrets,\sigcloss,\mscode,\gc}$
  \item $\dom(\mscode) \mathrel{\#} \ta$ or $\dom(\mscode) \subseteq \ta$
  \item If $w = \sealed{\sigma,\vsc}$ with $\dom(\mscode) \subseteq \ta$, $\sigma \in \sigcloss$ and $\exec{\vsc}$, then
    $\vsc$ behaves reasonably up to $n$ steps.
  \item $A_{\mathrm{non-linear}}\mathrel{\#} \ta$
  \end{itemize}

  then $\npair{(w,w)} \in \lrv(W)$
\end{lemma}
\begin{proof}
  By induction on the judgement
  \begin{equation*}
    \dom(\mscode), A_{\mathrm{non-linear}}, \sigrets,\sigcloss \vdash_{\mathrm{comp-export}} w
  \end{equation*}
  We have the following cases:
  \begin{itemize}
  \item $w = \sealed{\sigma,\vsc}$, $\vsc = ((\rx,\normal),\baddr,\eaddr,\aaddr)$, $[\baddr,\eaddr]\subseteq \dom(\mscode)$, $\sigma \in \sigcloss$:

    By definition of $\lrv$ and by choosing island $r = r_{\mathrm{code}}$, it suffices to prove that
    \begin{itemize}
    \item $\npair[n']{(\vsc,\vsc)} \in H^\mathrm{code,\square}_\sigma \; \sigrets \; \sigcloss \; \code \; \gc \; \sigma \; \xi^{-1}(W)$ for all $n' < n$:

      By definition of $\codereg{\sigrets,\sigcloss,\mscode,\gc}$ and $H^\mathrm{code,\square}_\sigma$ and because we know that $\sigma\in\sigcloss$, it suffices to prove that one of the following holds: 
      \begin{itemize}
      \item $(\dom(\code) \mathrel{\#} \ta$ and $\npair[n']{(\vsc,\vsc)} \in \lrv \; \xi(\xi^{-1}(W))$
      \item $(\dom(\code) \subseteq \ta$ and $\npair[n']{(\vsc,\vsc)} \in \lrvtrusted \; \xi(\xi^{-1}(W))$ and $\exec{\vsc}$
      \end{itemize}
      Take $\trust = \untrusted$ iff $\dom(\mscode \mathrel{\#} \ta)$ and $\trust = \trusted$ iff $\dom(\mscode) \subseteq \ta$.
      Since $\xi(\xi^{-1}(W)) = W$ and because we know by assumption that $\dom(\mscode \mathrel{\#} \ta)$ or $\dom(\mscode) \subseteq \ta$, it suffices to prove that $\npair[n']{(\vsc,\vsc)} \in \lrvg{\trust}(W)$.

      This last fact follows from Lemma~\ref{lem:ftlr-component-code-ptrs}, using Lemma~\ref{lem:downwards-closed} and the definition of $\xReadCond{}$ with the assumption that $\pwheap(r_{\mathrm{code}}) = \codereg{\sigrets,\sigcloss,\mscode,\gc}$.

    \item $(\isLinear{\vsc} \text{ iff } \isLinear{\vsc})$: trivially fine

    \item
      For all $W' \future \purePart{W}$, $W_o$, $n' < n$, $\npair[n']{(\vsc_S',\vsc'_T)} \in H^\mathrm{code,\square}_\sigma \; \sigrets \; \sigcloss \; \code \; \gc \; \sigma \; \xi^{-1}(W_o)$, we have that
      \begin{equation*}
        \npair[n']{(\vsc,\src{\vsc_S'},\vsc,\vsc_T')} \in \lrexj(W'\oplus W_o))
      \end{equation*}


      From one of our assumptions, we know that $\vsc$ behaves reasonably up to $n$ steps if $\dom(\mscode)\subseteq \ta$.

      Theorem~\ref{thm:ftlr} now tells us that $\npair{(\vsc,\vsc)} \in \lre(\purePart{W})$, using the definition of $\xReadCond{}$ and the assumption that $\pwheap(r_{\mathrm{code}}) = \codereg{\sigrets,\sigcloss,\mscode,\gc}$.

      Since $\sigma\in\sigcloss$, it follows from $\npair[n']{(\vsc_S',\vsc'_T)} \in H^\mathrm{code,\square}_\sigma \; \sigrets \; \sigcloss \; \code \; \gc \; \sigma \; \xi^{-1}(W_o)$ that
      \begin{equation*}
        \npair[n']{(\vsc_S',\vsc'_T)} \in \lrvg{\trust}(W_o)
      \end{equation*}
      with $\trust = \untrusted$ iff ($\dom(\code \mathrel{\#} \ta)$ or $\nonExec{\vsc_S'}$) and $\trust = \trusted$ iff ($\dom(\code) \subseteq \ta$ and $\exec{\vsc_S'}$).

      The result now follows from Lemma~\ref{lem:lre-implies-lrexj}.
    \end{itemize}

  \item
    \begin{equation*}
      \dom(\mscode),\emptyset,A_{\mathrm{non-linear}},\sigrets,\sigcloss \vdash_{\mathrm{comp-value}} w
    \end{equation*}

    In this case, the result follows from Lemma~\ref{lem:ftlr-comp-value}, using the fact that $\readCond$ and $\writeCond$ follow trivially for the empty $A_{\mathrm{own}}$ and a similar observation about the empty imports.
  \end{itemize}
\end{proof}

\begin{lemma}[FTLR for components]
  \label{lem:ftlr-comps}
  If
  \begin{itemize}
  \item $\gc = (\ta,\stkb,\gsigrets,\gsigcloss)$
  \item $\var{comp}$ is a well-formed component, i.e. $\ta \vdash \var{comp}$
  \item One of the following holds:
    \begin{itemize}
    \item $\dom(\var{comp}.\mscode) \subseteq \ta$ and $\var{comp}$ is a reasonable component (see Section~\ref{sec:reasonability})
    \item $\dom(\var{comp}.\mscode) \mathrel{\#} \ta$
    \end{itemize}
  \item $\sigrets \subseteq \gsigrets$ and $\sigcloss \subseteq \gsigcloss$
  \end{itemize}
  Then there exists a $W$ such that
  \begin{itemize}
  \item $\npair{(\var{comp},\var{comp})} \in \lrcomp(W)$
  \item $\dom(\pwheap)$ can be chosen to not include any finite set of region names 
  \item $\dom(\pwfree) =
    \dom(\pwpriv)=\emptyset$
  \end{itemize}
\end{lemma}
\begin{proof}
  If $\var{comp}$ is of the form $(\var{comp}_0,c_{\mathrm{main},c}, c_{\mathrm{main},d})$, then we have (by definition of component well-formedness) that $c_{\mathrm{main},c}, c_{\mathrm{main},d} \subseteq \var{comp}_0.\overline{c_{\mathrm{export}}}$, as required by $\lrcomp(W)$.
  Hence, we can restrict our attention to components of the form $\var{comp}_0$.

  Take
  \begin{equation*}
    \var{comp} = (\mscode\uplus\mspad,\msdata,\overline{\var{import}},\overline{\var{export}},\sigrets,\sigcloss,A_\linear)
  \end{equation*}
  Note that confusingly, $\var{comp}.\mscode = \mscode\uplus\mspad$, so the assumption about $\dom(\var{comp}.\mscode)$ should be interpreted properly.
  We then know from $\ta \vdash \var{comp}$ that 
  \begin{itemize}
  \item $\dom(\mscode) = [\baddr,\eaddr]$
  \item $[\baddr-1,\eaddr+1] \mathrel{\#} \dom(\msdata)$
  \item $\mspad = [\baddr-1\mapsto 0] \uplus [\eaddr+1 \mapsto 0]$
  \item $\sigrets,\sigcloss,\ta\vdash_{\mathrm{comp-code}} \mscode$
  \item $\exists A_\mathrm{own} : \dom(\msdata) \rightarrow \powerset{\dom(\msdata)}$
  \item $\dom(\msdata) = A_{\mathrm{non-linear}} \uplus A_\linear$
  \item $A_\linear = \biguplus_{a \in \dom(\msdata)} A_{\mathrm{own}}(a)$
  \item
    \begin{equation*}
      \forall a \in \dom(\msdata)\ldotp \dom(\mscode),A_{\mathrm{own}}(a),A_{\mathrm{non-linear}},\sigrets,\sigcloss \vdash_{\mathrm{comp-value}} \msdata(a)
    \end{equation*}
  \item
    \begin{equation*}
      \overline{\dom(\mscode), A_{\mathrm{non-linear}}, \sigrets,\sigcloss \vdash_{\mathrm{comp-export}} c_{\mathrm{export}}}
    \end{equation*}
  \item $(\dom(\mscode) \subseteq \ta) \vee (\dom(\mscode) \mathrel{\#} \ta \wedge \sigrets = \emptyset)$
  \item $\dom(\msdata) \mathrel{\#} \ta$
  \end{itemize}

  Take $\overline{\var{import}} = \overline{a_{\mathrm{import}} \mapsfrom s_{\mathrm{import}}}$.

  Now take $W$ such that
  \begin{itemize}
  \item $\dom(\pwfree) = \dom(\pwpriv) = \emptyset$
  \item $\pwheap(r_\code) =  \codereg[\mathrm{code},\square]{\sigrets,\sigcloss,\mscode,\gc}$
  \item there exist $r_{\mathrm{addr}} : \dom(\msdata) \mapsto \dom(\pwheap)$ such that
    \begin{itemize}
    \item For all $\aaddr\in\dom(\msdata)$, $\pwheap(r_{\mathrm{addr}}(\aaddr)) = \stdreg{\{\aaddr\},\gc}{\lin}$ with ($\lin = \pure$ if $\aaddr\not\in A_\linear$) and ($\lin = \spatialo$ if $\aaddr\in A_\linear$)
    \end{itemize}
  \end{itemize}
  with $r_{\mathrm{code}}$ and $r_a$ chosen according to the given restriction.
  
  It now remains to prove that $\npair{(\var{comp},\var{comp})} \in \lrcomp{W}$.
  Take $W' \future W$ and assume that
  \begin{equation*}
    \overline{\npair[n']{(w_{\mathrm{import}},w_{\mathrm{import}})}} \in \lrv(\purePart{W'}) \text{ for all $n' < n$}\\
  \end{equation*}
  and
  \begin{equation*}
    \msdata' = \msdata{}[\overline{a_{\mathrm{import}} \mapsto w_{\mathrm{import}}}] \\
  \end{equation*}
  We need to show that
  \begin{itemize}
  \item
    \begin{equation*}
      \npair{(\sigrets\uplus\sigcloss,\mspad\uplus\mscode\uplus \msdata', \mspad\uplus\mscode\uplus\msdata')} \in \lrheap(\pwheap)(W')
    \end{equation*}

   Take $R_\ms : \dom(\activeReg{\pwheap}) \fun \MemSeg \times \MemSeg$ such that
    \begin{itemize}
    \item $R_\ms(r_\addr) = (\msdata|_{\{a\}}, \msdata|_{\{a\}})$.
    \item $R_\ms(r_\var{code}) = (\mspad\uplus\mscode,\mspad\uplus\mscode)$
    \end{itemize}
    Take $R_W' : \dom(\activeReg{\pwheap}) \fun \World$ such that
    \begin{itemize}
    \item $R_W'(r_\addr) = \purePart{W}[\var{heap}.r_{\addr'}\mapsto \stdreg{\{\aaddr\},\gc}{\spatialo}]_{\addr' \in A_\var{own}(\addr)}$ 
    \item $R_W'(r_\var{code}) = \purePart{W}$
    \end{itemize}
    Since then $W = \oplus_{r \in \dom(\activeReg{\pwheap})} R_W'(r)$, we can use Lemma~\ref{lem:oplus-future-distr} to construct an $R_W$ with
    \begin{equation*}
      W' = \oplus_{r \in \dom(\activeReg{\pwheap})} R_W(r)
    \end{equation*}
    and $R_W(r) \future R_W'(r)$ for all $r$.
    Finally, take $R_\var{seal} : \dom(\activeReg{\pwheap}) \fun \powerset{\Seal}$ to satisfy
    \begin{itemize}
    \item $R_\var{seal}(r_\addr) = \emptyset$
    \item $R_\var{seal}(r_\var{code}) = \sigrets\uplus \sigcloss$.
    \end{itemize}

    We then need to prove that 
    \begin{itemize}
    \item $\npair[n']{(\msdata'|_{\{\addr\}},\msdata'|_{\{\addr\}})} \in \stdreg{\{\aaddr\},\gc}{\lin}.H~ \xi^{-1}(R_W(r_\addr))$ for all $n' < n$:
      Take $n' < n$.
      By definition, it suffices to prove that,
      $\npair[n']{(\msdata'(\addr),\msdata'(\addr))} \in \lrv(R_W(r_\addr))$.

      If $(a \mapsfrom \_) \in \overline{\var{import}}$, then we know that $\msdata'(a)$ is the corresponding $w_{\mathrm{import}}$ and the result is fine by the assumption that
      \begin{equation*}
        \overline{\npair[n']{(w_{\mathrm{import}},w_{\mathrm{import}})}} \in \lrv(\purePart{W'})
      \end{equation*}
      together with \ref{lem:monotonicity} and the fact that $R_W(\addr) \future \purePart{W'}$ (by Lemma~\ref{lem:auth-partial-betterthannothing}).

      Otherwise, $\msdata'(a) = \msdata(a)$ and we know from our assumptions that
      \begin{equation*}
        \dom(\mscode),A_{\mathrm{own}}(a),A_{\mathrm{non-linear}},\sigrets,\sigcloss \vdash_{\mathrm{comp-value}} \msdata(a)
      \end{equation*}
      
      By Lemma~\ref{lem:ftlr-comp-value}, it then suffices to prove that
      \begin{itemize}
      \item $\npair[n']{A_{\mathrm{non-linear}}} \in \readCond{\normal,R_W(\addr)}$: follows by Lemma~\ref{lem:monotonicity} using the fact that $R_W(a) \future R_W'(a)$ and by definition, using the choice of $W$, $R_W'(\addr)$ 
      \item $\npair[n']{A_{\mathrm{non-linear}}} \in \writeCond{\normal,R_W(\addr)}$: follows by Lemma~\ref{lem:monotonicity} using the fact that $R_W(a) \future R_W'(a)$ and by definition, using the choice of $W$, $R_W'(\addr)$ 
      \item $\npair[n']{A_{\mathrm{own}}(\addr)} \in \readCond{\linear,R_W(\addr)}$: follows by Lemma~\ref{lem:monotonicity} using the fact that $R_W(a) \future R_W'(a)$ and by definition, using the choice of $W$, $R_W'(\addr)$ 
      \item $\npair[n']{A_{\mathrm{own}}(\addr)} \in \writeCond{\linear,R_W(\addr)}$: follows by Lemma~\ref{lem:monotonicity} using the fact that $R_W(a) \future R_W'(a)$ and by definition, using the choice of $W$, $R_W'(\addr)$ 
      \item $\pwheap[R_W(\addr)](r_{\mathrm{code}}) = \codereg{\sigrets,\sigcloss,\mscode,\gc}$: follows from the choice of $W$ and $R_W'(\addr)$ and the fact that $R_W(a) \future R_W'(a)$
      \item $\dom(\mscode) \mathrel{\#} \ta$ or $\dom(\mscode) \subseteq \ta$: by assumption.
      \end{itemize}
    \item $\npair[n']{(\mspad\uplus\mscode,\mspad\uplus\mscode)} \in  \codereg[\mathrm{code},\square]{\sigrets,\sigcloss,\mscode,\gc}.H~ \xi^{-1}(\purePart{W'})$ for all $n' < n$:
      Take $n' < n$.
      We take $\trust = \trusted$ if $[\baddr-1,\eaddr+1] \subseteq \ta$ and $\trust =\untrusted$ otherwise.
      By definition, we need to prove that:
      
      \begin{itemize}
      \item $\dom(\code) = [\baddr,\eaddr]$: by assumption.
      \item $([\baddr-1,\eaddr+1]\subseteq \ta \wedge \trust = \trusted) \vee ([\baddr-1,\eaddr+1]\mathrel{\#} \ta \wedge \sigrets = \emptyset \wedge \trust =\untrusted)$: by assumption and choice of $\trust$.
      \item $\sigrets,\sigcloss,\ta\vdash_{\mathrm{comp-code}} \code$: by assumption.
      \item $\mspad = [\baddr-1 \mapsto 0] \uplus [\eaddr + 1 \mapsto 0]$:
        by assumption
      \item $\forall a \in \dom(\code) \ldotp \npair[n']{(\code(a),\code(a))} \in \lrvg{\trust}(\purePart{}(\xi(\xi^{-1}(\purePart{W'}))))$:
        Note first that $\purePart{}(\xi(\xi^{-1}(\purePart{W'}))) =\purePart{W'}$.

        By Lemma~\ref{lem:ftlr-comp-code}, it suffices to prove that:
        \begin{itemize}
        \item $\pwheap[W'](r_{\mathrm{code}}) = \codereg{\sigrets,\sigcloss,\mscode,\gc}$: follows by choice of $W$ and the fact that $W' \future W$.
        \item ($\dom(\mscode) \mathrel{\#} \ta$ and $\trust = \untrusted$) or ($\dom(\mscode) \subseteq \ta$ and $\trust = \trusted$): by assumption and choice of $\trust$
        \end{itemize}
      \end{itemize}
    \item
      \begin{equation*}
        \dom(\codereg[\mathrm{code},\square]{\sigrets,\sigcloss,\mscode,\gc}.H_\sigma) = \sigrets\uplus \sigcloss
      \end{equation*}
      This follows easily from the definition.
  \end{itemize}
  \item $\overline{\npair{(w_{\mathrm{export}},w_{\mathrm{export}})}} \in \lrv(\purePart{W'})$:

    We know from our assumptions that
    \begin{equation*}
      \overline{\dom(\mscode), A_{\mathrm{non-linear}}, \sigrets,\sigcloss \vdash_{\mathrm{comp-export}} w_{\mathrm{export}}}
    \end{equation*}

    By Lemma~\ref{lem:ftlr-comp-export}, it then suffices to prove that
    \begin{itemize}
    \item $\npair{A_{\mathrm{non-linear}}} \in \readCond{}(\normal,\purePart{W'})$: follows by Lemma~\ref{lem:monotonicity} using the fact that $R_W(a) \future R_W'(a)$ and by definition, using the choice of $W$, $R_W'(\addr)$ 
    \item $\npair{A_{\mathrm{non-linear}}} \in \writeCond{}(\normal,\purePart{W'})$: follows by Lemma~\ref{lem:monotonicity} using the fact that $R_W(a) \future R_W'(a)$ and by definition, using the choice of $W$, $R_W'(\addr)$ 
    \item $\pwheap[\purePart{W'}](r_{\mathrm{code}}) = \codereg{\sigrets,\sigcloss,\mscode,\gc}$: follows from the choice of $W$ and the fact that $W' \future W$
    \item $\dom(\mscode) \mathrel{\#} \ta$ or $\dom(\mscode) \subseteq \ta$: by assumption.
    \item If $w_{\mathrm{export}} = \sealed{\sigma,\vsc}$ with $\dom(\mscode) \subseteq \ta$, $\sigma \in \sigcloss$ and $\exec{\vsc}$, then $\vsc$ behaves reasonably up to $n$ steps:

      This follows directly from the fact that $\var{comp}$ is a reasonable component.

    \item $A_{\mathrm{non-linear}}\mathrel{\#} \ta$:
      This follows from the facts that $A_{\mathrm{non-linear}}\subseteq \dom(\msdata) \mathrel{\#} \ta$.
    \end{itemize}
  \end{itemize}
\end{proof}

Note: the trusted case of the above lemma can be considered as a compiler correctness result.
The untrusted case can be considered as a back-translation correctness result.

\subsection{Related execution configurations}
\label{sec:related-exec-confs}

\[
  \lrec(W) = \left\{ \npair{\left(
        \arraycolsep=0pt\array{l}
        (\ms_S,\reg_S,\stk,\ms_\stk),\\
        (\ms_T,\reg_T)\endarray\right)} \middle|
    \begin{array}{l}
      \gc = (\ta,\stkb) \tand\\
      \exists W_M,W_R,W_\pcreg \ldotp W = W_M \oplus W_R \oplus W_\pcreg \tand\\
      \npair{( (\reg_S(\pcreg),\reg_S(\rdata)), (\reg_T(\pcreg),\reg_T(\rdata)) )} \in \lrexj(W_\pcreg) \tand \\
      \reg_S(\pcreg) \neq \retptrc(\_) \wedge 
      \reg_S(\rdata) \neq \retptrd(\_) \wedge \\
      \nonExec{\reg_S(\rdata)} \wedge
      \nonExec{\reg_T(\rdata)}\\
      \memSat{\ms_S,\ms_\stk,\stk,\ms_T}{W_M} \tand\\
      \npair{\stpair{\reg}{\reg}} \in \lrr(\{\rdata\})(W_R)
    \end{array}
            \right\}
\]

\lau{TODO 18-07-2018: gc should probably be specified here, so they are not just picked at random.}
\lau{TODO 18-07-2018: Where is this lemma even used?}
\begin{lemma}[Compatibility lemma for initial execution configuration construction]
  \label{lem:compat-initial-ec}
  If
  \begin{itemize}
  \item $\npair{(\var{comp},\var{comp})} \in \lrcomp(W)$
  \item $\dom(\pwfree) = \dom(\pwpriv) = \emptyset$
  \item $\var{comp} \rightsquigarrow \Phi$,
  \end{itemize}
  then $\exists W' \future W$ such that for all $n' < n$
  \begin{itemize}
  \item $\npair[n']{(\Phi,\Phi)}\in\lrec(W')$
  \end{itemize}
\end{lemma}
\begin{proof}
  Take $\var{comp} = ((\mscode, \msdata, \overline{\var{import}}, \overline{\var{export}}, \sigrets, \sigcloss,A_\linear),  c_{\mathrm{main},c}, c_{\mathrm{main},d})$.
  Then we know from $\var{comp} \rightsquigarrow \Phi$ that
  \begin{itemize}
  \item $c_{\mathrm{main},c} = \sealed{\sigma_1, c_{\mathrm{main},c}'}$ 
  \item $c_{\mathrm{main},d} = \sealed{\sigma_2, c_{\mathrm{main},d}'}$
  \item $\sigma_1 = \sigma_2$
  \item $\reg(\pcreg) = c_{\mathrm{main},c}'$
  \item $\reg(\rdata) = c_{\mathrm{main},d}'$ 
  \item $\nonExec{c_{\mathrm{main},d}'}$
  \item $\src{\reg(\rstk) = \stkptr{\rw,\baddr_\stk,\eaddr_\stk,\eaddr_\stk}}$ 
  \item $\trg{\reg(\rstk) = ((\rw,\linear),\baddr_\stk,\eaddr_\stk,\eaddr_\stk)}$ 
  \item $\reg(\RegName \setminus \{\pcreg,\rdata,\rstk\}) = 0$
  \item $\range{\ms_\stk} = \{0\}$
  \item $\mem = \mscode \uplus \msdata \trg{\;\uplus\; \ms_\stk}$
  \item $[\baddr_\stk,\eaddr_\stk] = \dom(\ms_\stk )$
  \item $[\baddr_\stk-1,\eaddr_\stk+1] \mathrel{\#} (\dom(\mscode) \cup \dom(\msdata))$
  \item $\overline{\var{import}} = \emptyset$
  \item $\Phi = (\mem, \reg\src{, \emptyset, \ms_\stk})$
  \end{itemize}

  From $\npair{(\var{comp},\var{comp})} \in \lrcomp(W)$, we know that
  \begin{itemize}
  \item $\{(\_ \mapsto c_{\mathrm{main},c}), (\_ \mapsto c_{\mathrm{main},d})\}\subseteq \overline{\var{export}}$ and
  \item $\npair{(\var{comp}_0,\var{comp}_0)} \in \lrcomp(W)$ with
  \item  $\var{comp}_0 = (\mscode, \msdata, \overline{a_{\mathrm{import}}\mapsfrom s_{\mathrm{import}}}, \overline{s_{\mathrm{export}}\mapsto w_{\mathrm{export}}}, \sigrets, \sigcloss,A_\linear)$
  \end{itemize}
  From this, it follows that for all  $W' \future W$, if
  \begin{equation*}
    \overline{\npair[n']{(w_{\mathrm{import}},w_{\mathrm{import}})}} \in \lrv(\purePart{W'}) \text{ for all $n' < n$}
  \end{equation*}
  and $\msdata' = \msdata{}[\overline{a_{\mathrm{import}} \mapsto w_{\mathrm{import}}}]$
  then
  \begin{equation*}
    \npair{(\sigrets\uplus\sigcloss,\mscode\uplus \msdata', \mscode\uplus\msdata')} \in \lrheap(\pwheap)(W')
  \end{equation*}
  and
  \begin{equation*}
    \overline{\npair{(w_{\mathrm{export}},w_{\mathrm{export}})}} \in \lrv(\purePart{W'})
  \end{equation*}
  Since $\overline{w_{\mathrm{import}}} = \emptyset$, the former holds vacuously, $\msdata' = \msdata$ and the latter two results follow for every $W' \future W$.
  
  Now take $\overline{r_a}$ for $a \in [\baddr_\stk,\eaddr_\stk]$ arbitrary and 
  \begin{align*}
    \pwheap[W_\stk] &= \emptyset\\
    \pwpriv[W_\stk] &= \emptyset\\
    \dom(\pwfree[W_\stk]) &= \{r_a \mid a \in [\baddr_\stk,\eaddr_\stk]\}\\
    \pwfree[W_\stk](r_a) &= \stdreg{\{a\},\gc}{\spatialo}\\
    W' &= W \uplus W_\stk\\
    W_\pcreg &= \purePart{W'}\\
    W_R &= \purePart{W} \uplus W_\stk\\
    W_M &= W \uplus \purePart{W_\stk}
  \end{align*}
  Then $W' = W_M \oplus W_R \oplus W_\pcreg$ (by Lemmas~\ref{lem:purePart-duplicable},~\ref{lem:oplus-distr-uplus},~\ref{lem:oplus-assoc-comm} and~\ref{lem:uplus-assoc-comm}), $W_M \future W$ (by Lemma~\ref{lem:uplus-future}) and $W_\pcreg \future \purePart{W}$ (by Lemma~\ref{lem:uplus-future} and~\ref{lem:purePart-oplus}).

  By definition of $\lrec$, it suffices to prove that 
  \begin{itemize}
  \item 
    $\npair[n']{( (\reg(\pcreg),\reg(\rdata)), (\reg(\pcreg),\reg(\rdata)) )} \in \lrexj(W_\pcreg)$:
    Since $W' \future W$, we know from above that:
    \begin{equation*}
      \overline{\npair{(w_{\mathrm{export}},w_{\mathrm{export}})}} \in \lrv(\purePart{W'})
    \end{equation*}
    and we have defined $W_\pc = \purePart{W'}$ and we know that $\{c_{\mathrm{main},c}, c_{\mathrm{main},d}\}\subseteq
    \overline{c_{\mathrm{export}}}$.
    It follows that 
    \begin{align*}
      \npair{(c_{\mathrm{main},c},c_{\mathrm{main},c})} &\in \lrv(W_\pc)\\
      \npair{(c_{\mathrm{main},d},c_{\mathrm{main},d})} &\in \lrv(W_\pc)
    \end{align*}

    Since $W_\pc = \purePart{W'}$, Lemma~\ref{lem:safe-vals-safe-exec} tells us that for $n' < n$,
    \begin{equation*}
      \npair[n']{( (c_{\mathrm{main},c},c_{\mathrm{main},d}), (c_{\mathrm{main},c},c_{\mathrm{main},d}) )} \in \lrexj(W_\pcreg)
    \end{equation*}
    Since $\reg(\pcreg) = c_{\mathrm{main},c}'$ and $\reg(\rdata) = c_{\mathrm{main},d}'$, this is what we set out to prove.

  \item
    \begin{itemize}
    \item $\reg_S(\pcreg) \neq \retptrc(\_)$, 
    \item $\reg_S(\rdata) \neq \retptrd(\_)$, 
    \item $\nonExec{\reg_S(\rdata)}$ and 
    \item $\nonExec{\reg_T(\rdata)}$
    \end{itemize}
    This follows from the facts that $\reg(\pcreg) = c_{\mathrm{main},c}'$, $\reg(\rdata) = c_{\mathrm{main},d}'$, 
    $\nonExec{c_{\mathrm{main},d}'}$, $\{c_{\mathrm{main},c}, c_{\mathrm{main},d}\}\subseteq \overline{c_{\mathrm{export}}}$ and the fact that $\overline{c_{\mathrm{export}}}$ are also valid 

  \item $\memSat[n']{\ms_S,\ms_\stk,\stk,\ms_T}{W_M}$:
    Since $W_M \future W$, we have seen above that 
    \begin{equation*}
      \npair{(\sigrets\uplus\sigcloss,\mscode\uplus \msdata, \mscode\uplus\msdata)} \in \lrheap(\pwheap[W_M])(W_M)
    \end{equation*}
    and by Lemma~\ref{lem:downwards-closed} also
    \begin{equation*}
      \npair[n']{(\sigrets\uplus\sigcloss,\mscode\uplus \msdata, \mscode\uplus\msdata)} \in \lrheap(\pwheap[W_M])(W_M)
    \end{equation*}
    It suffices to prove that also
    \begin{itemize}
    \item $\memSatStack[n']{\emptyset,\emptyset}{\purePart{W_M}}$
    \item $\memSatFStack[n']{\ms_\stk,\ms_\stk}{\purePart{W_M}}$
    \end{itemize}
    The former follows vacuously.
    The latter follows by taking $R_\ms(r_a) = (\ms_\stk|_{\{a\}},\ms_\stk|_{\{a\}})$, $R_W(r_a) = \purePart{W_M}$ for $a \in [\baddr_\stk,\eaddr_\stk]$ by Lemma~\ref{lem:purePart-duplicable} and~\ref{lem:purePart-idempotent} if we can show that:
    \begin{equation*}
      \npair[n'']{(\ms_\stk|_{\{a\}},\ms_\stk|_{\{a\}})} \in \stdreg{\{a\},\gc}{\spatialo}.H~ \xi^{-1}(\purePart{W_M})
    \end{equation*}
    for all $n'' < n'$.
    By definition, it suffices to show that for all $n'' < n'$, we have that:
    \begin{equation*}
      \npair[n'']{(\ms_\stk(a),\ms_\stk(a)) \in \lrv(\purePart{W_M})}
    \end{equation*}
    But $\ms_\stk(a) = 0$, so this follows easily by definition.

  \item $\npair[n']{(\reg,\reg)} \in \lrr(\{\rdata\})(W_R)$:
    
    We have that $\reg(\RegName \setminus \{\pcreg,\rdata,\rstk\}) = 0$, so by definition, it suffices to prove that
    \begin{equation*}
      \npair[n']{(\reg(r_\stk),\reg(r_\stk))} \in \lrv(W_R)
    \end{equation*}

    Since $\src{\reg(\rstk) = \stkptr{\rw,\baddr_\stk,\eaddr_\stk,\eaddr_\stk}}$ and $\trg{\reg(\rstk) = ((\rw,\linear),\baddr_\stk,\eaddr_\stk,\eaddr_\stk)}$, it suffices to prove (by definition) that
    \begin{align*}
      \npair[n']{[b_\stk,e_\stk]} &\in \stackReadCond{W_R}\\
      \npair[n']{[b_\stk,e_\stk]} &\in \stackWriteCond{W_R}
    \end{align*}
    For both, we can take $S = \dom(\pwfree[W_\stk]) = \{r_a \mid a \in [\baddr_\stk,\eaddr_\stk]\}$, $R(r_\addr) = \{\addr\}$, and then it suffices to prove that for all $r_\addr$, $\pwfree[W_R](r_\addr).H \nsubeq \stdreg{\{\addr\},\gc}{\spatialo}$ resp. $\pwfree[W_R](r_\addr).H \nsupeq \stdreg{\{\addr\},\gc}{\spatialo}$.
    Both follow easily since 
    $\pwfree[W_R](r_a) = \pwfree[W_\stk](r_a) = \stdreg{\{a\},\gc}{\spatialo}$.
  \end{itemize}
\end{proof}

\begin{lemma}[Adequacy of execution configuration LR]
  \label{lem:adequacy}
  If
  \begin{itemize}
  \item $\npair{\stpair{\Phi}{\Phi}}\in\lrec[\preceq,\gc](W)$
  \item $i \leq n$
  \item $\Phi_S \sterm[i]{\gc}$
  \end{itemize}
  then $\Phi_T\term$.

  Also, if
  \begin{itemize}
  \item $\npair{\stpair{\Phi}{\Phi}}\in\lrec[\succeq,\gc](W)$
  \item $i \leq n$
  \item $\Phi_T \term[i]$,
  \end{itemize}
  then $\Phi_S\sterm{\gc}$.
\end{lemma}
\begin{proof}
  First, assume that 
  \begin{itemize}
  \item $\npair{\stpair{\Phi}{\Phi}}\in\lrec[\preceq](W)$
  \item $i \leq n$
  \item $\Phi_S \sterm[i]{\gc}$
  \end{itemize}

  Assume w.l.o.g. that $\Phi_S = (\ms_S,\reg_S,\stk,\ms_\stk)$, $\Phi_T = (\ms_T,\reg_T)$ and $\gc = (\ta,\stkb)$.
  Then it follows from $\npair{\stpair{\Phi}{\Phi}}\in\lrec[\preceq](W)$ that there exist $W_M,W_R,W_\pcreg$ such that
  \begin{itemize}
  \item $W = W_M \oplus W_R \oplus W_\pcreg$
  \item $\npair{( (\reg_S(\pcreg),\reg_S(\rdata)), (\reg_T(\pcreg),\reg_T(\rdata)) )} \in \lrexj[\preceq,\gc](W_\pcreg)$
  \item
    \begin{itemize}
    \item $\reg_S(\pcreg) \neq \retptrc(\_)$,
    \item $\reg_S(\rdata) \neq \retptrd(\_)$, 
    \item $\nonExec{\reg_S(\rdata)}$ and 
    \item $\nonExec{\reg_T(\rdata)}$
    \end{itemize}
  \item $\memSat{\ms_S,\ms_\stk,\stk,\ms_T}{W_M}$
  \item $\npair{\stpair{\reg}{\reg}} \in \lrr[\preceq,\gc](\{\rdata\})(W_R)$
  \end{itemize}

  We can then instantiate the conditions from $\lrexj$ with the other conditions from this list to obtain $\Phi_S',\Phi_T'$ such that
  \begin{itemize}
  \item $\Phi_S' = \xjumpResult{\reg_S(\pcreg)}{\reg_S(\rdata)}{\Phi_S}$ 
  \item $\Phi_T' = \xjumpResult{\reg_T(\pcreg)}{\reg_T(\rdata)}{\Phi_T}$
  \item $\npair{\left(\Phi_S', \Phi_T' \right)}\in \lrol$
  \end{itemize}

  Using the facts that $\reg_S(\pcreg) \neq \retptrc(\_)$, $\reg_S(\rdata) \neq \retptrd(\_)$ and $\nonExec{\reg_S(\rdata)}$, we know (by definition of $\xjumpResult{\_}{\_}{\_}$) that  
  \begin{itemize}
  \item $\Phi_S' = \xjumpResult{\reg_S(\pcreg)}{\reg_S(\rdata)}{\Phi_S} = \Phi_S\updReg{\pcreg}{\reg_S(\pcreg)}\updReg{\rdata}{\reg_S(\rdata)} = \Phi_S$ 
  \item $\Phi_T' = \xjumpResult{\reg_T(\pcreg)}{\reg_T(\rdata)}{\Phi_T} = \Phi_T\updReg{\pcreg}{\reg_T(\pcreg)}\updReg{\rdata}{\reg_T(\rdata)} = \Phi_T$ 
  \end{itemize}

  From $\npair{\left(\Phi_S, \Phi_T \right)}\in \lrol$, it follows immediately that if $\Phi_S \sterm[i]{\ta,\stkb}$ with $i \le n$, then $\Phi_T \term$ as required.

  The proof in the other direction is directly analogous.
\end{proof}

\begin{lemma}[Compatibility lemma for context plugging]
  \label{lem:compat-context-plug}
  If $\npair{\stpair{C}{C}}\in\lrcomp(W_1)$ and $\npair{\stpair{P}{P}} \in \lrcomp(W_2)$, then
  $\npair[n']{(\plug{C_S}{P_S},\plug{C_T}{P_T})} \in \lrec(W_1 \uplus W_2)$, for all $n' < n$.
\end{lemma}
\begin{proof}
  By definition, we have that $C_S \bowtie P_S \rightsquigarrow C_S[P_S]$ and $C_T \bowtie P_T \rightsquigarrow C_T[P_T]$.

  The result follows directly from Lemmas~\ref{lem:compat-linking} and~\ref{lem:compat-initial-ec}.
\end{proof}

\section{Full Abstraction}
\label{sec:full-abstraction}
\begin{definition}[Source language contextual equivalence]
  In the source language, we define that $\src{\var{comp}_1 \sconeq \var{comp}_2}$ iff
  \begin{equation*}
    \forall \src{\context} \ldotp \emptyset \vdash \src{\context} \Rightarrow \src{\plug{\context}{\var{comp}_1} \sterm{\ta[,1],\stkb_1}} \Leftrightarrow \src{\plug{\context}{\var{comp}_2} \sterm{\ta[,2],\stkb_2}}
  \end{equation*}
  with $\src{\ta[,i]} = \src{\dom(\var{comp}_i.\mscode)}$
  \lau{06-07-2018: TODO: Dominique, did you resolve this?}
  \dominique{and stkbase?}
\end{definition}

Note that we define source language contextual equivalence with respect to contexts that are not in $\ta$.
This means that they are unable to perform calls.
We believe this fits with the goal of this work: allow programmers (or better: authors of previous compiler passes) to reason about their target language programs under a special perspective, where all calls can be interpreted as calls that actually behave in a well-bracketed way by the operational semantics.
This special perspective is defined by overlaying a different operational semantics on the existing code: the source language semantics.
The fact that source contexts cannot make calls themselves is no problem: authors of previous compiler passes should only be able to take the perspective that their own calls are guaranteed to be well-bracketed.
It does not matter for them whether calls in the rest of the system are guaranteed to be well-bracketed.
Note also that if the trusted code hands out closures, the context can still invoke them with an xjmp, rather than a call.
That xjmp can even be the one in the implementation of call if the context uses that implementation. 
This works perfectly fine, except that the context does not get any well-bracketedness guarantees, but that doesn't matter.

\begin{definition}[Target language contextual equivalence]
  In the target language, we define that $\trg{\var{comp}_1 \tconeq \var{comp}_2}$ iff
  \begin{equation*}
    \forall \trg{\context} \ldotp \emptyset \vdash \trg{\context} \Rightarrow \trg{\plug{\context}{\var{comp}_1} \term} \Leftrightarrow \trg{\plug{\context}{\var{comp}_2} \term}
  \end{equation*}
\end{definition}

\begin{theorem}
  \label{thm:full-abstraction}
  For reasonable, well-formed components $\var{comp}_1$ and $\var{comp}_2$ (with respect to $\ta[,i] = \dom(\var{comp}_i.\mscode)$, respectively), we have
  \begin{gather*}
    \src{\var{comp}_1} \sconeq \src{\var{comp}_2}\\
    \Updownarrow\\
    \src{\var{comp}_1} \tconeq \src{\var{comp}_2}
  \end{gather*}
\end{theorem}
\begin{proof}
  \begin{itemize}
  \item Consider first the upward arrow.
    Assume $\src{\var{comp}_1} \tconeq \src{\var{comp}_2}$.

    Take a $\src{\context}$ such that $\emptyset \vdash \src{\context}$, take $\src{\ta[,i]}
    = \src{\dom(\var{comp}_i.\mscode)}$, $\gsigrets_i = \var{comp}_i.\sigrets$ and
    $\gsigcloss_i = \var{comp}_i.\sigcloss$, $\gc_i = (\ta[,i],\stkb_i,\gsigrets_i,\gsigcloss_i)$ and we will prove that
    $\src{\plug{\context}{\var{comp}_1} \sterm{\gc_1}} \Leftrightarrow
    \src{\plug{\context}{\var{comp}_2} \sterm{\gc_2}}$.

    By symmetry, we can assume w.l.o.g. that $\src{\plug{\context}{\var{comp}_1} \sterm{\gc_1}}$ and prove that $\src{\plug{\context}{\var{comp}_2} \sterm{\gc_2}}$.
    Note that this implies that $\src{\context}$ is a valid context for both $\src{\var{comp}_1}$ and $\src{\var{comp}_2}$.

    First, we show that also $\plug{\context}{\var{comp}_1} \trg{\term}$.
    Take $n$ the amount of steps in the termination of $\src{\plug{\context}{\var{comp}_1} \sterm{\gc_1}}$.
    It follows from Lemma~\ref{lem:ftlr-comps} that $\npair[n+1]{(\var{comp}_1,\var{comp}_1)} \in \lrcomp[\preceq,\gc_1](W_1)$ for some $W_1$ with $\dom(\pwfree) = \dom(\pwpriv) = \emptyset$.
    It also follows from the same Lemma~\ref{lem:ftlr-comps} and Lemma~\ref{lem:context-trusted-addr-ampl} that $\npair[n+1]{(\context,\context)} \in \lrcomp[\preceq,\gc_1](W_1')$ for some $W_1'$ that we can choose such that $W_1 \uplus W_1'$ is defined.
    Lemma~\ref{lem:compat-context-plug} then tells us that $\npair{(\plug{\context}{\var{comp}_1}, \plug{\context}{\var{comp}_1})} \in \lrec[\preceq,\gc_1](W_1\uplus W_1')$
    Together with $\src{\plug{\context}{\var{comp}_1} \sterm[n]{\gc_1}}$, Lemma~\ref{lem:adequacy} then tells us that $\plug{\context}{\var{comp}_1} \trg{\term}$.

    It follows from $\src{\var{comp}_1} \tconeq \src{\var{comp}_2}$ that also $\plug{\context}{\var{comp}_2} \trg{\term}$.

    It now remains to show that also $\src{\plug{\context}{\var{comp}_2} \sterm{\gc_2}}$.
    Take $n'$ the amount of steps in the termination of $\plug{\context}{\var{comp}_2} \trg{\term}$.
    It follows from Lemma~\ref{lem:ftlr-comps} that $\npair[n'+1]{(\var{comp}_2,\var{comp}_2)} \in \lrcomp[\succeq,\gc_2](W_2)$ for some $W_2$ with $\dom(\pwfree) = \dom(\pwpriv) = \emptyset$.
    It also follows from the same Lemma~\ref{lem:ftlr-comps} and Lemma~\ref{lem:context-trusted-addr-ampl} that $\npair[n'+1]{(\context,\context)} \in \lrcomp[\succeq,\gc_2](W_2')$ for some $W_2'$ that we can choose such that $W_2 \uplus W_2'$ is defined.
    Lemma~\ref{lem:compat-context-plug} then tells us that $\npair[n']{(\plug{\context}{\var{comp}_2}, \plug{\context}{\var{comp}_2})} \in \lrec[\succeq,\gc_2](W_2\uplus W_2')$
    Together with $\plug{\context}{\var{comp}_2} \trg{\term[n']}$, Lemma~\ref{lem:adequacy} then tells us that $\src{\plug{\context}{\var{comp}_2} \sterm{\gc_2}}$, concluding this direction of the proof.

  \item The downward arrow is similar.

    Assume $\src{\var{comp}_1} \sconeq \src{\var{comp}_2}$.
    Take $\src{\ta[,i]} = \src{\dom(\var{comp}_i.\mscode)}$, $\gsigrets_i = \var{comp}_i.\sigrets$ and $\gsigcloss_i = \var{comp}_i.\sigcloss$, $\gc_i = (\ta[,i],\stkb_i,\gsigrets_i,\gsigcloss_i)$.

    Take a $\trg{\context}$ such that $\emptyset \vdash \trg{\context}$ and we will prove that
    $\trg{\plug{\context}{\var{comp}_1} \term} \Leftrightarrow
    \trg{\plug{\context}{\var{comp}_2} \term}$.

    By symmetry, we can assume w.l.o.g. that $\trg{\plug{\context}{\var{comp}_1} \term}$ and prove that $\trg{\plug{\context}{\var{comp}_2} \term}$.
    Note that this implies that $\trg{\context}$ is a valid context for both $\trg{\var{comp}_1}$ and $\trg{\var{comp}_2}$.

    First, we show that also $\plug{\context}{\var{comp}_1} \src{\sterm{\gc_1}}$.
    Take $n$ the amount of steps in the termination of $\plug{\context}{\var{comp}_1} \trg{\term}$.
    It follows from Lemma~\ref{lem:ftlr-comps} that $\npair[n+1]{(\var{comp}_1,\var{comp}_1)} \in \lrcomp[\succeq,\gc_1](W_1)$ for some $W_1$ with $\dom(\pwfree) = \dom(\pwpriv) = \emptyset$.
    It also follows from the same Lemma~\ref{lem:ftlr-comps} and Lemma~\ref{lem:context-trusted-addr-ampl} that $\npair[n+1]{(\context,\context)} \in \lrcomp[\succeq,\gc_1](W_1')$ for some $W_1'$ that we can choose such that $W_1 \uplus W_1'$ is defined.
    Lemma~\ref{lem:compat-context-plug} then tells us that $\npair{(\plug{\context}{\var{comp}_1}, \plug{\context}{\var{comp}_1})} \in \lrec[\succeq,\gc_1](W_1\uplus W_1')$
    Together with $\plug{\context}{\var{comp}_1} \trg{\term[n]}$, Lemma~\ref{lem:adequacy} then tells us that $\plug{\context}{\var{comp}_1} \src{\sterm{\gc_1}}$.

    It follows from $\src{\var{comp}_1} \sconeq \src{\var{comp}_2}$ that also $\plug{\context}{\var{comp}_2} \src{\sterm{\gc_2}}$.

    It now remains to show that also $\plug{\context}{\var{comp}_2} \trg{\term}$.
    Take $n'$ the amount of steps in the termination of $\plug{\context}{\var{comp}_2} \src{\sterm{\gc_2}}$.
    It follows from Lemma~\ref{lem:ftlr-comps} that $\npair[n'+1]{(\var{comp}_2,\var{comp}_2)} \in \lrcomp[\preceq,\gc_2](W_2)$ for some $W_2$ with $\dom(\pwfree) = \dom(\pwpriv) = \emptyset$.
    It also follows from the same Lemma~\ref{lem:ftlr-comps} and Lemma~\ref{lem:context-trusted-addr-ampl} that $\npair[n'+1]{(\context,\context)} \in \lrcomp[\preceq,\gc_2](W_2')$ for some $W_2'$ that we can choose such that $W_2 \uplus W_2'$ is defined.
    Lemma~\ref{lem:compat-context-plug} then tells us that
    \begin{equation*}
      \npair[n']{(\plug{\context}{\var{comp}_2}, \plug{\context}{\var{comp}_2})} \in \lrec[\preceq,\gc_2](W_2\uplus W_2')
    \end{equation*}
    Together with $\plug{\context}{\var{comp}_2} \trg{\term[n']}$, Lemma~\ref{lem:adequacy} then tells us that $\plug{\context}{\var{comp}_2} \trg{\term}$, concluding the second direction of the proof.
  \end{itemize}
\end{proof}

\section{Lemmas}
\begin{lemma}
  \label{lem:unique-h-sigma}
  If $r,r' \in \dom(\pwheap)$ and $\pwheap(r) = (\pure,\_,H_\sigma)$
  and $\pwheap(r') = (\pure,\_,H_\sigma')$ and $\sigma \in \dom(H_\sigma)$
  and $\sigma \in \dom(H_\sigma')$ and $\memSat{\ms_S,\stk,\ms_\stk,\ms_T}{W}$,
  then
  \[
    H_\sigma = H_\sigma' \text{ and } r = r'
  \]
\end{lemma}
\begin{proof}
  This follows easily by definition of $\memSat{\ms_S,\stk,\ms_\stk,\ms_T}{W}$ and $\lrheap$.
\end{proof}

\dominique{Why is the below lemma needed?}
\begin{lemma}
  if $\npair{[\baddr,\eaddr]} \in \stackReadCond{W}$, and
  $\npair{[\baddr,\eaddr]} \in \stackWriteCond{W}$, and\\
  $\memSat{\ms_S,\stk,\ms_\stk,\ms_T}{W}$, then
  \[
    \begin{array}{l}
      \exists S \subseteq \addressable{\linear,\pwfree} \ldotp \\
      \quad \exists R : S \fun \powerset{\nats} \ldotp \\
      \qquad \biguplus R(r) = [\baddr,\eaddr] \wedge \\
      \qquad \forall r \in S \ldotp \pwfree(r).H \nequal \stdreg{R(r),\gc}{\pur}.H \wedge \\
      \qquad \quad |R(r)| = 1 \wedge \\
      \qquad \quad \pwfree(r) \text{ is address-stratified}
    \end{array}
  \]
\end{lemma}
\begin{proof}
  By assumption we get
  \begin{enumproof}
    \item $S_R \subseteq \addressable{\linear,\pwfree}$
    \item $R_R : S_R \fun \powerset{\nats}$
    \item $S_W \subseteq \addressable{\linear,\pwfree}$
    \item $R_W : S_W \fun \powerset{\nats}$
  \end{enumproof}
  such that
  \begin{enumproof}[resume]
    \item $\biguplus R_R(r) = [\baddr,\eaddr]$
    \item $\forall r \in S_R \ldotp \pwfree(r).H \nsubeq
      \stdreg{R_R(r),\gc}{\pur}.H \text{ and } |R_R(r)|=1$
    \item $\biguplus R_W(r) = [\baddr,\eaddr]$
    \item $\forall r \in S_W \ldotp \pwfree(r).H \nsupeq
      \stdreg{R_W(r),\gc}{\pur}.H \text{ and } |R_R(r)|=1 \text{ and } \pwfree(r)
      \text{ is address-stratified}$
  \end{enumproof}
  Now pick $S = S_R \cup S_W$ and show
  \[
    \forall r \in S \ldotp R_R(r) = R_W(r)
  \]
  Assuming for contradiction $\exists r \in S \ldotp R_R(r) \neq R_W(r)$, we have $\pwfree(r).H \nsubeq \stdreg{*,\gc}{\pur}.H$ and $\pwfree(r).H \nsupeq \stdreg{-,\gc}{\pur}.H$ for two distinct singleton sets $*$ and $-$.
  Due to the memory satisfaction assumption $\pwfree(r)$ is non-empty and it is trivial to show that so is the standard regions.
  From the above, we can conclude that $\pwfree(r)$ should contain some $\ms$
  with $\dom(\ms) = *$ and $\dom(\ms) = -$ which contradicts the address
  stratification assumption.

  Now pick
  \[
    R = R_R(r) \text{ for $r \in S$}
  \]
  We first show $\biguplus R(r) \supseteq [\baddr,\eaddr]$.
  For $\aaddr \in [\baddr,\eaddr]$ we know that there exists $r_1$ and $r_2$ such that $R_W(r_1) = R_R(r_2) = \{\aaddr\}$.
  It suffices to show $r_1 = r_2$.
  To this end assume the contrary.
  This, the address-stratification assumption on all write regions, the $n$-subset of a standard region for read regions, and memory satisfaction lead to a contradiction as any part of memory only can be governed by one region.
  Address-stratification and the singleton requirement follows trivially from
  the assumptions from the write region. Finally the $n$-equality follows from Lemma~\ref{lem:nsub-and-nsup-std}.
\end{proof}

\begin{lemma}
  \label{lem:nsub-and-nsup-std}
  If $\iota.H \nsubeq \stdreg{A,\gc}{v}.H$ and $\iota.H \nsupeq \stdreg{A,\gc}{v}.H$, then $\iota.H \nequal  \stdreg{A,\gc}{v}.H$
\end{lemma}
\begin{proof}
  By definition.
\end{proof}

\section{Proofs}
\subsection{Lemmas}
In this section, I have listed lemmas that seem to be necessary for the FTLR proof.
\begin{lemma}[Downwards closure of relations]
  \label{lem:downwards-closed}
  If $n' \leq n$, then
  \begin{itemize}
  \item If $\npair{A} \in \readCond{\lin,W}$, then $\npair[n']{A} \in \readCond{\lin,W}$.
  \item If $\npair{A} \in \stackReadCond{\lin,W}$, then $\npair[n']{A} \in \stackReadCond{\lin,W}$.
  \item If $\npair{A} \in \writeCond{\lin,W}$, then $\npair[n']{A} \in \writeCond{\lin,W}$.
  \item If $\npair{A} \in \stackWriteCond{\lin,W}$, then $\npair[n']{A} \in \stackWriteCond{\lin,W}$.
  \item If $\npair{A} \in \execCond{\lin,W}$, then $\npair[n']{A} \in \execCond{\lin,W}$.
  \item If $\npair{A} \in \xReadCond{W}$, then $\npair[n']{A} \in \xReadCond{W}$.
  \item If $\npair{(w,w')} \in \lrvg{\trust}(W)$, then $\npair[n']{(w,w')} \in \lrvg{\trust}(W)$.
  \item If $\npair{(w,w')} \in \lrrg{\trust}(W)$, then $\npair[n']{(w,w')} \in \lrrg{\trust}(W)$.
  \item If $\memSatFStack{ms_\stk,\ms_T}{W}$, then $\memSatFStack[n']{ms_\stk,\ms_T}{W}$.
  \item If $\memSatStack{ms_\stk,\ms_T}{W}$, then $\memSatStack[n']{ms_\stk,\ms_T}{W}$.
  \item If $\npair{(\overline{\sigma},\ms_S,\ms_T)} \in \lrheap{\pwheap}(W)$, then $\npair[n']{(\overline{\sigma},\ms_S,\ms_T)} \in \lrheap{\pwheap}(W)$.
  \item $\memSat{(\ms_S,\stk,\ms_\stk,\ms_T)}{W}$, then also $\memSat[n']{(\ms_S,\stk,\ms_\stk,\ms_T)}{W}$.
  \end{itemize}
\end{lemma}
\begin{proof}
  The properties for $\readCond{}$, $\stackReadCond{}$, $\writeCond{}$, $\stackWriteCond{}$, $\execCond{}$, $\xReadCond{}$ follow easily by definition.
  The property for $\lrv$ follows from the others and by definition.
  The property for $\lrr$ follows directly from the one for $\lrv$.
  
  The property for $\lrheap$, $\memSatFStack{\_,\_}{\_}$, $\memSatStack{\_,\_}{\_}$ follows by their definition from the quantifications over $n' < n$, and the property for $(\_,\_,\_,\_) :_n W$ follows from those.
\end{proof}

\begin{lemma}[Properties of n-equality of worlds]
  \label{lem:n-equality-props}
  If $W_1 \nequal W_2$ then
  \begin{itemize}
  \item $\purePart{W_1} \nequal \purePart{W_2}$
  \item If $W_1' \future W_1$ then there exists a $W_2'$ such that $W_2' \nequal W_1'$ and $W_2' \future W_2$.
  \item If $W_1' \oplus W_1$ is defined, then there exists a $W_2' \nequal W_1'$ and $W_2' \oplus W_2$ is defined.
  \item If $W_1 = W_1' \oplus W_1''$ then there exist $W_2' \nequal W_1'$ and $W_2'' \nequal W_1''$ such that $W_2 = W_2' \oplus W_2''$.
  \item If $W_1' \nequal W_2'$ then $W_1\oplus W_1' \nequal W_2 \oplus W_2'$.
  \item $\xi^{-1}(W_1) \nequal[n-1] \xi^{-1}(W_2)$
  \end{itemize}
\end{lemma}
\begin{proof}
  Easy to prove by unfolding definitions and making unsurprising choices for existentially quantified worlds.
\end{proof}

\begin{lemma}[Non-expansiveness of relations]
  \label{lem:non-expansiveness}
  If $W \nequal W'$, then
  \begin{itemize}
  \item $\readCond{W} \nequal \readCond{W'}$.
  \item If $\stackReadCond{\lin,W} \nequal \stackReadCond{\lin,W'}$.
  \item If $\writeCond{\lin,W} \nequal \writeCond{\lin,W'}$.
  \item If $\stackWriteCond{\lin,W} \nequal \stackWriteCond{\lin,W'}$
  \item If $\xReadCond{W} \nequal \xReadCond{W'}$.
  \item $\lre(W) \nequal \lre(W')$.
  \item $\lrexj(W) \nequal \lrexj(W')$.
  \item If $\execCond{\lin,W} \nequal \execCond{\lin,W'}$.
  \item If $\lrvg{\trust}(W) \nequal \lrvg{\trust}(W')$.
  \item If $\lrrg{\trust}(W) \nequal \lrrg{\trust}(W')$.
  \item If $\memSatFStack{ms_\stk,\ms_T}{W}$, then $\memSatFStack{ms_\stk,\ms_T}{W'}$.
  \item If $\memSatStack{ms_\stk,\ms_T}{W}$, then $\memSatStack{ms_\stk,\ms_T}{W'}$.
  \item If $\npair{(\overline{\sigma},\ms_S,\ms_T)} \in \lrheap{\pwheap}(W)$, then $\npair{(\overline{\sigma},\ms_S,\ms_T)} \in \lrheap{\pwheap[W']}(W')$.
  \item $\memSat{(\ms_S,\stk,\ms_\stk,\ms_T)}{W}$ iff $\memSat{(\ms_S,\stk,\ms_\stk,\ms_T)}{W'}$.
  \end{itemize}
\end{lemma}
\begin{proof}
  The properties for $\readCond{}$, $\stackReadCond{}$, $\writeCond{}$, $\stackWriteCond{}$, $\xReadCond{}$ follow from the fact that these are defined to use the world only for comparing regions using $\nequal$, $\nsubeq$ and $\nsupeq$. 

  The property for $\lre$ and $\lrexj$ follow from Lemma~\ref{lem:n-equality-props}.

  The property for $\execCond{}$ follows from Lemma~\ref{lem:n-equality-props} and the property for $\lre$.

  The property for $\lrv$ follows from the other properties, by definition, from Lemma~\ref{lem:n-equality-props}, and non-expansiveness of regions.

  The property of $\lrr$ follows from the one for $\lrv$.

  The property for $\lrheap$, $\memSatFStack{\_,\_}{\_}$, $\memSatStack{\_,\_}{\_}$ and $(\_,\_,\_,\_) :_n W$ follows from the non-expansiveness of regions in the world and Lemma~\ref{lem:n-equality-props}.
\end{proof}

\begin{lemma}[World monotonicity of relations]
  \label{lem:monotonicity}
  For all $n$, $W' \future W$, we have that
  \begin{itemize}
  \item If $(w_1,w_2) \in H_\sigma~\sigma~W$, then $(w_1,w_2) \in H_\sigma~\sigma~W'$.
    \lau{This should hold for any valid $H_\sigma$, but we need to prove it for specific instantiations.}
  \item If $\npair{(w_1,w_2)} \in \lrv(W)$, then $\npair{(w_1,w_2)} \in
    \lrv(W')$.
  \item If $\npair{(\overline{\sigma},\ms_S,\ms_T)} \in \lrheap(\pwheap[W_1])(W)$, then $\npair{(\overline{\sigma},\ms_S,\ms_T)} \in
    \lrheap(\pwheap[W_1])(W')$.
\end{itemize}
\end{lemma}
\begin{proof}
  Follows from the definitions.
  Note: the proof for $\lrheap$ relies on Lemma~\ref{lem:oplus-future-distr}.
\end{proof}

\begin{lemma}[$\lror$ closed under target language antireduction]
  For all $\Phi_S$, $\Phi_T$, $\Phi_T'$, $j$, $n$, if
\[
  \Phi_T \nstep[j]{} \Phi_T' \text{ and } \npair[n-j]{(\Phi_S,\Phi_T')} \in \lror,
\]
then
\[
  \npair{(\Phi_S,\Phi_T)} \in \lror
\]
\end{lemma}
\begin{proof}
  Special case of Lemma~\ref{lem:lro-anti-red-gen}.
\end{proof}

\begin{lemma}[$\lrol$ closed under source language antireduction]
  For all $\Phi_S$, $\Phi_T$, $\Phi_T'$, $j$, $n$, if
\[
  \Phi_S \nstep[j]{\gc} \Phi_S' \text{ and } \npair[n-j]{(\Phi_S',\Phi_T)} \in \lrol,
\]
then
\[
  \npair{(\Phi_S,\Phi_T)} \in \lrol
\]
\end{lemma}
\begin{proof}
  Special case of Lemma~\ref{lem:lro-anti-red-gen}.
\end{proof}

\begin{lemma}[$\lro$ closed under antireduction (generalised previous lemma)]
  \label{lem:lro-anti-red-gen}
  For all $\Phi_S$, $\Phi_S'$, $\Phi_T$, $\Phi_T'$, $j_S,j_T$, $n$, if
  \begin{itemize}
  \item $\Phi_S \nstep[j_S]{\gc} \Phi_S'$
  \item $\Phi_T \nstep[j_T]{} \Phi_T'$
  \item $\npair{(\Phi_S',\Phi_T')} \in \lro$
  \end{itemize}
  then
  \[
    \npair[n+\min(j_S,j_T)]{(\Phi_S,\Phi_T)} \in \lro
  \]
\end{lemma}
\begin{proof}
  Our two languages are deterministic\dominique{Lemma needed for this, I guess..}, so we have that $\Phi_S \sterm{\ta,\stkb}$ iff $\Phi_S \sterm{\ta,\stkb}$ and $\Phi_T \term$ iff $\Phi_T' \term$.
  It is also easy to show that if $\Phi_S \sterm[i + \min(j_S,j_T)]{\ta,\stkb}$, then $\Phi_S' \sterm[i]{\ta,\stkb}$ and if $\Phi_T \term[i+\min(j_S,j_T)]$, then $\Phi_T' \term[i]$.
  The result then follows easily by definition of $\lro[\preceq,\gc]$ and $\lro[\preceq,\gc]$.
\end{proof}

\begin{lemma}[Capability safety doesn't depend on address]
  \label{lem:cap-in-lrv-regardless-of-addr}
  If
  \[
    \npair{((\perm,\lin),\baddr,\eaddr,\aaddr),((\perm,\lin),\baddr,\eaddr,\aaddr)} \in \lrvg{\trust}(W)
  \]
  then
  \[
    \npair{((\perm,\lin),\baddr,\eaddr,\aaddr'),((\perm,\lin),\baddr,\eaddr,\aaddr')} \in \lrvg{\trust}(W)
  \]
\end{lemma}
\begin{proof}
  Direct form the definition of $\lrvg{\trust}$.
\end{proof}

\begin{lemma}[Stack capability safety doesn't depend on address]
\label{lem:stkptr-in-lrv-regardless-of-addr}
  If
  \[
    \npair{((\stkptr{\perm,\baddr,\eaddr,\aaddr},(\perm,\lin),\baddr,\eaddr,\aaddr))} \in \lrvg{\trust}(W)
  \]
  then
  \[
    \npair{(\stkptr{\perm,\baddr,\eaddr,\aaddr'},((\perm,\lin),\baddr,\eaddr,\aaddr'))} \in \lrvg{\trust}(W)
  \]
\end{lemma}
\begin{proof}
  Direct from the definition of $\lrvg{\trust}$.
\end{proof}

\begin{lemma}[Seal safety doesn't depend on current seal]
  \label{lem:seal-in-lrv-regardless-of-cur-seal}
  If
  \[
    \npair{(\seal{\sigma_\baddr,\sigma_\eaddr,\sigma},\seal{\sigma_\baddr,\sigma_\eaddr,\sigma})} \in \lrvg{\trust}(W)
  \]
  then
  \[
    \npair{(\seal{\sigma_\baddr,\sigma_\eaddr,\sigma'},\seal{\sigma_\baddr,\sigma_\eaddr,\sigma'})} \in \lrvg{\trust}(W)
  \]
\end{lemma}
\begin{proof}
  Direct from the definition of $\lrvg{\trust}$.
\end{proof}

\begin{lemma}[Capability safety monotone w.r.t. permission]
  \label{lem:cap-in-lrv-mono-perm}
  If
  \[
    \npair{((\perm,\lin),\baddr,\eaddr,\aaddr),((\perm,\lin),\baddr,\eaddr,\aaddr)} \in \lrvg{\trust}(W)
  \]
  and
  \[
    \perm' \sqsubseteq \perm
  \]
  then
  \[
    \npair{((\perm',\lin),\baddr,\eaddr,\aaddr),((\perm',\lin),\baddr,\eaddr,\aaddr)} \in \lrvg{\trust}(W)
  \]
\end{lemma}
\begin{proof}
  Direct from the definition of $\sqsubseteq$ and $\lrvg{\trust}$.
\end{proof}

\begin{lemma}[Capability splitting retains safety for normal capabilities]
  \label{lem:splitting-safety-normal}
  If
  \begin{itemize}
  \item $\npair{(c,c)} \in \lrvg{\trust}(W)$
  \item $c = ((\perm,\lin),\baddr,\eaddr,\aaddr)$
  \item $b \le s < e$
  \item $c_1 = ((\perm,\lin),\baddr,s,\aaddr)$
  \item $c_2 = ((\perm,\lin),s+1,\eaddr,\aaddr)$
  \item $c_3 = \linCons{c}$
  \item $\memSat{\ms_S,\stk,\ms_\stk,\ms_T}{W \oplus \purePart{W_1 \oplus W_2 \oplus W_3}}$
  \end{itemize}
  then there exist $W_1,W_2$,$W_3$ such that
  \begin{itemize}
  \item $W = W_1 \oplus W_2 \oplus W_3$
  \item $\npair{(c_1,c_1)} \in \lrvg{\trust}(W_1)$
  \item $\npair{(c_2,c_2)} \in \lrvg{\trust}(W_2)$
  \item $\npair{(c_3,c_3)} \in \lrvg{\trust}(W_3)$
  \end{itemize}
\end{lemma}
\begin{proof}

  If $\lin = \normal$, then we pick $W_3 = W$ and $W_1 = W_2 = \purePart{W_3}$ which easily satisfies $W = W_1 \oplus W_2 \oplus W_3$.

  Assuming $\gc = (\ta,\gc,\gsigrets,\gsigcloss)$, we know by assumption $\npair{(c,c)} \in \lrvg{\trust}(W)$ that either
\begin{itemize}
\item $\ta \subseteq [\baddr,\eaddr]$; or
\item $\ta \# [\addr,\eaddr]$
\end{itemize}
In the former case, the result follows from Lemma~\ref{lem:non-linear-pure} and Lemma~\ref{lem:conds-shrinkable}.

In the latter case, we know that $[\baddr,s],[s+1,\eaddr] \neq \emptyset$ because $b \le s < e$, so the result follows from Lemma~\ref{lem:non-linear-pure} and Lemma~\ref{lem:conds-shrinkable}.

If $\lin = \linear$, then 
we know from $\npair{(c,c)} \in \lrvg{\trust}(W)$ and w.l.o.g we can assume
\begin{itemize}
\item $\npair{[\baddr,\eaddr]} \in \readCond{\linear,W}$
\item $\npair{[\baddr,\eaddr]} \in \writeCond{\linear,W}$
\end{itemize}
from this, we get $S_{\var{read}} \subseteq \addressable{\linear,\pwheap{W}}$ and
$R_{\var{read}} : S_{\var{read}} \fun \Worlds$ such that
\begin{itemize}
\item $\biguplus_{r\in S_{\var{read}}} R_{\var{read}}(r) \supseteq A$
\item $\forall r\ldotp |R_{\var{read}}(r)|  = 1$
\item $\forall r \in S_{\var{read}} \ldotp \pwheap(r).H \nsubeq \stdreg{R_{\var{read}}(r),\gc}{\pur}.H$
\end{itemize}
and
$S_{\var{write}} \subseteq \addressable{\linear,\pwheap{W}}$ and
$R_{\var{write}} : S_{\var{write}} \fun \Worlds$ such that
\begin{itemize}
\item $\biguplus_{r\in S_{\var{write}}} R_{\var{write}}(r) \supseteq A$
\item $\forall r\ldotp |R_{\var{write}}(r)|  = 1$
\item $\forall r \in S_{\var{write}} \ldotp \pwheap(r).H \nsupeq
  \stdreg{R_{\var{write}}(r),\gc}{\pur}.H \wedge \pwheap(r) $ is address-stratified
\end{itemize}

Now we would like to show $R_{\var{read}}^{-1}([\baddr,\eaddr]) = R_{\var{write}}^{-1}([\baddr,\eaddr])$.
We know that the two sets are the same size as both $R$'s map to singleton sets.
This means that there exists $r \neq r'$ s.t.\ $R_{\var{read}}(r) =
R_{\var{write}}(r') = \{\aaddr'\}$ for $\aaddr' \in [\baddr,\eaddr]$.

By definition of the standard region, we know that for any $\hat{W}$ and any $\ms$ and $\ms'$ where $\npair[n-1]{(\ms,\ms')} \in \pwheap(r)$ we have $\dom(\ms)= \dom(\ms') = \{a'\}$.

By assumption $\pwheap(r)$ is address-stratified which means that for any $\hat{W}$ and any $\ms$ and $\ms'$ where $\npair[n-1]{(\ms,\ms')} \in \pwheap(r)$ we have $\dom(\ms)e= \dom(\ms') = \{a'\}$.

By the memory satisfaction, the memory must be split into disjointed parts that
each satisfy a region. With two regions that require a memory segment pair with
the same domain, we cannot satisfy all the regions, so we must have

Now pick $W_1$ as the world that owns $R_{\var{read}}^{-1}([\baddr,s])$, $W_2$
the world that owns $R_{\var{read}}^{-1}([s+1,\eaddr])$, and $W_3$ as the world
that owns the remaining regions of $W$.

It is clearly the case that $W = W_1 \oplus W_2 \oplus W_3$.

In this case, we need to show
\[
  \npair{(0,0)} \in \lrvg{\trust}(W_3)
\]
which is trivially the case.

For the remaining, it suffices to show
\begin{itemize}
\item $\npair{[\baddr,s]} \in \readCond{\linear,W_1}$
\item $\npair{[\baddr,s]} \in \writeCond{\linear,W_1}$
\item $\npair{[s+1,\eaddr]} \in \readCond{\linear,W_2}$
\item $\npair{[s+1,\eaddr]} \in \writeCond{\linear,W_2}$
\end{itemize}
which follows by assumption.
\end{proof}

\begin{lemma}[Capability splitting retains safety for stack capabilities]
  \label{lem:splitting-safety-stack}
  If
  \begin{itemize}
  \item $\npair{(c_S,c_T)} \in \lrvg{\trust}(W)$
  \item $c_S = \stkptr{\perm,\baddr,\eaddr,\aaddr}$
  \item $c_T = ((\perm,\linear),\baddr,\eaddr,\aaddr)$
  \item $b \le n \le e$
  \item $c_{1,S} = \stkptr{\perm,\baddr,n,\aaddr}$
  \item $c_{1,T} = ((\perm,\linear),\baddr,n,\aaddr)$
  \item $c_{2,S} = \stkptr{\perm,n,\eaddr,\aaddr}$
  \item $c_{2,T} = ((\perm,\linear),n,\eaddr,\aaddr)$
  \item $c_3 = 0$
  \end{itemize}
  then there exist $W_1,W_2$,$W_3$ such that
  \begin{itemize}
  \item $W = W_1 \oplus W_2 \oplus W_3$
  \item $\npair{(c_{1,S},c_{1,T})} \in \lrvg{\trust}(W_1)$
  \item $\npair{(c_{2,S},c_{2,T})} \in \lrvg{\trust}(W_2)$
  \item $\npair{(c_3,c_3)} \in \lrvg{\trust}(W_3)$
  \end{itemize}
\end{lemma}
\begin{proof}
  Similar to the proof of Lemma~\ref{lem:splitting-safety-normal}.
\end{proof}

\begin{lemma}[Capability splicing retains safety for normal capabilities]
  \label{lem:splicing-safety-normal}
  If
  \begin{itemize}
  \item $\npair{(c_1,c_1)} \in \lrvg{\trust}(W_1)$
  \item $\npair{(c_2,c_2)} \in \lrvg{\trust}(W_2)$
  \item $c_1 = ((\perm,\lin),\baddr,m,\aaddr)$
  \item $c_2 = ((\perm,\lin),m+1,\eaddr,\aaddr)$
  \item $c = ((\perm,\lin),\baddr,\eaddr,\aaddr)$
  \item $b \le m \le e$
  \item $c_1' = \linCons{c_1}$
  \item $c_2' = \linCons{c_2}$
  \item $W_1\oplus W_2 \oplus W_M$ is defined
  \item $\memSat{\ms_S,\stk,\ms_\stk,\ms_T}{W_M}$
  \end{itemize}
  then there exist $W_1',W_2',W_3'$ such that
  \begin{itemize}
  \item $W_1 \oplus W_2 = W_1' \oplus W_2' \oplus W_3'$
  \item $\npair{(c,c)} \in \lrvg{\trust}(W_3')$
  \item $\npair{(c_1',c_1')} \in \lrvg{\trust}(W_1')$
  \item $\npair{(c_2',c_2')} \in \lrvg{\trust}(W_2')$
  \end{itemize}
\end{lemma}
\begin{proof}
  From $\npair{(c_1,c_1)} \in \lrvg{\trust}(W_1)$ and $\npair{(c_2,c_2)} \in \lrvg{\trust}(W_2)$, it follows that either ($[b,m] \mathrel{\#} \ta$ or $[b,m] \subseteq \ta$) and also either ($[m+1,e] \mathrel{\#} \ta$ or $[m+1,e] \subseteq \ta$).

  Consider first the case where either $[b,m] \subseteq \ta$ or $[m+1,e]\subseteq \ta$.
  Then by definition of $\lrvg{\trust}$, we have that $\lin = \normal$, $\trust = \trusted$ and $\npair{[b,m]} \in \xReadCond{W_1}$ or $\npair{[m+1,e]} \in \xReadCond{W_2}$, respectively.
  It follows by definition of $\xReadCond{}$ and $\codereg{}$ that $[\baddr-1,m+1]$ or $[m,\eaddr+1] \subseteq \ta$, respectively, so that it is impossible that $[m+1,\eaddr]\mathrel{\#} \ta$ or $[\baddr,m] \mathrel{\#} \ta$ respectively.
  In other words, we must have that both $[b,m] \subseteq \ta$ and $[m+1,e]\subseteq \ta$.

  From Lemma~\ref{lem:conds-splicable}, it follows that also $\npair{[\baddr,\eaddr]}\in\xReadCond{W_1\oplus W_2}$ and from Lemma~\ref{lem:non-linear-pure}, it follows that $\npair{[\baddr,\eaddr]}\in\xReadCond{}(\purePart{W_1\oplus W_2})$.
  Finally, we have that $\npair{(c,c)}\in\lrvg{\trust}(\purePart{W_1\oplus W_2})$.
  We can take $W_3' = \purePart{W_1\oplus W_2}$, $W_1' = W_1$ and $W_2' = W_2$ and the remaining proof obligations follow by assumption and by Lemma~\ref{lem:purePart-duplicable} and~\ref{lem:oplus-assoc-comm}.
  
  Now consider the case that both $[b,m] \mathrel{\#} \ta$ and $[m+1,e]\mathrel{\#} \ta$.
  The results now follow by definition of $\lrv$, using Lemma~\ref{lem:conds-splicable} and~\ref{lem:non-linear-pure}, taking $W_3' = W_1 \oplus W_2$ and $W_1' = \purePart{W_1}$ and $W_2' = \purePart{W_2}$.
\end{proof}

\begin{lemma}[Capability splicing retains safety for stack capabilities]
  \label{lem:splicing-safety-stack}
  If
  \begin{itemize}
  \item $\npair{(c_{1,S},c_{1,T})} \in \lrvg{\trust}(W_1)$
  \item $\npair{(c_{2,S},c_{2,T})} \in \lrvg{\trust}(W_2)$
  \item $b \le m \le e$
  \item $c_{1,S} = \stkptr{\perm,\baddr,m,\aaddr}$
  \item $c_{1,T} = ((\perm,\linear),\baddr,m,\aaddr)$
  \item $c_{2,S} = \stkptr{\perm,m+1,\eaddr,\aaddr}$
  \item $c_{2,T} = ((\perm,\linear),m+1,\eaddr,\aaddr)$
  \item $c_S = \stkptr{\perm,\baddr,\eaddr,\aaddr}$
  \item $c_T = ((\perm,\linear),\baddr,\eaddr,\aaddr)$
  \end{itemize}
  then we have that
  \begin{itemize}
  \item $\npair{(c_S,c_T)} \in \lrvg{\trust}(W_1 \oplus W_2)$
  \end{itemize}
\end{lemma}
\begin{proof}
  This follows easily by definition of $\lrv$ and using Lemma~\ref{lem:conds-splicable}.
\end{proof}

\begin{lemma}[Stack capability safety monotone w.r.t. permission]
\label{lem:stkptr-in-lrv-mono-perm}
  If
  \[
    \npair{(\stkptr{\perm,\baddr,\eaddr,\aaddr},((\perm,\lin),\baddr,\eaddr,\aaddr))} \in \lrvg{\trust}(W)
  \]
  and
  \[
    \perm' \sqsubseteq \perm
  \]
  then
  \[
    \npair{(\stkptr{\perm',\baddr,\eaddr,\aaddr},((\perm',\lin),\baddr,\eaddr,\aaddr))} \in \lrvg{\trust}(W)
  \]
\end{lemma}
\begin{proof}
  Follows directly from the definition.
\end{proof}

\begin{lemma}[readCondition works]
  \label{lem:readcond-writecond-work}
  If
  \begin{itemize}
  \item $\memSat{(\ms_S,\stk,\ms_\stk,\ms_T)}{W_M}$
  \item $\npair{(b,e)}\in\readCond{l,W}$
  \item $a \in [b,e]$
  \item $W \oplus W_M$ is defined
  \item $n' < n$
  \end{itemize}
  Then $\npair[n']{(\ms_S(a),\ms_T(a))} \in \lrv(W')$ for some $W'$ such that $W_M = W' \oplus W_M'$.
  Additionally, if
  \begin{itemize}
  \item $\npair{(b,e)}\in\writeCond{l,W}$
  \end{itemize}
  Then $\memSat{(\ms_S[a\mapsto 0],\stk,\ms_\stk,\ms_T[a\mapsto 0])}{W_M'}$.
\end{lemma}
\begin{proof}
  From $\npair{(b,e)}\in\readCond{l,W}$, we get an $S \subseteq \addressable{\lin,
    \pwheap}$, an $R : S \rightarrow \powerset{\nats}$ with $\biguplus_{r \in S} R(r) \supseteq [\baddr,\eaddr]$ and $\pwheap(r).H
  \nsubeq \stdreg{R(r),\gc}{\pur}.H$ for all $r \in S$.

  Since $a \in [b,e]$, there is a unique $r \in S$ such that $a \in R(r)$.

  Since $W \oplus W_M$ is defined, we have that $r \in \dom(\pwheap) =
  \dom(\pwheap[W_M])$ and $\pwheap(r) \oplus \pwheap[W_M](r)$ is defined.

  From $\memSat{(\ms_S,\stk,\ms_\stk,\ms_T)}{W_M}$, we get that
  $\stk = (\opc_0,\ms_0):: \dots :: (\opc_m,\ms_m)$,
  $\ms_S \uplus \ms_\stk \uplus \ms_0 \uplus \dots \uplus \ms_m$ is defined,
  $W_M = W_{\var{stack}} \oplus W_{\var{free\_stack}} \oplus W_{\var{heap}}$ and
  $\exists \ms_\var{T,stack}, \ms_\var{T,free\_stack}, \ms_\var{T,heap}$, $\ms_{T,f}$, $\ms_{S,f}$, $\ms_S'$ such that
  \begin{itemize}
  \item $\ms_S =\ms_{S,f} \uplus \ms_S'$
  \item $\ms_T = \ms_\var{T,stack} \uplus \ms_\var{T,free\_stack} \uplus
    \ms_\var{T,heap} \uplus \ms_{T,f}$
  \item $\memSatStack{\stk,\ms_\var{T,stack}}{W_{\var{stack}}}$
  \item $\memSatFStack{\ms_\stk,\ms_\var{T,free\_stack}}{W_{\var{free\_stack}}}$
  \item $\npair{(\overline{\sigma},\ms_S',\ms_\var{T,heap})} \in \lrheap(W_M)(W_\var{heap})$.
  \end{itemize}

  From $\npair{(\overline{\sigma},\ms_S',\ms_\var{T,heap})} \in \lrheap(W_M)(W_\var{heap})$, we get an
  $R_\ms : \dom(\activeReg{\pwheap[W_\var{heap}]}) \fun \MemSeg \times \MemSeg$,
  $\ms_\var{T,heap} = \biguplus_{r \in \dom(\activeReg{\pwheap[W_\var{heap}]})} \pi_2(R_\ms(r))$,
  $\ms_S = \biguplus_{r \in \dom(\activeReg{\pwheap[W_\var{heap}]})} \pi_1(R_\ms(r))$,
  $\exists R_W : \dom(\activeReg{\pwheap[W_\var{heap}]}) \fun \World\ldotp$
  $W_\var{heap} = \oplus_{r \in \dom(\activeReg{\pwheap[W_\var{heap}]})} R_W(r)$,
  $\forall r \in \activeReg{\pwheap[W_M]}$, we have that
  $\npair[n']{R_\ms(r)} \in  \pwheap[W_\var{heap}](r).H \; \xi^{-1}(R_W(r))$ for all $n' < n$.
  We also get an $R_\var{seal} : \dom(\activeReg{\pwheap}) \fun \powerset{\Seal}$ such that  $\biguplus_{r \in \dom(\activeReg{\pwheap})} R_\var{seal}(r)) \subseteq \overline{\sigma}$ and $\dom(\pwheap(r).H_\sigma) = R_\var{seal}(r)$.

  We have that $r \in \addressable{\lin, \pwheap} \subseteq
  \activeReg{\pwheap[W_\var{heap}]}$, so $\npair[n']{R_\ms(r)} \in
  \pwheap[W_M](r).H \; \xi^{-1}(R_W(r))$.
  Because $\pwheap(r).H \nsubeq \stdreg{R(r),\gc}{\pur}.H$ and $W\oplus W_M$ is defined, it follows that
  also $\pwheap[W_M](r).H \nsubeq \stdreg{R(r),\gc}{\pur}.H$. This
  means that also $\npair[n']{R_\ms(r)} \in
  H^{\mathrm{std}}_{R(r)}\; \xi^{-1}(R_W(r))$.

  From this, it follows that $\dom(R_\ms(r).1) = \dom(R_\ms(r).2) = R(r)$ and we get a $S : R(r) \fun \World$ with $\xi(\xi^{-1}(R_W(r))) = \oplus_{\aaddr \in R(r)} S(\aaddr)$ and $\forall \aaddr \in R(r) \ldotp \npair[n']{(\ms_S(\aaddr),\ms_T(\aaddr))} \in \lrv(S(\aaddr))$.

  Since $\aaddr \in R(r)$, we have that $\npair[n']{(\ms_S(\aaddr),\ms_T(\aaddr))} \in \lrv(S(\aaddr))$ and
  we can take $W_M' = W_r' \oplus W_{\var{heap}}' \oplus (W_{\var{stack}} \oplus W_{\var{free\_stack}})$ with $W_r' = \oplus_{\aaddr \in (R(r)\setminus \{\aaddr\})} S(\aaddr)$ and $W_\var{heap}' = \oplus_{r' \in (\dom(\activeReg{\pwheap[W_\var{heap}]})\setminus \{r\})} R_W(r')$, and get
  \begin{align*}
    S(\aaddr) \oplus W_M'
    &=S(\aaddr) \oplus (W_r' \oplus W_\var{heap}'\oplus (W_{\var{stack}} \oplus W_{\var{free\_stack}}))\\
    &=
    \oplus_{\aaddr \in R(r)} S(\aaddr) \oplus W_{\var{heap}}'\oplus (W_{\var{stack}} \oplus W_{\var{free\_stack}})\\
    &=
    \xi(\xi^{-1}(R_W(r))) \oplus W_\var{heap}'\oplus (W_{\var{stack}} \oplus W_{\var{free\_stack}})\\
    &=
    R_W(r) \oplus W_\var{heap}'\oplus (W_{\var{stack}} \oplus W_{\var{free\_stack}})\\
    &=
    \oplus_{r \in \dom(\activeReg{\pwheap[W_\var{heap}]})} R_W(r)\oplus (W_{\var{stack}} \oplus W_{\var{free\_stack}})\\
    &=
      W_\var{heap} \oplus (W_{\var{stack}} \oplus W_{\var{free\_stack}})\\
      &= W_M
  \end{align*}

  Additionally, if $\npair{(b,e)}\in\writeCond{l,W}$, then
  we get an $S'\subseteq \addressable{\lin, \pwheap} \subseteq \activeReg{\pwheap[W_M]}$, an $R' : S' \fun \powerset{\nats}$ such that $\biguplus_{r \in S'} R'(r) \supseteq [\baddr,\eaddr]$ and 
  for all $r \in S'$, $\pwheap(r).H \nsupeq \stdreg{R'(r),\gc}{\pur}.H$ and
  $\pwheap(r)$ is address-stratified.

  Since $a \in [b,e]$, there is an $r' \in S'$ such that $a \in R'(r')$ .

  Because $W \oplus W_M$ is defined, it follows that also
  $\pwheap[W_M](r').H \nsupeq \stdreg{R'(r'),\gc}{\pur}.H$ and
  $\pwheap[W_M](r') $ is address-stratified.

  It follows that $r = r'$ because $\dom(R_\ms(r).1) = \dom(R_\ms(r).2) = R(r) \ni a$ and $\dom(R_\ms(r').2) = \dom(R_\ms(r').1) = R'(r') \ni a$ and all the $R_\ms(r).1$ and $R_\ms(r).2$ are disjoint.

  We have that $\npair[n']{(R_\ms(r).1\update{a}{0},R_\ms(r).2\update{a}{0})} \in \pwheap[W_\var{heap}](r).H \; \xi^{-1}(W_r')$ for all $n' < n$ because $\pwheap[W_M](r)$ is address-stratified and $\pwheap[W_M](r).H \nsupeq \stdreg{R'(r),\gc}{\pur}.H$.

  From this, it follows that $\npair{(\overline{\sigma},\ms_S\update{a}{0},\ms_\var{T,heap}\update{a}{0})} \in \lrheap(W_M)(W_r' \oplus W_\var{heap}')$ and finally
  \begin{equation*}
    \npair{(\overline{\sigma},\ms_S\update{a}{0},\stk,\ms_\stk,\ms_T\update{a}{0})} \in \lrheap(W_M')(W_r' \oplus W_{\var{heap}}')
  \end{equation*}
\end{proof}

\begin{lemma}[stackReadCondition works]
  \label{lem:stackreadcond-stackwritecond-work}
  If
  \begin{itemize}
  \item $\memSat{(\ms_S,\stk,\ms_\stk,\ms_T)}{W_M}$
  \item $\npair{(b,e)}\in\stackReadCond{W}$
  \item $a \in [b,e]$
  \item $W \oplus W_M$ is defined
  \item $n' < n$
  \end{itemize}
  Then $\npair[n']{(\ms_\stk(a),\ms_T(a))} \in \lrv(W')$ for some $W'$ such that $W_M = W' \oplus W_M'$.
  Additionally, if
  \begin{itemize}
  \item $\npair{(b,e)}\in\stackWriteCond{W}$
  \end{itemize}
  Then $\memSat{(\ms_S,\stk,\ms_\stk[a\mapsto 0],\ms_T[a\mapsto 0])}{W_M'}$.
\end{lemma}
\begin{proof}
  From $\npair{(b,e)}\in\stackReadCond{W}$, we get an $S \subseteq \addressable{\lin,
    \pwfree}$, an $R : S \rightarrow \powerset{\nats}$ with $\biguplus_{r \in S} R(r) \supseteq [\baddr,\eaddr]$ and $\pwfree(r).H
  \nsubeq \stdreg{R(r),\gc}{\pur}.H$ for all $r \in S$.

  Since $a \in [b,e]$, there is a unique $r \in S$ such that $a \in R(r)$.

  Since $W \oplus W_M$ is defined, we have that $r \in \dom(\pwfree) =
  \dom(\pwfree[W_M])$ and $\pwfree(r) \oplus \pwfree[W_M](r)$ is defined.

  From $\memSat{(\ms_S,\stk,\ms_\stk,\ms_T)}{W_M}$, we get that
  $\stk = (\opc_0,\ms_0):: \dots :: (\opc_m,\ms_m)$,
  $\ms_S \uplus \ms_\stk \uplus \ms_0 \uplus \dots \uplus \ms_m$ is defined,
  $W_M = W_{\var{stack}} \oplus W_{\var{free\_stack}} \oplus W_{\var{heap}}$ and
  $\exists \ms_\var{T,stack}, \ms_\var{T,free\_stack}, \ms_\var{T,heap}$, $\ms_{T,f}$, $\ms_{S,f}$, $\ms_S'$ such that
  \begin{itemize}
  \item $\ms_S =\ms_{S,f} \uplus \ms_S'$
  \item $\ms_T = \ms_\var{T,stack} \uplus \ms_\var{T,free\_stack} \uplus
    \ms_\var{T,heap} \uplus \ms_{T,f}$
  \item $\memSatStack{\stk,\ms_\var{T,stack}}{W_{\var{stack}}}$
  \item $\memSatFStack{\ms_\stk,\ms_\var{T,free\_stack}}{W_{\var{free\_stack}}}$
  \item $\npair{(\overline{\sigma},\ms_S',\ms_\var{T,heap})}\in\lrheap(\pwheap[W_M])(W_\var{heap})$.
  \end{itemize}

  From $\memSatFStack{\ms_\stk,\ms_\var{T,free\_stack}}{W_{\var{free\_stack}}}$, we get an
  \begin{itemize}
  \item $R_\ms : \dom(\activeReg{\pwfree[W_\var{free\_stack}]}) \fun \MemSeg \times \MemSeg$, 
  \item $\ms_\var{T,free\_stack} = \biguplus_{r \in \dom(\activeReg{\pwfree[W_\var{free\_stack}]})} \pi_2(R_\ms(r))$, 
  \item $\ms_\stk = \biguplus_{r \in \dom(\activeReg{\pwfree[W_\var{free\_stack}]})} \pi_1(R_\ms(r))$, 
  \item $\stkb \in \dom(\ms_\var{T,free\_stack}) \wedge \stkb \in \dom(\ms_\stk)$, 
  \item $\exists R_W : \dom(\activeReg{\pwfree[W_\var{free\_stack}]}) \fun \World\ldotp$ 
  \item $W_\var{free\_stack} = \oplus_{r \in \dom(\activeReg{\pwfree[W_\var{free\_stack}]})} R_W(r)$
  \end{itemize}
  and for all
  $r \in \activeReg{\pwfree[W_\var{free\_stack}]}$, we have that
  $\npair[n']{R_\ms(r)} \in  \pwfree[W_\var{free\_stack}](r).H \; \xi^{-1}(R_W(r))$ for all $n' < n$.

  We have that $r \in \addressable{\lin, \pwfree} \subseteq \activeReg{\pwfree[W_\var{free\_stack}]}$, so $\npair[n']{R_\ms(r)} \in \pwfree[W_\var{free\_stack}](r).H \; \xi^{-1}(R_W(r))$.
  Because $\pwfree(r).H \nsubeq \stdreg{R(r),\gc}{\pur}.H$ and $W\oplus W_M = W \oplus (W_\var{free\_stack} \oplus W_\var{free\_stack} \oplus W_\var{stack})$ is defined, it follows that also $\pwfree[W_\var{free\_stack}](r).H \nsubeq \stdreg{R(r),\gc}{\pur}.H$.
  This means that also $\npair[n']{R_\ms(r)} \in H^{\mathrm{std}}_{R(r)}\; \xi^{-1}(R_W(r))$.

  From this, it follows that $\dom(R_\ms(r).1) = \dom(R_\ms(r).2) = R(r)$ and we get a $S : R(r) \fun \World$ with $\xi(\xi^{-1}(R_W(r))) = \oplus_{\aaddr \in R(r)} S(\aaddr)$ and $\forall \aaddr \in R(r) \ldotp \npair[n']{(\ms_\stk(\aaddr),\ms_T(\aaddr))} \in \lrv(S(\aaddr))$.

  Since $\aaddr \in R(r)$, we have that $\npair[n']{(\ms_\stk(\aaddr),\ms_T(\aaddr))} \in \lrv(S(\aaddr))$ and
  we can take $W_M' = W_r' \oplus W_{\var{heap}} \oplus (W_{\var{stack}} \oplus W_{\var{free\_stack}}')$ with $W_r' = \oplus_{\aaddr \in (R(r)\setminus \{\aaddr\})} S(\aaddr)$ and $W_\var{free\_stack}' = \oplus_{r' \in (\dom(\activeReg{\pwheap[W_\var{heap}]})\setminus \{r\})} R_W(r')$, and get
  \begin{align*}
    S(\aaddr) \oplus W_M'
    &=S(\aaddr) \oplus (W_r' \oplus W_{\var{free\_stack}}' \oplus (W_{\var{stack}}\oplus W_\var{heap}))\\
    &=
    \oplus_{\aaddr \in R(r)} S(\aaddr) \oplus W_{\var{free\_stack}}' \oplus (W_{\var{stack}}\oplus W_{\var{heap}})\\
    &=
    \xi(\xi^{-1}(R_W(r))) \oplus W_{\var{free\_stack}}' \oplus (W_{\var{stack}} \oplus W_\var{heap})\\
    &=
    R_W(r) \oplus W_{\var{free\_stack}}' \oplus (W_{\var{stack}} \oplus W_\var{heap})\\
    &=
    \oplus_{r \in \dom(\activeReg{\pwheap[W_\var{free\_stack}]})} R_W(r)\oplus (W_{\var{stack}} \oplus W_{\var{heap}})\\
    &=
      W_\var{free\_stack} \oplus (W_{\var{stack}} \oplus W_{\var{heap}})\\
      &= W_M
  \end{align*}

  Additionally, if $\npair{(b,e)}\in\stackWriteCond{W}$, then
  we get an 
  \begin{equation*}
    S'\subseteq \addressable{\lin, \pwfree} \subseteq \activeReg{\pwheap[W_\var{free\_stack}]} \text{,}
  \end{equation*}
  and an $R' : S' \fun \powerset{\nats}$ such that $\biguplus_{r \in S'} R'(r) \supseteq [\baddr,\eaddr]$ and
  for all $r \in S'$, $\pwfree(r).H \nsupeq \stdreg{R'(r),\gc}{\pur}.H$ and
  $\pwfree(r)$ is address-stratified.

  Since $a \in [b,e]$, there is an $r' \in S'$ such that $a \in R'(r')$ .

  Because $W \oplus W_M$ is defined, it follows that also
  $\pwfree[W_\var{free\_stack}](r').H \nsupeq \stdreg{R'(r'),\gc}{\pur}.H$ and
  $\pwfree[W_\var{free\_stack}](r')$ is address-stratified.

  It follows that $r = r'$ because $\dom(R_\ms(r).1) = \dom(R_\ms(r).2) = R(r) \ni a$ and $ \dom(R_\ms(r').2) = \dom(R_\ms(r').1) = R'(r') \ni a$ and all the $R_\ms(r).1$ and $R_\ms(r).2$ are disjoint.

  We have that $\npair[n']{(R_\ms(r).1\update{a}{0},R_\ms(r).2\update{a}{0})} \in \pwfree[W_\var{free\_stack}](r).H \; \xi^{-1}(W_r')$ for all $n' < n$ because $\pwfree[W_\var{free\_stack}](r)$ is address-stratified and $\pwfree[W_\var{free\_stack}](r).H \nsupeq \stdreg{R'(r),\gc}{\pur}.H$.

  From this, it follows that $\memSatFStack{\ms_\stk\update{a}{0},\ms_\var{T,heap}\update{a}{0}}{W_r' \oplus W_\var{free\_stack}'}$ and finally $\memSat{(\ms_S,\stk,\ms_\stk\update{a}{0},\ms_T\update{a}{0})}{W_M'}$.
\end{proof}

\begin{lemma}[load from regular capability works]
  \label{lem:readcond-cap-works}
  If
  \begin{itemize}
  \item $\memSat{(\ms_S,\stk,\ms_\stk,\ms_T)}{W_M}$
  \item $c = ((\perm,\lin),\baddr,\eaddr,\aaddr)$
  \item $c' = ((\perm',\lin'),\baddr',\eaddr',\aaddr')$
  \item $\perm \in \readAllowed{}$, $\perm' \in \readAllowed{}$
  \item $\npair{(c,c')}\in\lrvg{\trust}(H_\sigma,W)$
  \item $W \oplus W_M$ is defined
  \item $n' < n$
  \item $w_S = \linCons{\ms_S(\aaddr)}$, $w_T = \linCons{\ms_T(\aaddr')}$
  \item $\linConsPerm{\perm}{\ms_S(\aaddr)}$, $\linConsPerm{\perm'}{\ms_T(\aaddr')}$
  \item $\aaddr \in [\baddr,\eaddr]$
  \item $\aaddr' \in [\baddr',\eaddr']$
  \end{itemize}
  Then $\exists W', W_M'\ldotp$
  \begin{itemize}
  \item $W_M = W' \oplus W_M'$
  \item $\memSat[n']{(\ms_S\update{\aaddr}{w_S},\stk,\ms_\stk,\ms_T\update{\aaddr'}{w_T})}{W_M'}$
  \item $\npair[n']{(\ms_S(\aaddr),\ms_T(\aaddr'))} \in \lrvg{\trust}(W')$
  \end{itemize}
\end{lemma}
\begin{proof}
  Consider first the case that $\npair{(c,c')}\in\lrv(W)$.

  From $\npair{(c,c')}\in\lrv(H_\sigma,W)$ with $c = ((\perm,\lin),\baddr,\eaddr,\aaddr)$, $c' = ((\perm',\lin'),\baddr',\eaddr',\aaddr')$, $\perm \in \readAllowed{}$ and $\perm' \in \readAllowed{}$, we get that $\baddr = \baddr'$, $\eaddr = \eaddr'$ and $\aaddr = \aaddr'$ and $\npair{(\baddr,\eaddr)} \in \readCond{\lin,W}$.

  Lemma~\ref{lem:readcond-writecond-work} then gives us a $W'$ and $W_M'$ such that $W_M = W' \oplus W_M'$ and $\npair[n']{(\ms_S(\aaddr),\ms_T(\aaddr'))} \in \lrv(W')$.

  it remains to prove that $\memSat[n']{(\ms_S\update{\aaddr}{w_S},\stk,\ms_\stk,\ms_T\update{\aaddr'}{w_T})}{W_M'}$.
  We have to distinguish the case that $\isLinear{\ms_S(\aaddr)}$ and the opposite case.
  \begin{itemize}
  \item case $\isLinear{\ms_S(\aaddr)}$: then $\linCons{\ms_S(\aaddr)} = 0$ and it follows from $\npair[n']{(\ms_S(\aaddr),\ms_T(\aaddr'))}\in\lrv(W)$ that also $\isLinear{\ms_T(\aaddr')}$ and $\linCons{\ms_T(\aaddr')} = 0$.

    From $\linConsPerm{\perm}{\ms_S(\aaddr)}$ and $\linConsPerm{\perm'}{\ms_T(\aaddr')}$, we then also get that $\perm,\perm' \in \writeAllowed{}$ and from $\npair{(c,c')}\in\lrv(W)$, it then follows that $\npair{(\baddr,eaddr)} \in \writeCond{\lin,W}$.
    From the ``Additionally, if..'' case in Lemma~\ref{lem:readcond-writecond-work} with Lemma~\ref{lem:downwards-closed}, we then get that $\memSat[n']{(\ms_S\update{\aaddr}{w_S},\stk,\ms_\stk,\ms_T\update{\aaddr'}{w_T})}{W_M'}$.

  \item case $\neg\isLinear{\ms_S(\aaddr)}$: then $\linCons{\ms_S(\aaddr)} = \ms_S(\aaddr)$ and it follows from $\npair[n']{(\ms_S(\aaddr),\ms_T(\aaddr'))}\in\lrv(W)$ that also $\neg\isLinear{\ms_T(\aaddr')}$ \dominique{Do we need a lemma for this?}
    and $\linCons{\ms_T(\aaddr')} = \ms_T(\aaddr')$.
    The fact that $\memSat[n']{(\ms_S\update{\aaddr}{w_S},\stk,\ms_\stk,\ms_T\update{\aaddr'}{w_T})}{W_M'}$ then follows simply by downwards closure of memory satisfaction, i.e.~Lemma~\ref{lem:downwards-closed}.
  \end{itemize}

  Now consider the case that $\npair{(c,c')}\in(\lrvg{\trusted}(W) \setminus \lrv(W))$. then we have that $[b,e] \subseteq \ta$ and $\npair{[b,e]} \in \xReadCond{W}$.
  By definition of $\xReadCond{}$, there is an $r_a$ such that $\pwheap(r) \nequal \codereg{\_,\_,\code,\gc}$, $a \in \dom(\code)$.
  By definition of $\codereg{}$ and using the fact that $\dom(\code) \ni a \in [b,e] \subseteq \ta$, we know that $\dom(\code) \subseteq \ta$ and $\_,\_\vdash_{\mathrm{comp-code}} \code$.
  From the fact that $\memSat{(\ms_S,\stk,\ms_\stk,\ms_T)}{W_M}$, together with the definition of $\codereg{}$, we know that $\npair[n']{(\Phi.\mem(a),\Phi.\mem(a))} \in \lrvg{\trusted}(\purePart{W_M})$ and we have that $\purePart{W_M} = \purePart{W}$ by Lemma~\ref{lem:purePart-oplus}. 
  From the fact that $\_,\_\vdash_{\mathrm{comp-code}} \code$, we get that $\neg \isLinear{\code(a)}$.
  
  We can then take $W_M' = W_M$, $W' = W$ and get the required results from what we have proven above and using Lemma~\ref{lem:downwards-closed}.
\end{proof}

\begin{lemma}[load from stack capability works]
  \label{lem:load-stack-cap-works}
  If
  \begin{itemize}
  \item $\memSat{(\ms_S,\stk,\ms_\stk,\ms_T)}{W_M}$
  \item $c = \stkptr{\perm,\baddr,\eaddr,\aaddr}$
  \item $c' = ((\perm',\lin'),\baddr',\eaddr',\aaddr')$
  \item $\perm \in \readAllowed{}$ or $\perm' \in \readAllowed{}$
  \item $\npair{(c,c')}\in\lrvg{\trust}(H_\sigma,W)$
  \item $W \oplus W_M$ is defined
  \item $n' < n$
  \item ($\aaddr \in [\baddr,\eaddr]$ or $\aaddr' \in [\baddr',\eaddr']$)
  \item $w_S = \linCons{\ms_S(\aaddr)}$, $w_T = \linCons{\ms_T(\aaddr')}$
  \item $\linConsPerm{\perm}{\ms_S(\aaddr)}$, $\linConsPerm{\perm'}{\ms_T(\aaddr')}$
  \end{itemize}
  Then $\exists W', W_M'\ldotp$
  \begin{itemize}
  \item $W_M = W' \oplus W_M'$
  \item $\memSat[n']{(\ms_S,\stk,\ms_\stk\update{\aaddr}{w_S},\ms_T\update{\aaddr'}{w_T})}{W_M'}$
  \item $\npair[n']{(\ms_\stk(a),\ms_T(a))} \in \lrv(W')$
  \end{itemize}
\end{lemma}
\begin{proof}
  From $\npair{(c,c')}\in\lrv(H_\sigma,W)$ with $c = \stkptr{\perm,\baddr,\eaddr,\aaddr}$, $c' = ((\perm',\lin'),\baddr',\eaddr',\aaddr')$, ($\perm \in \readAllowed{}$ or $\perm' \in \readAllowed{}$), we get that $\perm = \perm'$, $\lin' = \linear$, $\baddr = \baddr'$, $\eaddr = \eaddr'$ and $\aaddr = \aaddr'$ and $\npair{(\baddr,\eaddr)} \in \stackReadCond{W}$.

  From Lemma~\ref{lem:stackreadcond-stackwritecond-work}, we then get that $\npair[n']{(\ms_\stk(a),\ms_T(a))}\in\lrv(W')$ for some $W'$ such that $W_M = W' \oplus W_M'$.

  It remains to prove that $\memSat{(\ms_S,\stk,\ms_\stk\update{\aaddr}{w_S},\ms_T\update{\aaddr'}{w_T})}{W_M'}$.
  We have to distinguish the case that $\isLinear{\ms_\stk(\aaddr)}$ and the opposite case.
  \begin{itemize}
  \item case $\isLinear{\ms_\stk(\aaddr)}$: then $\linCons{\ms_\stk(\aaddr)} = 0$ and it follows from
    \begin{equation*}
      \npair[n']{(\ms_\stk(a),\ms_T(a))}\in\lrv(H_\sigma,W)
    \end{equation*}
    that also $\isLinear{\ms_T(\aaddr')}$ \dominique{Do we need a lemma for this?} and $\linCons{\ms_T(\aaddr')} = 0$.
    From $\linConsPerm{\perm}{\ms_S(\aaddr)}$ and $\linConsPerm{\perm'}{\ms_T(\aaddr')}$, we then also get that $\perm = \perm' \in \writeAllowed{}$ and from $\npair{(c,c')}\in\lrvg{\trust}(H_\sigma,W)$, it then follows that $\npair{(\baddr,eaddr)} \in \stackWriteCond{W}$.
    From the ``Additionally, if..'' case in Lemma~\ref{lem:readcond-writecond-work} with Lemma~\ref{lem:downwards-closed}, we then get that $\memSat[n']{(\ms_S,\stk,\ms_\stk\update{\aaddr}{0},\ms_T\update{\aaddr'}{0})}{W_M'}$.

  \item case $\neg\isLinear{\ms_\stk(\aaddr)}$: then $\linCons{\ms_\stk(\aaddr)} = \ms_\stk(\aaddr)$ and it follows from
    \begin{equation*}
      \npair[n']{(\ms_\stk(a),\ms_T(a))}\in\lrv(H_\sigma,W)
    \end{equation*}
    that also $\neg\isLinear{\ms_T(\aaddr')}$ \dominique{Do we need a lemma for this?}
    and $\linCons{\ms_T(\aaddr')} = \ms_T(\aaddr')$.
    The fact that $\memSat[n']{(\ms_S\update{\aaddr}{w_S},\stk,\ms_\stk,\ms_T\update{\aaddr'}{w_T})}{W_M'}$ then follows simply by downwards closure of memory satisfaction, i.e.\ Lemma~\ref{lem:downwards-closed}.
  \end{itemize}
\end{proof}

\begin{lemma}[Store to regular capability works]
  \label{lem:store-reg-works}
  If
  \begin{itemize}
  \item $\memSat{(\ms_S,\stk,\ms_\stk,\ms_T)}{W_M}$
  \item $c_1 = ((\perm,\lin),\baddr,\eaddr,\aaddr)$
  \item $c_1' = ((\perm',\lin'),\baddr',\eaddr',\aaddr')$
  \item $\perm \in \writeAllowed{}$, $\perm' \in \writeAllowed{}$
  \item $\aaddr \in [\baddr,\eaddr]$
  \item $\aaddr' \in [\baddr',\eaddr']$
  \item $\npair{(c_1,c_1')}\in\lrvg{\trusted}(H_\sigma,W_1)$
  \item $\npair{(c_2,c_2')}\in\lrv(H_\sigma,W_2)$
  \item $W_1 \oplus W_2 \oplus W_M$ is defined
  \item $n' < n$
  \item $c_3 = \linCons{c_2}$, $c_3' = \linCons{c_2'}$
  \end{itemize}
  Then $\exists W_2', W_M'\ldotp$
  \begin{itemize}
  \item $W_2 \oplus W_M = W_2' \oplus W_M'$
  \item $\memSat[n']{(\ms_S\update{\aaddr}{c_2},\stk,\ms_\stk,\ms_T\update{\aaddr'}{c_2'})}{W_M'}$
  \item $\npair[n']{(c_3,c_3')} \in \lrv(W_2')$
  \end{itemize}
\end{lemma}
\begin{proof}
  From $\npair{(c_1,c_1')}\in\lrvg{\trusted}(W_1)$ and $\perm \in \writeAllowed{}$, it follows that $c_1 = c_1'$ and $\npair{[b,e]} \in \writeCond{\lin,W_1}$. 

  We then get a $r$ and $A$, such that $a \in A$, $\pwheap[W_1](r).H \nsupeq \stdreg{A,\gc}{v}.H$ and $\pwheap[W_1](r)$ is address-stratified.

  If we decompose the judgement that $\memSat{(\ms_S,\stk,\ms_\stk,\ms_T)}{W_M}$, then we get some $\npair[n-1]{\ms_S|_A, \ms_T|_A} \in \stdreg{A,\gc}{v}.H~W_{M,A}$ for some $W_M = W_{M,A} \oplus W_{M,R}$.
  If we define $W_M'$ as $W_M \oplus W_2$ and $W_2' = \purePart{W_2}$, then we can use the properties about $\pwheap[W_1](r)$ above to show that  $\memSat[n']{(\ms_S\update{\aaddr}{c_2},\stk,\ms_\stk,\ms_T\update{\aaddr'}{c_2'})}{W_M'}$.
  The fact that $\npair[n']{(c_3,c_3')} \in \lrv(W_2')$ follows from Lemma~\ref{lem:non-linear-pure}.
\end{proof}

\begin{lemma}[Store to stack capability works]
  \label{lem:store-stack-works}
  If
  \begin{itemize}
  \item $\memSat{(\ms_S,\stk,\ms_\stk,\ms_T)}{W_M}$
  \item $c_1 = \stkptr{\perm,\baddr,\eaddr,\aaddr}$
  \item $c_1' = ((\perm',\linear),\baddr',\eaddr',\aaddr')$
  \item $\perm \in \writeAllowed{}$, $\perm' \in \writeAllowed{}$
  \item $\aaddr \in [\baddr,\eaddr]$
  \item $\aaddr' \in [\baddr',\eaddr']$
  \item $\npair{(c_1,c_1')}\in\lrvg{\trusted}(H_\sigma,W_1)$
  \item $\npair{(c_2,c_2')}\in\lrv(H_\sigma,W_2)$
  \item $W_1 \oplus W_2 \oplus W_M$ is defined
  \item $n' < n$
  \item $c_3 = \linCons{c_2}$, $c_3' = \linCons{c_2'}$
  \end{itemize}
  Then $\exists W_2', W_M'\ldotp$
  \begin{itemize}
  \item $W_2 \oplus W_M = W_2' \oplus W_M'$
  \item $\memSat[n']{(\ms_S,\stk,\ms_\stk\update{\aaddr}{c_2},\ms_T\update{\aaddr'}{c_2'})}{W_M'}$
  \item $\npair[n']{(c_3,c_3')} \in \lrv(W_2')$
  \end{itemize}
\end{lemma}
\begin{proof}
  From $\npair{(c_1,c_1')}\in\lrvg{\trusted}(W_1)$ and $\perm \in \writeAllowed{}$, it follows that $\perm = \perm'$, $\baddr = \baddr'$, $\eaddr = \eaddr'$, $\aaddr = \aaddr'$ and $\npair{[b,e]} \in \stackWriteCond{W_1}$. 

  We then get a $r$ and $A$, such that $a \in A$, $\pwfree[W_1](r).H \nsupeq \stdreg{A,\gc}{v}.H$ and $\pwfree[W_1](r)$ is address-stratified.

  If we decompose the judgement that $\memSat{(\ms_S,\stk,\ms_\stk,\ms_T)}{W_M}$, then we get some $\npair[n-1]{\ms_\stk|_A, \ms_T|_A} \in \stdreg{A,\gc}{v}.H~W_{M,A}$ for some $W_M = W_{M,A} \oplus W_{M,R}$.
  If we define $W_M'$ as $W_M \oplus W_2$ and $W_2' = \purePart{W_2}$, then we can use the properties about $\pwfree[W_1](r)$ above to show that  $\memSat[n']{(\ms_S,\stk,\ms_\stk\update{\aaddr}{c_2},\ms_T\update{\aaddr'}{c_2'})}{W_M'}$.
  The fact that $\npair[n']{(c_3,c_3')} \in \lrv(W_2')$ follows from Lemma~\ref{lem:non-linear-pure}.
\end{proof}

\begin{lemma}
  \label{lem:sealed-lrv-lrexj}
  If
  \begin{itemize}
  \item $\npair{(\sealed{\sigma,\vsc_1},\sealed{\sigma,\vsc_1'})} \in \lrvg{\trust}(W_{R,1})$
  \item $\npair{(\sealed{\sigma,\vsc_2},\sealed{\sigma,\vsc_2'})} \in \lrvg{\trust}(W_{R,2})$ 
  \item $\nonExec{\vsc_1'}$ and $\nonExec{\vsc_2'}$
  \item $W_{R,1} \oplus W_{R,2} \oplus W_M$ is defined
  \item $\memSat{\ms_S,\stk,\ms_\stk,\ms_T}{W_M}$
  \end{itemize}
  Then
  \begin{itemize}
  \item $\npair[n-1]{(\vsc_1,\vsc_2,\vsc_1',\vsc_2')} \in \lrexj(W_{R,1}\oplus W_{R,2})$
  \end{itemize}
\end{lemma}
\begin{proof}
  From $\npair{(\sealed{\sigma,\vsc_1},\sealed{\sigma,\vsc_1'})} \in \lrvg{\trust}(W_{R,1})$, we know that
  \begin{itemize}
  \item $(\isLinear{\src{\vsc_1}} \text{ iff } \isLinear{\vsc_1'})$ 
  \item $\exists r \in \dom(\pwheap), \sigrets,\sigcloss,\mscode \ldotp$ 
  \item $\pwheap[W_{R,1}](r) = (\pure,\_,H_\sigma)$ 
  \item $H_\sigma \nequal H^\mathrm{code,\square}_\sigma \; \sigrets \; \sigcloss \; \mscode \; \gc$
  \item $\npair[n']{\stpair[.]{\vsc_1}{\vsc_1'}} \in H_\sigma \; \sigma \; \xi^{-1}(W_{R,1})$ for all $n' < n$
  \item If $(\isLinear{\src{\vsc_1}}$ then for all $W' \future W_{R,1}$, $W_o$, $n' < n$, $\npair[n']{\stpair[.]{\vsc_2}{\vsc_2'}} \in H_\sigma \; \sigma \; \xi^{-1}(W_o)$, we have that
    \begin{equation*}
      \npair[n']{\src{\vsc_1},\src{\vsc_2},\vsc_1',\vsc_2'} \in \lrexj(W'\oplus W_o))
    \end{equation*}
  \item If $(\nonLinear{\src{\vsc_1}}$ then for all $W' \future \purePart{W_{R,1}}$, $W_o$, $n' < n$, $\npair[n']{\stpair[.]{\vsc_2}{\vsc_2'}} \in H_\sigma \; \sigma \; \xi^{-1}(W_o)$, we have that
    \begin{equation*}
      \npair[n']{\src{\vsc_1},\src{\vsc_2},\vsc_1',\vsc_2'} \in \lrexj(W'\oplus W_o))
    \end{equation*}
  \end{itemize}

  From $\npair{(\sealed{\sigma,\vsc_2},\sealed{\sigma,\vsc_2'})} \in \lrvg{\trust}(W_{R,2})$, we know that
  \begin{itemize}
  \item $(\isLinear{\src{\vsc_2}} \text{ iff } \isLinear{\vsc_2'})$ 
  \item $\exists r' \in \dom(\pwheap), \sigrets,\sigcloss,\mscode \ldotp$ 
  \item $\pwheap[W_{R,2}](r') = (\pure,\_,H_\sigma)$ 
  \item $H_\sigma \nequal H^\mathrm{code,\square}_\sigma \; \sigrets \; \sigcloss \; \mscode \; \gc$
  \item $\npair[n']{\stpair[.]{\vsc_2}{\vsc_2'}} \in H_\sigma \; \sigma \; \xi^{-1}(W_{R,2})$ for all $n' < n$
  \end{itemize}

  From $\memSat{\ms_S,\stk,\ms_\stk,\ms_T}{W_M}$, we know that two different regions cannot have $H_\sigma$ defined for the same $\sigma$, so that $r = r'$.
  Both when $\isLinear{\vsc_1}$ and when $\nonLinear{\vsc_1}$ (using Lemma~\ref{lem:world-fut-purePart}), we know that for all $W' \future W_{R,1}$, $W_o$, $n' < n$, $\vsc_2$, $\vsc_2'$, $\npair[n']{\stpair[.]{\vsc_2}{\vsc_2'}} \in H_\sigma \; \sigma \; \xi^{-1}(W_o)$, we have that
  \begin{equation*}
    \npair[n']{\src{\vsc_1},\src{\vsc_2},\vsc_1',\vsc_2'} \in \lrexj(W'\oplus W_o))
  \end{equation*}
  We can now instantiate this fact with $W' = W_{R,1}$, $W_o = W_{R,2}$, $n' = n-1$, and the fact that $\npair[n-1]{\stpair[.]{\vsc_2}{\vsc_2'}} \in H_\sigma \; \sigma \; \xi^{-1}(W_{R,2})$ (see above) to conclude
  \begin{equation*}
    \npair[n-1]{\src{\vsc_1},\src{\vsc_2},\vsc_1',\vsc_2'} \in \lrexj(W_{R,1}\oplus W_{R,2})\text{.}
  \end{equation*}
\end{proof}

\subsection{FTLR proof}

\begin{lemma}
  \label{lem:ftlr-internal-lemma}
  If
  \begin{itemize}
  \item One of the following sets of requirements holds:
    \begin{itemize}
    \item $\trust = \trusted$, $\Phi_S$ is reasonable up to $n$ steps and $[\baddr,\eaddr] \subseteq \ta$
    \item $\trust = \untrusted$ and $[b,e] \mathrel{\#} \ta$ and $\npair{[\baddr,\eaddr]} \in \readCond{\normal,W_\pcreg}$
    \end{itemize}
  \item $\Phi_S(\pcreg) = \Phi_T(\pcreg) = ((\rx,\normal),\baddr,\eaddr,a)$
  \item $\npair{[\baddr,\eaddr]} \in \xReadCond{W_\pcreg}$
  \item $\npair{(\Phi_S.\reg,\Phi_T.\reg)} \in \lrrg{\trust}(W_R)$
  \item $\memSat{\Phi_S.\mem,\Phi_S.\stk,\Phi_S.\ms_\stk,\Phi_T.\mem}{W_M}$
  \item $W_\pcreg \oplus W_R \oplus W_M$ is defined.
  \item Theorem~\ref{thm:ftlr} holds for all $n' < n$.
  \end{itemize}
  Then
  \[
    \npair{(\Phi_S,\Phi_T)} \in \lro
  \]
\end{lemma}
\begin{proof}
  In the following we will use the following: $\reg_S = \Phi_S.\reg$, $\reg_T = \Phi_T.\reg$, $\ms_S = \Phi_S.\mem$, $\ms_\stk = \Phi_S.\ms_\stk$, and $\stk = \Phi_S.\stk$.

  By complete induction on $n$, i.e. we can assume that the lemma already holds for $n' < n$.

  In order to prove this, we first, we prove that one of the following holds:
  \begin{enumproof}
  \item \label{case:ftlr:failed} $\Phi_S \nstep[j]{\gc} \failed$ and $\Phi_T
    \step[i] \failed$ for some $i$,$j$.
  \item \label{case:ftlr:halted} $\Phi_S \step[\gc] \halted$ and $\Phi_T \step\halted$
  \item \label{case:ftlr:get-and-arith-op} All of the following hold: (includes simple cases: gettype, geta, getb,
    gete, getp, getl, lt, plus, minus, one case of move that can be handled
    uniformly)
    \begin{itemize}
    \item $\Phi_S \step[\gc] \Phi_S'$
    \item $\Phi_T \step \Phi_T'$
    \item $\Phi_S$ does not point to $\scall{\offpc,\offsigma}{r_1}{r_2}$ or $\txjmp{r_1}{r_2}$
    \item $\Phi_S' = \updPcAddr{\Phi_S\updReg{r}{z}} \neq \failed$
    \item $\Phi_T' = \updPcAddr{\Phi_T\updReg{r}{z}} \neq \failed$
    \item $r \neq \pcreg$
    \item $z \in \ints$
    \end{itemize}
  \item \label{case:ftlr:call} All of the following hold:
    \begin{itemize}
    \item $\callCond{\Phi_S,r_1,r_2,\offpc,\offsigma,\aaddr}$
    \item $r_1 \neq \rtmp{1}$ and $r_2 \neq \rtmp{1}$
    \item $\baddr \leq \aaddr \tand \aaddr + \calllen-1 \leq \eaddr$
    \item $\exec{\Phi_S(\pcreg)}$
    \item for all $i = 0..\calllen-1$, $\ms_T(\aaddr+i) = \ms_S(\aaddr+i) \in \ints$.
    \item $[\baddr,\eaddr] \subseteq \ta$
    \item $\Phi_S \step[\gc] \Phi_S' \neq \failed$
    \item $\Phi_T \step \Phi_T' \neq \failed$
    \item $\Phi_S(r_1) = \sealed{\sigma,c_1}$
    \item $\Phi_S(r_2) = \sealed{\sigma,c_2}$
    \item $\nonExec{c_2}$
    \item $\Phi_S(\rstk) = \stkptr{\rw, \baddr_\stk, \eaddr_\stk, \aaddr_\stk}$
    \item $\baddr_\stk < \aaddr_\stk \leq \eaddr_\stk$
    \item $\Phi_S(\pcreg) = ((\perm,\normal),\baddr,\eaddr,\aaddr)$
    \item $w_1 = \linCons{c_1}$ and $w_2 = \linCons{c_2}$
    \item \[
        \begin{split}
          \Phi_S''.\reg = \Phi_S.reg      &[\rstk \mapsto \stkptr{\rw,\baddr_\stk,\aaddr_\stk-1,\aaddr_\stk-1}] \\
          &[\rretc \mapsto \sealed{\sigma',\retptrc(\baddr,\eaddr,\aaddr+\calllen)}]\\
          &[\rretd \mapsto \sealed{\sigma',\retptrd(\aaddr_\stk,\eaddr_\stk)}] \\
          &[r_1,r_2 \mapsto w_1,w_2] \\
          &[\rtmp{1} \mapsto 0]
        \end{split}
      \]
    \item $\Phi_S''.\mem = \Phi_S.\mem$
    \item $\ms_{\stk\_\priv,S} = \Phi_S.\ms_\stk |_{[\aaddr_\stk,\eaddr_\stk]}\update{\aaddr_\stk}{42}$
    \item $\Phi_S''.\stk = ((\aaddr+\calllen),\ms_{\stk\_\priv,S}) :: \Phi_S.\stk$
    \item $\Phi_S''.\ms_\stk = \Phi_S.\ms_\stk -
      \Phi_S.\ms_\stk|_{[\aaddr_\stk,\eaddr_\stk]}$
    \item $\baddr \leq \aaddr+\offpc \leq \eaddr$
    \item $\mem(\aaddr+\offpc) = \seal{\sigma_\baddr,\sigma_\eaddr,\sigma_\aaddr}$
    \item $\sigma' = \sigma_\aaddr + \offsigma$
    \item $\sigma_\aaddr \leq \sigma' \leq \sigma_\eaddr$

    \item $\Phi_S' =\var{xjumpResult}\left(c_1,c_2, \Phi_S'' \right)$

    \end{itemize}
  \item \label{case:ftlr:move-cca-etc} All of the following hold: (includes cap-manipulation cases: move, cca,
    restrict, seta2b, cseal, split, splice, that can be handled mostly
    uniformly)
    \begin{itemize}
    \item $\Phi_S \step[\gc] \Phi_S'$
    \item $\Phi_T \step \Phi_T'$
    \item $\Phi_S$ does not point to $\scall{\offpc,\offsigma}{r_1}{r_2}$ or $\txjmp{r_1}{r_2}$
    \item $\Phi_S' = \updPcAddr{\Phi_S\updReg{r_1\cdots r_k}{w_1\cdots w_k}}\neq
      \failed$
    \item $\Phi_T' = \updPcAddr{\Phi_T\updReg{r_1\cdots r_k}{w_1'\cdots w_k'}}
      \neq \failed$
    \item $r_i \neq \pcreg$ for all $i$
    \item One of the following holds:
      \begin{enumproof}
      \item \label{case:ftlr:move-cca-etc:cca-normal} (restrict,cca,seta2b) $w_1 = ((\perm',\lin),\baddr,\eaddr,\aaddr')$,
        $w_1' = ((\perm',\lin),\baddr,\eaddr,\aaddr')$, $\Phi_S(r_1) =
        ((\perm,\lin),\baddr,\eaddr,\aaddr)$ and $\Phi_T(r_1) =
        ((\perm,\lin),\baddr,\eaddr,\aaddr)$ and $\perm' \sqsubseteq\perm$, $k =
        1$
      \item \label{case:ftlr:move-cca-etc:cca-stack}(restrict,cca,seta2b)
        \begin{itemize}
        \item $w_1 = \stkptr{\perm',\baddr,\eaddr,\aaddr'}$, 
        \item $w_1' = ((\perm',\linear),\baddr,\eaddr,\aaddr')$, 
        \item $\Phi_S(r_1) = \stkptr{\perm,\baddr,\eaddr,\aaddr}$ and 
        \item $\Phi_T(r_1) = ((\perm,\linear),\baddr,\eaddr,\aaddr)$ and 
        \item $\perm' \sqsubseteq\perm$, 
        \item $k = 1$
      \end{itemize}

      \item \label{case:ftlr:move-cca-etc:cca-seal}(cca,seta2b) $w_1 = \seal{\sigma_\baddr,\sigma_\eaddr,\sigma'}$,
        $w_1' = \seal{\sigma_\baddr,\sigma_\eaddr,\sigma'}$, $\Phi_S(r_1) =
        \seal{\sigma_\baddr,\sigma_\eaddr,\sigma}$ and $\Phi_T(r_1) =
        \seal{\sigma_\baddr,\sigma_\eaddr,\sigma}$, $k = 1$
      \item \label{case:ftlr:move-cca-etc:move}(move) $w_1 = \Phi_S(r_2)$, $w_1' = \Phi_T(r_2)$, and $w_2 = \linCons{w_1}$, $w_2' = \linCons{w_1'}$ and $r_1 \neq r_2$ and $k = 2$.
      \item \label{case:ftlr:move-cca-etc:cseal}(cseal) $w_1 = \sealed{\sigma,\Phi_S(r_1)}$, $w_1' = \sealed{\sigma,\Phi_T(r_1)}$, $\Phi_S(r_2) = \Phi_T(r_2) = \seal{\sigma_\baddr,\sigma_\eaddr,\sigma}$, $\sigma_\baddr \le \sigma \le \sigma_\eaddr$, and $k=1$, and $\Phi_S$ points to $\src{\tcseal{r_1}{r_2}}$.
      \item \label{case:ftlr:move-cca-etc:split-normal} (split) $\Phi_T(r_3) = \Phi_S(r_3) =
        ((\perm,\lin),\baddr,\eaddr,\aaddr)$, $\baddr \leq s$, $s < \eaddr$,
        $w_1 = w_1' = ((\perm,\lin),\baddr,n,\aaddr)$, $w_2 = w_2' =
        ((\perm,\lin),n+1,\eaddr,\aaddr)$, $w_3 = w_3' = \linCons{\Phi_S(r_3)}$, $k=3$.
      \item \label{case:ftlr:move-cca-etc:split-seal} (split) $\Phi_T(r_3) = \Phi_S(r_3) =
        \seal{\sigma_\baddr,\sigma_\eaddr,\sigma}$, $\sigma_\baddr \leq s$, $s <
        \sigma_\eaddr$, $w_1 = w_1' = \seal{\sigma_\baddr,s,\sigma}$, $w_2 =
        w_2' = \seal{s+1,\sigma_\eaddr,\sigma}$, $k=2$.
      \item \label{case:ftlr:move-cca-etc:split-stack} (split) $\Phi_S(r_3) = \stkptr{\perm,\baddr,\eaddr,\aaddr}$,
        $\Phi_T(r_3) = ((\perm,\linear),\baddr,\eaddr,\aaddr)$, $\baddr \leq n$,
        $n < \eaddr$, $w_1 = \stkptr{\perm,\baddr,n,\aaddr}$, $w_1' =
        ((\perm,\linear),\baddr,n,\aaddr)$, $w_2 =
        \stkptr{\perm,n+1,\eaddr,\aaddr}$,

        $w_2' = ((\perm,\linear),n+1,\eaddr,\aaddr)$, $k=2$.
      \item \label{case:ftlr:move-cca-etc:splice-stack}(splice) $\Phi_T(r_2) = ((\perm,\linear),\baddr_2,\eaddr_2,\_)$,
        $\Phi_S(r_2) = \stkptr{\perm,\baddr_2,\eaddr_2,\_}$, and

        $\Phi_T(r_3) =
        ((\perm,\linear),\eaddr_2+1,\eaddr_3,\aaddr_3)$, $\Phi_S(r_3) =
        \stkptr{\perm,\eaddr_2+1,\eaddr_3,\aaddr_3}$, and $\baddr_2 \leq
        \eaddr_2$, and $\eaddr_2+1 \leq \eaddr_3$, and
        $w_1=((\perm,\linear),\baddr_2,\eaddr_3,\aaddr_3)$, $w_1'=
        \stkptr{\perm,\baddr_2,\eaddr_3,\aaddr_3}$, and $w_2 = w_2' = w_3 = w_3'
        = 0$, and $k=3$
      \item \label{case:ftlr:move-cca-etc:splice-normal}(splice)
        \begin{itemize}
        \item $\Phi_T(r_2) = \Phi_S(r_2) = ((\perm,\lin),\baddr_2,\eaddr_2,\_)$,
          and 
        \item $\Phi_T(r_3) = \Phi_S(r_3) =
          ((\perm,\lin),\eaddr_2+1,\eaddr_3,\aaddr_3)$, and 
        \item $\baddr_2 \leq
          \eaddr_2$, and 
        \item $\eaddr_2+1 \leq \eaddr_3$, and
        \item $w_1=w_1'=((\perm,\lin),\baddr_2,\eaddr_3,\aaddr_3)$, and 
        \item $w_2 = w_2'
          = \linCons{\Phi_S(r_2)}$ and 
        \item $w_3 = w_3' = \linCons{\Phi_S(r_3)}$, and
        \item $k=3$
        \end{itemize}

      \item \label{case:ftlr:move-cca-etc:splice-seal} (splice) $\Phi_T(r_2) = \Phi_S(r_2) =
        \seal{\sigma_{\baddr,2},\sigma_{\eaddr,2},\_}$, and $\Phi_T(r_3) =
        \Phi_S(r_3) =
        \seal{\sigma_{\eaddr,2}+1,\sigma_{\eaddr,3},\sigma_{\aaddr,3}}$, and
        $\sigma_{\baddr,2} \leq \sigma_{\eaddr,2}$, and $\sigma_{\eaddr,2} + 1
        \leq \sigma_{\eaddr,3}$and $w_1 = w_1' = \seal{\sigma_{\baddr,2},
          \sigma_{\eaddr,3}, \sigma}$
      \item \label{case:ftlr:move-cca-etc:jnz-z} (jnz zero case, noop move) $k = 0$
      \end{enumproof}
    \end{itemize}
  \item \label{case:ftlr:store-load-mem} All of the following hold: (includes memory-manipulation cases: store,
    load, that can be handled mostly uniformly)
    \begin{itemize}
    \item $\Phi_S \step[\gc] \Phi_S'$
    \item $\Phi_T \step \Phi_T'$
    \item $\Phi_S$ does not point to $\scall{\offpc,\offsigma}{r_1}{r_2}$ or $\txjmp{r_1}{r_2}$
    \item $\Phi_S' =
      \updPcAddr{\Phi_S\updReg{r_1,r_2}{w_1,w_2}\update{\mem.\aaddr}{w}}$
    \item $\Phi_T' =
      \updPcAddr{\Phi_T\updReg{r_1',r_2'}{w_1',w_2'}\update{\mem.\aaddr}{w'}}$
    \item $r_i \neq \pcreg$ for all $i$
    \item One of the following hold:
      \begin{enumproof}
      \item (store) $w_1 = w_1' = \Phi_S(r_1) = \Phi_T(r_1) =
        ((\perm,\lin),\baddr,\eaddr,\aaddr)$, and $\perm \in \writeAllowed{}$,
        and
        $\withinBounds{w_1}$, and\\
        $w = \Phi_S(r_2)$, and $w' = \Phi_T(r_2)$, and $w_2 = \linCons{\Phi_S(r_2)}$, $w_2' = \linCons{\Phi_T(r_2)}$.
      \item (load) $w_2 = w_2' = \Phi_T(r_2) = \Phi_S(r_2) =
        ((\perm,\lin),\baddr,\eaddr,\aaddr)$, and $\perm \in \readAllowed{}$,
        $\withinBounds{((\perm,\lin),\baddr,\eaddr,\aaddr)}$, and \\
        $w_1 = \Phi_S.\mem(\aaddr)$, and $w_1' = \Phi_T.\mem(\aaddr)$, and \\
        $w = \linCons{w_1}$, $w' = \linCons{w_1'}$, $\linConsPerm{\perm}{w_1}$, $\linConsPerm{\perm}{w_1'}$
      \end{enumproof}
    \end{itemize}
  \item \label{case:ftlr:store-load-stack} All of the following hold: (includes memory-manipulation cases: store,
    load, that can be handled mostly uniformly (stack))
    \begin{itemize}
    \item $\Phi_S \step[\gc]\Phi_S'$
    \item $\Phi_T \step \Phi_T'$
    \item $\Phi_S$ does not point to $\scall{\offpc,\offsigma}{r_1}{r_2}$ or $\txjmp{r_1}{r_2}$
    \item $\Phi_S' =
      \updPcAddr{\Phi_S\updReg{r_1,r_2}{w_1,w_2}\update{\ms_\stk.\aaddr}{w}}$
    \item $\Phi_T' =
      \updPcAddr{\Phi_T\updReg{r_1',r_2'}{w_1',w_2'}\update{\ms_\stk.\aaddr}{w'}}$
    \item $r_i \neq \pcreg$ for all $i$
    \item One of the following hold:
      \begin{enumproof}
      \item (store) $w_1 = \Phi_T(r_1) =
        ((\perm,\linear),\baddr,\eaddr,\aaddr)$, $w_1' = \Phi_S(r_1) =
        \stkptr{\perm,\baddr,\eaddr,\aaddr}$, and $\perm \in \writeAllowed{}$,
        and
        $\withinBounds{w_1}$, and\\
        $w = \Phi_S(r_2)$, and $w' = \Phi_T(r_2)$, and $w_2 = \linCons{\Phi_S(r_2)}$, $w_2' = \linCons{\Phi_T(r_2)}$.
      \item (load) $w_2' = \Phi_T(r_2) =
        ((\perm,\linear),\baddr,\eaddr,\aaddr)$, and $w_2 = \Phi_S(r_2) =
        \stkptr{\perm,\baddr,\eaddr,\aaddr}$, and $\perm \in \readAllowed{}$,
        $\withinBounds{((\perm,\lin),\baddr,\eaddr,\aaddr)}$, and $\aaddr \in
        \dom(\Phi.\ms_\stk)$, and $\aaddr \in \dom(\Phi.\ms_\stk)$, and\\
        $w_1 = \Phi_S.\ms_\stk(\aaddr)$, and $w_1' = \Phi_T.\ms_\stk(\aaddr)$, and $w = \linCons{w_1}$, $w' = \linCons{w_1'}$, $\linConsPerm{\perm}{w_1}$ and $\linConsPerm{\perm}{w_1'}$.
      \end{enumproof}
    \end{itemize}
  \item \label{case:ftlr:jump} All of the following hold: (includes control-flow manipulation cases:
    jmp, jnz, xjmp, that can be handled mostly uniformly)
    \begin{itemize}
    \item $\Phi_S \step[\gc] \Phi_S'$
    \item $\Phi_T \step \Phi_T'$
    \item $\Phi_S$ does not point to $\scall{\offpc,\offsigma}{r_1}{r_2}$ or $\txjmp{r_1}{r_2}$
    \item One of the following holds
      \begin{enumproof}
      \item(jmp,jnz)
        \begin{itemize}
        \item $\Phi_S' = \Phi_S\updReg{\pcreg,r_1}{\Phi_S(r_1),w_1}$ and
        \item $\Phi_T' = \Phi_T\updReg{\pcreg,r_1'}{\Phi_T(r_1),w_1'}$ and
        \item $\Phi_S(r_1) = \Phi_T(r_1) = ((\perm_1,\lin_1),\baddr_1,\eaddr_1,\aaddr_1)$,
        \item $\exec{\Phi_S(r_1)}$, 
        \item $\withinBounds{\Phi_S(r_1)}$,
        \item $w_1 = \linCons{\Phi_S(r_1)}$ and 
        \item $w_1' = \linCons{\Phi_T(r_1)}$
      \end{itemize}

      \item(xjmp)
        \begin{itemize}
        \item $\Phi_S(r_1) = \sealed{\sigma,c_1}$ and 
        \item $\Phi_S(r_2) = \sealed{\sigma,c_2}$ and 
        \item $\Phi_T(r_1) = \sealed{\sigma,c_1'}$ and 
        \item $\Phi_T(r_2) = \sealed{\sigma,c_2'}$ and 
        \item $c_1' \neq \retptrc(\_)$ and 
        \item $c_2' \neq \retptrd(\_)$ and 
        \item $\nonExec{\Phi_S(r_2)}$ and 
        \item $\nonExec{\Phi_T(r_2)}$ and 
        \item $\Phi_S'' = \Phi_S\updReg{r_1,r_2}{\linCons{c_1},\linCons{c_2}}$ and 
        \item $\Phi_S' = \xjumpResult{c_1}{c_2}{\Phi_S''}$ and 
        \item $\Phi_T'' = \Phi_T\updReg{r_1,r_2}{\linCons{c_1'},\linCons{c_2'}}$ and 
        \item $\Phi_T' = \xjumpResult{c_1'}{c_2'}{\Phi_T''}$
        \end{itemize}

      \end{enumproof}
    \end{itemize}
  \end{enumproof}
  The above follows from a careful analysis of the cases of the operational semantics, using the following facts:
  \begin{itemize}
  \item $[b,e] \subseteq \ta$ or $[b,e] \mathrel{\#} \ta$ (by assumption of this lemma)
  \item $\Phi_S.\reg(\pcreg) = \Phi_T.\reg(\pcreg)$ (by assumption of this lemma)
  \item $\Phi_S.\mem(a) = \Phi_T.\mem(a)$ if $a \in [b,e]$: follows from
    $\npair{[\baddr,\eaddr]} \in \xReadCond{W_\pcreg}$, the fact that $W_\pcreg
    \oplus W_R \oplus W_M$ is defined and
    $\memSat{\Phi_S.\mem,\Phi_S.\stk,\Phi_S.\ms_\stk,\Phi_T.\mem}{W_M}$
  \item $\Phi_S.\mem(a\cdots a+\calllen-1) = \Phi_T.\mem(a \cdots a+\calllen -1)$ if $[a\cdots a+\calllen-1] \subseteq [b,e]$ (follows similarly).
  \item $\Phi_S.\reg(r) = ((\perm_1,\lin_1),\baddr_1,\eaddr_1,\aaddr_1)$ implies that $\Phi_S.\reg(r) = \Phi_T.\reg(r)$ and $a_1 \in \dom(\Phi_S.\mem)$ (follows from $\npair{(\Phi_S.\reg,\Phi_T.\reg)} \in \lrrg{\trust}(W_R)$).
  \item $\Phi_S.\reg(r) = z$ implies that $\Phi_S.\reg(r) = \Phi_T.\reg(r)$ (follows from $\npair{(\Phi_S.\reg,\Phi_T.\reg)} \in \lrrg{\trust}(W_R)$).
  \item $\Phi_S.\reg(r) = \seal{\sigma_b,\sigma_e,\sigma}$ implies that $\Phi_S.\reg(r) = \Phi_T.\reg(r)$ (follows from $\npair{(\Phi_S.\reg,\Phi_T.\reg)} \in \lrrg{\trust}(W_R)$).
  \item $\Phi_S.\reg(r) = \sealed{\sigma,\vsc}$ implies that $\Phi_T.\reg(r) = \sealed{\sigma,\vsc'}$ (follows from $\npair{(\Phi_S.\reg,\Phi_T.\reg)} \in \lrrg{\trust}(W_R)$).
  \item $\Phi_S.\reg(r) \neq \retptrd(\_,\_)$ (follows from $\npair{(\Phi_S.\reg,\Phi_T.\reg)} \in \lrrg{\trust}(W_R)$).
  \item $\Phi_S.\reg(r) \neq \retptrc(\_,\_,\_)$ (follows from $\npair{(\Phi_S.\reg,\Phi_T.\reg)} \in \lrrg{\trust}(W_R)$).
  \item $\Phi_S.\reg(r) = \stkptr{\perm_1,\baddr_1,\eaddr_1,\aaddr_1}$ implies that $\Phi_T.\reg(r) = ((\perm_1,\linear),b_1,e_1,a_1)$ and $a \in \dom(\Phi_S.\ms_\stk)$ (follows from $\npair{(\Phi_S.\reg,\Phi_T.\reg)} \in \lrrg{\trust}(W_R)$).
  \item Similar facts about the current address, base and end pointer, permissions and linearity of all register capabilities being equal (follows from $\npair{(\Phi_S.\reg,\Phi_T.\reg)} \in \lrrg{\trust}(W_R)$).
  \end{itemize}

  By the above observation, we know that $\Phi_S \step[\gc] \Phi_S'$ and $\Phi_T \step
  \Phi_T'$ for some $\Phi_S'$ and $\Phi_T'$.
  According to Lemma~\ref{lem:lro-anti-red-gen}, it suffices to show:
  \[
    \npair[n-1]{(\Phi_S',\Phi_T')} \in \lro
  \]
  Consider each of the possible cases:

  In case \ref{case:ftlr:failed}, both executions go to $\failed$.
  In this case, the result follows vacuously by definition of $\lro[\preceq,\gc]$ and $\lro[\preceq,\gc]$.
  \\\\
  In case \ref{case:ftlr:halted}, both source and target configuration halts in 0
  steps, so both directions of $\lro$ are trivially satisfied.
  \\\\
  In case \ref{case:ftlr:get-and-arith-op}, we use the induction hypothesis to conclude
  \[
    \npair[n-1]{(\Phi_S',\Phi_T')} \in \lro
  \]
  from the following facts:
  \begin{itemize}
  \item One of the following sets of requirements holds:
    \begin{itemize}
    \item $\trust = \trusted$, $\Phi_S$ is reasonable up to $n-1$ steps and $[\baddr,\eaddr] \subseteq \ta$
    \item $\trust = \untrusted$ and $[b,e] \mathrel{\#} \ta$ and $\npair[n-1]{[\baddr,\eaddr]} \in \readCond{\normal,W_\pcreg}$
    \end{itemize}
    This follows from the corresponding assumption of this lemma, the fact that $\Phi_S \step[\gc] \Phi_S'$ and $\Phi_S$ does not point to $\scall{\offpc,\offsigma}{r_1}{r_2}$ or $\txjmp{r_1}{r_2}$ and \ref{lem:downwards-closed}.
  \item $\Phi_S'(\pcreg) = \Phi_T'(\pcreg) = ((\rx,\normal),\baddr,\eaddr,\_)$:
    Follows from the definition of $\updPcAddr{}$ and the assumptions of this case.
  \item $\npair[n-1]{[\baddr,\eaddr]} \in \xReadCond{W_\pcreg}$:
    Follows from the corresponding assumption of this lemma using Lemma~\ref{lem:downwards-closed}.
  \item $\npair[n-1]{(\Phi_S.\reg,\Phi_T.\reg)} \in \lrrg{\trust}(W_R)$:
    Follows from the corresponding assumption of this lemma using Lemma~\ref{lem:downwards-closed}, the definition of $\lrr$ and the fact that integers are always in $\lrvg{\trust}$ by definition.
  \item $\memSat[n-1]{\Phi_S.\mem,\Phi_S.\ms_\stk,\Phi_S.\stk,\Phi_T.\mem}{W_M}$:
    Follows from the corresponding assumption of this lemma using Lemma~\ref{lem:downwards-closed}.
  \item Theorem~\ref{thm:ftlr} holds for all $n'' < n-1$:
    follows from the corresponding assumption of this lemma, since $n - 1 < n$.
  \end{itemize}

  In case \ref{case:ftlr:call}, we may assume $r_1 \neq r_2$\footnote{If the register contains a data capability, then the execution fails in the step after the jump. If it is an executable capability, then the xjump fails as it does not permit executable capabilities for the data part.} and $r_i \neq \pcreg$ for $i \in \{1,2\}$\footnote{The pc is executable which causes the xjump to fail.} as this will cause the execution to fail.
  We need to let the target execution catch up.
  That is $\Phi_T' \nstep[15]{} \Phi_T''$ for
  \[
    \Phi_T'' = \Phi_T
    \begin{array}[t]{l}
      \update{\mem.\aaddr_\stk}{42} \\
      \updReg{\rstk}{((\rw,\linear),\baddr_\stk,\aaddr_\stk-1,\aaddr_\stk-1)} \\
      \updReg{\rretd}{\sealed{\sigma',((\rw,\linear),\aaddr_\stk,\eaddr_\stk,\aaddr_\stk-1)}} \\
      \updReg{\rretc}{\sealed{\sigma',((\perm,\normal),\baddr,\eaddr,\aaddr+\retoffset)}} \\
      \updReg{\rtmp{1}}{0}\\
      \updReg{r_1,r_2}{\linCons{\Phi_T(r_1)},\linCons{\Phi_T(r_2)}}\\
      \updReg{\pcreg}{c_1'}\\
      \updReg{\rdata}{c_2'}
    \end{array}
  \]
  where $\Phi_T(r_i) = \sealed{\sigma,c_i'} \text{ for $i\in\{1,2\}$}$ and $\retoffset = 15$ which is the offset to the return code.
  Now using Lemma~\ref{lem:lro-anti-red-gen} again, it suffices to show
  \[
    \npair[n-1]{(\Phi_S',\Phi_T'')} \in \lro
  \]
  By assumption, we have $\npair[n]{(\reg_S(r_i),\reg_T(r_i))} \in \lrvg{\trust}(W_{R,i})$ for some $W_{R,i}$ with $i \in \{1,2\}$.
  We know the capabilities in $r_1$ and $r_2$ are sealed capabilities, and by
  Lemma~\ref{lem:unique-h-sigma} and the definition of $\lrvg{\trust}$ we get
  $\npair[n-1]{(c_2,c_2')} \in H_\sigma \; \sigma \; \xi^{-1}(W_{R,2})$
  and w.l.o.g.\

  \begin{equation}
    \forall W' \future W_{R,1}, W_o, n' < n, \npair[n']{(c_2,c_2')} \in H_\sigma \; \sigma \; \xi^{-1}(W_o) \ldotp \npair[n']{c_1,c_2,c_1',c_2'} \in \lrexj(W'\oplus W_o)\label{eq:ftlr:call:sealed-code-assump}
  \end{equation}
  Now take $W_o = W_{R,2}$, take $n' = n-1$ and construct $W_{R,1}'$ as follows:

  By Lemma~\ref{lem:nsub-and-nsup-std} and the safety assumption on the register-file, there exists $S \supseteq [\baddr_\stk, \eaddr_\stk]$ such that for some $R : S \fun \powerset{\nats}$ we have $\biguplus_{r \in S} R(r) \supseteq [\baddr_\stk,\eaddr_\stk]$ and for all $r \in S$, $\pwfree[W_{R,1}](r).H \nequal \stdreg{R(r),\gc}{\pur}.H$ and $|R(r)| = 1$ and $\pwfree[W_{R,1}](r)$ is address-stratified.
  Now take $r_\var{priv\_stk}$ fresh and define
  \[
    W_{R,1}' = W_{R,1}[\free.R^{-1}([\aaddr_\stk,\eaddr_\stk]) \mapsto \revoked][\priv.r_\var{priv\_stk} \mapsto (\stareg[(\ms_{\stk\_\priv,S},\Phi_T''.\mem|_{[\aaddr_\stk,\eaddr_\stk]}),\gc]{\spa,\square},\aaddr+\calllen)]
  \]
  We know $W_{R,1}' \future W_{R,1}$ as the revoked regions must have been $\spatial$ in $W_{R,q}$ (as they are owned by the part of the world assigned to the stack-register in the register-file relation).
  The static region for the private stack is an extension of the old world.

  Pick this world as $W'$ in Eq~\ref{eq:ftlr:call:sealed-code-assump}.
  Let $W_{R,2}'$ be the same world but with the ownership of $W_{R,2}$ and pick it for $W_o$.
  Now observe that also $W_{R,2}' \future W_{R,2}$ and use monotonicity of $H_\sigma$ with the above facts to get
  \[
    \npair[n-1]{c_1,c_2,c_1',c_2'} \in \lrexj(W_{R,1}'\oplus W_{R,2}')
  \]

  Now pick register files and memories such that they form $\Phi_S''$ (defined in
  the assumptions) and $\Phi_T''$ and for $W_R'$ and $W_M'$ (we define them below) show
  \begin{enumerate}[label=\roman*.]
  \item \label{case:ftlr:scall:world} $W_{R,1}'\oplus W_{R,2}' \oplus W_R' \oplus W_M'$ is defined.
  \item \label{case:ftlr:scall:reg} $\npair[n-1]{(\Phi_S''.\reg,\Phi_T''.\reg)} \in \lrr(\{\rdata\})(W_R')$
  \item \label{case:ftlr:scall:mem} $\memSat[n-1]{\Phi_S''.\mem,\Phi_S''.\stk,\Phi_S''.\ms_\stk,\Phi_T''.\mem}{W_M'}$
  \end{enumerate}
  to get
  \[
    \npair[n-1]{(\Phi_S',\Phi_T'')} \in \lro
  \]
  as desired.

  It remains to show \ref{case:ftlr:scall:world}-\ref{case:ftlr:scall:mem}, but first we note that we can deduce the following: From the assumption $\npair{[\baddr,\eaddr]} \in \xReadCond{W}$ we get $r \in \addressable{\normal, \pwheap}$ and $\mscode$ such that $\dom(\mscode) \supseteq [\baddr,\eaddr]$ and $\pwheap(r) \nequal \codereg{\sigrets',\_,\mscode,\gc}$.
  Further by $\memSat{\ms_S,\stk,\ms_\stk,\ms_T}{W}$, we know that
  \begin{equation}
    \npair{(\mscode\uplus\mspad,\mscode\uplus\mspad)} \in H^{\var{code}} \; \sigrets' \; \sigcloss' \; \mscode \; (\ta,\_,\gsigrets,\gsigcloss)~ W_r\label{eq:ftlr:rretd:Hcode}
  \end{equation}
  where $\sigma' \in \sigrets'$ and $\dom(\mscode) \supseteq [\baddr,\eaddr]$ and $\aaddr + \offpc \in \dom(\mscode)$ and $W = W_r \oplus \_$.
  That is: $\mscode$ contains the call we are considering.
  This also entails
  \begin{equation}
    \label{eq:ftlr:code-comp-jud}
    \sigrets',\sigcloss' \vdash_\var{code\_comp} \mscode
  \end{equation}
  from \ref{eq:ftlr:rretd:Hcode} we also get
\begin{itemize}
\item $\dom(\mscode\uplus\mspad) \subseteq \ta$

  Otherwise we would have $\dom(\mscode\uplus\mspad) \# \ta$ which would contradict $\ta\supseteq [b,e] \subseteq\dom(\mscode\uplus\mspad)$

\item $\sigrets' \subseteq \gsigrets$

  Follows from the above.
\end{itemize}
  \lau{TODO: Make the above more precise and possibly put into a lemma (also needed in the next case).}

  Case \ref{case:ftlr:scall:world}: Pick $W_R'$ and $W_M'$ to have the regions of $W_{R,1}'$, but where $W_R'$ owns $r_\var{priv\_stk}$ and otherwise has the ownership of $W_R$ and $W$ with the exception of the regions owned by $W_{R,1}'$ and $W_{R,2}'$.
  $W_M'$ has the ownership of $W_M$.
  Case \ref{case:ftlr:scall:world} follows from assumption $W \oplus W_R \oplus W_M$.
  The only changes to the worlds is that some ownership has been shifted from $W_R$ to $W_{R,1}$ and $W_{R,2}$ and the ownership for $W$ now belongs to $W_R'$.
  In other words, no ownership has been duplicated.

  Case \ref{case:ftlr:scall:reg}:

  First, from reasonability of $\Phi_S$, we get that $\Phi.\reg(r)$ is reasonable in memory $\Phi.\mem$ and free stack $\Phi.\ms_\stk$ up to $n-1$ steps for all $r \neq \pcreg$.
  Lemma~\ref{lem:trusted-and-reasonable-is-untrusted} then tells us that
  $\npair[n-1]{(\Phi_S.\reg,\Phi_T.\reg)} \in \lrr(W_R)$.

  We then need to split the ownership of $W_R'$.
  From assumption $\npair[n]{(\Phi_S.reg,\Phi_T.\reg)} \in \lrr(W_R)$, we get a way to split the ownership of $W_R$.
  We take this as the starting point, but with the following changes: regions $r_1$ and $r_2$ maps to worlds with no ownership (i.e.\ $\purePart{W_{R,1}'}$).
  region $\rstk$ maps to a world with the same ownership, but of course without the now revoked regions.
  Region $\rretd$ maps to a world that owns private $r_\var{priv\_stk}$ region.
  Finally, $\rretc$ maps to a world with the ownership of $W$.

  we split the world in the same way for the registers that remain unchanged, and we get from Lemma \ref{lem:downwards-closed} and \ref{lem:monotonicity} that
  $\npair[n-1]{(\Phi_S''.\reg,\Phi_T''.\reg)} \in \lrrg{\trust}(\{\rdata,\rstk,\rtmp{1},r_1,r_2,\rretd,\rretc\})(W_R'')$ for the appropriate $W_R''$.

  To obtain $\npair[n-1]{(\Phi_S''.\reg,\Phi_T''.\reg)} \in \lrr(\{\rdata\})(W_R')$, it remains to prove the following cases:
  \begin{description}
  \item[Case $\rstk$:] Show:
    \[
      \npair[n-1]{(\stkptr{\rw,\baddr_\stk,\aaddr_\stk-1,\aaddr_\stk-1},((\rw,\linear),\baddr_\stk,\aaddr_\stk-1,\aaddr_\stk-1))} \in \lrv(W_{R,\rstk}')
    \]
    We know by Assumption $\npair[n]{(\Phi_S.reg,\Phi_T.\reg)} \in \lrrg{\trust}(W_R)$, we know
    \[
      \npair[n]{(\stkptr{\rw,\baddr_\stk,\eaddr_\stk,\aaddr_\stk},((\rw, \linear), \baddr_\stk ,\eaddr_\stk, \aaddr_\stk)} \in \lrvg{\trust}(W_{R,\stk})
    \]
    which by the $\mathit{stackReadCondition}$ gives $S \subseteq \addressable{\linear,\pwfree}$ and $R : S \fun \powerset{\nats}$.
    We need to pick an $S'$ to argue
    \[
      \npair[n-1]{[\baddr_\stk,\aaddr_\stk-1]} \in \stackReadCond{W_{R,\rstk}'}
    \]
    To this end pick $S' = R^{-1}(\left(\bigcup_{r\in S} R(r) \right) \setminus [\aaddr_\stk,\eaddr_\stk])$ and $R'$ to be $R$ limited to $S'$.
    As we exclude all the revoked regions and the ownership otherwise remains the
    same in $W_{R,\rstk}'$ as in $W_{R,\rstk}$, the regions in $S'$ are exactly the same in the new world as they were in $W_{R,\rstk}$.
    So what we need to show follows immediately from
    \[
      \npair[n]{[\baddr_\stk,\eaddr_\stk]} \in \stackReadCond{W_{R,\rstk}}
    \]
    and by Lemma~\ref{lem:downwards-closed}.

    We show
    \[
      \npair[n-1]{[\baddr_\stk,\aaddr_\stk-1]} \in \stackWriteCond{W_{R,\rstk}'}
    \]
    in the same way.

  \item[Case $\rtmp{1}$:] Show:
    \[
      \npair[n-1]{(0,0)} \in \lrv(W_{R,\rtmp{1}}')
    \]
    Follows immediately from the definition.

  \item[Case $r_1$,$r_2$:] The two cases are symmetric, so we just show the $r_1$ case:
    \[
      \npair[n-1]{(\linCons{c_1},\linCons{\Phi_T(r_1)})} \in \lrv(W_{R,r_1}')
    \]
    Follows from the related register files assumption and Lemma~\ref{lem:lincons-lrv} and the fact that for sealed capabilities if they are in $\lrvg{\trusted}$, then they are in the $\lrv$ part.

  \item[Case $\rretd$:] We have to show
    \[
      \npair[n-1]{(\sealed{\sigma',\retptrd(\aaddr_\stk,\eaddr_\stk)},\sealed{\sigma',((\rw,\linear),\aaddr_\stk,\eaddr_\stk,\aaddr_\stk-1)})} \in \lrv(W_{R,\rretd}')
    \]
    where $W_{R,\rretd}'$ is the part of $W_R'$ with ownership over $r_\var{priv\_stk}$.

    Use the $r=r_{\mscode}$ as the witness.
    The $\xReadCond{}$ gives us $\pwheap(r) \nequal \codereg{\_,\_,\mscode,\gc}$.
    As it is a pure region, it is also present in the future world we consider.
    We now have to show:
    \begin{enumerate}[label=\alph*)]
    \item for $n'' < n-1$ we have
      \[\npair[n'']{\left(\array{l}\retptrd(\aaddr_\stk,\eaddr_\stk),\\((\rw,\linear),\aaddr_\stk,\eaddr_\stk,\aaddr_\stk-1)\endarray\right)}
        \in H^\mathrm{code,\square}_\sigma \; \sigrets' \; \sigcloss' \; \mscode'
        \; (\ta,\stkb) \; \sigma' \; \xi^{-1}(W_{R,\rretd}')
      \]
      First, we already know $\dom(\mscode) \subseteq \ta$ and $\sigrets' \subseteq \gsigrets$.
      Now pick $r_\var{priv\_stk}$ as the witness.
      We immediately get $\dom(\ms_{\stk\_\priv,S}) = \dom(\Phi_T''.\mem|_{[\aaddr_\stk,\eaddr_\stk]}) = [\aaddr_\stk,\eaddr_\stk]$.
      Next, $\decInstr{\mscode'(\aaddr,\aaddr+\calllen-1)} = \overline{\scall{\offpc,\offsigma}{r_1}{r_2}}$ follows from $\callCond{\Phi_S,r_1,r_2,\offpc,\offsigma,\aaddr}$.
      Finally, $\ms_\code(\aaddr+\offpc) = \seal{\sigma_\baddr,\sigma_\eaddr}$ with $\sigma' = \sigma_\baddr + \offpc \in \sigrets'$ which follows from $\sigrets',\sigcloss' \vdash_{\mathrm{comp-code}} \mscode$ and the fact that the call is there.
    \item $\isLinear{\sealed{\sigma',\retptrd(\aaddr_\stk,\eaddr_\stk)}}$ iff
      $\isLinear{\sealed{\sigma',((\rw,\linear),\aaddr_\stk,\eaddr_\stk,\aaddr_\stk)}}$\\
      Trivial, both are linear.
    \item
      \[
        \begin{array}{l}
          \forall W'' \future \purePart{W_{R,\rretd}'}, W_o, n'' < n-1,\\
          \quad\npair[n'']{(\vsc_S',\vsc_T')}
          \in  H^\mathrm{code,\square}_\sigma \; \sigrets' \; \sigcloss' \; \mscode'
          \; (\ta,\stkb) \; \sigma' \; \xi^{-1}(W_o) \ldotp \\
          \qquad\npair[n'']{\retptrd(\aaddr_\stk,\eaddr_\stk),\vsc_S',((\rw,\linear),\aaddr_\stk,\eaddr_\stk,\aaddr_\stk-1),\vsc_T'} \in
          \lrexj(W'\oplus W_o))
        \end{array}
      \]
      Trivial as both configurations fails.
    \end{enumerate}

  \item[Case $\rretc$:]
    We have to show
    \[
      \npair[n-1]{\left(\array{l}\sealed{\sigma',\retptrc(\baddr,\eaddr,\aaddr +
            \calllen)},\\\sealed{\sigma',((\rx,\normal),\baddr,\eaddr,\aaddr + \retoffset)}\endarray\right)} \in \lrv(W_{R,\rretc}')
    \]
    where $W_{R,\rretc}'$ has the ownership of $W$.
    Just as in the previous case, we know that for some region $r$ there exists an $\mscode$ such that: $\pwheap(r) \nequal \codereg{\sigrets',\_,\mscode,\gc}$ where $\sigma' \in \sigrets'$ and $\dom(\mscode) \supseteq [\baddr,\eaddr]$ and $\aaddr + \offpc \in \dom(\mscode)$ and $\sigrets' \subseteq \gsigrets$.
    It follows easily from the definition of $H_\sigma^\var{code}$ that
    \begin{equation}
      \label{eq:code-caps-in-h-sigma}
      \npair[n'']{\left(\array{l}\retptrc(\baddr,\eaddr,\aaddr + \calllen),\\((\rx,\normal),\baddr,\eaddr,\aaddr + \retoffset)\endarray\right)} \in H^\mathrm{code,\square}_\sigma \; \sigrets' \; \sigcloss' \; \mscode \; (\ta,\stkb) \; \sigma' \; \xi^{-1}(W_{R,\rretc}')
    \end{equation}
    for $n'' < n-1$.
    Both capabilities are non-linear, so
    \[
      \isLinear{\retptrc(\baddr,\eaddr,\aaddr + \calllen)}\text{ iff } \isLinear{((\rx,\normal),\baddr,\eaddr,\aaddr + \retoffset)}
    \]
    is indeed the case.

    Finally we need to show:
    \[
      \begin{array}{l}
        \forall W'' \future \purePart{W_{R,\rretc}'}, W_o, n'' < n-1,\\
        \quad  \npair[n'']{(\vsc_S',\vsc'_T)} \in H^\mathrm{code,\square}_\sigma \; \sigrets' \; \sigcloss' \; \mscode \; (\ta,\stkb) \; \sigma' \; \xi^{-1}(W_o) \ldotp \\
        \qquad \npair[n'']{\retptrc(\baddr,\eaddr,\aaddr + \calllen),\vsc_S',((\rx,\normal),\baddr,\eaddr,\aaddr + \retoffset),\vsc_T'} \in \lrexj(W''\oplus W_o))
      \end{array}
    \]

    To this end let $W'' \future \purePart{W_{R,\rretc}'}$ and $W_o$ be given s.t. $W'' \oplus W_o$ is defined.
    Further, let $\npair[n'']{(\vsc_S',\vsc'_T)} \in H^\mathrm{code,\square}_\sigma \; \sigrets' \; \sigcloss' \; \mscode \; (\ta,\stkb) \; \sigma' \; \xi^{-1}(W_o)$ be given and show
    \[
      \npair[n'']{\retptrc(\baddr,\eaddr,\aaddr + \calllen),\vsc_S',((\rx,\normal),\baddr,\eaddr,\aaddr + \retoffset),\vsc_T'} \in \lrexj(W''\oplus W_o))
    \]

    Now let $n''' \leq n''$ be given along with $\reg[3]_S$, $\reg[3]_T$, $\ms^{(3)}_S$, $\ms^{(3)}_T$,$\ms^{(3)}_\stk$, $\stk^{(3)}$, $W_R''$, and $W_M''$ such that
    \begin{itemize}
    \item $W'' \oplus W_o \oplus W_R'' \oplus W_M''$ is defined
    \item $\memSat[n''']{\ms^{(3)}_S,\ms^{(3)}_\stk, \stk^{(3)}, \ms^{(3)}_T}{W_M''}$
    \item $\npair[n''']{(\reg[3]_S,\reg[3]_T)} \in \lrr(W_R'')$
    \end{itemize}

    Based on $H_\sigma^\var{code}$, there are three possible values for $\vsc_S'$ and $\vsc_T'$.
    In the first case, $\vsc_S'$ is a $\retptrc$ and $\vsc_T'$ is a capability with permission $\rx$.
    In this case, $\var{xjumpResult}$ will produce failed configurations which are trivially in the observation relation.
    In the next case, it is required that $\sigma' \in \sigcloss'$, but this
    cannot be the case as $\sigma' \in \sigrets'$ and we have $\sigrets',\sigcloss' \vdash_{\mathrm{comp-code}} \mscode$, which implies that $\sigcloss' \mathrel{\#} \sigrets'$.

    This leaves us with one final case, namely $\vsc_S'= \retptrd(\baddr_\stk',\eaddr_\stk')$ and $\vsc_T'=((\rw,\linear),\baddr_\stk',\eaddr_\stk',\baddr_\stk'-1)$.
    Further we know
    \[
      \begin{array}{l}
        \exists r \in \addressable{\linear,\pwpriv[W_o]} \ldotp \pwpriv[W_o](r).H \nequal (\stareg[(\ms_{\priv,S}^{(3)},\ms_{\priv,T}^{(3)}),\gc]{\spao,\square}, \aaddr'+\calllen) \tand \\
        \quad \dom(\ms_{\priv,S}^{(3)}) = \dom(\ms_{\priv,T}^{(3)}) = [\baddr_\stk',\eaddr_\stk'] \tand\\
        \quad \decInstr{\code([\aaddr',\aaddr' + \calllen-1])} = \overline{\scall{\offpc,\offsigma}{r_1}{r_2}} \tand \\
        \quad \code(\aaddr'+\offpc) = \seal{\sigma_b,\sigma_e,\sigma_b} \tand \sigma = \sigma_b + \offsigma \in \sigrets
      \end{array}
    \]
    Call this region $r'$.
    By the above, the fact that the code capability pair is in $H_\sigma^\var{code}$ (\ref{eq:code-caps-in-h-sigma}) and Lemma~\ref{lem:unique-ret-seals}, we get $\aaddr' = \aaddr$.
    This means that $W_M''$ and $W_R''$ have this region.

    We know that the two register-files are related which in particular means that the values in register $\rstk$ are related.
    Now consider the following cases: 
    \begin{itemize}
    \item $\reg_S^{(3)}(\rstk) \neq \stkptr{\_,\_,\_,\_}$\\
      In this case due to $\reg[3]_S$ being related to $\reg[3]_T$, there are three cases we need to consider.
      In all cases, the source configuration will fail because the value in the stack register is not a stack capability.
      If we can argue that the target configuration will also fail, then the two are in the observation relation.
      First, if $\reg[3]_T = \sealed{\sigma_\var{ret\_\stk},\vsc_T''}$, then the return code will fail when the base address (a sealed capability has no base address, so the instruction returns $-1$) is compared to $\stkb$.
      Second, if $\reg[3]_T = \seal{\_,\_,\_}$, then the target execution fails when it attempts to splice this seal with $\vsc_T'$ (which we know is not a $\seal{}$ capability).
      Finally, $\reg[3]_T = ((\perm_\var{ret\_\stk},\lin_\var{ret\_\stk}),\baddr_\var{ret\_\stk},\eaddr_\var{ret\_\stk},\_)$ and $\npair[n''']{[\baddr_\var{ret\_\stk},\eaddr_\var{ret\_\stk}]} \in \readCond{\lin_\var{ret\_\stk},W_{R,\rstk}''}$ is satisfied.
      This means that it is satisfied by some heap region, but by the memory satisfaction assumption $\stkb$ must be in the free stack part of the world.
      This means that the execution will fail that $\stkb$ check.

    \item $\reg_S^{(3)}(\rstk) = \stkptr{\perm_\var{ret\_\stk},\baddr_\var{ret\_\stk},\_,\_}$ and $\baddr_\var{ret\_\stk} \neq \stkb$\\
      Here the source side will fail the xjmp as the base address is not $\stkb$.
      Similarly on the target side, the return code will fail the $\stkb$ check.

    \item $\reg_S^{(3)}(\rstk) = \stkptr{\_,\baddr_\var{ret\_\stk},\eaddr_\var{ret\_\stk},\_}$ and $\baddr_\var{ret\_\stk} = \stkb$ and either $\eaddr_\var{ret\_\stk} +1 \neq \baddr_\stk' $ or $\perm_\var{ret\_\stk} \neq \rw$ or
      $\baddr_\var{ret\_\stk} > \eaddr_\var{ret\_\stk}$.\\
      In this case, the source configuration will fail as one of the conditions in $\var{xjumpResult}$ will not be met.
      On the target side, the splice will fail as either the two capabilities being spliced don't line up, the permissions don't match, or the range of authority is empty, respectively.

    \item $\reg_S^{(3)}(\rstk) = \stkptr{\perm_\var{ret\_\stk},\baddr_\var{ret\_\stk},\eaddr_\var{ret\_\stk},\_}$ and $\baddr_\var{ret\_\stk} = \stkb$ and $\eaddr_\var{ret\_\stk} +1 = \baddr_\stk' $ and $\perm_\var{ret\_\stk} = \rw$ and $\baddr_\var{ret\_\stk} \leq \eaddr_\var{ret\_\stk}$.
      \\
      We would like to show that $\ms_{\priv,S}^{(3)}$ is the top most stack frame and that $r'$ governs it.
      By the memory satisfaction assumption and the presence of $r'$ in $W_M''$ we know that $\stk^{(3)}$ is non-empty.
      By the memory satisfaction on the private stack, the following must be the case:
      \[
        \begin{array}{l}
          \stk^{(3)} = (\opc_0,\ms_0), \dots (\opc_m,\ms_m) \wedge \\
          \forall i \in \{0,\dots,m\} \ldotp (\dom(\ms_i) \neq \emptyset \wedge\\
          \quad \forall i < j \ldotp \forall a \in \dom(\ms_i) \ldotp \forall a' \in \dom(\ms_j) \ldotp \stkb < a < a')\\
        \end{array}
      \]

      Assume for contradiction $\ms_{\priv,S}^{(3)}$ is not the top frame. In that
      case $\dom(\ms_0) \neq \emptyset$ and $\forall a \in \dom(\ms_0) \ldotp
      \ldotp \stkb < a < b_\stk$ at the same time, we know
      \[
        \npair[n'']{(\stkptr{\rw,\stkb,\eaddr_\var{ret\_\stk},\_},((\rw,\linear),\stkb,\eaddr_\var{ret\_\stk},\_))}
        \in \lrv(W_{R,\rstk}'')
      \]
      which means that the free stack part of the world contains a region that at least governs $[\stkb,\eaddr_\var{ret\_\stk}]$.
      Combine this with $\eaddr_\var{ret\_\stk} +1 = \baddr_\stk' $, we can conclude that no such address can exist in $\ms_0$, so it must be empty, but this cannot be the case either.
      Therefore, the top stack frame must contain $\ms_{\priv,S}^{(3)}$.

      Further, due to the disjointedness required by memory satisfaction, it must be $r'$ that governs this stack frame.
      This also means that we have $\opc_0 =  \aaddr+\calllen$.
      With this, we have all the requirements for $\var{xjumpResult}$ satisfied on
      both sides which allows us to pick the necessary configurations:
      \begin{multline*}
        \Phi^{(4)}_S=\xjumpResult{r_1}{r_2}{(\ms_S^{(3)},\reg[3]_S,(\ms_{\priv,S}^{(3)},\aaddr+\calllen)
          :: \stk^{(3)}_\var{rest},\ms_\stk^{(3)})} =\\
        (\ms_S^{(3)},\reg[3]_S
        \arraycolsep=0pt
        \begin{array}[t]{l}
          \updReg{\pcreg}{((\rx,\normal),\baddr,\eaddr,\aaddr+\calllen)} \\
          \updReg{\rdata}{0} \\
          \updReg{\rstk}{\stkptr{\rw,\stkb,e_\stk',\eaddr_\var{ret\_stk}+1}} \\
          \updReg{\rtmp{1}}{0} \\
          \updReg{\rtmp{2}}{0}
        \end{array}
        ,\stk^{(3)}_\var{rest},\ms_{\priv,S}^{(3)} \uplus \ms_\stk^{(3)})
      \end{multline*}
      and
      \begin{multline*}
        \Phi^{(4)}_T=\xjumpResult{r_1}{r_2}{(\ms_T^{(3)},\reg[3]_T)} =\\
        (\ms_T^{(3)},\reg[3]_S
        \arraycolsep=0pt
        \begin{array}[t]{l}
          \updReg{\pcreg}{((\rx,\normal),\baddr,\eaddr,\aaddr+\retoffset)} \\
          \updReg{\rdata}{((\rw,\linear),\stkb,\eaddr_\var{ret\_stk})}
        \end{array})
      \end{multline*}
      It now remains to show
      \[
        \npair[n''']{(\Phi^{(4)}_S,\Phi^{(4)}_T)} \in \lro
      \]
      Use Lemma~\ref{lem:lro-anti-red-gen} by which it suffices to show the
      following two things:
      \begin{itemize}
      \item
        \[
          \Phi^{(4)}_T \step[l] \Phi^{(5)}_T = (\ms_T^{(3)},\reg[3]_S
          \arraycolsep=0pt
          \begin{array}[t]{l}
            \updReg{\pcreg}{((\rx,\normal),\baddr,\eaddr,\aaddr+\calllen)} \\
            \updReg{\rdata}{0} \\
            \updReg{\rstk}{((\rw,\linear),\stkb,e_\stk',\eaddr_\var{ret\_stk}+1)} \\
            \updReg{\rtmp{1}}{0} \\
            \updReg{\rtmp{2}}{0}
          \end{array})
        \]
        for some number of steps $l$. This follows immediately from the
        operational semantics.
      \item
        \[
          \npair[n''']{(\Phi^{(4)}_S,\Phi^{(5)}_T)} \in \lro
        \]
        For fresh $r_{b_\stk'}\dots r_{e_\stk'}$ and $W'''$ defined as
        \[
          W''' = W''[\priv.r' \mapsto \revoked,\free.r_{b_\stk'}\dots r_{e_\stk'}\mapsto \stdreg{\{b_\stk'\},\gc}{\spa} \dots \stdreg{\{e_\stk'\},\gc}{\spa}]
        \]
        and $W'''_R$ the same as $W'''$, but with the ownership of $W''_R$ as well as for the regions $r_{b_\stk'}\dots r_{e_\stk'}$,
        and $W'''_M$ the same as $W'''$ but with the ownership of $W''_M$,
        we show the following:
        \begin{enumproof}
        \item $\npair[n''']{(\Phi_S^{(4)}.\reg,\Phi_T^{(5)}.\reg)} \in \lrr(W'''_R)$
        \item $\memSat[n''']{\Phi_S^{(4)}.\mem,\Phi_S^{(4)}.\stk,\Phi_S^{(4)}.\ms_\stk,\Phi_T^{(5)}.\mem}{W'''_M}$ 
        \item $W''' \oplus W'''_R \oplus W'''_M$ defined
        \end{enumproof}

        Assuming the above, we use our assumption that Theorem~\ref{thm:ftlr} holds for all $n' < n$ to get
        \[
          \npair[n''']{(((\rx,\normal),\baddr,\eaddr,\aaddr+\calllen),((\rx,\normal),\baddr,\eaddr,\aaddr+\calllen))} \in \lre(W''')
        \]
        Note that from assumption ``$\Phi_S$ reasonable up to $n$ steps'' and Lemma~\ref{lem:ec-reasonable-downwards-closed} and $n''' \le n'' < n - 1$ we get ``$\Phi_S$ reasonable up to $n'''+1$ steps'' from which it follows that ``$\Phi_S(\pcreg) + \calllen$ reasonable up to $n'''$ steps''.
        Using this along with the register-file safety and memory satisfaction,
        we get
        \[
          \npair[n''']{(\Phi^{(4)}_S,\Phi^{(5)}_T)} \in \lro
        \]
        as desired. We need to show the three things we skipped:

        Show:
        \[
          \npair[n''']{(\Phi_S^{(4)}.\reg,\Phi_T^{(5)}.\reg)} \in \lrr(W'''_R)
        \]
        We will split the world like in $\npair[n''']{(\reg_S^{(3)},\reg_T^{(3)})} \in \lrr(W)$, but where $r_\stk$ also get the ownership of $r_{b_\stk'}\dots r_{e_\stk'}$.
        We need to show the following
        \begin{itemize}
        \item Case $\rdata$, $\rtmp{1}$, $\rtmp{2}$:

          Trivial.
        \item Case $r \not\in \{\rdata,\rtmp{1},\rtmp{2},\pcreg,\rstk\}$:

          We have $\Phi^{(4)}_S.\reg(r) = \reg^{(3)}_S(r)$ and $\Phi^{(5)}_T.\reg(r) = \reg^{(3)}_T(r)$.
          We already know:
          \[\npair[n''']{(\reg^{(3)}_S(r),\reg^{(3)}_T(r))} \in \lrr(W''_{R,r}),\]
          This is true for some $W_{R,r}''$ which does not have the ownership of the regions that are revoked in $W_R'''$, so that we can take the corresponding $W_{R,r}'''$ and have $W_{R,r}''' \future W_{R,r}''$.
          From Lemma~\ref{lem:monotonicity}, we then get
          \[
            \npair[n''']{(\reg^{(3)}_S(r),\reg^{(3)}_T(r))} \in \lrr(W'''_{R,r})
          \]
          as desired.
        \item Case $\rstk$:

          For this case, we need to show
          \[
            \npair[n''']{[\stkb,e_\stk']} \in \stackReadCond{W'''_{R,\rstk}}
          \]
          and
          \[
            \npair[n''']{[\stkb,e_\stk']} \in \stackWriteCond{W'''_{R,\rstk}}
          \]
          For $W'''_{R,\rstk}$ that owns the same as $W''_{R,\rstk}$ as well as $r_{b_\stk'},\dots,r_{e_\stk'}$.

          For the first part, we use $\npair[n''']{(\reg_S^{(3)}(\rstk),\reg_T^{(3)}(\rstk))} \in \lrv(W_{R,\rstk}'')$ and the fact that the stack capability must have $\rw$ permission from which it follows that
          \[
            \npair[n''']{[\stkb,e_\var{ret\_stk}]} \in \stackReadCond{W''_{R,\rstk}}
          \]
          which gives us $S_\free \subseteq \addressable{\linear,\pwfree}$ and $R_\free : S_\free \fun \powerset{\nats}$ such that
          \begin{itemize}
          \item $\forall r \in S_\free \ldotp |R(r)| = 1$
          \item $\biguplus_{r \in S_\free} R(r) \subseteq [\stkb,e_\var{ret\_stk}]$
          \item $\forall r \in S_\free \ldotp \pwfree(r) \nsupeq \stdreg{R(r),\gc}{\spao}$
          \end{itemize}
          Now pick $S_\var{read} = S_\free \cup \{r_{b_\stk'},\dots,r_{e_\stk'}\}$ and
          \[
            R_\var{read}(r) =
            \begin{cases}
              R_\var{read}(r) & r \in S_\var{read} \\
              \{a\} & r \in \{r_{b_\stk'},\dots,r_{e_\stk'}\} \wedge r = r_a
            \end{cases}
          \]
          It is clearly the case that $\forall r \in S_\var{read} \ldotp |R_\var{read}(r)| = 1$ and $\biguplus_{r \in S_\var{read}} R(r) \subseteq [\stkb,e_\var{ret\_stk}]$.
          For $\forall r \in S_\var{free} \ldotp \pwfree(r) \nsupeq \stdreg{R(r),\gc}{\spao}$ it follows from $\forall r \in S_\free \ldotp \pwfree(r) \nsupeq \stdreg{R(r),\gc}{\spao}$ and $\pwfree(r) = \stdreg{R(r),\gc}{\spao}$ for $r \in \{r_{b_\stk'},\dots,r_{e_\stk'}\}$.

          The $\var{stackReadCondition}$ is shown in the same way, but we also use Lemma~\ref{lem:stdreg-singleton-addr-strat} to argue that the new regions are address-stratified.
        \end{itemize}

        Show:
        \begin{equation}
          \memSat[n''']{\ms^{(3)}_S,\Phi_S^{(4)}.\stk,\ms_\stk^{(3)} \uplus \ms_{\priv,S}^{(3)} ,\ms^{(3)}_T}{W'''_M}\label{eq:ftlr:case:last-mem-sat}
        \end{equation}

        From the assumption $\memSat[n''']{\ms^{(3)}_S,\ms^{(3)}_\stk, \stk^{(3)}, \ms^{(3)}_T}{W_M''}$, we know
        \begin{itemize}
        \item $\stk = (\addr + \calllen,\ms_{\priv,S}^{(3)})::(\opc_1,\ms_1):: \dots :: (\opc_m,\ms_m)$
        \item $\ms^{(3)}_S \uplus \ms^{(3)}_\stk \uplus \ms_{\priv,S}^{(3)} \uplus \ms_1 \uplus \dots \uplus \ms_m  \text{ is defined}$
        \item $W''_M = W_{\var{stack}}'' \oplus W_{\var{free\_stack}}'' \oplus W_{\var{heap}}''$
        \item $\ms_\var{T,stack}, \ms_\var{T,free\_stack}, \ms_\var{T,heap}, \ms_{T,f}, \ms_{S,f}, \ms_S',\overline{\sigma}$ such that
          \begin{itemize}
          \item $\ms^{(3)}_S = \ms_{f,S} \uplus \ms_S'$
          \item $\ms^{(3)}_T = \ms_\var{T,stack} \uplus \ms_\var{T,free\_stack} \uplus \ms_\var{T,heap} \uplus \ms_{T,f}$
          \item $\memSatStack[n''']{\stk,\ms_\var{T,stack}}{W_{\var{stack}}''}$
          \item $\memSatFStack[n''']{\ms^{(3)}_\stk,\ms_\var{T,free\_stack}}{W_{\var{free\_stack}}''}$
          \item $\npair[n''']{(\overline{\sigma},\ms_S',\ms_\var{T,heap})} \in \lrheap(\pwheap)(W_{\var{heap}}'')$
          \end{itemize}
        \end{itemize}
        We will pick the same things to show \ref{eq:ftlr:case:last-mem-sat} with a few changes.
        We have to show
        \begin{itemize}
        \item $\Phi_S^{(4)}.\stk = (\opc_1,\ms_1):: \dots :: (\opc_m,\ms_m)$

          By the memory satisfaction assumption and the change to the stack.
        \item $\ms^{(3)}_S \uplus \ms^{(3)}_\stk \uplus \ms_{\priv,S}^{(3)} \uplus \ms_1 \uplus \dots \uplus \ms_m \text{ is defined}$

          By the memory satisfaction assumption.
        \item $W'''_M = W_{\var{stack}}''' \oplus W_{\var{free\_stack}}''' \oplus W_{\var{heap}}'''$

          Define the new worlds to have the ownership of their $W_M''$ counterparts except $W_{\var{stack}}'''$ does not take ownership of the regions used for the safety of addresses $b_\stk',\dots,e_\stk'$.
          This ownership goes to $W_{\var{free\_stack}}'''$ instead.
        \item We need to pick partitions of $\ms^{(3)}_T$ and a frame for $\ms^{(3)}_S$.
          We pick the same as we get from the memory satisfaction assumption except we pick the free stack partition of the target memory to be $\ms_\var{T,free\_stack} \uplus \ms_{\priv,S}^{(3)}$ and the stack partition to be $\ms_\var{T,stack}\setminus \ms_\var{T,stack} | _{\dom(\ms_{\priv,S}^{(3)})}$
        \item $\ms^{(3)}_S = \ms_{f,S} \uplus \ms_S'$

          By assumption.
        \item $\ms^{(3)}_T = \ms_\var{T,stack} \uplus \ms_\var{T,free\_stack} \uplus \ms_\var{T,heap} \uplus \ms_{T,f}$

          By assumption and the fact that the only change is that we moved part of the stack to the free stack.
        \item $\memSatStack[n''']{(\opc_1,\ms_1):: \dots :: (\opc_m,\ms_m),\ms_\var{T,stack}\setminus \ms_\var{T,stack} | _{\dom(\ms_{\priv,S}^{(3)})}}{W_{\var{stack}}'''}$

          Follows easily from the private stack satisfaction assumption.
          The distribution functions from the assumption are simply limited to forget about the now revoked region and extend the world partition to be on $W_M'''$, but with the same ownership as in the one we had in the assumption.
          All the new partitions are future worlds of the old ones as none of them owned the revoked region (it was owned by $W_R''$).
        \item $\memSatFStack[n''']{\ms^{(3)}_\stk \uplus \ms_{\priv,S}^{(3)},\ms_\var{T,free\_stack} \uplus \ms_\var{T,stack} | _{\dom(\ms_{\priv,S}^{(3)})}}{W_{\var{free\_stack}}'''}$

          From the $\memSatFStack[n''']{\ms^{(3)}_\stk,\ms_\var{T,free\_stack}}{W_{\var{free\_stack}}''}$ assumption we get
          \begin{itemize}
          \item $R_\ms : \dom(\activeReg{W_{\var{free\_stack}}''}) \fun \MemSeg \times \MemSeg$ and
          \item $R_W : \dom(\activeReg{W_{\var{free\_stack}}''}.\free) \fun \Worlds$
          \end{itemize}
          for which
          \begin{itemize}
          \item $\ms^{(3)}_\stk = \biguplus_{r \in \dom(\activeReg{W_{\var{free\_stack}}''.\free})} \pi_1(R_\ms(r))$
          \item $\ms_\var{T,free\_stack} = \biguplus_{r \in \dom(\activeReg{W_{\var{free\_stack}}''.\free})} \pi_2(R_\ms(r))$
          \item $\stkb \in \dom(\ms^{(3)}_\stk)$ and $\stkb \in \dom(\ms_\var{T,free\_stack})$
          \item $W_{\var{free\_stack}}'' \bigoplus_{r \in \dom(\activeReg{W_{\var{free\_stack}}''.\free})} R_W(r)$
          \item $\forall r \in \dom(\activeReg{W_{\var{free\_stack}}''.\free}), n'''' < n'''\ldotp \npair[n'''']{R_\ms(r)} \in W_{\var{free\_stack}}''.\free(r).H \; \xi^{-1}(R_W(r))$
          \end{itemize}
          Now pick
          \[
            R_\ms'(r) =
            \begin{cases}
              (\ms_{\priv,S}^{(3)}|_{\{a'\}} , \ms_\var{T,stack}|_{\{a'\}}) & r_{a'} \in \{r_{b_\stk'}\dots r_{e_\stk'}\}\\
              R_\ms(r) & \totherwise
            \end{cases}
          \]
          and
          \[
            R_W'(r) =
            \begin{cases}
              W'''_{a'} & r_{a'} \in \{r_{b_\stk'}\dots r_{e_\stk'}\} \\
              R_W(r)[\priv.r' \mapsto \revoked,\free.r_{b_\stk'}\dots r_{e_\stk'}\mapsto \stdreg{\{b_\stk'\},\gc}{\spa} \dots \stdreg{\{e_\stk'\},\gc}{\spa}]& \totherwise
            \end{cases}
          \]
          for $W'''_{a'}$ constructed as follows: $W_{\var{stack}}''$ is the part of the world given to the stack judgement in the assumption.
          This world is split into a number of parts to satisfy the memory interpretation of each frame.
          Say $W_{top,\var{stack}}''$ is used for the top stack frame in the assumption.
          The top stack frame is governed by a static region, so by definition it is split into parts that satisfy each of the addresses.
          That is for $a' \in \{b_\stk' \dots e_\stk'\}$, $W_{a'}''$ is the part that makes the value in memory satisfy the value relation.
          Now let $W_{a'}'''$ be $W'''$ but with the ownership of $W_{a'}''$

          It is easy to see that the $R_\ms'$ constructs the two memories, and
          from the assumption, we also get that the base stack address is in there.

          It remains to show
          \[
            \forall r \in \dom(\activeReg{W_{\var{free\_stack}}'''.\free}),n'''' < n''' \ldotp \npair[n'''']{R_\ms(r)} \in W_{\var{free\_stack}}'''.\free(r).H \; \xi^{-1}(R_W(r))
          \]
          for $r \in \dom(\activeReg{W_{\var{free\_stack}}'''.\free}) \setminus \{r_{b_\stk'}\dots r_{e_\stk'}\}$, it follows from monotonicity of the $H$ (memory interpretation) function.
          For $r_{a'} \in \{r_{b_\stk'}\dots r_{e_\stk'}\}$ and $n'''' < n'''$, we need to show
          \[
            \npair[n'''']{(\ms_{\priv,S}^{(3)}(a') , \ms_\var{T,stack}|(a'))} \in \stdreg{\{a\},\gc}{\spa}.H (\xi^{-1}R_W(r_{a'}))
          \]
          which amounts to showing $\dom(\ms_{\priv,S}^{(3)}|_{\{a'\}}) = \dom(\ms_\var{T,stack}|_{\{a'\}}) \{a'\}$, which is the case, and
          \[
            \npair[n'''']{(\ms_{\priv,S}^{(3)}(a') , \ms_\var{T,stack}|(a'))} \in \lrv (W_{a'}''')
          \]
          Which follows from Lemma~\ref{lem:monotonicity} and the fact that we have a memory satisfaction assumption in which $\ms_{\priv,S}^{(3)}$ is governed by a standard static region.

        \item $\npair[n''']{(\overline{\sigma},\ms_S',\ms_\var{T,heap})} \in \lrheap(\pwheap)(W_{\var{heap}}''')$

          Follows by Lemma~\ref{lem:monotonicity} and the heap satisfaction assumption.
        \end{itemize}

        Argue:
        \[
          W''' \oplus W'''_R \oplus W'''_M \text{ is defined.}
        \]
        This follows from the assumption $W'' \oplus W''_R \oplus W''_M$ and the fact that each of the new worlds is constructed from one of the past worlds and only one of them claims the ownership of the new regions.
      \end{itemize}
    \end{itemize}

\end{description}
This concludes case \ref{case:ftlr:scall:reg}

  Case \ref{case:ftlr:scall:mem}:
  we need to show:
  \[
    \memSat[n-1]{\Phi_S''.\mem,\Phi_S''.\stk,\Phi_S''.\ms_\stk,\Phi_T''.\mem}{W_M'}
  \]
  which amounts to
  \[
    \memSat[n-1]{\Phi_S.\mem,((\aaddr+\calllen),\ms_{\stk\_\priv,S}) :: \Phi_S,\Phi_S.\ms_\stk - \Phi_S.\ms_\stk|_{[\aaddr_\stk,\eaddr_\stk]},\Phi_T''.\mem}{W_M'}
  \]
  for $\ms_{\stk\_\priv,S} = \Phi_S.\ms_\stk |_{[\aaddr_\stk,\eaddr_\stk]}\update{\aaddr_\stk}{42}$.

  In order to show this, we will first show the following:
  \begin{itemize}
  \item $\baddr_\stk = \stkb$

    We know $\Phi_S$ is reasonable up to $n$ steps.
    Further, from $\callCond{\Phi_S,r_1,r_2,\offpc,\offsigma,a}$ and $b \leq a$ and $a + \calllen -1 \leq e$ and $[b,e] \subseteq \ta$ we can conclude that $\Phi_S$ points to $\src{\scall{\offpc,\offsigma}{r_1}{r_2}}$ in $\ta$.
    By the \emph{Guarantee stack base address before call} we then know $\Phi_S(r_\stk) = \stkptr{\_,\stkb,\_,\_}$.
  \end{itemize}

  By assumption we have
  $\memSat{\Phi_S.\mem,\Phi_S.\stk,\Phi_S.\ms_\stk,\Phi_T.\mem}{W_M}$ which
  gives us $\ms_\var{T,stack}, \ms_\var{T,free\_stack},$ $\ms_\var{T,heap},
  \ms_{T,f}, \ms_{S,f}, \ms_S', W_{M,\var{stack}}, W_{M,\var{free\_stack}},
  W_{M,\var{heap}}$ such that:
  \begin{itemize}
  \item $\Phi_S.\stk = (\opc_0, \ms_0) \dots (\opc_m,\ms_m)$
  \item $\ms_0 \uplus \dots \uplus \ms_m \uplus \Phi_S.\mem \uplus \Phi_S.\ms_\stk$
  \item $W_M = W_{M,\var{stack}} \oplus W_{M,\var{free\_stack}} \oplus W_{M,\var{heap}}$
  \item $\Phi_S.\mem = \ms_{f,S} \uplus \ms_S'$
  \item $\Phi_T.\mem = \ms_\var{T,stack} \uplus \ms_\var{T,free\_stack} \uplus
    \ms_\var{T,heap} \uplus \ms_{T,f}$
  \item $\memSatStack{\Phi_S.\stk,\ms_\var{T,stack}}{W_{M,\var{stack}}}$
  \item $\memSatFStack{\Phi_S.\ms_\stk,\ms_\var{T,free\_stack}}{W_{M,\var{free\_stack}}}$
  \item $\npair{(\overline{\sigma},\ms_S',\ms_\var{T,heap})} \in \lrheap(\pwheap[W_M])(W_{M,\var{heap}})$
  \end{itemize}

  In order to prove the memory satisfaction, we pick the same memories except for the following changes:
  \begin{itemize}
  \item The target free stack partition: $\ms_\var{T,free\_stack} \setminus \ms_\var{T,free\_stack}|_{[\aaddr_\stk,\eaddr_\stk]}$
  \item The target private stack partition: $\ms_\var{T,stack} \uplus \ms_\var{T,free\_stack}|_{[\aaddr_\stk,\eaddr_\stk]}\update{\aaddr_\stk}{42}$
  \end{itemize}
  and the worlds
  \begin{itemize}
  \item Free stack world: $W_{M,\var{free\_stack}}'$ is $W_{R,1}'$ with the ownership of $W_{M,\var{free\_stack}}$ except that it gives up any ownership that we used for the safety of the addresses $[\aaddr_\stk,\eaddr_\stk]$.
  \item Private stack world: $W_{M,\var{stack}}'$ is $W_{R,1}'$ with the ownership of $W_{M,\var{stack}}$ except that it takes the ownership that $W_{M,\var{free\_stack}}$ used for the safety of the addresses $[\aaddr_\stk,\eaddr_\stk]$.
  \item Heap world: $W_{M,\var{heap}}'$ is $W_{R,1}'$ with the ownership of $W_{M,\var{heap}}$ which gives us $W_{M,\var{heap}}' \future W_{M,\var{heap}}$.
  \end{itemize}
  We now need to show:
  \begin{itemize}
  \item $\stk = ((\aaddr+\calllen),\ms_{\stk\_\priv,S}) :: (\opc_0, \ms_0) \dots (\opc_m,\ms_m)$

    Trivial.
  \item $\ms_{\stk\_\priv,S} \uplus \ms_0 \uplus \dots \uplus \ms_m \uplus \Phi_S.\mem \uplus \Phi_S.\ms_\stk - \Phi_S.\ms_\stk|_{[\aaddr_\stk,\eaddr_\stk]}$

    Follows by assumption and the fact that we have shuffled around some memory\footnote{Thanks to the register-file safety assumption we know for sure that the memory we remove from the free stack is actually there.}.
  \item $W_M' = W_{M,\var{stack}}' \oplus W_{M,\var{free\_stack}}' \oplus W_{M,\var{heap}}'$

    This follows by assumption and the fact that all ownership we have added to a world has been removed from another.
  \item $\Phi_S.\mem = \ms_{f,S} \uplus \ms_S'$

    By assumption.
  \item $\Phi_T.\mem = \ms_\var{T,stack} \uplus \ms_\var{T,free\_stack}|_{[\aaddr_\stk,\eaddr_\stk]}\update{\aaddr_\stk}{42} \uplus \ms_\var{T,free\_stack} \setminus \ms_\var{T,free\_stack}|_{[\aaddr_\stk,\eaddr_\stk]} \uplus \ms_\var{T,heap} \uplus \ms_{T,f}$

    Follows by assumption and the fact that we have shuffled around some memory\footnote{Thanks to the register-file safety assumption we know for sure that the memory we remove from the free stack is actually there.}.
  \item $\memSatStack[n-1]{((\aaddr+\calllen),\ms_{\stk\_\priv,S}) :: \Phi_S,\ms_\var{T,stack} \uplus \ms_\var{T,free\_stack}|_{[\aaddr_\stk,\eaddr_\stk]}\update{\aaddr_\stk}{42}}{W_{M,\var{stack}}'}$:

    For most of the conditions, they follow from the previous stack satisfaction assumption.
    The only challenge is to argue that the new stack frame satisfies the conditions.
    First of all, we know $\ms_{\stk\_\priv,S}$ is non-empty as it at least contains address $\aaddr_\stk$.
    Next, we need to argue that
    \[
      \begin{array}{l}
        \forall i \in \{0,\dots,m\} \ldotp \\
        \quad \forall a \in \dom(\ms_{\stk\_\priv,S}) \ldotp \forall a' \in \dom(\ms_j) \ldotp \stkb < a < a') \wedge\\
      \end{array}
    \]
    The first bit, $\stkb < a$ for $a \in \dom(\ms_{\stk\_\priv,S})$ follows from the fact that $a_\stk$ is the smallest address of $a \in \dom(\ms_{\stk\_\priv,S})$ and by assumption $\stkb < a_\stk$.

    Now assume for contradiction that there exists $a' \in \dom(\ms_j)$ for some $j$ such that $a' \le a$ for some $a \in \dom(\ms_{\stk\_\priv,S})$.
    By assumption we have $\stkb < a'$ so $a' \in [\stkb+1,a] \subseteq [\stkb+1,a_\stk] \subseteq [\stkb,e_\stk]$ which means that $a'$ is an address governed by a stack pointer.
    By the register-safety assumption this must mean that it is an address of the free part of the stack.
    At the same time, it must be an address of the private stack because $\ms_j$ is part of the stack from the original configuration.
    This contradicts the initial memory satisfaction assumption as the different parts must be disjointed.

    From the stack satisfaction assumption, we get $R_\ms$ and $R_W$. Pick
    \[
      R_\ms'(r) =
      \begin{cases}
        (\ms_{\stk\_\priv,S}, (\aaddr+\calllen),\ms_\var{T,free\_stack}|_{[\aaddr_\stk,\eaddr_\stk]}\update{\aaddr_\stk}{42})& \text{for $r = r_\var{priv\_stk}$}\\
        R_\ms(r) & \totherwise
      \end{cases}
    \]
    and for $R_W'$ pick
    \[
      R_W'(r) =
      \begin{cases}
        \begin{aligned}
          W_{M,\var{free\_stack},[\aaddr_\stk,\eaddr_\stk]}&[\free.R^{-1}([\aaddr_\stk,\eaddr_\stk]) \mapsto \revoked]\\
          &[\priv.r_\var{priv\_stk} \mapsto (\stareg[(\ms_{\stk\_\priv,S},\Phi_T''.\mem|_{[\aaddr_\stk,\eaddr_\stk]}),\gc]{\spa,\square},\aaddr+\calllen)]
        \end{aligned}

        & \text{for $r = r_\var{priv\_stk}$} \\
        \begin{aligned}
          R_W(r)&[\free.R^{-1}([\aaddr_\stk,\eaddr_\stk]) \mapsto \revoked]\\
          &[\priv.r_\var{priv\_stk} \mapsto (\stareg[(\ms_{\stk\_\priv,S},\Phi_T''.\mem|_{[\aaddr_\stk,\eaddr_\stk]}),\gc]{\spa,\square},\aaddr+\calllen)]
        \end{aligned} &
      \end{cases}
    \]
    where $W_{M,\var{free\_stack},[\aaddr_\stk,\eaddr_\stk]}$ is the world that the free stack assumption uses to satisfy that range of addresses.

    Most of the condition trivially holds. 
    The only one that requires some argumentation is
    \[
      \npair[n']{(\ms_{\stk\_\priv,S},\ms_\var{T,free\_stack}|_{[\aaddr_\stk,\eaddr_\stk]}\update{\aaddr_\stk}{42})}
      \in
      \stareg[(\ms_{\stk\_\priv,S},\Phi_T''.\mem|_{[\aaddr_\stk,\eaddr_\stk]}),\gc]{\spa,\square}.H
      \; \xi^{-1}(R_W'(r_\var{priv\_stk}))
    \]
    for all $n' < n-1$.
    Where $R_W'(r_\var{priv\_stk})$ is the part of $W_{M,\var{stack}}'$ with the ownership used for addresses $[\aaddr_\stk, \eaddr_\stk]$ in the memory satisfaction assumption.

    As we have the safe register-file assumption, we know that the stack capability is safe.
    This means that addresses $[\aaddr_\stk,\eaddr_\stk]$ must be part of the free stack.
    Further, $\Phi''_T.\mem = \Phi_T.\mem\update{\aaddr_\stk}{42}$ and $\ms_\var{T,free\_stack} \subseteq \Phi_T.\mem$ and $\dom(\ms_\var{T,free\_stack}) \supseteq [\aaddr_\stk,\eaddr_\stk]$ from which it follows that the memories are equal to the one of the static region.

    It remains to show
    \[
      \forall \aaddr \in [\aaddr_\stk,\eaddr_\stk] \ldotp \npair[n']{(\ms_{\stk\_\priv,S}(\aaddr),\ms_\var{T,free\_stack}|_{[\aaddr_\stk,\eaddr_\stk]}\update{\aaddr_\stk}{42}(\aaddr))} \in \lrv(W_{R,\aaddr}')
    \]
    for $n' < n-1$.
    For $\aaddr = \aaddr_\stk$ it is trivial as we have to show
    \[
      \npair[n-1]{(42,42)} \in \lrv(W_{R,\aaddr_\stk}')
    \]
    for $\aaddr \in [\aaddr_\stk+1,\eaddr_\stk]$ we need to show
    \[
      \npair[n']{(\ms_{\stk\_\priv,S}(\aaddr),\ms_\var{T,free\_stack}|_{[\aaddr_\stk,\eaddr_\stk]}\update{\aaddr_\stk}{42}(\aaddr))} \in \lrv(W_{R,\aaddr}')
    \]
    for $n' < n-1$.
    If we can show
    \[
      \npair[n']{(\Phi_S.\ms_\stk(\aaddr),\ms_\var{T,free\_stack}(\aaddr))} \in \lrv(W_{R,\aaddr})
    \]
    for $n' < n-1$, then we are done by monotonicity of $\lrv$.

    By assumption we know $\npair{(\Phi_S.\reg,\Phi_T.\reg)} \in \lrr(W_R)$ which entails $\npair{(\Phi_S.\reg(\rstk),\Phi_T.\reg(\rstk))} \in \lrr(W_{R,\rstk})$.
    We know $\Phi_S.\reg(\rstk) = \stkptr{\rw,\stkb,\eaddr_\stk,\aaddr_\stk}$, so by the definition of $\lrv$ we get
    \[
      \npair{[\stkb,\eaddr_\stk]} \in \stackReadCond{W_{R,\rstk}}
    \]
    which in turn gives us $S_{\rstk} \subseteq \addressable{\linear,\pwfree}$ and
    $R_{\rstk} : S_{\rstk} \fun \powerset{\nats}$ for which
    \begin{itemize}
    \item $\forall r \in S_{\rstk} \ldotp |R_{\rstk}(r)| = 1$
    \item $\uplus_{r \in S_{\rstk}} R_{\rstk}(r) \supseteq [\stkb,\eaddr_\stk]$
    \item $\forall r \in S_{\rstk} \ldotp W_{R,\rstk}(r).H \nsubeq \stdreg{R_{\rstk}(r),\gc}{\spao}$
    \end{itemize}
    Further, we know $\memSatFStack{\Phi_S.\ms_\stk,\ms_\var{T,free\_stack}}{W_{M,\var{free\_stack}}.\free}$ which means that we have $R_\ms : \dom(\activeReg{W_{M,\var{free\_stack}}}) \fun MemSeg \times \MemSeg$ and $R_W : \dom(W_{M,\var{free\_stack}}.\free) \fun \Worlds$.
    $R_W$ distributes the ownership of $W_{M,\var{free\_stack}}$ and $R_\ms$ partitions the memories.

    We know that all regions in $S_{\rstk}$ govern singleton memory segments, so $R_\ms$ must map to singleton memory segment pairs for $r \in S_{\rstk}$.
    Further, by definition of the free stack satisfaction for $r \in S_{\rstk}$ we have
    \[
      \npair[n']{R_\ms(r)} \in W_{R,\rstk}(r).H \; \xi^{-1}(R_W(r))
    \]
    for $n' < n$
    which entails
    \[
      \npair[n-1]{R_\ms(r)} \in \stdreg{R_{\rstk}(r),\gc}{\spao}.H \; \xi^{-1}(R_W(r))
    \]
    which entails
    \[
      \npair[n-1]{R_\ms(r)(a)} \in \lrv (R_W(r))
    \]
    for $a \in R_{\rstk(r)}$.

    Now, using Lemma~\ref{lem:downwards-closed}, this is exactly what we wanted to show because $R_W(r)$ is what we picked as $W_{R,\aaddr}$.
  \item $\memSatFStack[n-1]{\Phi_S.\ms_\stk - \Phi_S.\ms_\stk|_{[\aaddr_\stk,\eaddr_\stk]},\ms_\var{T,free\_stack} \setminus \ms_\var{T,free\_stack}|_{[\aaddr_\stk,\eaddr_\stk]}}{W_{M,\var{free\_stack}}'}$:

    From the safety assumption on the stack capability, we can deduce a number of things:
    \begin{itemize}
    \item $[\aaddr_\stk,\eaddr_\stk]$ must have been part of the free stack
    \item for every address in $[\aaddr_\stk,\eaddr_\stk]$ there is a region for that singleton memory segment.
    \end{itemize}
    The first part means that we do indeed remove all of the memory we try to subtract.
    The latter means that we can reuse the same split of the remaining memory and the world ownership as we get from assumption $\memSatFStack{\Phi_S.\ms_\stk,\ms_\var{T,free\_stack}}{W_{M,\var{free\_stack}}}$.
    Using this, the result follows from monotonicity of the $H$ function.

  \item $\npair[n-1]{(\overline{\sigma},\ms_S',\ms_\var{T,heap})} \in \lrheap(\pwheap[W_M'])(W_{M,\var{heap}}')$:

    Follows from Lemma~\ref{lem:heap-sat-unchanged-future}, Lemma~\ref{lem:downwards-closed}, the fact that the heap part of the world remains unchanged and $W_{M,\var{heap}}' \future W_{M,\var{heap}}$.
  \end{itemize}

  Case \ref{case:ftlr:move-cca-etc}:

  First show that
  \[
    \npair[n-1]{(((\rx,\normal),\baddr,\eaddr,\aaddr),((\rx,\normal),\baddr,\eaddr,\aaddr))} \in \lrvg{\trust}(\purePart{W_R})
  \]
  First observe that by Lemma~\ref{lem:purePart-oplus} $\purePart{W_R} = \purePart{W_\pcreg}$.
  Further, by assumption, Lemma~\ref{lem:downwards-closed}, Lemma~\ref{lem:purePart-oplus}, and Lemma~\ref{lem:non-linear-pure}, we have
  \[
    \npair[n-1]{[\baddr,\eaddr]} \in \xReadCond{}(\purePart{W_R})
  \]

  If $\trust = \trusted$, then by assumption we have $[\baddr,\eaddr] \subseteq \ta$ which means that all the conditions for are met for the capability pair to be in the trusted part of the value relation.

  If $\trust = \untrusted$, then we need to show that the capability pair is in the untrusted part of the value relation which means that we need to show:
  \begin{itemize}
  \item $\npair[n-1]{[\baddr,\eaddr]} \in \readCond{}(\normal,\purePart{W_R})$

    This follows by assumption, Lemma~\ref{lem:non-linear-pure}, Lemma~\ref{lem:purePart-oplus}, and Lemma~\ref{lem:downwards-closed}.
  \item $\npair[n-1]{[\baddr,\eaddr]} \in \xReadCond{}(\purePart{W_R})$

    We already showed this.
  \item $\npair[n-1]{[\baddr,\eaddr]} \in \execCond{}(\purePart{W_R})$

    To this end let $W' \future \purePart{W_R}$ and $n' < n-1$ and $\aaddr' \in [\baddr',\eaddr']$ be given, and show
    \[
      \npair[n']{(((\rx,\normal),\baddr',\eaddr',\aaddr'),((\rx,\normal),\baddr',\eaddr',\aaddr'))} \in \lre(W')
    \]
    This follows immediately from the FTLR: we know that $[b,e] \mathrel{\#} \ta$ since $\trust = \untrusted$ and we know that $n' < n - 1 < n$.
  \end{itemize}
  
  Now show
  \[
    \npair[n-1]{(\Phi_S'.\reg,\Phi_T'.\reg)} \in \lrrg{\trust}(W_R)
  \]
  Note that we know that $\Phi_S.\reg(\pcreg)$ is not linear (which will sometimes help to eliminate some cases).

  By assumption we have
  \[
    \npair[n]{(\Phi_S.\reg,\Phi_T.\reg)} \in \lrrg{\trust}(W_R)
  \]
  which gives us $R_R : (\RegName \setminus \{\pcreg\}) \fun \World$ such that $W_R = \bigoplus_{r\in \RegName \setminus \{\pcreg\}} R_R(r)$ and for all $r$ in $\RegName \setminus \{\pcreg\}$ we have $\npair[n]{(\Phi_S.\reg(r),\Phi_T.\reg(r))} \in \lrrg{\trust}(R_R(r))$.
  To this end, we need to consider each of cases
  \ref{case:ftlr:move-cca-etc:cca-normal}-\ref{case:ftlr:move-cca-etc:splice-seal}:

  \lau{TODO: check that all the below cases take care of the trusted case.}
  \begin{itemize}
  \item Case \ref{case:ftlr:move-cca-etc:cca-normal}:

    Pick $R_R$ as the ownership distribution.
    For $r \neq r_1$ it follows by assumption and Lemma~\ref{lem:downwards-closed}.
    For $r = r_1$ it also follows by assumption and Lemma~\ref{lem:downwards-closed}, \ref{lem:cap-in-lrv-regardless-of-addr}, and \ref{lem:cap-in-lrv-mono-perm}.
  \item Case \ref{case:ftlr:move-cca-etc:cca-stack}:

    Pick $R_R$ as the ownership distribution.
    For $r \neq r_1$ it follows by assumption and Lemma~\ref{lem:downwards-closed}.
    For $r = r_1$ it also follows by assumption and Lemma~\ref{lem:downwards-closed}, \ref{lem:stkptr-in-lrv-regardless-of-addr}, and \ref{lem:stkptr-in-lrv-mono-perm}.
  \item Case \ref{case:ftlr:move-cca-etc:cca-seal}:

    Pick $R_R$ as the ownership distribution.
    For $r \neq r_1$ it follows by assumption and Lemma~\ref{lem:downwards-closed}.
    For $r = r_1$ it also follows by assumption and Lemma~\ref{lem:downwards-closed} and \ref{lem:seal-in-lrv-regardless-of-cur-seal}.

  \item Case \ref{case:ftlr:move-cca-etc:move}:
    Pick the ownership distribution based on the linearity of $w_2$: If $\isLinear{w_2}$, then pick
    \[
      R_R'(r) =
      \begin{cases}
        R_R(r_2) \oplus R_R(r_1) & r = r_1 \\
        \purePart{W_R}         & r = r_2 \\
        R_R(r)                 & \totherwise
      \end{cases}
    \]
    if $\neg \isLinear{w_2}$, then pick
    \[
      R_R'(r) = R_R(r)
    \]

    In the case where $\isLinear{w_2}$, we may assume $r_2 \neq \pcreg$. We need to show (assuming $r_1 \neq \pcreg$)
    \[
      \npair[n-1]{(\Phi_S'(r_1),\Phi_T'(r_1))} \in \lrvg{\trust}(R_R(r_2) \oplus R_R(r_1))
    \]
    which is 
    \[
      \npair[n-1]{(\Phi_S(r_2),\Phi_T(r_2))} \in \lrvg{\trust}(R_R(r_2) \oplus R_R(r_1))
    \]
    this follows by assumption and Lemma~\ref{lem:monotonicity} and~\ref{lem:oplus-future}.

    We also need to show
    \[
      \npair[n-1]{(\Phi_S'(r_2),\Phi_T'(r_2))} \in \lrvg{\trust}(\purePart{W_R})
    \]
    which is trivial as $\Phi_S'(r_2) = \Phi_T'(r_2) = 0$. 

    Finally for $r \neq r_1,r_2,\pcreg$ 
    \[
      \npair[n-1]{(\Phi_S'(r),\Phi_T'(r))} \in \lrvg{\trust}(R_R(r))
    \]
    Follows by assumption and Lemma~\ref{lem:downwards-closed}.

    In the case where $\neg \isLinear{w_2}$

    If $r_2 \neq \pcreg$, then for $r \neq r_1$
    \[
      \npair[n-1]{(\Phi_S'(r),\Phi_T'(r))} \in \lrvg{\trust}(R_R(r))
    \]
    Follows by assumption and Lemma~\ref{lem:downwards-closed}.
    
    For $r \neq r_1,r_2$
    \[
      \npair[n-1]{(\Phi_S'(r_1),\Phi_T'(r_1))} \in \lrvg{\trust}(R_R(r_1))
    \]
    amounts to
    \[
      \npair[n-1]{(\Phi_S(r_2),\Phi_T(r_2))} \in \lrvg{\trust}(R_R(r_1))
    \]
    By assumption and Lemma~\ref{lem:purePart-oplus} and \ref{lem:non-linear-pure}.

    If $r_2 = \pcreg$, show
    \[
      \npair[n-1]{(\Phi_S(\pcreg),\Phi_T(\pcreg))} \in \lrvg{\trust}(R_R(r_1))
    \]
    which follows from Lemma~\ref{lem:monotonicity} and \ref{lem:world-fut-purePart} and what we have proven about the pc.
  \item Case \ref{case:ftlr:move-cca-etc:cseal}:

    Pick $R_R' = R_R$.
    For $r \neq r_1,r_2$, we have
    \[
      \npair[n-1]{(\Phi_S(r),\Phi_T(r))} \in \lrv(R_R(r))
    \]
    by assumption and Lemma~\ref{lem:downwards-closed}.

    For $r_1$ use Lemma~\ref{lem:lrv-relates-linearity} and consider the following 2 cases:
    \begin{itemize}
    \item $\ta \# [\baddr,\eaddr]$: by assumption this entails $\trust = \untrusted$.

      By assumption we have $\npair{(\Phi_S(r_2),\Phi_T(r_2))} \in \lrvg{\untrusted}(R_R(r_2))$ which gives us the following facts:
      \begin{itemize}
      \item $[\sigma_\baddr,\sigma_\eaddr] \# (\gsigrets \uplus \gsigcloss)$
      \item $\forall \sigma' \in [\sigma_\baddr,\sigma_\eaddr] \ldotp \exists r \in \pwheap[R_R(r_2)] \ldotp \pwheap[R_R(r_2)](r) = (\pure,\_,H_\sigma) \wedge H_\sigma \; \sigma' \nequal (\lrvg{\untrusted} \circ \xi)$
      \end{itemize}
      
      This means that for $\sigma$ there is a region $r$ for which $\pwheap[R_R(r_2)](r) = (\pure,\_,H_\sigma)$ and $H_\sigma \; \sigma' \nequal (\lrvg{\untrusted} \circ \xi)$.
      Pick this as the region in the sealed case and $\sigrets' = \emptyset$,
      $\sigcloss = [\sigma_\baddr,\sigma_\eaddr]$, and $\mscode = [\baddr,\eaddr] \mapsto 0$.
      We now need to show the following:
      \begin{itemize}
      \item $H_\sigma \; \sigma \nequal H^\mathrm{code,\square}_\sigma \; \sigrets \; \sigcloss \; \mscode \; \gc \; \sigma$

        To this end let $\hat{W}$ be given and show
        \[
          H_\sigma \; \sigma \; \hat{W} \nequal H^\mathrm{code,\square}_\sigma \; \sigrets \; \sigcloss \; \mscode \; \gc \; \sigma \; \hat{W}
        \]
        By transitivity of $n$-equality it suffices to show $H_\sigma \; \sigma' \; \hat{W} \nequal \lrvg{\untrusted} \circ \xi(\hat{W})$, which follows by assumption and $H^\mathrm{code,\square}_\sigma \; \sigrets \; \sigcloss \; \mscode \; \gc \; \sigma \; \hat{W} \nequal \lrvg{\untrusted} \circ \xi(\hat{W})$ which follows by definition of $H^\mathrm{code,\square}_\sigma$ and the fact that $\sigma \in \sigcloss$ and $\ta \# [\baddr,\eaddr]$.

      \item $\npair[n']{(\Phi_S(r_1),\Phi_T(r_1))} \in H_\sigma \; \sigma \; \xi^{-1}(R_R(r_1)) \text{ for all $n' < n$}$
        
        which corresponds to showing $\npair[n']{(\Phi_S(r_1),\Phi_T(r_1))} \in \lrvg{\untrusted}(R_R(r_1))$ which is true by assumption and Lemma~\ref{lem:downwards-closed}.
      \item $(\isLinear{\vsc_S} \Rightarrow \forall W' \future R_R(r_1), W_o, n' < n-1, \npair[n']{(\vsc_S',\vsc'_T)} \in H_\sigma \; \sigma \; \xi^{-1}(W_o) \ldotp \npair[n']{\vsc_S,\vsc_S',\vsc_T,\vsc_T'} \in \lrexj(W'\oplus W_o))$

        Let $W' \future R_R(r_1)$, $W_o$, $n' < n$, $\npair[n']{(\vsc_S',\vsc'_T)} \in \codereg{\sigrets,\sigcloss,\mscode,\gc}.H_\sigma \; \sigma \; \xi^{-1}(W_o)$.
        Further let $n'' \le n'$ be given and assume $\npair[n'']{(\reg_S,\reg_T)} \in \lrrg{\untrusted}(\{\rdata\})(W_R)$, $\memSat{\ms_S,\stk,\ms_\stk,\ms_T}{W_M}$.

        Now consider the following cases:
        \begin{itemize}
        \item $\exec{\vsc_S'}$:

          In this case, pick $\Phi_S' = \Phi_T' = \failed$ (as xjump fails).
          It is trivial to show $\npair{\Phi_S',\Phi_T'} \in \lro$.
        \item $\nonExec{\vsc_S'}$:

          In this case consider what $\Phi_S(r_1)$ is:\\
          If $\Phi_S(r_1) = ((\perm,\lin),\baddr,\eaddr,\aaddr)$ and $\perm \in \{\rwx, \rx\}$:

          If $\perm = \rwx$, then we have a contradiction with $\npair{(\Phi_S(r_1),\Phi_T(r_1))} \in \lrvg{\untrusted}(R_R(r_1))$.

          If $\perm = \rx$, then by Lemma~\ref{lem:safe-cap-exec-is-normal} we have a contradiction with $\isLinear{\Phi_S(r_1)}$.
          \\
          \\
          Otherwise (not $\Phi_S(r_1) = ((\perm,\lin),\baddr,\eaddr,\aaddr)$ and $\perm \in \{\rwx, \rx\}$):\\

          In this case, pick $\Phi_S' = (\ms_S,\reg_S\update{\pcreg}{\Phi_S(r_1)}\update{\rdata}{\vsc_S'},\stk,\ms_\stk)$ and $\Phi_T' = (\ms_T,\reg_T\update{\pcreg}{\Phi_T(r_1)}\update{\rdata}{\vsc_T'})$.
          In this case, the next step of execution fails which makes it trivial
          to show $\npair{\Phi_S',\Phi_T'} \in \lro$.
        \end{itemize}
      \item If $\nonLinear{\src{\vsc_S}}$ then for all $W' \future \purePart{R_R(r_2)}$, $W_o$, $n' < n-1$, $\npair[n']{(\vsc_S',\vsc'_T)} \in H_\sigma \; \sigma \; \xi^{-1}(W_o)$ we have that $\npair[n']{\vsc_S,\vsc_S',\vsc_T,\vsc_T'} \in \lrexj(W'\oplus W_o)$

        Let $W' \future \purePart{R_R(r_1)}$, $W_o$, $n' < n$, $\npair[n']{(\vsc_S',\vsc'_T)} \in \codereg{\sigrets,\sigcloss,\mscode,\gc}.H_\sigma \; \sigma \; \xi^{-1}(W_o)$.
        Further let $n'' \le n'$ be given and assume that
        \begin{itemize}
        \item $\npair[n'']{(\reg_S,\reg_T)} \in
          \lrrg{\untrusted}(\{\rdata\})(W_R)$
        \item $\memSat{\ms_S,\stk,\ms_\stk,\ms_T}{W_M}$
        \end{itemize}

        Now consider the following cases:
        \begin{itemize}
        \item $\exec{\vsc_S'}$:

          In this case, pick $\Phi_S' = \Phi_T' = \failed$ (as xjump fails).
          It is trivial to show $\npair{\Phi_S',\Phi_T'} \in \lro$.
        \item $\nonExec{\vsc_S'}$:

          In this case consider what $\Phi_S(r_1)$ is.\\
          If $\Phi_S(r_1) = ((\perm',\lin'),\baddr',\eaddr',\aaddr')$ and $\perm' \in \{\rwx, \rx\}$:

          If $\perm' = \rwx$, then we have a contradiction with $\npair{(\Phi_S(r_1),\Phi_T(r_1))} \in \lrvg{\untrusted}(R_R(r_1))$.

          If $\perm' = \rx$, then the result follows from \ref{lem:untrusted-sealed-capability}.
        \end{itemize}
      \end{itemize}
    \item $\ta \subseteq [\baddr,\eaddr]$: by assumption this entails $\trust = \trusted$.

      Further, we know $\Phi_S$ points to $\src{\tcseal{r_1}{r_2}}$ in $\ta$, so by the reasonability assumption on $\Phi_S$, we know $\sigma \in \gsigcloss$ and one of the following holds:
      \begin{itemize}
      \item $\exec{\Phi(r_1)}$ and $\Phi(r_1)$ behaves reasonably up to $n - 1$ steps.
      \item $\nonExec{\Phi(r_1)}$ and $\Phi(r_1)$ is reasonable up to $n - 1$ steps in memory $\Phi.\ms$ and free stack $\Phi.\ms_\stk$
      \end{itemize}
      By $\sigma \in \gsigcloss$ we can conclude that $\npair{(\Phi_S(r_2),\Phi_T(r_2))} \not\in \lrvg{\untrusted}(W)$ which means that the assumption $\npair{(\Phi_S(r_2),\Phi_T(r_2))} \in \lrvg{\trusted}(W)$ gives us $r \in \dom(\pwheap)$ such that $\pwheap(r) \nequal \codereg{\sigrets,\sigcloss,\mscode,\gc}$, $[\sigma_\baddr,\sigma_\eaddr] \subseteq (\sigrets \cup \sigcloss)$ and $\sigrets \subseteq \gsigrets$ and $\sigcloss \subseteq \gsigcloss$, $\dom(\mscode) \subseteq \ta$.
      
      Now pick $r$ and show
      \begin{itemize}
      \item $\npair[n']{(\Phi_S(r_1),\Phi_T(r_1))} \in \codereg{\sigrets,\sigcloss,\mscode,\gc}.H_\sigma \; \sigma \; \xi^{-1}(R_R(r_1)) \text{ for all $n' < n-1$}$

        This amounts to showing (1) assuming $\exec{\Phi_S(r_1)}$ show $\npair[n']{(\Phi_S(r_1),\Phi_T(r_1))} \in \lrvg{\trusted}(R_R(r_1))$, which is true by assumption and Lemma~\ref{lem:downwards-closed}.

        (2) Assuming $\nonExec{\Phi_S(r_1)}$ show $\npair[n']{(\Phi_S(r_1),\Phi_T(r_1))} \in \lrvg{\untrusted}(R_R(r_1))$, which follows by assumption, Lemma~\ref{lem:trusted-and-reasonable-is-untrusted} and Lemma~\ref{lem:downwards-closed}.
      \item
        \[
          \begin{array}{l}
            (\isLinear{\Phi_S(r_1)} \Rightarrow \\
            \quad\forall W' \future R_R(r_1), W_o, n' < n, \npair[n']{(\vsc_S',\vsc'_T)} \in \codereg{\sigrets,\sigcloss,\mscode,\gc}.H_\sigma \; \sigma \; \xi^{-1}(W_o) \ldotp \\
            \qquad \npair[n']{\Phi_S(r_1),\vsc_S',\Phi_T(r_1),\vsc_T'} \in \lrexj(W'\oplus W_o))
          \end{array}
        \]
        
        Let $W' \future R_R(r_1)$, $W_o$, $n' < n$, $\npair[n']{(\vsc_S',\vsc'_T)} \in \codereg{\sigrets,\sigcloss,\mscode,\gc}.H_\sigma \; \sigma \; \xi^{-1}(W_o)$.
        Further let $n'' \le n'$ be given and assume $\npair[n'']{(\reg_S,\reg_T)} \in \lrrg{\untrusted}(\{\rdata\})(W_R)$, $\memSat{\ms_S,\stk,\ms_\stk,\ms_T}{W_M}$.

        Now consider the following cases:
        \begin{itemize}
        \item $\exec{\vsc_S'}$:

          In this case, pick $\Phi_S' = \Phi_T' = \failed$ (as xjump fails).
          It is trivial to show $\npair{\Phi_S',\Phi_T'} \in \lro$.
        \item $\nonExec{\vsc_S'}$:

          In this case consider what $\Phi_S(r_1)$ is:\\
          If $\Phi_S(r_1) = ((\perm,\lin),\baddr,\eaddr,\aaddr)$ and $\perm \in \{\rwx, \rx\}$:

          If $\perm = \rwx$, then we have a contradiction with $\npair{(\Phi_S(r_1),\Phi_T(r_1))} \in \lrvg{\trusted}(R_R(r_1))$.

          If $\perm = \rx$, then by Lemma~\ref{lem:safe-cap-exec-is-normal} we have a contradiction with $\isLinear{\Phi_S(r_1)}$.
          \\
          \\
          Otherwise (not $\Phi_S(r_1) = ((\perm,\lin),\baddr,\eaddr,\aaddr)$ and $\perm \in \{\rwx, \rx\}$):\\

          In this case, pick $\Phi_S' = (\ms_S,\reg_S\update{\pcreg}{\Phi_S(r_1)}\update{\rdata}{\vsc_S'},\stk,\ms_\stk)$ and $\Phi_T' = (\ms_T,\reg_T\update{\pcreg}{\Phi_T(r_1)}\update{\rdata}{\vsc_T'})$.
          In this case, the next step of execution fails which makes it trivial to show $\npair{\Phi_S',\Phi_T'} \in \lro$.
        \end{itemize}
      \item
        \[
          \begin{array}{l}
            (\nonLinear{\Phi_S(r_1)} \Rightarrow \\
            \quad \forall W' \future \purePart{R_R(r_1)}, W_o, n' < n, \npair[n']{(\vsc_S',\vsc'_T)} \in \codereg{\sigrets,\sigcloss,\mscode,\gc}.H_\sigma \; \sigma \; \xi^{-1}(W_o) \ldotp \\
            \qquad \npair[n']{\Phi_S(r_1),\vsc_S',\Phi_T(r_1),\vsc_T'} \in \lrexj(W'\oplus W_o))
          \end{array}
        \]

        Let $W' \future \purePart{R_R(r_1)}$, $W_o$, $n' < n$, $\npair[n']{(\vsc_S',\vsc'_T)} \in \codereg{\sigrets,\sigcloss,\mscode,\gc}.H_\sigma \; \sigma \; \xi^{-1}(W_o)$.
        Further let $n'' \le n'$ be given and assume $\npair[n'']{(\reg_S,\reg_T)} \in \lrrg{\untrusted}(\{\rdata\})(W_R)$, $\memSat{\ms_S,\stk,\ms_\stk,\ms_T}{W_M}$.

        Now consider the following cases:
        \begin{itemize}
        \item $\exec{\vsc_S'}$:

          In this case, pick $\Phi_S' = \Phi_T' = \failed$ (as xjump fails).
          It is trivial to show $\npair{\Phi_S',\Phi_T'} \in \lro$.
        \item $\nonExec{\vsc_S'}$:

          In this case consider what $\Phi_S(r_1)$ is.\\
          If $\Phi_S(r_1) = ((\perm',\lin'),\baddr',\eaddr',\aaddr')$ and $\perm' \in \{\rwx, \rx\}$:

          If $\perm' = \rwx$, then we have a contradiction with $\npair{(\Phi_S(r_1),\Phi_T(r_1))} \in \lrvg{\trusted}(R_R(r_1))$.

          If $\perm' = \rx$, then in the case where $\npair{(\Phi_S(r_1),\Phi_T(r_1))} \in \lrvg{\untrusted}(R_R(r_1))$, the result follows from \ref{lem:untrusted-sealed-capability}.

          In the case where $\npair{(\Phi_S(r_1),\Phi_T(r_1))} \not\in \lrvg{\untrusted}(R_R(r_1))$, we want to show $(\ms_S,\reg_S\update{\pcreg}{\Phi_S(r_1)}\update{\rdata}{\vsc_S'},\stk,\ms_\stk)$ behaves reasonable up to $n-1$ steps.

          By assumption we have $\Phi_S(r_1)$ behaves reasonably up to $n-1$ steps, so it SFTS $\reg_S(r)$ is reasonable up to $n-1$ steps in memory $\ms_S$ and free stack $\ms_\stk$ for $r \neq \pcreg$ which follows from Lemma~\ref{lem:untrusted-source-values-are-reasonable}.
          
          By assumption we have $\npair{[\baddr',\eaddr']} \in \xReadCond{R_R(r_1)}$, so using Lemma~\ref{lem:non-expansiveness} and Lemma~\ref{lem:monotonicity} as well as $[\baddr',\eaddr'] \subseteq \ta$.
          Using this with the other assumptions and the IH, we get:
          \[
            \npair[n-1]{
              \begin{multlined}
                ((\ms_S,\reg_S\update{\pcreg}{\Phi_S(r_1)}\update{\rdata}{\vsc_S'},\stk,\ms_\stk),\\
                (\ms_T,\reg_T\update{\pcreg}{\Phi_T(r_1)}\update{\rdata}{\vsc_T'}))
              \end{multlined}
            } \in \lro
          \]

          Otherwise (not $\Phi_S(r_1) = ((\perm',\lin'),\baddr',\eaddr',\aaddr')$ and $\perm' \in \{\rwx, \rx\}$):

          In this case, pick $\Phi_S' = (\ms_S,\reg_S\update{\pcreg}{\Phi_S(r_1)}\update{\rdata}{\vsc_S'},\stk,\ms_\stk)$ and $\Phi_T' = (\ms_T,\reg_T\update{\pcreg}{\Phi_T(r_1)}\update{\rdata}{\vsc_T'})$.
          In this case, the next step of execution fails which makes it trivial to show $\npair{\Phi_S',\Phi_T'} \in \lro$.
        \end{itemize}
      \end{itemize}
    \end{itemize}
  \item Case \ref{case:ftlr:move-cca-etc:split-normal}:
    From $\npair{(\Phi_S(r_3),\Phi_S(r_3))} \in \lrvg{\trust}(R_R(r_3))$, Lemma~\ref{lem:splitting-safety-normal} gives us $W_1$, $W_2$ and $W_3$ such that $R_R(r_3) = W_1 \oplus W_2 \oplus W_3$ and $\npair{(w_1,w_1) \in \lrvg{\trust}(W_1)}$, $\npair{(w_2,w_2) \in \lrvg{\trust}(W_2)}$ and $\npair{(w_3,w_3) \in \lrvg{\trust}(W_3)}$.
    We take $R_R'(r_1) = W_1$, $R_R'(r_2) = W_2$ and $R_R'(r_3) = W_3 \oplus R_R(r_1) \oplus R_R(r_2)$ and $R_R' = R_R(r)$ elsewhere.

    By the above points, by assumption and using Lemma~\ref{lem:downwards-closed}, we then have for all $r$ that:
    \[
      \npair[n-1]{(\Phi_S'(r), \Phi_T'(r))} \in \lrvg{\trust}(R_R'(r))
    \]
  \item Case \ref{case:ftlr:move-cca-etc:split-seal}:
    In this case, we can take $R_R' = R_R$ and use Lemma~\ref{lem:purePart-oplus} to give us that all $\purePart{R_R(r)}$ are equal.
    For $r = r_1, r_2, r_3$, we then get easily by definition that
    \[
      \npair[n-1]{(\Phi_S'(r), \Phi_T'(r))} \in \lrvg{\trust}(R_R'(r))
    \]
    and for other registers, it follows by Lemma~\ref{lem:downwards-closed}.
  \item Case \ref{case:ftlr:move-cca-etc:split-stack}:
    First we only consider registers $r_1,r_2,r_3$.

    From $\npair{(\Phi_S(r_3),\Phi_T(r_3))} \in \lrvg{\trust}(R_R(r_3))$, Lemma~\ref{lem:splitting-safety-stack} gives us $W_1$, $W_2$ and $W_3$ such that $R_R(r_3) = W_1 \oplus W_2 \oplus W_3$ and $\npair{(w_1,w_1) \in \lrvg{\trust}(W_1)}$, $\npair{(w_2,w_2) \in \lrvg{\trust}(W_2)}$ and $\npair{(w_3,w_3) \in \lrvg{\trust}(W_3)}$.
    We take $R_R'(r_1) = W_1$, $R_R'(r_2) = W_2$ and $R_R'(r_3) = W_3 \oplus R_R(r_1) \oplus R_R(r_2)$ and $R_R' = R_R(r)$ elsewhere.

    By the above points, by assumption and using Lemma~\ref{lem:downwards-closed}, we then have for all $r$ that:
    \[
      \npair[n-1]{(\Phi_S'(r), \Phi_T'(r))} \in \lrvg{\trust}(R_R'(r))
    \]

  \item Case \ref{case:ftlr:move-cca-etc:splice-stack}:

    From $\npair{(\Phi_S(r_2),\Phi_T(r_2))} \in \lrvg{\trust}(R_R(r_2))$ and $\npair{(\Phi_S(r_3),\Phi_T(r_3))} \in \lrvg{\trust}(R_R(r_3))$, Lemma~\ref{lem:splicing-safety-stack} tells us that $\npair{(w_1,w_1) \in \lrvg{\trust}(R_R(r_2)\oplus R_R(r_3))}$.
    Since $w_2 = w_2'=w_3=w_3'=0$, it's clear that 
    \begin{itemize}
    \item $\npair{(w_2,w_2') \in \lrvg{\trust}(R_R(r_1))}$ and
    \item $\npair{(w_3,w_3') \in \lrvg{\trust}(\purePart{R_R(r_3)})}$
    \end{itemize}
    We take $R_R'(r_1) = R_R(r_2)\oplus R_R(r_3)$, $R_R'(r_2) = R_R(r_1)$ and $R_R'(r_3) = \purePart{R_R(r_3)}$ and $R_R' = R_R(r)$ elsewhere.

    By the above points, by assumption and using Lemma~\ref{lem:downwards-closed}, we then have for all $r$ that:
    \[
      \npair[n-1]{(\Phi_S'(r), \Phi_T'(r))} \in \lrvg{\trust}(R_R'(r))
    \]

  \item Case \ref{case:ftlr:move-cca-etc:splice-normal}:
    
    We start by arguing the safety of $r_1$:

    From $\npair{(\Phi_S(r_2),\Phi_T(r_2))} \in \lrvg{\trust}(R_R(r_2))$ and $\npair{(\Phi_S(r_3),\Phi_T(r_3))} \in \lrvg{\trust}(R_R(r_3))$, Lemma~\ref{lem:splicing-safety-normal} gives us $W_1',W_2',W_3'$ such that $R_R(r_2)\oplus R_R(r_3) = W_1'\oplus W_2' \oplus W_3'$ and
    $\npair{(w_i,w_i) \in \lrvg{\trust}(W_i')}$ for $i = 1,2,3$.
    We take $R_R'(r_1) = W_1'\oplus R_R(r_1)$ (which is defined because $R_R(r_1) \oplus (R_R(r_2)\oplus R_R(r_3))$ is defined), $R_R'(r_2) = W_2'$ and $R_R'(r_3) = W_3'$ and $R_R' = R_R(r)$ elsewhere.

    By the above points, by assumption and using Lemma~\ref{lem:downwards-closed}, we then have for all $r$ that:
    \[
      \npair[n-1]{(\Phi_S'(r), \Phi_T'(r))} \in \lrvg{\trust}(R_R'(r))
    \]

  \item Case \ref{case:ftlr:move-cca-etc:splice-seal}:

    In this case, we can take $R_R' = R_R$ and use Lemma~\ref{lem:purePart-oplus} to give us that all $\purePart{R_R(r)}$ are equal.
    For $r = r_1, r_2, r_3$, we then get easily by definition that
    \[
      \npair[n-1]{(\Phi_S'(r), \Phi_T'(r))} \in \lrvg{\trust}(R_R'(r))
    \]
    and for other registers, it follows by Lemma~\ref{lem:downwards-closed}.

  \item Case \ref{case:ftlr:move-cca-etc:jnz-z}:

    In this case, we can take $R_R' = R_R$ and for all registers, the result follows by Lemma~\ref{lem:downwards-closed}.
  \end{itemize}

  By Lemma~\ref{lem:lro-anti-red-gen} it suffices to show
  \[
    \npair[n-1]{\Phi_S',\Phi_T'} \in \lro
  \]
  which follows from the induction hypothesis.
  In this case, the IH is applicable because we have the following:
  \begin{itemize}
  \item One of the following sets of requirements holds:
    \begin{itemize}
    \item $\trust = \trusted$, $\Phi_S'$ is reasonable up to $n-1$ steps and $[\baddr,\eaddr] \subseteq \dom(\mscode) = \ta$
    \item $\trust = \untrusted$, $[b,e] \mathrel{\#} \ta$ and $\npair[n-1]{[\baddr,\eaddr]} \in \readCond{\normal,W_\pcreg}$
    \end{itemize}

    We know one of the following holds:
    \begin{itemize}
    \item $\trust = \trusted$, $\Phi_S$ is reasonable up to $n$ steps and $[\baddr,\eaddr] \subseteq \dom(\mscode) = \ta$
    \item $\trust = \untrusted$, $[b,e] \mathrel{\#} \ta$ and $\npair{[\baddr,\eaddr]} \in \readCond{\normal,W_\pcreg}$
    \end{itemize}
    If the latter is the case, then the result follows by Lemma~\ref{lem:downwards-closed}.

    If the former holds, then it follows by definition of execution configuration reasonability, using the fact that $\Phi_S$ does not point to $\scall{\offpc,\offsigma}{r_1}{r_2}$ or $\txjmp{r_1}{r_2}$.

  \item $\Phi_S'(\pcreg) = \Phi_T'(\pcreg) = ((\rx,\normal),\baddr,\eaddr,\_)$:
    Follows by definition of $\updPcAddr{}$ and using the fact that $r_i \neq \pcreg$ for all $i$ and the corresponding assumption of thi slemma.

  \item $\npair[n-1]{[\baddr,\eaddr]} \in \xReadCond{W_\pcreg}$:
    Follows by Lemma~\ref{lem:downwards-closed} from the corresponding assumption of this lemma.

  \item $\npair[n-1]{(\Phi_S'.\reg,\Phi_T'.\reg)} \in \lrrg{\trust}(W_R)$:
    See above.

  \item $\memSat[n-1]{\Phi_S'.\mem,\Phi_S'.\stk,\Phi_S'.\ms_\stk,\Phi_T'.\mem}{W_M}$:
    These components are all unchanged from $\Phi_S$ and $\Phi_T$, so the result follows by Lemma~\ref{lem:downwards-closed} from the corresponding assumption of this lemma.

  \item $W_\pcreg \oplus W_R \oplus W_M$ is defined:
    Follows from the corresponding assumption of this lemma.

  \item Theorem~\ref{thm:ftlr} holds for all $n' < n-1$:
    follows from the corresponding assumption of this lemma.
  \end{itemize}

  Case \ref{case:ftlr:store-load-mem}:

  By Lemma~\ref{lem:lro-anti-red-gen} it suffices to show
  \[
    \npair[n-1]{\Phi_S',\Phi_T'} \in \lro
  \]

  First, we show that for some $W_R'$ and $W_M'$ such that $W_R \oplus W_M = W_R' \oplus W_M'$, we have that $\npair[n-1]{(\Phi_S'.\reg,\Phi_T'.\reg)} \in \lrrg{\trust}(W_R')$ and $\memSat[n-1]{\Phi_S'.\mem,\Phi_S'.\stk,\Phi_S'.\ms_\stk,\Phi_T'.\mem}{W_M'}$.
  We have that $\Phi_S' = \updPcAddr{\Phi_S\updReg{r_1,r_2}{w_1,w_2}\update{\mem.\aaddr}{w}}$ and $\Phi_T' = \updPcAddr{\Phi_T\updReg{r_1',r_2'}{w_1',w_2'}\update{\mem.\aaddr}{w'}}$ and we distinguish the following two cases:
  \begin{itemize}
  \item (store) $w_1 = w_1' = \Phi_S(r_1) = \Phi_T(r_1) =
    ((\perm,\lin),\baddr,\eaddr,\aaddr)$, and $\perm \in \writeAllowed{}$,
    and
    $\withinBounds{w_1}$, and
    $w = \Phi_S(r_2)$, and $w' = \Phi_T(r_2)$, and $w_2 = \linCons{\Phi_S(r_2)}$, $w_2' = \linCons{\Phi_T(r_2)}$,  $r_2 \neq \pcreg$:

    From $\npair{(\Phi_S.\reg,\Phi_T.\reg)} \in \lrrg{\trust}(W_R)$, we know that
    \begin{itemize}
    \item $\npair{(w,w')} \in \lrvg{\trust}(W_{R,2})$
    \item $\npair{(w_2,w_2')} \in \lrvg{\trust}(\purePart{W_{R,2}]})$ and
    \item $\npair{(\Phi_S.\reg\update{r_2}{w_2},\Phi_T.\reg\update{r_2}{w_2'})} \in \lrrg{\trust}(W_R')$
    \end{itemize}
    with $W_R =W_R' \oplus W_{R,2}$ (using Lemma~\ref{lem:purePart-oplus} and~\ref{lem:purePart-duplicable}).

    From reasonability of $\Phi_S$, we know that $\Phi.\reg(r_2)$ is reasonable in memory $\Phi.\mem$ up to $n-1$ steps.
    Lemma~\ref{lem:trusted-and-reasonable-is-untrusted} then tells us that $\npair{(w,w')} \in \lrv(W_{R,2})$.

    The result then follows from Lemma~\ref{lem:store-reg-works}, redistributing ownership in the obvious way.

  \item (load) $w_2 = w_2' = \Phi_T(r_2) = \Phi_S(r_2) =
    ((\perm,\lin),\baddr,\eaddr,\aaddr)$, and $\perm \in \readAllowed{}$,
    $\withinBounds{((\perm,\lin),\baddr,\eaddr,\aaddr)}$, and
    $w_1 = \Phi_S.\mem(\aaddr)$, and $w_1' = \Phi_T.\mem(\aaddr)$, and
    $w = \linCons{w_1}$, $w' = \linCons{w_1'}$, $\linConsPerm{\perm}{w_1}$, $\linConsPerm{\perm}{w_1'}$ and $r_1 \neq \pcreg$

    Follows from Lemma~\ref{lem:readcond-cap-works}, redistributing ownership in the obvious way.

  \end{itemize}

  By the induction hypothesis, it suffices to prove that:
  \begin{itemize}
  \item One of the following sets of requirements holds:
    \begin{itemize}
    \item $\trust = \trusted$, $\Phi_S'$ is reasonable up to $n$ steps and $[\baddr,\eaddr] \subseteq \dom(\mscode) = \ta$
    \item $\trust = \untrusted$ and $[b,e] \mathrel{\#} \ta$ and $\npair{[\baddr,\eaddr]} \in \readCond{\normal,W_\pcreg}$
    \end{itemize}:

    Follows by the same assumption of this lemma, in the first case using the definition of execution configuration reasonability, using the fact that $\Phi_S$ does not point to $\scall{\offpc,\offsigma}{r_1}{r_2}$ or $\txjmp{r_1}{r_2}$.

  \item $\Phi_S'(\pcreg) = \Phi_T'(\pcreg) = ((\rx,\normal),\baddr,\eaddr,\_)$:

    Follows by definition of $\updPcAddr{}$ and the fact that stores from pc into memory and loads into pc are not allowed.

  \item $\npair[n-1]{[\baddr,\eaddr]} \in \xReadCond{W_\pcreg}$:

    Follows by the same assumption of this lemma, using Lemma~\ref{lem:downwards-closed}.

  \item $\npair[n-1]{(\Phi_S'.\reg,\Phi_T'.\reg)} \in \lrrg{\trust}(W_R')$:

    See above.

  \item $\memSat[n-1]{\Phi_S'.\mem,\Phi_S'.\stk,\Phi_S'.\ms_\stk,\Phi_T'.\mem}{W_M'}$:

    See above.

  \item $W_\pcreg \oplus W_R' \oplus W_M'$ is defined.

  \item Theorem~\ref{thm:ftlr} holds for all $n' < n-1$:

    Follows by the same assumption of this lemma.
  \end{itemize}

  Case \ref{case:ftlr:store-load-stack}

  By Lemma~\ref{lem:lro-anti-red-gen} it suffices to show
  \[
    \npair[n-1]{\Phi_S',\Phi_T'} \in \lro
  \]

  First, we show that for some $W_R'$ and $W_M'$ such that $W_R \oplus W_M = W_R' \oplus W_M'$, we have that $\npair[n-1]{(\Phi_S'.\reg,\Phi_T'.\reg)} \in \lrrg{\trust}(W_R')$ and $\memSat[n-1]{\Phi_S'.\mem,\Phi_S'.\stk,\Phi_S'.\ms_\stk,\Phi_T'.\mem}{W_M'}$.
  We have that $\Phi_S' = \updPcAddr{\Phi_S\updReg{r_1,r_2}{w_1,w_2}\update{\ms_\stk.\aaddr}{w}}$, $\Phi_T' = \updPcAddr{\Phi_T\updReg{r_1',r_2'}{w_1',w_2'}\update{\ms_\stk.\aaddr}{w'}}$ and we distinguish the following two cases:
  \begin{itemize}
  \item (store) $w_1 = \Phi_T(r_1) =
        ((\perm,\linear),\baddr,\eaddr,\aaddr)$, $w_1' = \Phi_S(r_1) =
        \stkptr{\perm,\baddr,\eaddr,\aaddr}$, and $\perm \in \writeAllowed{}$, and $\withinBounds{w_1}$, and $w = \Phi_S(r_2)$, and $w' = \Phi_T(r_2)$, and $w_2 = \linCons{\Phi_S(r_2)}$, $w_2' = \linCons{\Phi_T(r_2)}$:

        Follows from Lemma~\ref{lem:store-stack-works}, redistributing ownership in the obvious way.

  \item (load) $ w_2' = \Phi_T(r_2) =
        ((\perm,\linear),\baddr,\eaddr,\aaddr)$, and $w_2 = \Phi_S(r_2) =
        \stkptr{\perm,\baddr,\eaddr,\aaddr}$, and $\perm \in \readAllowed{}$,
        $\withinBounds{((\perm,\lin),\baddr,\eaddr,\aaddr)}$, and $\aaddr \in
        \dom(\Phi.\ms_\stk)$, and $\aaddr \in \dom(\Phi.\ms_\stk)$, and\\
        $w_1 = \Phi_S.\ms_\stk(\aaddr)$, and $w_1' = \Phi_T.\ms_\stk(\aaddr)$, and $w = \linCons{w_1}$, $w' = \linCons{w_1'}$, $\linConsPerm{\perm,w_1}$ and $\linConsPerm{\perm,w_1'}$ and $r_1 \neq \pcreg$:

    Follows from Lemma~\ref{lem:load-stack-cap-works}, redistributing ownership in the obvious way.

  \end{itemize}

  By the induction hypothesis, it suffices to prove that:
  \begin{itemize}
  \item One of the following sets of requirements holds:
    \begin{itemize}
    \item $\trust = \trusted$, $\Phi_S'$ is reasonable up to $n$ steps and $[\baddr,\eaddr] \subseteq \dom(\mscode) = \ta$
    \item $\trust = \untrusted$ and $[b,e] \mathrel{\#} \ta$ and $\npair{[\baddr,\eaddr]} \in \readCond{\normal,W_\pcreg}$
    \end{itemize}:

    Follows by the same assumption of this lemma, in the first case using the definition of execution configuration reasonability, using the fact that $\Phi_S$ does not point to $\scall{\offpc,\offsigma}{r_1}{r_2}$ or $\txjmp{r_1}{r_2}$.

  \item $\Phi_S'(\pcreg) = \Phi_T'(\pcreg) = ((\rx,\normal),\baddr,\eaddr,\_)$:

    Follows by definition of $\updPcAddr{}$ and the fact that stores from pc into memory and loads into pc are not allowed.

  \item $\npair[n-1]{[\baddr,\eaddr]} \in \xReadCond{W_\pcreg}$:

    Follows by the same assumption of this lemma, using Lemma~\ref{lem:downwards-closed}.

  \item $\npair[n-1]{(\Phi_S'.\reg,\Phi_T'.\reg)} \in \lrrg{\trust}(W_R')$:

    See above.

  \item $\memSat[n-1]{\Phi_S'.\mem,\Phi_S'.\stk,\Phi_S'.\ms_\stk,\Phi_T'.\mem}{W_M'}$:

    See above.

  \item $W_\pcreg \oplus W_R' \oplus W_M'$ is defined.

  \item Theorem~\ref{thm:ftlr} holds for all $n' < n-1$:

    Follows by the same assumption of this lemma.
  \end{itemize}

  Case \ref{case:ftlr:jump}

  We have that
  \begin{itemize}
  \item $\Phi_S \step[\gc] \Phi_S'$
  \item $\Phi_T \step \Phi_T'$
  \item $\Phi_S$ does not point to $\scall{\offpc,\offsigma}{r_1}{r_2}$ or $\txjmp{r_1}{r_2}$
  \item One of the following holds
    \begin{enumproof}
    \item(jmp,jnz) $\Phi_S' = \Phi_S\updReg{\pcreg,r_1}{\Phi_S(r_1),w_1}$ and
      $\Phi_T' = \Phi_T\updReg{\pcreg,r_1'}{\Phi_T(r_1),w_1'}$ and
      $\Phi_S(r_1) = \Phi_T(r_1) = ((\perm_1,\lin_1),\baddr_1,\eaddr_1,\aaddr_1)$,
      $\exec{\Phi_S(r_1)}$, $\withinBounds{\Phi_S(r_1)}$,
      $w_1 = \linCons{\Phi_S(r_1)}$ and $w_1' = \linCons{\Phi_T(r_1)}$
    \item(xjmp)
      \begin{itemize}
      \item $\Phi_S(r_1) = \sealed{\sigma,c_1}$ and 
      \item $\Phi_S(r_2) = \sealed{\sigma,c_2}$ and 
      \item $\Phi_T(r_1) = \sealed{\sigma,c_1'}$ and
      \item $\Phi_T(r_2) = \sealed{\sigma,c_2'}$ and
      \item $c_1' \neq \retptrc(\_)$ and 
      \item $c_2' \neq \retptrd(\_)$ and
      \item $\nonExec{\Phi_S(r_2)}$ and 
      \item $\nonExec{\Phi_T(r_2)}$ and
      \item $\Phi_S'' = \Phi_S\updReg{r_1,r_2}{\linCons{c_1},\linCons{c_2}}$ and 
      \item $\Phi_S' = \xjumpResult{c_1}{c_2}{\Phi_S''}$ and
      \item $\Phi_T'' = \Phi_T\updReg{r_1,r_2}{\linCons{c_1'},\linCons{c_2'}}$ and
      \item $\Phi_T' = \xjumpResult{c_1'}{c_2'}{\Phi_T''}$
    \end{itemize}

    \end{enumproof}
  \end{itemize}

  According to Lemma~\ref{lem:lro-anti-red-gen}, it suffices to show:
  \[
    \npair[n-1]{(\Phi_S',\Phi_T')} \in \lro
  \]
  If $n-1 = 0$, then this holds vacuously (by definition of $\lro[\preceq,\gc]$ and $\lro[\preceq,\gc]$), so we assume that $n-1 > 0$.

  In the first case (jmp,jnz), the fact that $\Phi_S$ is reasonable, with $\Phi_S \step[\gc] \Phi_S'$ and $\Phi_S$ does not point to $\scall{\offpc,\offsigma}{r_1}{r_2}$ or $\txjmp{r_1}{r_2}$, gives us that $\perm_1 = \perm$, $\lin_1 = \normal$, $b_1 = b$, $e_1 = e$.
  The fact that $\lin_1 =\normal$ implies that $w_1 =\Phi_S(r_1)$ and $w_1' = \Phi_T(r_1)$.

  The induction hypothesis tells us that it suffices to prove the following:
  \begin{itemize}
  \item One of the following sets of requirements holds:
    \begin{itemize}
    \item $\trust = \trusted$, $\Phi_S'$ is reasonable up to $n-1$ steps and $[\baddr,\eaddr] \subseteq \ta$
    \item $\trust = \untrusted$ and $[b,e] \mathrel{\#} \ta$ and $\npair[n-1]{[\baddr,\eaddr]} \in \readCond{\normal,W_\pcreg}$
    \end{itemize}
    This follows from the corresponding assumption of this lemma, the fact that $\Phi_S \step[\gc] \Phi_S'$ and $\Phi_S$ does not point to $\scall{\offpc,\offsigma}{r_1}{r_2}$ or $\txjmp{r_1}{r_2}$ and \ref{lem:downwards-closed}.
  \item $\Phi_S'(\pcreg) = \Phi_T'(\pcreg) = ((\rx,\normal),\baddr,\eaddr,\_)$:
    See above.
  \item $\npair[n-1]{[\baddr,\eaddr]} \in \xReadCond{W_\pcreg}$:
    Follows by Lemma~\ref{lem:downwards-closed} from the fact that $\npair{[b,e]} \in \xReadCond{W_\pcreg}$.
  \item $\npair[n-1]{(\Phi_S'.\reg,\Phi_T'.\reg)} \in \lrrg{\trust}(W_R)$:

    We have that $\Phi_S'.\reg(r) = \Phi_S.\reg(r)$ for all $r \neq \pcreg$, so the result follows by Lemma~\ref{lem:downwards-closed} from the corresponding assumption of this lemma.

  \item $\memSat[n-1]{\Phi_S'.\mem,\Phi_S'.\stk,\Phi_S'.\ms_\stk,\Phi_T'.\mem}{W_M}$:

    These components of $\Phi_S'$ and $\Phi_T'$ are all unmodified from $\Phi_S$ and $\Phi_T$, so the result follows by Lemma~\ref{lem:downwards-closed} from the corresponding assumption of this lemma.

  \item $W_\pcreg \oplus W_R \oplus W_M$ is defined: by assumption.

  \item Theorem~\ref{thm:ftlr} holds for all $n' < n-1$:
    Follows by the same assumption of this lemma.
  \end{itemize}

  In the second case (xjmp), we have that $\npair{(\Phi_S(r_1),\Phi_T(r_1))} \in \lrvg{\trust}(W_{R,1})$, $\npair{(\Phi_S(r_2),\Phi_T(r_2))} \in \lrvg{\trust}(W_{R,2})$ and $\npair{(\Phi_S'',\Phi_T'')} \in \lrrg{\trust}(\{r_1,r_2\})(W_R')$ for some $W_{R,1}, W_{R,2}, W_R'$ with $W_R = W_{R,1} \oplus W_{R,2} \oplus W_R'$.

  Using the facts that $\Phi_S(r_1) = \sealed{\sigma,c_1}$ and $\Phi_S(r_2) = \sealed{\sigma,c_2}$ and $\Phi_T(r_1) = \sealed{\sigma,c_1'}$ and $\Phi_T(r_2) = \sealed{\sigma,c_2'}$ and $\nonExec{\Phi_S(r_2)}$ and $\nonExec{\Phi_T(r_2)}$, and the above points, Lemma~\ref{lem:sealed-lrv-lrexj} tells us that $\npair[n-1]{(c_1,c_2,c_1',c_2')} \in \lrexj(W_{R,1} \oplus W_{R,2})$. 

  By definition of $\lrexj$, it suffices to prove that
  \begin{itemize}
  \item
    \begin{equation*}
      \npair[n-1]{\left(\array{l}\Phi_S.\reg\update{r_1,r_2}{\linCons{c_1},\linCons{c_2}},\\ \Phi_T.\reg\update{r_1,r_2}{\linCons{c_1'},\linCons{c_2'}}\endarray\right)} \in
      \lrr(\{\rdata\})(W_R')
    \end{equation*}

    It is easy to show that $\npair{\linCons{c_1},\linCons{c_1'}} \in \lrvg{\trust}(\purePart{W_{R,1}})$ and $\npair{\linCons{c_2},\linCons{c_2'}} \in \lrvg{\trust}(\purePart{W_{R,2}})$, by using Lemma~\ref{lem:non-linear-pure} in the non-linear case and the fact that $0$ is always related to itself by definition of $\lrvg{\trust}$ in the linear case.
    We also have that $\purePart{W_R'} = \purePart{W_{R,1}} = \purePart{W_{R,2}}$ by Lemma~\ref{lem:purePart-oplus} and the fact that $W_R = W_{R,1} \oplus W_{R,2} \oplus W_R'$, so it follows that $\npair{\linCons{c_1},\linCons{c_1'}} \in \lrvg{\trust}(\purePart{W_R'})$ and $\npair{\linCons{c_2},\linCons{c_2'}} \in \lrvg{\trust}(\purePart{W_R'})$.
    From this and the fact that $\npair{(\Phi_S'',\Phi_T'')} \in \lrrg{\trust}(\{r_1,r_2\})(W_R')$, it follows easily that
    \begin{equation*}
      \npair[n-1]{\left(\array{l}\Phi_S.\reg\update{r_1,r_2}{\linCons{c_1},\linCons{c_2}},\\ \Phi_T.\reg\update{r_1,r_2}{\linCons{c_1'},\linCons{c_2'}}\endarray\right)} \in
      \lrrg{\trust}(\{\rdata\})(W_R')
    \end{equation*}

    From the fact that $\Phi_S$ is reasonable up to $n$ steps, tells us that $\Phi.\reg(r)$ is reasonable in memory $\Phi.\mem$ and free stack $\Phi.\ms_\stk$ up to $n-1$ steps for all $r \neq \pcreg$.
    Lemma~\ref{lem:trusted-and-reasonable-is-untrusted} then tells us that
    \begin{equation*}
      \npair[n-1]{\left(\array{l}\Phi_S.\reg\update{r_1,r_2}{\linCons{c_1},\linCons{c_2}},\\ \Phi_T.\reg\update{r_1,r_2}{\linCons{c_1'},\linCons{c_2'}}\endarray\right)} \in
      \lrr(\{\rdata\})(W_R')
    \end{equation*}
    using the fact that $n-1 > 0$ (by assumption above), Theorem~\ref{thm:ftlr} holds up to $n-1$ steps (by assumption).

  \item $\memSat[n-1]{\ms_S,\stk,\ms_\stk,\ms_T}{W_M}$
    Follows by Lemma~\ref{lem:downwards-closed} from the corresponding assumption of this lemma.

  \item $W_M \oplus W_R' \oplus W_{R,1} \oplus W_{R,2}$ is defined.

    This follows easily from the facts that $W_R = W_{R,1} \oplus W_{R,2} \oplus W_R'$ (see above) and the assumption that $W_\pcreg \oplus W_R' \oplus W_M'$ is defined.
  \end{itemize}
\end{proof}

\begin{proof}[Proof of Theorem~\ref{thm:ftlr}]
  By complete induction over $n$\footnote{if $n=0$, then we have a contradiction with $\Phi_S\term[i]$ and $\Phi_T\term[i]$ when we get to $\lro$.}.
  Assume
  \begin{itemize}
  \item $\npair{[\baddr,\eaddr]} \in \xReadCond{W}$
  \end{itemize}
  and one of the following sets of requirements holds:
  \begin{enumerate}[label={\roman*)}]
  \item \begin{itemize}
    \item $[\baddr,\eaddr] \subseteq \ta$
    \item ${((\rx,\normal),\baddr,\eaddr,\aaddr)}$ behaves reasonably up to $n$ steps.
    \end{itemize}
  \item
    \begin{itemize}
    \item $[\baddr,\eaddr] \mathrel{\#} \ta$
    \end{itemize}
  \end{enumerate}
and show
\[
\npair{(c,c)} \in \lre(W)
\]
For $c=((\rx,\normal),\baddr,\eaddr,\aaddr)$.\\

Let $n' \leq n$ be given and assume
\begin{enumerate}
\item $\npair[n']{\stpair{\reg}{\reg}} \in \lrr(W_R)$ \label{item:reg-ass}
\item $\memSat[n']{\ms_S,\stk,\ms_\stk,\ms_T}{W_M}$\label{item:mem-ass}
\item $W \oplus W_R \oplus W_M$ is defined\label{item:worlds-def}
\end{enumerate}
Further let
\begin{itemize}
\item $\Phi_S = (\ms_S,\reg_S\update{\pcreg}{c},\stk,\ms_\stk)$
\item $\Phi_T = (\ms_T,\reg_T\update{\pcreg}{c})$
\end{itemize}
and show
\[
  \npair[n']{(\Phi_S,\Phi_T)} \in \lro
\]

By Lemma~\ref{lem:ftlr-internal-lemma}, taking $\trust = \trusted$ iff $[b,e] \subseteq \ta$ and $\trust = \untrusted$ otherwise, it suffices to show that:
\begin{itemize}
\item If $\trust =\trusted$ then $ \Phi_S$ is reasonable up to $n'$ steps.

  We know by assumption that $c$ behaves reasonably up to $n$ steps.

  By definition, it suffices to show that $\reg_S(r)$ is reasonable up to $n$ steps in memory $\ms_S$ and free stack $\ms_\stk$ for $r \nequal \pcreg$ and that $\ms_S$, $\ms_\stk$ and $\stk$ are all disjoint.

  Take an $r \neq \pcreg$ and $\gc = (\ta,\stkb,\gsigrets,\gsigcloss)$.
  By Lemma~\ref{lem:untrusted-source-values-are-reasonable}, it suffices to prove the following:
  \begin{itemize}
  \item $\npair{(w,\_)} \in \lrv(W_w)$: follows from $\npair[n']{\stpair{\reg}{\reg}} \in \lrr(W_R)$.
  \item $\memSat{\ms_S,\stk,\ms_\stk,\_}{W_M}$: by assumption.
  \item $\purePart{W_w} \oplus \purePart{W_M}$ is defined: By Lemma~\ref{lem:purePart-oplus}
  \end{itemize}

\item $\trust = \trusted \vee \npair[n']{[\baddr,\eaddr]} \in \readCond{\lin,W}$:
  If $\trust =\ untrusted$ then we know that $[b,e] \mathrel{\#} \ta$, so that $\npair{[\baddr,\eaddr]} \in \readCond{\lin,W}$ follows by Lemma~\ref{lem:xReadCond-outside-ta-implies-readCond}.

\item $\Phi_S(\pcreg) = \Phi_T(\pcreg) = ((\rx,\lin),\baddr,\eaddr,\aaddr)$:
  We know that $\Phi_S(\pcreg) = \Phi_T(\pcreg) = c = ((\rx,\lin),\baddr,\eaddr,\aaddr)$. 

\item $\npair[n']{[\baddr,\eaddr]} \in \xReadCond{W}$:
  Follows from $\npair{[\baddr,\eaddr]} \in \xReadCond{W}$ using Lemma~\ref{lem:downwards-closed}.

\item $\npair[n']{(\Phi_S.\reg,\Phi_T.\reg)} \in \lrrg{\trust}(W_R)$:
  Follows directly from $\npair[n']{\stpair{\reg}{\reg}} \in \lrr(W_R)$  since $\lrr(W_R) \subseteq \lrrg{\trust}(W_R)$.

\item $\memSat[n']{\Phi_S.\mem,\Phi_S.\ms_\stk,\Phi_S.\stk,\Phi_T.\mem}{W_M}$:
  By assumption.

\item $W \oplus W_R \oplus W_M$ is defined:
  By assumption.

\item Theorem~\ref{thm:ftlr} holds for all $n'' < n'$:
  Follows from our induction hypothesis since $n' \le n$.

\end{itemize}
\end{proof}

\begin{lemma}
  \label{lem:stack-sat-unchanged-future}
  If $\pwfree = \pwfree[W']$ and $W' \future W$ and
  $\memSatFStack{\ms_\stk,\ms_T}{W}$, then $\memSatFStack{\ms_\stk,\ms_T}{W'}$.
\end{lemma}
\begin{proof}
  Follows by inspecting the definition of $\memSatFStack{\ms_\stk,\ms_T}{W}$ and by the monotonicity of memory relations in the world.
\end{proof}

\begin{lemma}
  \label{lem:heap-sat-unchanged-future}
  If $\pwheap[\purePart{W}] = \pwheap[\purePart{W'}]$ and $W_{h}' \future W_h$ and $\npair{(\overline{\sigma},\ms_S,\ms_T)} \in \lrheap(\pwheap[W])(W_h)$, then $\npair{(\overline{\sigma},\ms_S,\ms_T)} \in\lrheap(\pwheap[W'])(W_h')$.
\end{lemma}
\begin{proof}
  Follows by inspecting the definition of $\lrheap$ and $\purePart{}$ and using the monotonicity of memory relations in the world.
\end{proof}

\begin{lemma}
  \label{lem:stdreg-singleton-addr-strat}
  $\stdreg{\{*\},\gc}{v}$ is address-stratified.
\end{lemma}
\begin{proof}
  Trivial.
\end{proof}

\begin{lemma}[Unique return seals]
  \label{lem:unique-ret-seals}
  If
  \begin{itemize}
  \item $\mscode([\aaddr..\aaddr+\calllen-1]) = \overline{\scall{\offpc,\offsigma}{r_1}{r_2}}$
  \item $\mscode([\aaddr'..\aaddr'+\calllen-1]) = \overline{\scall{\offpc',\offsigma'}{r_1'}{r_2'}}$
  \item $\mscode(\aaddr + \offpc) = \seal{\sigma_b,\sigma_e,\sigma_b}$ and $\sigma = \sigma_b + \offsigma$
  \item $\mscode(\aaddr' + \offpc') = \seal{\sigma_b',\sigma_e',\sigma_b'}$ and $\sigma = \sigma_b' + \offsigma'$
  \item $\sigrets,\sigcloss \vdash_{\mathrm{comp-code}} \mscode$
  \end{itemize}
  then
  \[
    \aaddr = \aaddr'
  \]
  and $\offpc = \offpc'$ and $\offsigma = \offsigma'$ and $r_1 = r_1'$ and $r_2 = r_2'$ and $\sigma_b = \sigma_b'$ and $\sigma_e = \sigma_e'$.
\end{lemma}
\begin{proof}
  From $\sigrets,\sigcloss \vdash_{\mathrm{comp-code}} \mscode$,
  we get a $d_\sigma : \dom(\mscode) \rightarrow \powerset{\Seal}$ such that the following holds:
  \begin{itemize}
  \item $\mscode$ has no hidden calls
  \item $\sigrets \mathrel{\#} \sigcloss$
  \item $\sigrets = \biguplus_{a \in \dom(\mscode)} d_\sigma(a)$
  \item $\forall a \in \dom(\mscode) \ldotp \sigrets,d_\sigma(a),\sigcloss \vdash_{\mathrm{comp-code}} \mscode,a$
  \item $\exists a \ldotp \mscode(a) = \seal{\sigma_\baddr,\sigma_\eaddr,\_} \wedge [\sigma_\baddr,\sigma_\eaddr] \neq \emptyset$
  \end{itemize}

  Particularly, we have $\sigrets,d_\sigma(a),\sigcloss \vdash_{\mathrm{comp-code}} \mscode,a$ and $\sigrets,d_\sigma(a'),\sigcloss \vdash_{\mathrm{comp-code}} \mscode,a'$.
  From these, and the facts that $\mscode([\aaddr..\aaddr+\calllen-1]) = \overline{\scall{\offpc,\offsigma}{r_1}{r_2}}$ and $\mscode([\aaddr'..\aaddr'+\calllen-1]) = \overline{\scall{\offpc',\offsigma'}{r_1'}{r_2'}}$, it follows that
  \begin{itemize}
  \item $\mscode(a+\offpc) = \seal{\sigma_\baddr,\sigma_\eaddr,\sigma_\baddr}$ and $\sigma_\baddr+\offsigma \in d_\sigma(a)$
  \item $\mscode(a'+\offpc') = \seal{\sigma_\baddr',\sigma_\eaddr',\sigma_\baddr'}$ and $ \sigma_\baddr'+\offsigma' \in d_\sigma(a')$
  \end{itemize}
  From the facts that $\sigma = \sigma_b + \offsigma = \sigma_b' + \offsigma'$, $\sigma_\baddr+\offsigma \in d_\sigma(a)$, $ \sigma_\baddr'+\offsigma' \in d_\sigma(a')$ and the disjointness of the $d_\sigma(a)$, we get that $a = a'$.
  With this fact, the rest of the proof obligations follow directly from the other assumed equations.
\end{proof}

\begin{lemma}
  \label{lem:lrv-relates-linearity}
  For all $W$, $n$, $w_1$, $w_2$ if
  \begin{itemize}
  \item $\npair{(w_1,w_2)} \in \lrvg{\trust}(W)$
  \end{itemize}
  then
  \[
    \isLinear{w_1} \text{ iff } \isLinear{w_2}
  \]
\end{lemma}
\begin{proof}
  (Trivial) We consider the possible cases for $w_1$ and $w_2$
  \begin{itemize}
  \item Case $w_1,w_2 \in \ints$: By definition of $\isLinear{}$.
  \item Case $w_1 = \sealed{\sigma,\vsc_1}$ and $w_2 = \sealed{\sigma,\vsc_2}$:
    From the assumption $\npair[n]{(w_1,w_2)} \in \lrv(W)$ we get the desired result.
  \item Case $w_1 = \seal{\_,\_,\_}$ and $w_2 = \seal{\_,\_,\_}$: By definition of $\isLinear{}$.
  \item Case $w_1 = \stkptr{\_,\_,\_,\_}$ and $w_2 = ((\_,\linear),\_,\_,\_)$: By definition of $\isLinear{}$, stack pointers are linear.
  \item Case $w_1 = ((\_,\lin),\_,\_,\_)$ and $w_2 = ((\_,\lin),\_,\_,\_)$: By definition of $\isLinear{}$.
  \end{itemize}
  If $\trust = \trusted$, then there are two more cases to consider, but like the above cases, they are trivial.
\end{proof}

\begin{lemma}
  \label{lem:lincons-lrv}
  For all $W' \future W$, $n$, $w_1$, $w_2$
   if
    \begin{itemize}
    \item $\isLinear{w_1}$ or $\isLinear{w_2}$
    \end{itemize} or
    \begin{itemize}
    \item $\npair{(w_1,w_2)} \in \lrv(W)$
    \item $\nonLinear{w_1}$ or $\nonLinear{w_2}$
    \end{itemize}
    then
    \[
      \npair{(\linCons{w_1},\linCons{w_2})} \in \lrv(\purePart{W'})
    \]
\end{lemma}
\begin{proof}
  For the first set of assumptions, we can conclude
  $\isLinear{w_1}$ and $\isLinear{w_2}$ by Lemma~\ref{lem:lrv-relates-linearity} which means that we need to argue
  \[
    \npair{(0,0)} \in \lrv(\purePart{W'})
  \]
  which is trivially true.

  For the second set of assumptions, we can conclude $\nonLinear{w_1}$ and
  $\nonLinear{w_2}$ by Lemma~\ref{lem:lrv-relates-linearity} which means that we
  need to show $\npair{(w_1,w_2)} \in \lrv(\purePart{W'})$ which is true by
  assumption and Lemma~\ref{lem:monotonicity} and \ref{lem:non-linear-pure}.
\end{proof}

\begin{lemma}
  \label{lem:safe-cap-exec-is-normal}
  If
  \begin{itemize}
  \item $c = ((\perm,\lin),\baddr,\eaddr,\aaddr)$
  \item $\perm = \{\rx,\rwx\}$
  \item $\npair{(c,\_)} \in \lrv{\trust}(W)$
  \end{itemize}
  then
  \[
    \lin = \normal
  \]
  
\end{lemma}
\begin{proof}
  Follows from the definition.
\end{proof}

\begin{lemma}
  \label{lem:untrusted-sealed-capability}
  If
  \begin{itemize}
  \item $c_c = ((\rx,\normal),\baddr,\eaddr,\aaddr)$
  \item $\aaddr \in [\baddr,\eaddr]$
  \item $\npair{[\baddr,\eaddr]} \in \execCond{W}$
  \item $\nonExec{c_d}$
  \item $\npair{(c_d,c_d')} \in \lrvg{\untrusted}(W_o)$
  \item $\npair{(\reg_S,\reg_T)} \in \lrrg{\untrusted}(\{\rdata\})(W_R)$
  \item $\memSat{\ms_S,\stk,\ms_\stk,\ms_T}{W_M}$
  \item $\Phi_S = (\ms_S,\reg_S\update{\pcreg}{c_c}\update{\rdata}{c_d},\stk,\ms_\stk)$
  \item $\Phi_T = (\ms_T,\reg_T\update{\pcreg}{c_c'}\update{\rdata}{c_d'})$
  \item $W' \future \purePart{W}$
  \item $W' \oplus W_o \oplus W_R \oplus W_M$ is defined
  \end{itemize}
  then
  \[
    \npair[n-1]{(\Phi_S,\Phi_T)} \in \lro
  \]
\end{lemma}
\begin{proof}
  From Lemma~\ref{lem:monotonicity} and Lemma~\ref{lem:non-linear-pure} we get
  \[
    \npair{[\baddr,\eaddr]} \in \execCond{W}
  \]
  from which we get
  \[
    \npair[n-1]{(c_c,c_c')} \in \lre(W')
  \]
  From $\npair{(\reg_S,\reg_T)} \in \lrrg{\untrusted}(\{\rdata\})(W_R)$, $\npair{(c_d,c_d')} \in \lrvg{\untrusted}(W_o)$, and $W' \oplus W_o \oplus W_R \oplus W_M$ is defined, we conclude
  \[
    \npair{(\reg_S\update{\rdata}{c_d},\reg_T\update{\rdata}{c_d})} \in \lrrg{\untrusted}(W_R \oplus W_o)
  \]
  along with $\memSat{\ms_S,\stk,\ms_\stk,\ms_T}{W_M}$ and Lemma~\ref{lem:non-expansiveness} we can now conclude
  \[
    \npair[n-1]{(\Phi_S,\Phi_T)} \in \lro
  \]
\end{proof}

\section{Notes}
\subsection{Notes on linear capabilities}
It seems reasonable to have enough instructions to let any program be able to make sufficient checks that it can verify that its execution won't fail.
With our current instruction set, a load may fail if a linear capability happens to be located at a memory address that on attempts to load from with a capability without write permission.
To make up for this, one could make an instruction that checks the linearity of a capability in memory without loading it.
It may not be practical to make such an instruction if linearity is kept track as a field on each capability, but it may be tractable if linearity tags are kept track of in a table.

\subsection{Calling convention design decisions}
\subsubsection{Returning the full stack}
When a callee return from a call, they must return all of the stack they were passed. If we omit this requirement, then we cannot guarantee well-bracketedness. The following is an example of circumventing well-bracketedness by keeping part of the stack:

\begin{itemize}
\item An adversary calls our trusted code with a call-back. 
\item Our code uses part of the stack and calls the callback with the rest of the stack.
\item The adversary splits that stack in two and saves the part of the stack adjacent to our stack in some persistent memory. The adversary calls us anew with a callback and the part of the stack they did not save. 
\item We use part of the stack we receive and call the callback with the rest of the stack. The adversary can now use the saved stack to return from the first call to the callback breaking well-bracketedness. The reason this is possible is because we do not check the size of the stack.
\end{itemize}
Figure~\ref{fig:ret-full-stk} illustrates the above example. In the figure, $f$ is the function to ``is'' and $g$ is the callback passed to ``us''. In the end, the part of the stack marked \emph{kept by adv} can be used to return from the first call-back (the call with the *). The stack that is used to return lines up with $\var{Us}_\priv$, so the return is successful.
\begin{figure}
  \centering
  \includegraphics[angle=90,width=\textwidth]{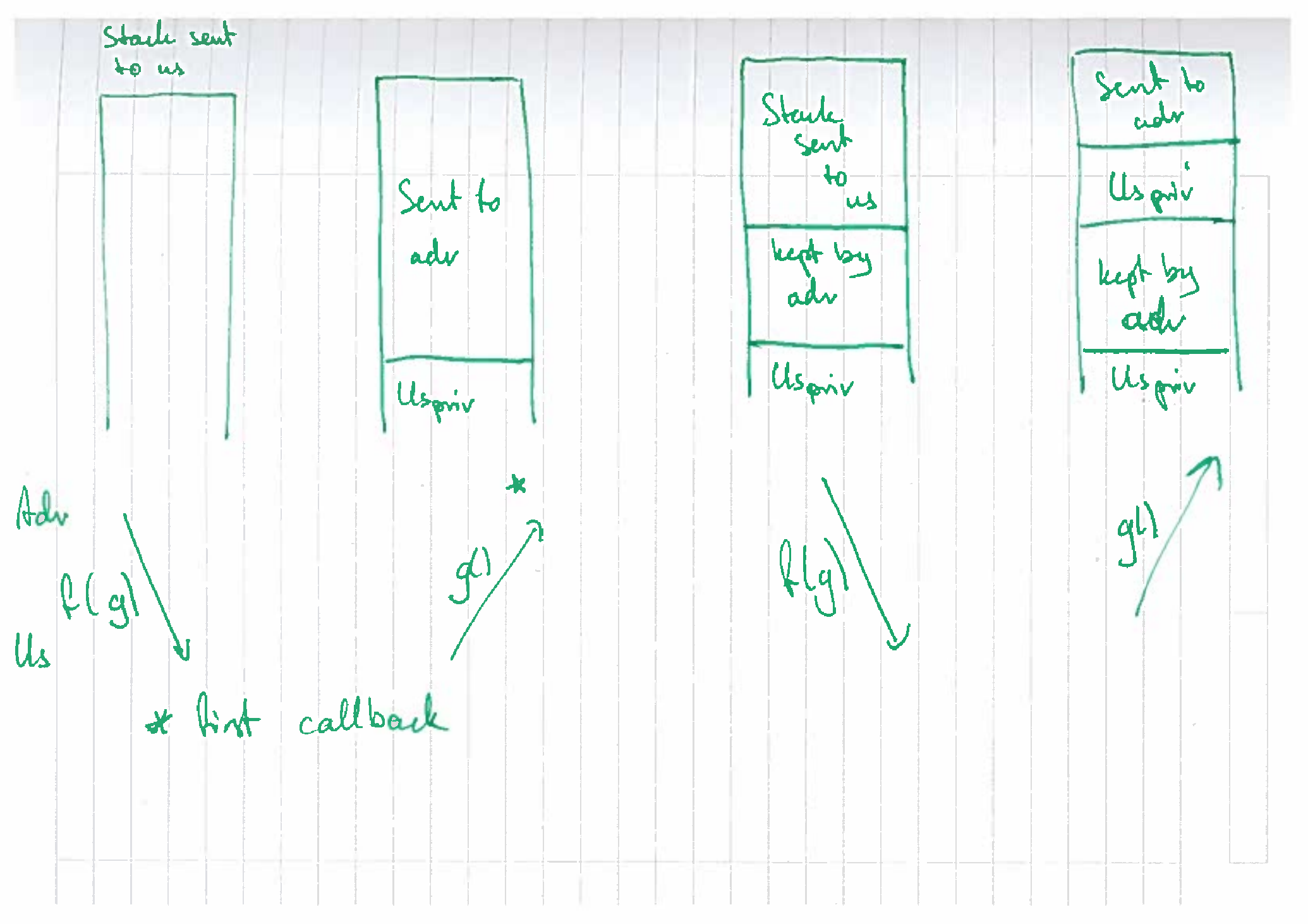}
  \caption{Illustration of the stack in the example that illustrates why the entire stack must be returned.}
  \label{fig:ret-full-stk}
\end{figure}

\subsubsection{Restriction on stack allocation}
We need to somehow make sure that it is always the same stack that is used. If we don't, then an adversary can simply split the stack in two, use one part for one call and the other for another call. At this point, they can return to either of the two calls - in other words, well-bracketedness is not enforced.

\begin{itemize}
  \item An adversary starts the execution. They split the stack in two and call us with one part (say the top part) along with a callback.
  \item We use part of the stack and call the adversary with the rest of the stack.
  \item The adversary calls us again this time using the other part of the stack (here the bottom part of the stack).
  \item Again, we use part of the stack and call the adversary with the rest of this part of the stack.
\end{itemize}
At this point, the adversary can return from either of the two calls. Swapping around the order in which the adversary uses the two parts of the stack changes nothing.

The example is illustrated in Figure~\ref{fig:stk-alloc}.

One way to solve this problem is to make the "top address" of the stack known. There are many ways to do this, but we have chosen to do the following:
\begin{itemize}
\item The stack grows downwards, so the ``last address'' of the stack is the base address of the initial stack capability.
\item The base address of the stack capability is a fixed address, so in the semantics, it will be expressed as a constant that is publicly known.
\item At some point before a call, it must be checked whether the stack we are using actually has the globally known base address (if not we must fail because we cannot trust this stack).
\item To ensure that the check is made, we require it in the semantic condition (at this moment of time not defined, so we are yet to see what it looks like).
\end{itemize}

\begin{figure}
  \centering
  \includegraphics[angle=270,trim={0 0 5cm 0},width=\textwidth]{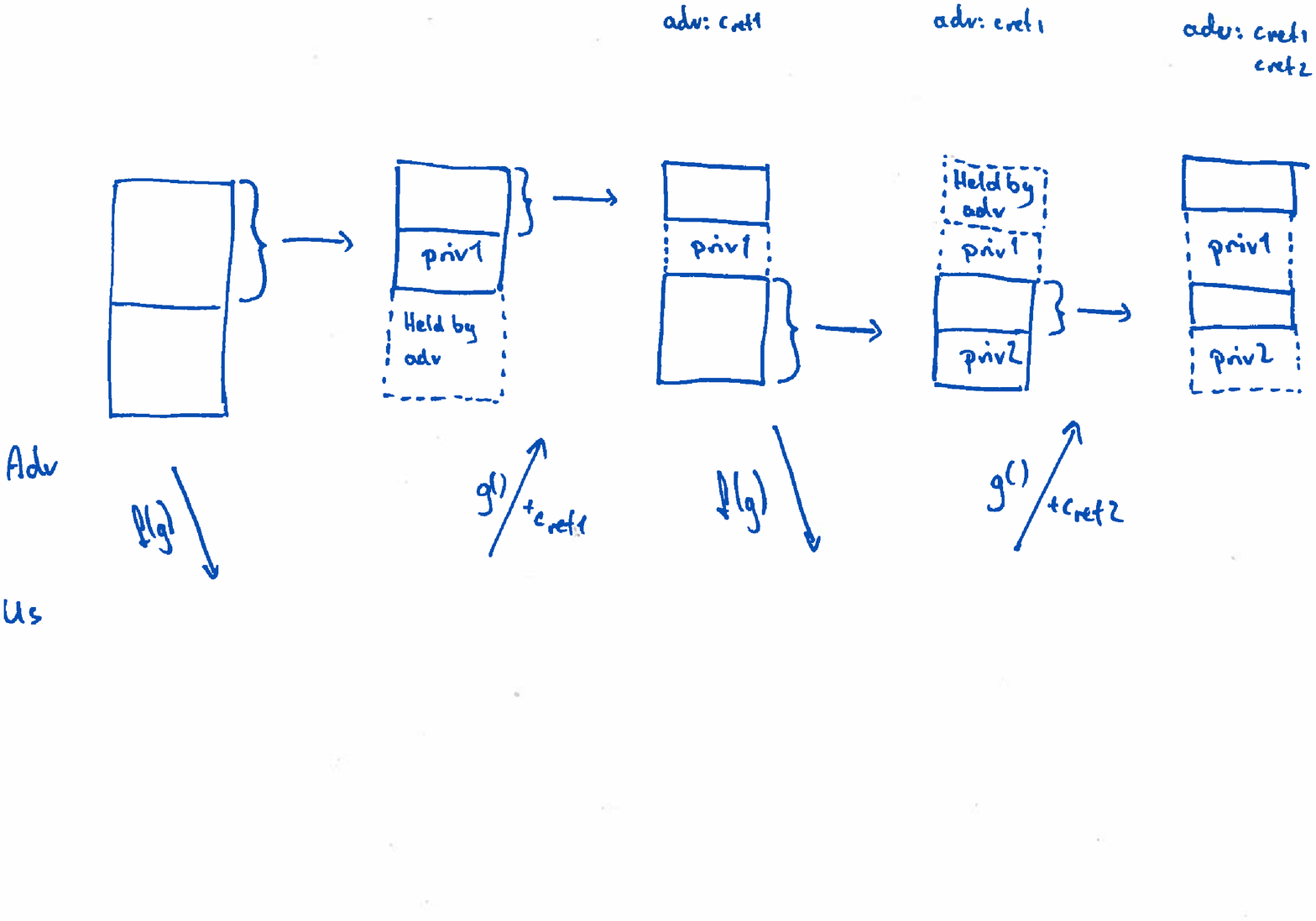}
  \caption{Illustration of the potential stack allocation issue.}
  \label{fig:stk-alloc}
\end{figure}

\section{Related Work}
\subsection{Conditional Full-Abstraction}
The idea of conditional full-abstraction was used by \citet{Juglaret2016} to define full abstraction for unsafe languages. Their definition requires both the programs and the context to be fully defined (i.e.\ not cause undefined behavior). If the programs are not required to be fully-defined, then anything can happen which makes it impossible to reason about.. In our work, undefined behavior marks cases that we do not want to consider because they should be excluded further up in the compilation chain. Further, if we have to take these cases into account, then we need to add checks which protects the trusted code against itself, but properly compiled code should not have to protect itself against itself.
\lau{03-10-2017: We can adjust this when we have actually done something.}

Update: followup paper to the above presented at PriSC 2018, perhaps published elsewhere?
\url{https://popl18.sigplan.org/event/prisc-2018-formally-secure-compilation-of-unsafe-low-level-components}

\bibliography{references}

\begin{thebibliography}{1}
\providecommand{\natexlab}[1]{#1}
\providecommand{\url}[1]{\texttt{#1}}
\expandafter\ifx\csname urlstyle\endcsname\relax
  \providecommand{\doi}[1]{doi: #1}\else
  \providecommand{\doi}{doi: \begingroup \urlstyle{rm}\Url}\fi

\bibitem[Juglaret et~al.(2016)Juglaret, Hritcu, de~Amorim, and
  Pierce]{Juglaret2016}
Yannis Juglaret, Catalin Hritcu, Arthur~Azevedo de~Amorim, and Benjamin~C.
  Pierce.
\newblock Beyond full abstraction: Formalizing the security guarantees of
  low-level compartmentalization.
\newblock \emph{CoRR}, abs/1602.04503, 2016.
\newblock URL \url{http://arxiv.org/abs/1602.04503}.

\end{thebibliography}
\end{document}